\documentclass[a4paper, 10pt, reqno, DIV=calc]{amsart}

\usepackage[utf8]{inputenc}
\usepackage[T1]{fontenc}
\usepackage{lmodern}
\usepackage{microtype}
\usepackage{longtable}
\usepackage{wrapfig}
\usepackage{amsmath, amsthm, amssymb, thmtools}
\usepackage{mathtools}
\usepackage{relsize}
\usepackage{mathrsfs}								
\usepackage{esint}									
\usepackage{enumitem}								
\usepackage{listings}								
\usepackage{stmaryrd}								
\usepackage[mathscr]{eucal}
\newcommand\mathcall[1]{\text{\usefont{U}{BOONDOX-cal}{m}{n}#1}}

\usepackage[dvipsnames]{xcolor}
\usepackage[linktocpage]{hyperref}
\definecolor{colorred}{HTML}{B00000}
\definecolor{colorgreen}{HTML}{258300}
\definecolor{colorblue}{HTML}{2e32fa}
\hypersetup{colorlinks=true, linkcolor=colorred, citecolor=colorgreen, urlcolor=black}								
\usepackage{xifthen}								
\usepackage{url}
\usepackage{booktabs}								
\usepackage{todonotes}
\usepackage{caption}
\usepackage{bbm}									
\usepackage{stmaryrd}								
\usepackage{upgreek}
\makeatletter
\newcommand\MyAutoefPhrasecolorGroup[1]{%
  \color@begingroup\color{MyCurrentcolor}#1\endgroup
}%
\def\HyRef@testreftype#1.#2\\{%
 \colorlet{MyCurrentcolor}{.}%
 \ltx@IfUndefined{#1autorefname}{%
   \ltx@IfUndefined{#1name}{%
     \HyRef@StripStar#1\\*\\\@nil{#1}%
     \ltx@IfUndefined{\HyRef@name autorefname}{%
       \ltx@IfUndefined{\HyRef@name name}{%
         \def\HyRef@currentHtag{}%
         \Hy@Warning{No autoref name for `#1'}%
       }{%
         \edef\HyRef@currentHtag{%
           \noexpand\MyAutoefPhrasecolorGroup{%
             \expandafter\noexpand\csname\HyRef@name name\endcsname
           }%
           \noexpand~%
         }%
       }%
     }{%
       \edef\HyRef@currentHtag{%
         \noexpand\MyAutoefPhrasecolorGroup{%
           \expandafter\noexpand
           \csname\HyRef@name autorefname\endcsname
         }%
         \noexpand~%
       }%
     }%
   }{%
     \edef\HyRef@currentHtag{%
       \noexpand\MyAutoefPhrasecolorGroup{%
         \expandafter\noexpand\csname#1name\endcsname
       }%
       \noexpand~%
     }%
   }%
 }{%
   \edef\HyRef@currentHtag{%
     \noexpand\MyAutoefPhrasecolorGroup{%
       \expandafter\noexpand\csname#1autorefname\endcsname
     }%
     \noexpand~%
   }%
 }%
}%
\makeatother

\numberwithin{equation}{section}

\newcommand{\nlb}{{\ensuremath{\textnormal{b}}}}
\newcommand{\nlc}{{\ensuremath{\textnormal{c}}}}
\newcommand{\nld}{{\ensuremath{\textnormal{d}}}}

\newcommand{\nlC}{{\ensuremath{\textnormal{C}}}}

\newcommand{\rmc}{{\ensuremath{\mathrm{c}}}}
\newcommand{\rmd}{{\ensuremath{\mathrm{d}}}}
\newcommand{\rme}{{\ensuremath{\mathrm{e}}}}

\newcommand{\rms}{{\ensuremath{\mathrm{s}}}}

\newcommand{\rmv}{{\ensuremath{\mathrm{v}}}}

\newcommand{\rmC}{{\ensuremath{\mathrm{C}}}}
\newcommand{\rmD}{{\ensuremath{\mathrm{D}}}}

\newcommand{\rmG}{{\ensuremath{\mathrm{G}}}}
\newcommand{\rmH}{{\ensuremath{\mathrm{H}}}}

\newcommand{\rmP}{{\ensuremath{\mathrm{P}}}}

\newcommand{\sfc}{{\ensuremath{\mathsf{c}}}}
\newcommand{\sfd}{{\ensuremath{\mathsf{d}}}}
\newcommand{\sfe}{{\ensuremath{\mathsf{e}}}}

\newcommand{\sfr}{{\ensuremath{\mathsf{r}}}}

\newcommand{\sfB}{{\ensuremath{\mathsf{B}}}}

\newcommand{\scrA}{{\ensuremath{\mathscr{A}}}}

\newcommand{\scrC}{{\ensuremath{\mathscr{C}}}}
\newcommand{\scrD}{{\ensuremath{\mathscr{D}}}}
\newcommand{\scrE}{{\ensuremath{\mathscr{E}}}}
\newcommand{\scrF}{{\ensuremath{\mathscr{F}}}}

\newcommand{\scrK}{{\ensuremath{\mathscr{K}}}}

\newcommand{\scrP}{{\ensuremath{\mathscr{P}}}}

\newcommand{\scrS}{{\ensuremath{\mathscr{S}}}}
\newcommand{\scrT}{{\ensuremath{\mathscr{T}}}}
\newcommand{\scrU}{{\ensuremath{\mathscr{U}}}}

\newcommand{\scrX}{{\ensuremath{\mathscr{X}}}}

\newcommand{\calL}{{\ensuremath{\mathcall{L}}}}

\newcommand{\bdalpha}{{\ensuremath{\boldsymbol{\alpha}}}}
\newcommand{\bdbeta}{{\ensuremath{\boldsymbol{\beta}}}}

\newcommand{\bdpi}{{\ensuremath{\boldsymbol{\pi}}}}

\newcommand{\bdsigma}{{\ensuremath{\boldsymbol{\sigma}}}}

\newcommand{\N}{\boldsymbol{\mathrm{N}}}						
\newcommand{\Q}{\boldsymbol{\mathrm{Q}}}						
\newcommand{\R}{\boldsymbol{\mathrm{R}}}						
\renewcommand{\d}{\,\mathrm{d}}				

\let\limsup\undefined
\let\liminf\undefined

\DeclareMathOperator*{\limsup}{limsup}		
\DeclareMathOperator*{\liminf}{liminf}		
\DeclareMathOperator{\supp}{spt}			
\DeclareMathOperator{\tr}{tr}				

\let\originalleft\left			
\let\originalright\right
\renewcommand{\left}{\mathopen{}\mathclose\bgroup\originalleft}
\renewcommand{\right}{\aftergroup\egroup\originalright}

\newcommand{\mapdef}[3][]{\ifthenelse{\isempty{#1}}{#2\quad\longmapsto\quad #3}{#1\colon\quad #2\quad\longmapsto\quad #3}}		
\newcommand{\der}[2][]{\ifthenelse{\isempty{#1}}{\frac{\nld}{\nld #2}}{\left.\frac{\nld}{\nld #2}\right\vert_{#1}}}				
\makeatletter
\newcommand{\checknarg}{\@ifnextchar\bgroup{\gobblenarg}{}}
\newcommand{\gobblenarg}[1]{\@ifnextchar\bgroup{,\ \! #1\gobblenarg}{,\ \! #1}}

\theoremstyle{definition}
\newtheorem{bump}{Bump}[section]

\theoremstyle{plain}
\newtheorem{theorem}[bump]{Theorem}
\newtheorem{proposition}[bump]{Proposition}
\newtheorem{definition}[bump]{Definition}
\newtheorem{lemma}[bump]{Lemma}
\newtheorem{corollary}[bump]{Corollary}
\newtheorem{assumption}[bump]{Assumption}

\theoremstyle{remark}
\newtheorem{remark}[bump]{Remark}
\newtheorem{example}[bump]{Example}

\newtheoremstyle{cited}
{\topsep}		
{\topsep}		
{\itshape}		
{}				
{\bfseries}		
{\textbf{.}}	
{.5em}			
{\thmname{#1} \thmnumber{#2} \thmnote{\normalfont#3}}		

\theoremstyle{cited}			

\makeatletter
\let\@fnsymbol\@arabic	 		
\makeatother

\makeatletter
\def\nonumberfootnote{\xdef\@thefnmark{}\@footnotetext}			
\makeatother

\newcommand{\mms}{\mathit{M}}				
\newcommand{\met}{\sfd}						
\newcommand{\meas}{\mathfrak{m}}			
\newcommand{\Leb}{\calL}					
\newcommand{\vol}{\mathrm{vol}}				
\newcommand{\OptGeo}{\mathrm{OptGeo}}		
\newcommand{\OptTGeo}{\mathrm{OptTGeo}}
\newcommand{\Id}{\mathrm{Id}}				

\newcommand{\ac}{{\mathrm{ac}}}

\newcommand{\RCD}{\mathrm{RCD}}				
\newcommand{\CD}{\mathrm{CD}}				
\newcommand{\TCD}{\mathrm{TCD}}
\newcommand{\wTCD}{\mathrm{wTCD}}
\newcommand{\TMCP}{\mathrm{TMCP}}

\newcommand{\bounded}{\nlb}					
\newcommand{\comp}{\nlc}					
\newcommand{\loc}{\mathrm{loc}}				

\newcommand{\pr}{\mathrm{proj}}				
\newcommand{\Ric}{\mathrm{Ric}}				
\newcommand{\Cont}{\nlC}					
\newcommand{\Ell}{\mathit{L}}				
\newcommand{\Lip}{\mathrm{Lip}}				
\newcommand{\Geo}{\mathrm{Geo}}				

\newcommand{\cov}{\textnormal{\texttt{p}}}

\newcommand{\Dom}{\scrD}					

\DeclareMathOperator{\Ent}{Ent}				
\DeclareMathOperator{\Hess}{Hess}			
\newcommand{\eval}{\sfe}					
\newcommand{\Restr}{\mathrm{restr}}			
\newcommand{\push}{\sharp}					

\newcommand{\One}{1}				

\newcommand{\Len}{\mathrm{Len}}

\newcommand{\TGeo}{\mathrm{TGeo}}

\newcommand{\tsep}{\uptau}
\newcommand{\hleq}{\mathbin{\hat{\leq}}}
\newcommand{\hll}{\mathbin{\hat{\ll}}}
\newcommand{\Lef}{{\mathrm{left}}}
\newcommand{\Rig}{{\mathrm{right}}}
\newcommand{\Sp}{\mathrm{sp}}
\newcommand{\mix}{{\mathrm{mix}}}

\usepackage{blindtext}

\allowdisplaybreaks

\providecommand{\bysame}{\leavevmode\hbox to3em{\hrulefill}\thinspace}

\let\oldtocsection=\tocsection

\let\oldtocsubsection=\tocsubsection

\let\oldtocsubsubsection=\tocsubsubsection

\renewcommand{\tocsection}[2]{\hspace{0em}\oldtocsection{#1}{#2}}
\renewcommand{\tocsubsection}[2]{\hspace{1em}\oldtocsubsection{#1}{#2}}
\renewcommand{\tocsubsubsection}[2]{\hspace{2em}\oldtocsubsubsection{#1}{#2}}

\newcommand{\nocontentsline}[3]{}
\newcommand{\tocless}[2]{\bgroup\let\addcontentsline=\nocontentsline#1{#2}\egroup}

\newcommand{\mres}{\mathbin{\vrule height 1.6ex depth 0pt width
0.13ex\vrule height 0.13ex depth 0pt width 1.3ex}}

\begin{document}

\title[Rényi's entropy on Lorentzian spaces]{Rényi's entropy on Lorentzian spaces. Timelike curvature-dimension conditions}
\author{Mathias Braun}
\address{Fields Institute for Research in Mathematical Sciences, 222 College Street, Toronto, Ontario M5T 3J1,
Canada}
\email{braun@math.toronto.edu}
\dedicatory{Dedicated to Professor Karl-Theodor Sturm on the\\ occasion of his sixtieth birthday}
\date{\today}
\thanks{The author acknowledges funding by the Fields Institute for Research in Mathematical
Sciences. He is sincerely grateful to Robert McCann for help- and fruitful discussions. Moreover, he warmly thanks Shin-ichi Ohta for pointing out an error in an earlier version of \autoref{Th:Smothh}. This research is supported in part by Canada Research Chairs Program funds and a Natural Sciences and Engineering Research Council of Canada Grant
(2020-04162) held by Robert McCann. A major part of this work was carried out while the author was a Postdoctoral Fellow at the Department of Mathematics at the University of Toronto from February 2022 to June 2022 under the supervision of Robert McCann.}

\begin{abstract} For a Lorentzian space measured by $\meas$ in the sense of  Kunzinger, Sämann, Cavalletti, and Mondino, we introduce and study syn\-thetic notions of timelike lower Ricci curvature bounds by $K\in\boldsymbol{\mathrm{R}}$ and upper dimension bounds by $N\in[1,\infty)$, namely the time\-like curvature-dimension con\-ditions $\smash{\mathrm{TCD}_p(K,N)}$ and $\smash{\mathrm{TCD}_p^*(K,N)}$ in weak and strong forms,  where $p\in (0,1)$, and  the timelike mea\-sure-contraction properties $\smash{\mathrm{TMCP}(K,N)}$ and $\smash{\mathrm{TMCP}^*(K,N)}$. These are formulated by convexity properties of the Rényi entropy with respect to $\mathfrak{m}$ along $\smash{\ell_p}$-geodesics of pro\-bability measures. 

We show many features of these notions, including their  compatibility with the smooth setting, sharp geometric inequalities, stability, equivalence of the named weak and strong versions, local-to-global properties, and uniqueness of chronological $\smash{\ell_p}$-optimal couplings and chronological $\smash{\ell_p}$-geodesics. We also prove the equivalence of $\smash{\mathrm{TCD}_p^*(K,N)}$ and $\smash{\mathrm{TMCP}^*(K,N)}$ to their respective entropic counterparts in the sense of Cavalletti and Mondino.

Some of these results are obtained under timelike $p$-es\-sential nonbranching, a concept which is a priori weaker than timelike nonbranching.
\end{abstract}

\maketitle
\thispagestyle{empty}

\tableofcontents

\addtocontents{toc}{\protect\setcounter{tocdepth}{1}}

\section{Introduction}\label{Ch:Intro}

Throughout this paper, let $(\mms,\met)$ be a nonempty, proper --- hence complete and separable --- metric space. All topological and measure-theoretic properties will be understood with respect to the topology induced by $\met$ and its induced Borel $\sigma$-algebra. 
Furthermore, let $\meas$ be a given Radon measure on $\mms$ with full support, symbolically $\supp\meas=\mms$.

\subsection*{Background} Mathematical general relativity has recently developed a promising novel research direction.  McCann \cite{mccann2020} pioneered lower boundedness of the Ricci tensor $\Ric$ of a smooth Lorentzian spacetime $\mms$ in  timelike directions to be equivalent to convexity properties of the Boltzmann entropy $\Ent_\meas$, for $\meas$ a possibly weighted Lorentzian volume measure, along ``chronological geodesics'' on the space of probability measures $\scrP(\mms)$ in the entropic sense \cite{erbar2015}. This includes the \emph{strong energy condition} of Hawking and Penrose \cite{hawking1966, hawking1970,penrose1965} which, for solutions to the Einstein equation with vanishing cosmological constant, reduces to timelike nonnegativity of $\Ric$. Mondino and Suhr \cite{mondinosuhr2018} pro\-vided a similar description of the Einstein equations, i.e.~timelike lower and upper Ricci curvature bounds, through synthetic optimal transport means. In a dif\-ferent direction, Kunzinger and Sämann \cite{kunzinger2018} set up the theory of Lorentzian pre-length, length, and geodesic spaces. These are singular analogues of smooth Lorentzian spacetimes, comparable to metric spaces generalizing Riemannian manifolds. (See \cite{sormani2016} for a related approach, and \cite{kste} for its nonsmooth extension.)  These discoveries let Caval\-letti and Mondino \cite{cavalletti2020}  introduce 
\begin{itemize}
\item the entropic timelike curvature-dimension condition $\smash{\TCD_p^e(K,N)}$, and 
\item the entropic timelike measure-contraction property $\smash{\TMCP^e(K,N)}$,
\end{itemize}
for given parameters $p\in (0,1)$, $K\in \R$, and $N\in [1,\infty)$. Both conditions concretize the meaning of a Lorentzian space measured by $\meas$ to admit  timelike Ricci curvature bounded from below by $K$ and dimension bounded from above by $N$. As indicated above, the inherent for\-mu\-lations are phrased by convexity properties of $\Ent_\meas$ along \emph{$\ell_p$-geodesics}, thus make sense a priori in more general frameworks than smooth Lorentzian spacetimes. Here, $\smash{\ell_p}$ is the so-called \emph{$p$-Lorentz--Wasserstein distance} on $\scrP(\mms)$, see \eqref{Eq:LORWAS},  first introduced in \cite{eckstein2017} and studied further in \cite{cavalletti2020,kellsuhr2020,mccann2020,mondinosuhr2018,
suhr2018}. Recently \cite{braun2022}, we have furthermore introduced an ``infinite-dimensional'' Lorentzian analogue akin to \cite{sturm2006a}.

This line of thought to synthetize Ricci curvature bounds is not new. In metric geometry, results for Riemannian manifolds \cite{cordero2001,otto2000,vonrenesse2005} analogous to \cite{mccann2020,mondinosuhr2018} lead Sturm \cite{sturm2006a,sturm2006b}, independently Lott and Villani \cite{lott2009}, to paraphrase lower Ricci curvature and upper dimension bounds for metric measure spaces by convexity of certain entropy functionals along geodesics in $\scrP(\mms)$ with respect to the so-called $2$-Wasserstein distance $W_2$. This was the starting point of an intense activity on so-called $\CD$ or $\RCD$ spaces, vivid to date. Among others, it produced a nonsmooth splitting theorem \cite{gigli2013}, structure theoretical insights \cite{brue2020,mondino2019}, a first and second order Sobolev calculus \cite{ambrosio2014a,ambrosio2014b,ambrosio2015,gigli2015,
gigli2018, savare2014}, and results about rigidity \cite{erbar2020,ketterer2015a, ketterer2015b} and stability \cite{ambrosio2017,giglimondino2015,suzuki2019}.

Synthetic timelike Ricci curvature bounds for Lorentzian spaces  through optimal transport as initiated by the mentioned  works  \cite{cavalletti2020,mccann2020,mondinosuhr2018} might lead to an equally flourishing theory. There are many smooth  results, both classical and recent, that might (or, occasionally, already do) admit generalizations to such spaces,  notably singularity theorems à la Hawking and Penrose \cite{alexander2019, cavalletti2020, graf2020, hawking1973, kunzinger2015, lu2021, woolgar2016} and à la Gannon and Lee \cite{gannon1975, gannon1976, lee1976, schinnerl2021}, as well as splitting theorems \cite{beem1985, eschenburg1988, galloway1989, lu2021b}. See the recent survey article \cite{steinbauer}. Already from the conceptual viewpoint, \cite{cavalletti2020} sheds a novel light to recent efforts to understand Lorentzian spacetimes with low regularity metrics, where curvature bounds are phrased distributionally thus far \cite{geroch1987,graf2020,steinbauer2009}; the setting of \cite{cavalletti2020,kunzinger2018} allows not only the potential metric tensor, but also the base space to be irregular. This extension is desirable since general relativity generally \emph{predicts} the presence of singularities under extreme physical conditions, e.g.~in black holes, causal Fermion systems \cite{finster2016, finster2018}, or causal sets \cite{bombelli1987}, which might not be approximable by smooth structures (even if so, they can be irregular). See also \cite{kunzinger2018,martin2006} and the introduction of \cite{cavalletti2020}. Moreover, by what we learned from  \cite{mondinosuhr2018} such a theory is a further step to understand the Einstein equations \cite{klainerman2015,rendall2002} in more singular, yet physically relevant settings \cite{burtscher2014, griffiths2015, lichnerowicz1955, mars1993, penrose1972, vickers1990, vickers2000}, which also relates to the \emph{censorship conjectures} \cite{christodoulou2009,dafermos2014}. For recently introduced synthetic notions of timelike \emph{upper} Ricci curvature bounds, see \cite{cavalletti2022,mondinosuhr2018}. 

Information-theoretic approaches to spacetime curvature, towards which we will take a new direction, are also based on  physical evidences and connections of gravity and entropies already discovered up to half a decade ago \cite{bardeen, bekenstein, jacobson, verlinde}.

\subsection*{Objective} The goal of this paper is to introduce and study timelike cur\-vature-dimension conditions and timelike measure-contraction properties for measured Lorentzian spaces by optimal transport, different from \cite{cavalletti2020}: instead of the Boltz\-mann entropy, we  use the \emph{Rényi entropy}, cf.~\eqref{Eq:RN} below. Our treatise constitutes the Lo\-rentzian counterpart of the works \cite{bacher2010,sturm2006b} for metric measure spaces, as is  \cite{cavalletti2020} relative to \cite{erbar2015}. Besides timelike geometric inequalities and stability properties, it will e.g.~allow for a pathwise interpretation and a local-to-global property, and it will lift to the universal cover of a measured Lorentzian space. Moreover, we extend and refine some results from \cite{cavalletti2020}. Even in the smooth case, our work seems the first to study the Rényi entropy within general relativity.

\subsection*{Setting} Before making all above indicated objects, notions, and results precise, we outline our framework. See \autoref{Sub:Lorentzian nonsmooth} and \autoref{Sec:OT Lorentzian} for details. 

Let $(\mms,\met,\ll,\leq,\tsep)$ be a Lorentzian pre-length space in the sense of \cite{kunzinger2018}. Here, $\ll$ and $\leq$ are generalized notions of \emph{chronology} and \emph{causality} between points in $\mms$. Moreover, $\tsep\colon \mms^2\to[0,\infty]$ is the so-called \emph{time separation function}. Roughly speaking, $\tsep(x,y)$ measures the maximal proper-time a point $x\in\mms$ in spacetime needs to travel to $y\in\mms$. Compared to $\CD$ spaces, $\tsep$ takes over the role of the metric in the formulation of timelike curvature-dimension conditions and time\-like measure-contraction properties, yet this analogy should be handled carefully. For instance, $\tsep$ satisfies  the \emph{reverse} triangle inequality \eqref{Eq:Reverse tau}, it is symmetric only in pathological cases, and $\tsep(x,y)=0$ if $y$ is not in the chronological future of $x$. 

Although some definitions and results in this article can be extended to more general settings, to streamline the exposition we will assume $(\mms,\met,\ll,\leq,\tsep)$ to be causally closed, $\scrK$-globally hyperbolic, regular, and geodesic, see \autoref{Ass:ASS}. Inter alia, this hypothesis ensures that $\tsep$ is finite and continuous. When endowed with a reference measure $\meas$ as initially fixed, for terminological cleanness, in this introduction we call $(\mms,\met,\meas,\ll,\leq,\tsep)$ \emph{measured Lorentzian space}.

For $p\in(0,1]$, define the total transport cost $\smash{\ell_p\colon\scrP(\mms)^2 \to [0,\infty]\cup\{-\infty\}}$ by
\begin{align}\label{Eq:LORWAS}
\ell_p(\mu,\nu) := \sup \Vert \tsep\Vert_{\Ell^p(\mms^2,\pi)}.
\end{align}
The supremum is taken over all \emph{causal} couplings $\pi\in\Pi(\mu,\nu)$, i.e.~$x\leq y$ for $\pi$-a.e. $(x,y)\in\mms^2$, subject to the usual convention $\sup \emptyset := -\infty$; if such a $\pi$ obeys $x\ll y$ for $\pi$-a.e.~$(x,y)\in\mms^2$, it is called \emph{chronological}. A causal coupling attaining this supremum is called \emph{$\smash{\ell_p}$-optimal}.
 Similarly as for $\tsep$, we interpret the quantity $\smash{\ell_p}$ as measuring the maximal proper-time a spacetime mass distribution $\mu$ needs to be moved to a mass distribution $\nu$ in a future-directed, causal manner. Moreover, $\smash{\ell_p}$ satisfies the reverse triangle inequality \eqref{Eq:Reverse triangle lp}.

The convexity properties making up our main definitions are understood along \emph{timelike proper-time parametrized $\smash{\ell_p}$-geodesics} in $\scrP(\mms)$, axiomatized as follows (see \autoref{Sub:GEO}, \autoref{Sub:Geodesics}, and \autoref{App:B}). A timelike geodesic on $\mms$ is a curve $\eta\in \Cont([0,1]; \mms)$ which is $\met$-Lipschitz continuous, obeys $\eta_s \ll \eta_t$ for every $s,t\in[0,1]$ with $s<t$, and maximizes the $\tsep$-length \eqref{Eq:Length tau def}, symbolically $\eta\in\TGeo(\mms)$. As for maximizing geodesics in metric geometry \cite{burago2001}, every $\eta\in\TGeo(\mms)$ has a proper-time reparametrization $\gamma\in\Cont([0,1];\mms)$, i.e.
\begin{align}\label{Eq:tau s t equ}
\tsep(\gamma_s,\gamma_t) = (t-s)\,\tsep(\gamma_0,\gamma_1)>0
\end{align}
for every $s,t\in[0,1]$ with $s<t$ \cite{kunzinger2018}. The inherent reparametrization procedure in\-duces a map $\sfr\colon \TGeo(\mms) \to \Cont([0,1];\mms)$ which is continuous, cf.~\autoref{Le:Cty reparametrization}. Given any $\mu_0,\mu_1\in\scrP(\mms)$, let $\smash{\OptTGeo_{\ell_p}(\mu_0,\mu_1)}$ be the set of all $\bdpi\in\scrP(\TGeo(\mms))$ such that $\smash{\pi := (\eval_0,\eval_1)_\push\bdpi}$ is a chronological $\smash{\ell_p}$-optimal coupling of $\mu_0$ and $\mu_1$, where $\eval_t\colon \Cont([0,1];\mms)\to \mms$ is the usual evaluation map $\eval_t(\gamma):= \gamma_t$ at $t\in[0,1]$. Our main \autoref{Ass:ASS} below  ensures that every $x,y\in\mms$ with $x\ll y$ can be connected by an element of $\TGeo(\mms)$. Set 
\begin{align*}
\OptTGeo_{\ell_p}^\tsep(\mu_0,\mu_1) := \sfr_\push\OptTGeo_{\ell_p}(\mu_0,\mu_1),
\end{align*}
and observe that all measures belonging to this set are concentrated on
\begin{align*}
\TGeo^\tsep(\mms) := \sfr(\TGeo(\mms))
\end{align*}
whose elements, in turn, obey \eqref{Eq:tau s t equ} for every $s,t\in[0,1]$ with $s< t$. 

\begin{definition}\label{Def:Termin} A weakly continuous curve $(\mu_t)_{t\in[0,1]}$ in $\scrP(\mms)$ is called
\begin{enumerate}[label=\textnormal{\alph*.}]
\item \emph{timelike $\smash{\ell_p}$-geodesic} if there exists a measure $\smash{\bdpi\in\OptTGeo_{\ell_p}(\mu_0,\mu_1)}$ such that $\mu_t = (\eval_t)_\push\bdpi$ for every $t\in[0,1]$, and
\item \emph{timelike proper-time parametrized $\smash{\ell_p}$-geodesic} provided there is a measure $\bdpi\in\smash{\OptTGeo_{\ell_p}^\tsep(\mu_0,\mu_1)}$ such that $\mu_t= (\eval_t)_\push\bdpi$ for every $t\in[0,1]$.
\end{enumerate}
\end{definition}

Our definition of an $\smash{\ell_p}$-geodesic is different from \cite{cavalletti2020,mccann2020}. Indeed, every timelike proper-time parametrized $\smash{\ell_p}$-geodesic $(\mu_t)_{t\in[0,1]}$ satisfies
\begin{align}\label{Eq:FBHAHB}
\ell_p(\mu_s,\mu_t) = (t-s)\,\ell_p(\mu_0,\mu_1)>0
\end{align}
for every $s,t\in[0,1]$ with $s<t$ and is thus an $\smash{\ell_p}$-geodesic as in \cite{cavalletti2020,mccann2020}. On the other hand, timelike $\smash{\ell_p}$-geodesics in our sense are implicit maximizers of some length functional, cf.~\autoref{Pr:Max}. Thus, \autoref{Def:Termin} reflects how geodesics are set up at the base space terminologically and conceptually \cite{kunzinger2018}. In general, it is more restrictive, as \cite{cavalletti2020,mccann2020} do not even assume weak con\-tinuity of the considered paths in $\scrP(\mms)$. (It is remarkable that this regularity is implied solely by \eqref{Eq:FBHAHB} in all relevant smooth cases, cf.~\autoref{Re:Compat}.) Our approach allows $\smash{\OptTGeo_{\ell_p}^\tsep(\mu_0,\mu_1)}$ to admit good compactness properties, cf.~\autoref{Le:Villani lemma for geodesic}. In\-deed, proper-time reparametrizations of timelike geodesics in $\mms$ are \emph{not} Lipschitz continuous \cite{kunzinger2018}, which makes it difficult to find compact subsets of $\Cont([0,1];\mms)$ satisfying \eqref{Eq:tau s t equ}. On the other hand, by \autoref{Ass:ASS} --- more precisely, a property termed \emph{non-total imprisonment} ---  all elements of $\TGeo(\mms)$ passing through a fixed compact set, e.g.~a causal diamond, have uniformly bounded Lipschitz constants. This entails excellent compactness properties of $\TGeo(\mms)$ and $\smash{\OptTGeo_{\ell_p}(\mu_0,\mu_1)}$ which transfer to $\smash{\TGeo^\tsep(\mms)}$ and $\smash{\OptTGeo_{\ell_p}^\tsep(\mu_0,\mu_1)}$ by continuity of $\sfr$ and $\smash{\sfr_\push}$, respectively.

Lastly, the basic object which will be hypothesized to have some convexity along the named curves in $\scrP(\mms)$ is the $N$-Rényi entropy
\begin{align}\label{Eq:RN}
\scrS_N(\mu) := -\int_\mms \rho^{-1/N}\d\mu = -\int_\mms \rho^{1-1/N}\d\meas,
\end{align}
where $\mu = \rho\,\meas + \mu_\perp$ is the Lebesgue decomposition of $\mu\in\scrP(\mms)$ with respect to $\meas$. Its properties are summarized in \autoref{Subsub:Renyi}. It only relates to $\Ent_\meas$ in the limit case $N\to\infty$, see \eqref{Eq:Ent SN limit}.

\subsection*{Contributions} 
Now we describe the content of our work in more detail.

$\blacktriangleright$ Given $p\in (0,1)$, in \autoref{Def:TCD*} and \autoref{Def:TCD} we introduce the sub\-se\-quent synthetic notions of the timelike Ricci curvature of $(\mms,\met,\meas,\ll,\leq,\tsep)$ being bounded from below by $K\in \R$, and its dimension being bounded from above by $N\in [1,\infty)$, respectively:
\begin{itemize}
\item the reduced timelike curvature-dimension condition $\smash{\TCD_p^*(K,N)}$,
\item the weak reduced timelike curvature-dimension condition $\smash{\wTCD_p^*(K,N)}$,
\item the timelike curvature-dimension condition $\smash{\TCD_p(K,N)}$, and
\item the weak timelike curvature-dimension condition $\smash{\wTCD_p(K,N)}$.
\end{itemize} 
These are Lorentzian versions of \cite{bacher2010,sturm2006b}  and are all formulated by convexity properties of the Rényi entropy along timelike proper-time parametrized $\smash{\ell_p}$-geodesics $(\mu_t)_{t\in[0,1]}$ between pairs $\smash{(\mu_0,\mu_1)\in\scrP(\mms)^2}$ of compactly supported, $\meas$-absolutely continuous mass distributions, symbolically $\smash{\mu_0,\mu_1\in\scrP_\comp^\ac(\mms,\meas)}$. The inherent  inequalities de\-fining $\smash{\TCD_p^*(K,N)}$ and $\smash{\wTCD_p^*(K,N)}$ use the distortion coefficients $\smash{\sigma_{K,N}^{(r)}}$, while $\smash{\TCD_p(K,N)}$ and $\smash{\wTCD_p(K,N)}$ employ the quantities $\smash{\tau_{K,N}^{(r)}}$; we refer to \autoref{Def:Dist coeff} below. The strong versions require the relevant estimate for Rényi's entropy to hold for all \emph{time\-like $p$-dualizable} endpoint pairs $(\mu_0,\mu_1)$, meaning that there exists an appropriate chronological $\smash{\ell_p}$-optimal coupling $\pi\in\Pi(\mu_0,\mu_1)$. The weak versions require $(\mu_0,\mu_1)$ to be \emph{strongly timelike $p$-dualizable}, which essentially asks \emph{all} $\smash{\ell_p}$-optimal couplings of $\mu_0$ and $\mu_1$ to be chronological. The latter notions have been introduced in \cite{cavalletti2020}, see \autoref{Def:TL DUAL} for their  precise definition.

In \autoref{Pr:TCD and TCD*}, we establish the general relations
\begin{align}\label{RKB}
\begin{split}
\TCD_p(K,N) &\ \Longrightarrow \ \TCD_p^*(K,N),\\
\wTCD_p(K,N) &\ \Longrightarrow \ \wTCD_p^*(K,N),\\
\TCD_p^*(K,N) &\ \Longrightarrow \ \TCD_p(K^*,N),\\
\wTCD_p^*(K,N) &\ \Longrightarrow \ \wTCD_p(K^*,N)
\end{split}
\end{align}
between these conditions, where $K^* := K(N-1)/N$, and assuming $K>0$ in the latter two cases.

$\blacktriangleright$ Our first major result is the compatibility of the above notions with timelike lower Ricci curvature bounds in the smooth framework, cf.~\autoref{Th:Equiv}. We also cover the weighted case which has recently lead to singularity and splitting theorems for so-called Bakry--Émery spacetimes \cite{woolgar2016}.


\begin{theorem}\label{Th:Smothh} Let $p\in(0,1)$, $K\in\R$, and $N\in[n,\infty)$, where $\smash{n\in\N_{\geq 2}}$. For the Lo\-rentz\-ian space induced by a smooth, $n$-dimensional Lorentzian spacetime $\mms$ and measured by the Lorentzian volume measure, any of $\smash{\TCD_p^*(K,N)}$ and $\smash{\wTCD_p^*(K,N)}$ holds if and only if the Ricci curvature of $\mms$ is bounded from below  by $K$ in all timelike directions.
\end{theorem}

In particular, for $K=0$ this gives a new characterization of the strong energy condition à la Hawking and Penrose, paralleling \cite{mccann2020}, and a new description of one inequality in the vacuum Einstein equations with zero  cosmological constant \cite{mondinosuhr2018}.

An earlier version of \autoref{Th:Smothh} claimed that one can include $\smash{\TCD_p(K,N)}$ and $\smash{\wTCD_p(K,N)}$ in this chain of equivalences. As thankfully pointed out to the author by Shin-ichi Ohta, the standard proof of $\CD(K,N)$ for smooth Riemannian manifolds \cite{ohta,sturm2006b} by decomposing the transport into a concave tangential and a $(K,N-1)$-concave orthogonal part relies on the symmetry of the transport potential, which is the gradient of an appropriate function (see also corresponding comments in \cite[p.~383]{villani2009}). This property is specific to the $2$-Wasserstein distance. In the Lorentzian setting, the  potential is not symmetric for any $p\in (0,1)$. By using a similar, yet slightly different approach, we have been able to close this gap in \cite{braunohta} (and to extend this result to general weighted Finsler spacetimes).


 
$\blacktriangleright$ The conditions $\smash{\wTCD_p(K,N)}$ and $\smash{\wTCD_p^*(K,N)}$ are shown to imply timelike geometric estimates à la
\begin{itemize}
\item Brunn--Minkowski, see \autoref{Pr:Brunn-Minkowski} and \autoref{Propos},
\item Bonnet--Myers, see \autoref{Cor:Bonnet-Myers} and \autoref{Cor:Reduced BM}, as well as
\item Bishop--Gromov, see \autoref{Th:BG} and \autoref{Th:Reduced BG}.
\end{itemize}
Here the advantage of $\smash{\wTCD_p(K,N)}$ is that it yields these estimates in \emph{sharp} form. Indeed, for a given  globally hyperbolic, $n$-dimensional Lipschitz spacetime $\mms$,  $\smash{n\in\N_{\geq 2}}$ --- our running example of a Lorentzian space of low regularity satisfying \autoref{Ass:ASS} below, cf.~\autoref{Ex:Low reg spt} ---, the timelike Bishop--Gromov inequality obtained from $\smash{\wTCD_p(K,N)}$ yields  
\begin{align*}
n = \dim^\tsep\mms \leq N
\end{align*}
for the so-called geometric dimension $\smash{\dim^\tsep\mms}$  introduced in \cite{mccann2021}, see \autoref{Cor:HD}. As inferred in \autoref{Re:HD2}, the same reasoning starting from $\wTCD_p^*(K,N)$ only gives $\smash{n=\dim^\tsep\mms \leq N+1}$, cf.~\eqref{RKB}. Corresponding geometric inequalities have been also obtained from the weak entropic timelike curvature-dimension condition $\smash{\wTCD_p^e(K,N)}$  \cite{cavalletti2020}, but do not yield sharp versions either \emph{a priori} (unless timelike nonbranching is assumed \cite[Sec.~5.3]{cavalletti2020}, which we avoid).

$\blacktriangleright$ Another feature of our timelike curvature-dimension conditions is their weak stability according to a natural Lorentzian notion of convergence introduced in \cite{cavalletti2020} following \cite{giglimondino2015}, concretized in \autoref{Th:Stability TCD}. These are analogues of corresponding stability results for their entropic versions \cite[Thm.~3.12, Thm.~3.14]{cavalletti2020}.

\begin{theorem}\label{Th:Asddd} Let $(\mms_k,\met_k,\meas_k,\ll_k,\leq_k,\tsep_k)_{k\in\N}$ be a sequence of measured Lorentz\-ian spaces obeying $\smash{\TCD_p^*(K,N)}$ for fixed $p\in (0,1)$, $K\in\R$, and $N\in[1,\infty)$ which converges to a measured Lorentzian  space $(\mms_\infty,\met_\infty,\meas_\infty,\ll_\infty,\leq_\infty,\tsep_\infty)$ in the sense of \autoref{Def:Convergence}.  Then the latter satisfies $\smash{\wTCD_p^*(K,N)}$. 

The same conclusion holds true when  replacing $\smash{\TCD_p^*(K,N)}$ and $\smash{\wTCD_p^*(K,N)}$ by $\smash{\TCD_p(K,N)}$ and $\smash{\wTCD_p(K,N)}$, respectively.
\end{theorem}

Hence, limits according to \autoref{Def:Convergence}  of sequences of smooth Lo\-rentz\-ian spacetimes satisfying the hypothesis of \autoref{Th:Smothh} are $\smash{\wTCD_p(K,N)}$.

\autoref{Th:Asddd} is where the a priori distinction between weak and strong versions of our timelike curvature-dimension conditions, as in its entropic counterpart from \cite{cavalletti2020}, comes into play. There are essentially two technical reasons which arise from the proof of \autoref{Th:Stability TCD}. First, given a strongly timelike $p$-dualizable pair in $\smash{\scrP_\comp^\ac(\mms_\infty,\meas_\infty)^2}$, the members $\smash{\mu_{k,0},\mu_{k,1}\in \scrP_\comp^\ac(\mms_k,\meas_k)}$, $k\in\N$, of correspon\-ding recovery sequences between which $\smash{\TCD_p^*(K,N)}$ applies can only be constructed to be timelike $p$-dua\-li\-zable in general. Second, in the passage as $k\to\infty$, our timelike curvature-dimension condition requires the involved limit objects (namely, timelike proper-time parametrized $\smash{\ell_p}$-geodesics and $\smash{\ell_p}$-optimal couplings), obtained by a compactness argument, to be chronological. However, unlike causality, chronology of either objects is not stable under weak limits in general, unless we know that \emph{every} $\smash{\ell_p}$-optimal coupling of the limit marginals is chronological.

$\blacktriangleright$ In \cite{cavalletti2020}, the most powerful consequences of $\smash{\TCD_p^e(K,N)}$ ---  e.g.~uniqueness of $\smash{\ell_p}$-optimal couplings and $\smash{\ell_p}$-geodesics, localization, and singularity theorems --- were obtained under the assumption of \emph{timelike nonbranching}. Roughly speaking, the latter means that elements of $\smash{\TGeo^\tsep(\mms)}$ do not admit forward or backward branching at intermediate times. However, this property might fail e.g.~for space\-times with regularity less than $\smash{\Cont^{1,1}}$ \cite{cavalletti2020}, even if a $\TCD$ condition holds \cite{garcia}. Hence, it may be lost in limit cases, e.g.~in the framework of \autoref{Th:Asddd}.

We introduce the following notion called \emph{timelike $p$-essential non\-branching}. See \autoref{Def:Essentially nonbranching} for the precise definition.

\begin{definition}\label{Def:pENB definition} Given $p\in (0,1]$, we term the space $(\mms,\met,\meas,\ll,\leq,\tsep)$ \emph{timelike $p$-essentially nonbranching} if every element of $\smash{\OptTGeo_{\ell_p}^\tsep(\scrP_\comp^\ac(\mms,\meas)^2)}$ is concentrated on a timelike nonbranching subset of $\smash{\TGeo^\tsep(\mms)}$.
\end{definition}

As stated in \autoref{Re:reversed}, timelike nonbranching implies timelike $p$-essentially nonbranching for every $p\in (0,1]$.

\autoref{Def:pENB definition} is based on the notion of \emph{essential nonbranching} metric measure spaces  \cite{rajala2014}. In \cite{rajala2014}, every $\RCD$ space was proven to be  essentially nonbranching. Since the $\RCD$ condition is stable with respect to measured Gromov--Hausdorff (mGH) convergence  \cite{giglimondino2015, sturm2006a}, this implicitly gives a stability result for the essential  nonbranching condition. For this reason, we also believe \autoref{Def:pENB definition} to be easier to deal with in  convergence questions, provided  a suitable Lorentzian version of infinitesimal Hilbertianity \cite{ambrosio2014a,ambrosio2014b,gigli2015} is found in the future.

Here, we obtain a local-to-global property for the $\smash{\TCD_p^*(K,N)}$ condition under timelike $p$-essential nonbranching, cf.~\autoref{Th:Local to global}. This constitutes a Lorentzian variant of \cite{bacher2010} and is a feature of the reduced timelike curvature-dimension condition; indeed, as detailed in \autoref{Eq:Blurt}, such a property seems subtle for $\smash{\TCD_p(K,N)}$. Furthermore, after \cite[Thm.~3.14]{erbar2015} we expect such a local-to-global property for $\smash{\TCD_p^e(K,N)}$ to hold \emph{a priori} only in a strong form.

$\blacktriangleright$ Another key result is the equivalence proof of all reduced timelike curvature-dimension conditions with all entropic ones from \cite{cavalletti2020} under timelike $p$-essential nonbranching, see \autoref{Th:Equivalence TCD* and TCDe}. Remarkably, there seems no clear way of directly passing from the reduced to the entropic $\TCD$ condition, or vice versa.

\begin{theorem}\label{Zums} Assume $(\mms,\met,\meas,\ll,\leq,\tsep)$ to obey \autoref{Ass:ASS} and to be timelike $p$-essentially nonbranching for some $p\in (0,1)$. Then the following statements are equivalent.
\begin{enumerate}[label=\textnormal{\textcolor{black}{(}\roman*\textcolor{black}{)}}]
\item The condition $\smash{\TCD_p^*(K,N)}$ holds.
\item The condition $\smash{\wTCD_p^*(K,N)}$ holds.
\item The condition $\smash{\TCD_p^e(K,N)}$ holds.
\item The condition $\smash{\wTCD_p^e(K,N)}$ holds.
\item\label{joh} For every timelike $p$-dualizable pair $(\mu_0,\mu_1) = (\rho_0\,\meas,\rho_1\,\meas)\in\scrP^\ac(\mms,\meas)^2$, there exists some $\smash{\bdpi\in\OptTGeo_{\ell_p}^\tsep(\mu_0,\mu_1)}$ such that for every $t\in[0,1]$, the measure $(\eval_t)_\push\bdpi$ is $\meas$-absolutely continuous, and its density $\rho_t$ with respect to $\meas$ obeys, for every $N'\geq N$, the inequality
\begin{align*}
\rho_t(\gamma_t)^{-1/N'}\geq \sigma_{K,N'}^{(1-t)}(\tsep(\gamma_0,\gamma_1))\,\rho_0(\gamma_0)^{-1/N'} + \sigma_{K,N'}^{(t)}(\tsep(\gamma_0,\gamma_1))\,\rho_1(\gamma_1)^{-1/N'}
\end{align*}
for $\bdpi$-a.e.~$\gamma\in\TGeo^\tsep(\mms)$.
\end{enumerate}
\end{theorem}

Apart from the smooth case \cite{mccann2020,mondinosuhr2018}, also the equivalence of $\smash{\TCD_p^e(K,N)}$ and $\smash{\wTCD_p^e(K,N)}$ is new, even in the more restrictive  timelike nonbranching case. In fact, as indicated we establish both inherent equivalences of $\smash{\TCD_p^*(K,N)}$ and $\smash{\wTCD_p^*(K,N)}$ as well as $\smash{\TCD_p^e(K,N)}$ and $\smash{\wTCD_p^e(K,N)}$ only through their respective one-to-one correspondence with \ref{joh}. The latter point constitutes a \emph{pathwise} form of either time\-like curvature-dimension condition. 
Also, \ref{joh} provides density bounds for interpolating timelike proper-time parametrized $\smash{\ell_p}$-geodesics which, under $\smash{\wTCD_p^e(K,N)}$, were obtained in \cite{braun2022} using variational methods. (In \cite{braun2022}, we have not assumed timelike essential nonbranching, though.) Compare this with \autoref{Sub:Good TCD} and \autoref{Sub:Good} below, where we employ similar methods to construct so-called ``good geodesics''.

An analogous pathwise form of $\smash{\TCD_p(K,N)}$ and $\smash{\wTCD_p(K,N)}$ --- however, with no direct link to \cite{cavalletti2020} --- is established in \autoref{Th:Equiv TCD with geo}.

In addition, in \cite[Thm.~5.6]{cavalletti2020} a synthetic Hawking singularity theorem has been obtained under more general assumptions than $\smash{\TCD_p^e(K,N)}$. As a byproduct of \autoref{Zums} and \eqref{RKB}, in the genuine timelike nonbranching case this directly transfers to our setting; in the statement \cite[Thm.~5.6]{cavalletti2020}, one simply has to replace the entropic versions from \cite{cavalletti2020} by any of ours.

$\blacktriangleright$ Lastly, we introduce the synthetic notions of
\begin{itemize}
\item the reduced timelike measure-contraction property $\smash{\TMCP^*(K,N)}$, and
\item the timelike measure-contraction property $\smash{\TMCP(K,N)}$.
\end{itemize}
These are versions of the respective timelike curvature-dimension conditions where convexity of the Rényi entropy is not required along the chronological $\smash{\ell_p}$-optimal transport between two $\meas$-absolutely continuous measures, but rather from an $\meas$-absolutely continuous distribution to a Dirac mass in its chronological future. For metric measure spaces, the corresponding notion was introduced in \cite{ohta2007, sturm2006b}; the entropic Lorentzian version called $\smash{\TMCP^e(K,N)}$ is due to \cite{cavalletti2020}. We note that both properties do not depend on the transport exponent $p\in (0,1)$, cf.~\autoref{Re:Indep transp exp}.

In \autoref{Pr:TMCP to TMCP*} and  \autoref{Pr:TMCP to TCD} we prove the following correspondences, where $K^* := K(N-1)/N$ is defined as for \eqref{RKB}:
\begin{align*}
\TMCP(K,N) &\ \Longrightarrow \ \TMCP^*(K,N),\\
\TMCP^*(K,N) &\ \Longrightarrow \ \TMCP(K^*,N),\\
\wTCD_p^*(K,N) &\ \Longrightarrow\ \TMCP^*(K,N),\\
\wTCD_p(K,N) &\ \Longrightarrow\ \TMCP(K,N),
\end{align*}
again assuming $K>0$ in the latter two cases.  
Furthermore, we extend the above mentioned geometric inequalities, cf.~\autoref{Re:Geom inequ TMCP}, we prove the full stability of both properties under the notion of convergence from \autoref{Def:Convergence}, cf.~\autoref{Th:Stability TMCP}, and we show the equivalence of $\smash{\TMCP^*(K,N)}$ with $\smash{\TMCP^e(K,N)}$ under timelike essential nonbranching, cf.~\autoref{Th:Equivalence TMCP* and TMCPe}. We also give pathwise characterizations of $\smash{\TMCP^*(K,N)}$ and $\smash{\TMCP(K,N)}$, see \autoref{Th:Equivalence TMCP* and TMCPe} and \autoref{Th:Equivalence TMCP}. 

Lastly, in \autoref{Th:Uniqueness couplings} and \autoref{Th:Uniqueness geodesics} we establish  the uniqueness of chronological $\smash{\ell_p}$-optimal couplings as well as  timelike proper-time parametrized $\smash{\ell_p}$-optimal geodesics for initial measures in $\smash{\scrP_\comp^\ac(\mms,\meas)}$, again under timelike essential nonbranching. 
This extends the smooth outcomes from \cite{mccann2020}. Our reasoning follows \cite{cavalletti2020} (itself reliant on \cite{cavalletti2017}); in fact, we extend the corresponding uniqueness results obtained therein under $\smash{\TMCP^e(K,N)}$ from the timelike nonbranching to the timelike essentially nonbranching case, cf.~\autoref{Re:From TNB to TENB}. However, let us stress that even in the timelike nonbranching case, our uniqueness results do \emph{not} follow by combining \autoref{Zums} and  \cite{cavalletti2020}. Rather, \autoref{Th:Uniqueness couplings} and \autoref{Th:Uniqueness geodesics} are \emph{prerequisites} for the proof of \autoref{Zums}, cf.~\autoref{Re:Uniq}.

\subsection*{Organization} In \autoref{Ch:prel}, we collect all necessary background material about measured Lorentzian spaces and optimal transport theory on these. \autoref{Ch:TCD conditions} is devoted to a detailed investigation of our timelike curvature-dimension conditions. \autoref{Ch:TMCP} studies the timelike measure-contraction property. \autoref{App:Smooth} contains the smooth proof of \autoref{Th:Smothh}, and \autoref{App:B} collects some technical results evolving around \autoref{Def:Termin}.

\addtocontents{toc}{\protect\setcounter{tocdepth}{2}}

\section{Preliminaries}\label{Ch:prel}

\subsection{Volume distortion coefficients} We start by introducing the \emph{volume distortion coefficients} that  make up our main definitions, and by listing their basic  properties. We refer to \cite{bacher2010, erbar2015,sturm2006b} for details.

Given $\kappa\in\R$, $t\in[0,1]$, and $\vartheta\geq 0$,  we define the quantities
\begin{align}\label{Eq:DIS COF}
\begin{split}
\mathfrak{s}_\kappa(\vartheta) &:= \begin{cases}
\dfrac{\sin(\!\sqrt{\kappa}\,\vartheta)}{\sqrt{\kappa}} & \text{if }\kappa >0,\\
\vartheta & \text{if } \kappa =0,\\
\dfrac{\sinh(\!\sqrt{-\kappa}\,\vartheta)}{\sqrt{-\kappa}} & \text{otherwise},
\end{cases}\\
\sigma_\kappa^{(t)}(\vartheta) &:= \begin{cases} \infty\vspace*{0.15cm} & \textnormal{if }\kappa\,\vartheta^2 \geq \pi^2,\\
t & \text{if }\kappa\,\vartheta^2 = 0,\\
\dfrac{\mathfrak{s}_\kappa(t\,\vartheta)}{\mathfrak{s}_\kappa(\vartheta)} & \textnormal{otherwise},
\end{cases}
\end{split}
\end{align}

\begin{definition}\label{Def:Dist coeff} For $K\in\R$ and $N\in[1,\infty)$, slightly abusing notation we set
\begin{align*}
\sigma_{K,N}^{(t)}(\vartheta) &:= \sigma_{K/N}^{(t)}(\vartheta),\\
\tau_{K,N}^{(t)}(\vartheta) &:= t^{1/N}\,\sigma_{K,N-1}^{(t)}(\vartheta)^{1-1/N}.
\end{align*}
\end{definition}

Note that for every $t\in(0,1)$ and every $\vartheta > 0$, $\smash{\sigma_{K,N}^{(t)}(\vartheta)}$ is continuous in $(K,N)\in \R\times[1,\infty)$, nondecreasing in $K$, and nonincreasing in $N$ \cite[Rem.~2.2]{bacher2010}. Analogous claims apply to the quantity $\smash{\tau_{K,N}^{(t)}(\vartheta)}$ \cite[p.~138]{sturm2006b}. Furthermore, for every $t\in [0,1]$, every $\vartheta\geq 0$, and every $\kappa\in (-\infty,\pi^2/\vartheta^2)$,
\begin{align}\label{Eq:Distortion coeff property}
\sigma_\kappa^{(t)}(\vartheta) = \sigma_{\kappa\vartheta^2}^{(t)}(1).
\end{align}

\begin{remark}\label{Re:Lower bounds sigma} We recall the following elementary inequality from \cite[Rem.~2.3]{cavalletti2017}: for every $K\in\R$, every $N\in[1,\infty)$, every $t\in[0,1]$, and every $\vartheta\geq 0$,
\begin{align*}
\sigma_{K,N}^{(t)}(\vartheta) \geq t\,\rme^{-(1-t)\vartheta\sqrt{K^-/N}}.
\end{align*}
\end{remark}

Lastly, in view of \autoref{Th:Equivalence TCD* and TCDe} and \autoref{Th:Equivalence TMCP* and TMCPe}, given any $t\in[0,1]$ let us define  $\rmG_t\colon \R^2 \times (-\infty,\pi^2)\to (-\infty,\infty]$ and $\rmH_t\colon \R\times (-\infty,\pi^2)\to (-\infty,\infty]$ by
\begin{align}\label{Eq:GtHt}
\begin{split}
\rmG_t(x,y,\kappa) &:= \log\!\big[\sigma_\kappa^{(1-t)}(1)\,\rme^x + \sigma_\kappa^{(t)}(1)\,\rme^y\big]\\
\rmH_t(x,\kappa) &:= \log\!\big[\sigma_\kappa^{(1-t)}(1)\,\rme^x\big] = \log\sigma_\kappa^{(1-t)}(1) + x.
\end{split}
\end{align}
Then the functions $\rmG_t$ and $\rmH_t$ are jointly convex \cite[Lem.~2.11]{erbar2015}. 

\subsection{Nonsmooth Lorentzian spaces}\label{Sub:Lorentzian nonsmooth} We continue with a concise digression on the theory of nonsmooth Lorentzian (pre-length, length, and geodesic) spaces. We refer to \cite{cavalletti2020, kunzinger2018} for  details, proofs, and examples about the corresponding notions.

\subsubsection{Chronology and causality} Let us fix a preorder $\leq$  and a transitive relation $\ll$, contained in $\leq$, on $\mms$. The triple $(\mms,\ll,\leq)$ is called \emph{causal space} \cite[Def. 2.1]{kunzinger2018}. We say that $x,y\in\mms$ are \emph{timelike} or \emph{causally} related if $x\ll y$ or $x\leq y$, respectively. We write $x<y$ whenever $x\leq y$ yet $x\neq y$. Define the sets $\smash{\mms_\ll^2}$ and $\smash{\mms_\leq^2}$ to consist of all pairs $(x,y)\in\mms^2$ with $x\ll y$  and $x\leq y$, respectively.


\begin{definition} We term the triple $(\mms,\ll,\leq)$ 
\begin{enumerate}[label=\textnormal{\alph*.}]
\item \emph{causally closed} if the relation $\leq$ is closed, i.e.~$\smash{\mms_\leq^2}$ is closed in $\mms^2$, and
\item \emph{locally causally closed} if every point in $\mms$ has some open neighborhood $U\subset\mms$ such that $\smash{\bar{U}^2\cap\mms_\leq^2}$ is relatively  closed in $\smash{\bar{U}^2}$.
\end{enumerate}
\end{definition}

Given any  $A\subset\mms$, we  define  the \emph{chro\-no\-logical future} $I^+(A)\subset\mms$ and the \emph{causal future} $J^+(A)\subset\mms$ of $A$ \cite[Def.~2.3]{kunzinger2018} through
\begin{align*}
I^+(A) &:= \{y\in\mms : x\ll y\textnormal{ for some }x\in A\},\\
J^+(A) &:= \{y\in\mms : x\leq y\textnormal{ for some }x\in A\}.
\end{align*}
Analogously, we define the \emph{chronological past} $I^-(A)$ and the \emph{causal past} $J^-(A)$ of $A$. Employing a slight abuse of notation, given  $x\in\mms$ we write $\smash{I^\pm(x) := I^\pm(\{x\})}$ and $\smash{J^\pm(x) := J^\pm(\{x\})}$. For a Borel measure $\mu$ on $\mms$, we shall write $\smash{I^\pm(\mu) := I^\pm(\supp\mu)}$ and $\smash{J^\pm(\mu)} := \smash{J^\pm(\supp\mu)}$. For all these objects, set $I(A,B) := I^+(A) \cap I^-(B)$, and define $I(x,y)$, $I(\mu,\nu)$, $J(A,B)$, $J(x,y)$, and $J(\mu,\nu)$ in an analogous way. 

\subsubsection{Lorentzian pre-length spaces}  A function $\tsep\colon\mms^2 \to [0,\infty]$ is called \emph{time separation function} \cite[Def.~2.8]{kunzinger2018} if it is lower semicontinuous, and for every $x,y,z\in\mms$ it satisfies the properties
\begin{enumerate}[label=\textnormal{\alph*.}]
\item $\tsep(x,y) = 0$ if $x\not\leq y$,
\item $\tsep(x,y)>0$ if and only if $x\ll y$, and
\item if $x\leq y\leq z$ we have the \emph{reverse triangle inequality}
\begin{align}\label{Eq:Reverse tau}
\tsep(x,z) \geq \tsep(x,y) + \tsep(y,z).
\end{align}
\end{enumerate}
The existence of such a $\tsep$ implies that $\ll$ is an \emph{open} relation \cite[Prop.~2.13]{kunzinger2018}; in particular, the set $\smash{I^\pm(A)}$ is open for every $A\subset\mms$ \cite[Lem.~2.12]{kunzinger2018}.

\begin{definition}\label{Def:LLSSP} A \emph{Lorentzian pre-length space} is a quintuple $(\mms,\met,\ll,\leq,\tsep)$ which consists of a causal space $(\mms,\ll,\leq)$ endowed with a metric $\met$ and a time separation function $\tsep$ as introduced above.
\end{definition}

\subsubsection{Length of curves}\label{Sub:Length curves} Let $(\mms,\met,\ll,\leq,\tsep)$ be a given Lorentzian pre-length space. A \emph{path} designates a  map $\gamma\colon [a,b]\to\mms$, where $a,b\in\R$ with $a<b$; if $\gamma$ is continuous, we call $\gamma$ a \emph{curve}. Such a curve $\gamma$ is \emph{\textnormal{(}future-directed\textnormal{)} timelike} or \emph{\textnormal{(}future-directed\textnormal{)} causal} if it  parametrized by constant $\met$-speed --- in particular $\met$-Lipschitz continuous ---, and $\gamma_s \ll \gamma_t$ or $\gamma_s \leq \gamma_t$ for every $s,t\in[a,b]$ with $s< t$, respectively. 
(Evidently, causal and timelike curves can be defined without requiring Lipschitz continuity. Usually, though, we regard the latter as implicitly given, and we always mention explicitly  when Lipschitz continuity can be dropped.) It is  \emph{null} if it is causal and $\gamma_a \not\ll\gamma_b$. Similar notions evidently make sense for  \emph{past-directed} curves and their \emph{causal character} (that is, their property of being chronological, causal, or null). Unless explicitly stated otherwise, any curve of a specified causal character is assumed future-directed. 

The \emph{length} with respect to $\tsep$ of a path $\gamma\colon [a,b]\to \mms$ is 
\begin{align}\label{Eq:Length tau def}
\Len_\tsep(\gamma) := \inf\{\tsep(\gamma_{t_0},\gamma_{t_1}) + \dots + \tsep(\gamma_{t_{n-1}},\gamma_{t_n})\}, 
\end{align}
where the infimum is taken over all integers $n\in\N$ and all times $t_0,\dots,t_n \in[a,b]$ with $t_0=a$, $t_n=b$, and $t_i < t_{i+1}$ for every $i\in\{0,\dots,n-1\}$ \cite[Def.~2.24]{kunzinger2018}. It is additive \cite[Lem.~2.25]{kunzinger2018}, and reparametrizations do neither change causal characters \cite[Lem.~2.27]{kunzinger2018} nor the $\tsep$-length itself \cite[Lem.~2.28]{kunzinger2018}.

\subsubsection{Geodesics}\label{Sub:GEO} A causal curve $\eta\colon[0,1]\to\mms$ is called \emph{geodesic} (or \emph{maximal}) provided $\Len_\tsep(\eta) = \tsep(\eta_0,\eta_1)$ \cite[Def.~2.33]{kunzinger2018}. The space of all such curves is denoted $\Geo(\mms)$, and its subset of timelike curves is written $\TGeo(\mms)$.

\begin{definition}\label{Def:Geod base space} We call $(\mms,\met,\ll,\leq,\tsep)$ \emph{geodesic} if it is localizable in the sense of \autoref{Sub:Regula} below, and for every $x,y\in\mms$ with $x<y$, there exists a geodesic $\eta\in\Geo(\mms)$ starting at $x$ and ending at $y$.
\end{definition}

\begin{remark}\label{Re:Not mean} Note that a priori, geodesy does \emph{not} mean that points $x,y\in\mms$ with $x\ll y$ can be joined by a timelike geodesic.
\end{remark}

Every \emph{timelike} geodesic has a weak parametrization \cite[Def.~3.31]{kunzinger2018} by $\tsep$-arclength provided  $\tsep$ is continuous as well as $\tsep(x,x) = 0$ for every $x\in\mms$ \cite[Cor.~3.35]{kunzinger2018}. More precisely, given $\eta\in\TGeo(\mms)$ define $\psi_\eta\in \Cont([0,1];[0,1])$ by
\begin{align*}
\psi_\eta(t) := \frac{\Len_\tsep(\eta\big\vert_{[0,t]})}{\tsep(\eta_0,\eta_1)} = \frac{\tsep(\eta_0,\eta_t)}{\tsep(\eta_0,\eta_1)}.
\end{align*}
Consider the \emph{reparametrization map} $\sfr\colon \TGeo(\mms) \to \Cont([0,1];\mms)$ given by
\begin{align*}
(\sfr\circ\eta)_t := \gamma_{\psi_\eta^{-1}(t)};
\end{align*}
see also \autoref{App:B}. Then for every $\eta\in\TGeo(\mms)$, the curve $\gamma := \sfr\circ \eta$ obeys
\begin{align}\label{Eq:Proper time par}
\tsep(\gamma_s,\gamma_t) = (t-s)\,\tsep(\gamma_0,\gamma_1)>0
\end{align}
for every $s,t\in[0,1]$ with $s<t$. Note that elements of 
\begin{align*}
\TGeo^\tsep(\mms) := \sfr(\TGeo(\mms))
\end{align*}
might not be $\met$-Lipschitz continuous any more \cite[p.~424]{kunzinger2018}. On the other hand, as shown in \autoref{Le:Cty reparametrization}, $\sfr$ is con\-tinuous with respect to the uniform topology, which will be important for our axiomatization of $\smash{\ell_p}$-geodesics and the ``transference'' of compactness properties in $\TGeo(\mms)$ to $\smash{\TGeo^\tsep(\mms)}$, see \autoref{Sub:Geodesics} below. In order to not interrupt the exposition, we have deferred corresponding definitions and results to \autoref{App:B}.

\subsubsection{Regularity}\label{Sub:Regula} Geodesy does not exclude the occurrence of null segments of geodesics between points which are in timelike relation, hence does not guarantee the existence of an element of $\TGeo(\mms)$ joining these; cf.~\autoref{Re:Not mean}. Among other things, ruling out these issues is the motivation behind our hypothesis on $(\mms,\met,\ll,\leq,\tsep)$ to be regular, cf.~\autoref{Ass:ASS}.

The hypotheses of the results mentioned in this and the next subsection vary and are not streamlined. However, all statements will eventually hold under our mentioned standing, more restrictive \autoref{Ass:ASS} below.

The following notions of localizability --- akin to the smooth concept of convex neighborhoods \cite[p.~417]{kunzinger2018} --- and regularity \cite[Def.~3.16]{kunzinger2018} are rather  technical. They will not be worked with explicitly in the sequel, yet we incorporate them for completeness. More relevant are their consequences, summarized in \autoref{Re:Conse}.

\begin{definition}\label{Def:Local reg} We call $(\mms,\met,\ll,\leq,\tsep)$ \emph{localizable} if every $x\in\mms$ has an open neighborhood $\Omega_x\subset\mms$ with the following properties.
\begin{enumerate}[label=\textnormal{\alph*\textcolor{black}{.}}]
\item There is $c>0$ such that for every causal curve $\gamma\colon[0,1]\to \mms$ with image contained in $\Omega_x$, $\Len_\met(\gamma) \leq c$.
\item There exists a continuous function $\smash{\omega_x\colon\Omega_x^2\to [0,\infty)}$ such that the quintuple $\smash{(\Omega_x,\met\big\vert_{\Omega_x^2},\ll\!\big\vert_{\Omega_x^2},\leq\!\big\vert_{\Omega_x^2},\omega_x)}$ forms a Lorentzian pre-length space according to \autoref{Def:LLSSP}, and for every $y\in\Omega_x$, 
\begin{align*}
I^\pm(y)\cap \Omega_x \neq \emptyset.
\end{align*}
\item\label{EqSBD} For every $y,z\in\Omega_x$ there exists $\smash{\gamma^{y,z}\in\Geo(\Omega_x)}$ connecting $y$ to $z$ with
\begin{align*}
\Len_\tsep(\gamma^{y,z}) = \omega_x(y,z) \leq \tsep(y,z).
\end{align*}
\end{enumerate}

It is called \emph{regular} if it is localizable and
\begin{enumerate}[label=\textnormal{\alph*.}]
\setcounter{enumi}{3}
\item whenever $\smash{y,z\in\Omega_x}$ satisfy $y\ll z$ then some \textnormal{(}equivalently, every\textnormal{)} curve $\smash{\gamma^{y,z}}$ as in \ref{EqSBD} is timelike and has strictly greater $\tsep$-length than any causal curve from $y$ to $z$ with image contained in $\Omega_x$ which contains a null segment.
\end{enumerate}
\end{definition}

\begin{remark}\label{Re:Conse} Here we list the most important consequences of \autoref{Def:Local reg}.
\begin{itemize}
\item The different notions of the singular adaptations \cite[Def.~2.35, Lem.~2.38]{kunzinger2018} of the well-known strong causality condition \cite[Def.~14.11]{oneill1983}, compare with \autoref{Sec:GHYP} below, coincide under localizability \cite[Thm.~3.26]{kunzinger2018}.
\item Under localizability, strong causality, and local causal closedness, the length functional \eqref{Eq:Length tau def} is upper semicontinuous with respect to uniform convergence of causal, not necessarily Lipschitz continuous curves \cite[Prop.~3.17]{kunzinger2018}.
\item Under regularity, every element of $\Geo(\mms)$ has a causal character, i.e.~the following dichotomy holds \cite[Thm.~3.18]{kunzinger2018}. Every $\gamma\in\Geo(\mms)$  is either timelike or null; in particular,
\begin{align*}
\gamma\in\TGeo(\mms) \ \Longleftrightarrow \ \tsep(\gamma_0,\gamma_1)>0. 
\end{align*}
\item As a consequence of the last two items, geodesics will  good compactness properties under \autoref{Ass:ASS}, cf.~\autoref{Sub:Measura}.
\end{itemize}
\end{remark}

\subsubsection{Global hyperbolicity}\label{Sec:GHYP} Following \cite[Sec.~1.1]{cavalletti2020}, we term $(\mms,\met,\ll,\leq,\tsep)$  \emph{non-to\-tally imprisoning} if for every compact $C\subset \mms$ there exists a constant $c>0$ such that the $\met$-arclength of any causal curve contained in $C$ is no larger than $c$.

\begin{definition}\label{Def:GHYPPPPP} The space $(\mms,\met,\ll,\leq,\tsep)$ is 
\begin{enumerate}[label=\textnormal{\alph*.}]
\item  \emph{globally hyperbolic} if it is non-totally imprisoning and  the causal diamond $J(x,y)$ is compact for every $x,y\in\mms$, and
\item  \emph{$\scrK$-globally hyperbolic} if it is non-totally imprisoning and the causal diamond $J(C_0,C_1)$ is compact for all compact $C_0,C_1\subset\mms$.
\end{enumerate}
\end{definition}

\begin{remark}\label{Re:TRRE} Following \cite[Def.~3.22]{kunzinger2018}, we term the Lorentzian pre-length space $(\mms,\met,\ll,\leq,\tsep)$  \emph{Lorentzian length space} if it is locally causally closed, causally path-connected, localizable, and satisfies
\begin{align*}
\tsep(x,y) = \sup\{\Len_\tsep(\gamma) : \gamma\colon[0,1]\to\mms \textnormal{ causal with }\gamma_0 = x \textnormal{ and }\gamma_1=y\}\cup\{0\}
\end{align*}
for every $x,y\in\mms$. Here causal path-connectedness \cite[Def.~3.1]{kunzinger2018} means that every $x,y\in\mms$ with $x<y$ can be joined by a causal curve, and at least one such curve is timelike as soon as $x\ll y$.

By \autoref{Re:Conse}, a locally causally closed, regular Lorentzian geodesic space is a Lorentzian length space; indeed, any geodesic connecting two timelike related points must be timelike, which yields causal path-connectedness.  In particular, the causal ladder \cite[Thm.~3.26]{kunzinger2018} applies to the former, see below. Conversely, by the nonsmooth  Avez--Seifert  theorem \cite[Thm.~3.30]{kunzinger2018}, in the locally causally closed, globally hyperbolic case, Lorentzian length spaces are geodesic.
\end{remark}

We list some further properties of ($\scrK$-)globally hyperbolic Lorentzian pre-length spaces. If $(\mms,\met,\ll,\leq,\tsep)$ is locally causally closed, globally hyperbolic, and $I^\pm(x) \neq \emptyset$ for every $x\in \mms$ --- e.g.~when $(\mms,\met,\ll,\leq,\tsep)$ is localizable --- then $\scrK$-global hyperbolicity holds, cf.~\cite[Lem.~1.5]{cavalletti2020} and \cite[Thm.~3.7, Cor.~3.8]{min}. On the other hand, every locally causally closed, $\scrK$-globally hyperbolic Lorentzian geodesic space is in fact causally closed \cite[Lem.~1.6]{cavalletti2020}. By \cite[Def.~3.25, Thm.~3.26]{kunzinger2018}, global hyperbolicity implies the nonsmooth analogue of  strong causality  from  \autoref{Re:Conse}. By \autoref{Re:TRRE}, for locally causally closed, regular Lorentzian geodesic spaces,   $\tsep$ is finite and continuous \cite[Thm.~3.28]{kunzinger2018}. In \cite{burtscher2021}, a singular analogue of Geroch's characterization \cite{geroch1970} of global hyperbolicity via Cauchy time functions is proven.

The following is a class of examples fitting \autoref{Ass:ASS} below.

\begin{example}[Spacetimes of low regularity]\label{Ex:Low reg spt} Let $\mms$ be a smooth spacetime endowed with a Lipschitz  continuous, globally hyperbolic metric $\langle\cdot,\cdot\rangle$. Here, global hyperbolicity is a priori intended in the classical sense of \cite[p.~411]{oneill1983}, i.e.~strong causality plus compactness of causal diamonds between points. The Lipschitz condition ensures that $(\mms,\langle\cdot,\cdot\rangle)$ is causally plain \cite[Def.~1.16]{grant2012} by \cite[Cor.~1.17]{grant2012}. Hence, by \cite[Thm.~5.12]{kunzinger2018} it induces a canonical Lorentzian length space which is causally closed \cite[Prop.~3.3]{smnn}. Hence, by the already mentioned causal ladder \cite[Thm.~3.26]{kunzinger2018}, global hyperbolicity in the sense of \autoref{Def:GHYPPPPP} holds (in fact, is equivalent to its spacetime counterpart). Applying \cite[Prop.~3.4]{smnn} then  yields $\scrK$-global hyperbolicity; geodesy follows from \autoref{Re:TRRE}. Lastly, regularity follows from \cite[Thm.~1.1]{grafling}; this is another instance where $\langle\cdot,\cdot\rangle$ being $\smash{\Cont^{0,1}}$ matters.
\end{example}

\subsection{Optimal transport on Lorentzian spaces}\label{Sec:OT Lorentzian} Next, let us briefly review the   theory of optimal transport on the class of spaces introduced above \cite{cavalletti2020}. We refer to \cite{eckstein2017,kellsuhr2020, mccann2020, mondinosuhr2018, suhr2018} for prior developments in the smooth case.

\subsubsection{Basic probabilistic notation} Let $\scrP(\mms)$ denote the set of all Borel probability measures on $\mms$. Let $\scrP_\comp(\mms)$ and $\scrP^\ac(\mms,\meas)$ be its subsets consisting of all compactly supported and $\meas$-absolutely continuous elements, respectively; set $\smash{\scrP_\comp^\ac(\mms,\meas)} := \smash{\scrP_\comp(\mms)\cap\scrP^\ac(\mms,\meas)}$. Given any $\mu\in\scrP(\mms)$, by $\mu_\perp$ we mean the $\meas$-singular part in the corresponding Lebesgue decomposition of $\mu\in\scrP(\mms)$.

For a Borel map $F\colon \mms\to \mms'$ into a metric space $(\mms',\met')$, given any $\mu\in\scrP(\mms)$ the measure $F_\push\mu \in\scrP(\mms')$ designates the usual \emph{push-forward}  of $\mu$ under $F$, defined by the formula $F_\push\mu[B] := \smash{\mu\big[F^{-1}(B)\big]}$ for every Borel set $B\subset\mms'$.

Given  $\mu,\nu\in\scrP(\mms)$, let $\Pi(\mu,\nu)$ denote the set of all \emph{couplings} of $\mu$ and $\nu$, i.e.~all $\pi\in\scrP(\mms^2)$ with $\pi[\, \cdot\times\mms] = \mu$ and $\pi[\mms\times\cdot\, ] = \nu$.

With $\Cont([0,1];\mms)$ denoting the set of all curves $\gamma\colon [0,1]\to\mms$, endowed with the uniform topology, for $t\in [0,1]$ the so-called \emph{evaluation map} $\eval_t\colon \Cont([0,1];\mms) \to \mms$ is defined through $\eval_t(\gamma) := \gamma_t$.

\subsubsection{Weak convergence} For convenience, let us briefly recall basic notions and results about weak convergence of probability measures now.

A sequence $(\mu_n)_{n\in\N}$ in $\scrP(\mms)$ converges \emph{weakly} to $\mu\in\scrP(\mms)$ \cite[p.~7]{billingsley} if for every bounded continuous function $\varphi\colon \mms\to\R$, symbolically $\varphi\in \Cont_\bounded(\mms)$,
\begin{align*}
\lim_{n\to\infty}\int_\mms\varphi\d\mu_n = \int_\mms\varphi\d\mu.
\end{align*}
Recall that weak convergence is induced by a metric \cite[Rem.~5.1.1]{ambrosio2008}.

The following two classical results will be used frequently, e.g.~in the proofs of \autoref{Th:Stability TCD} and \autoref{Th:Stability TMCP}. \autoref{Th:EIN}, often called \emph{Portmanteau's theorem} \cite[Thm.~2.1]{billingsley}, relates weak convergence to ``setwise convergence'' of sequences of probability measures, whereas \autoref{Th:ZWE}, the \emph{Prokhorov theorem} \cite[Thm.~5.1, Thm.~5.2]{billingsley}, arms us with a very useful compactness criterion.

\begin{theorem}\label{Th:EIN} For every sequence $(\mu_n)_{n\in\N}$ in $\scrP(\mms)$ and every $\mu\in\scrP(\mms)$, the following are equivalent.
\begin{enumerate}[label=\textnormal{(\roman*)}]
\item \textnormal{\textbf{Weak convergence.}} The sequence $(\mu_n)_{n\in\N}$ converges weakly to $\mu$.
\item \textnormal{\textbf{Upper semicontinuity on closed sets.}} For every closed $C\subset\mms$, 
\begin{align*}
\limsup_{n\to\infty}\mu_n[C]\leq \mu[C].
\end{align*}
\item \textnormal{\textbf{Lower semicontinuity on open sets.}}For every open $U\subset\mms$, 
\begin{align*}
\mu[U]\leq\liminf_{n\to\infty}\mu_n[U].
\end{align*}
\item \textnormal{\textbf{Continuity sets.}} For every Borel set $E\subset\mms$ with $\mu[\partial E]=0$,
\begin{align*}
\lim_{n\to\infty}\mu_n[E]=\mu[E].
\end{align*}
\end{enumerate}
\end{theorem}

\begin{theorem}\label{Th:ZWE} Assume $\met$ to be complete and separable. Then a set $\scrC\subset\scrP(\mms)$ is relatively compact with respect to the weak topology if and only if it is \emph{tight}, i.e.~for every $\varepsilon > 0$ there exists a compact set $C\subset\mms$ with
\begin{align*}
\sup\{\mu[C^\sfc] : \mu\in \scrC\}\leq \varepsilon.
\end{align*}
\end{theorem}

\subsubsection{Chronological and causal couplings} Let $(\mms,\met,\ll,\leq,\tsep)$ be a Lorentzian pre-length space. We define the (possibly empty) set $\Pi_\ll(\mu,\nu)$ of all \emph{chronological couplings} of two probability measures $\mu,\nu\in\scrP(\mms)$ to consist of all $\pi\in\Pi(\mu,\nu)$ such that $\smash{\pi[\mms_\ll^2]=1}$. The set $\Pi_\leq(\mu,\nu)$ of all \emph{causal couplings} of $\mu$ and $\nu$ is defined analogously. Under causal closedness, clearly $\pi\in\Pi_\leq(\mu,\nu)$ if and only if $\pi\in\Pi(\mu,\nu)$ as well as $\smash{\supp\pi\subset\mms_\leq^2}$; an analogous statement  holds for the locally causally closed case if $\mu,\nu\in\scrP_\comp(\mms)$.

A chronological or causal coupling of $\mu,\nu\in\scrP(\mms)$  intuitively describes a way of transporting an infinitesimal mass portion $\rmd\mu(x)$ to an infinitesimal mass portion $\rmd\nu(y)$ as to guarantee $x\ll y$ or $x\leq y$, respectively.

We call $\mu,\nu\in\scrP(\mms)$ \emph{chronologically related} if $\Pi_\ll(\mu,\nu)\neq\emptyset$.

\subsubsection{The $\smash{\ell_p}$-optimal transport problem} Given an exponent $p\in (0,1]$ and following the expositions in  \cite{cavalletti2020,mccann2020} we will adopt the conventions
\begin{align*}
\sup\emptyset &:= -\infty,\\
(-\infty)^p &:= (-\infty)^{1/p} := -\infty,\\
\infty -\infty &:=-\infty. 
\end{align*}

Let the total transport cost function $\ell_p\colon\scrP(\mms)^2\to [0,\infty]\cup\{-\infty\}$ be given by
\begin{align}\label{Eq:Totalcost}
\begin{split}
\ell_p(\mu,\nu) &:=  \sup\{\Vert \tsep\Vert_{\Ell^p(\mms^2,\pi)} : \pi \in \Pi_\leq(\mu,\nu)\}\\
&\phantom{:}= \sup\{\Vert l\Vert_{\Ell^p(\mms^2,\pi)} : \pi\in\Pi(\mu,\nu)\},
\end{split}
\end{align}
where $l\colon \mms^2 \to [0,\infty]\cup\{-\infty\}$ is defined through
\begin{align*}
l^p(x,y)  := \begin{cases} \tsep^p(x,y) & \textnormal{if }x\leq y,\\
-\infty & \textnormal{otherwise}.
\end{cases}
\end{align*}

\begin{remark}
The sets of maximizers for both suprema defining $\smash{\ell_p(\mu,\nu)}$ coincide, including the case $\smash{\Pi_\leq(\mu,\nu)=\emptyset}$. One advantage of the  formulation via $l$ is that under (local) causal closedness and global hyperbolicity, $l^p$ is upper semicontinuous. In this case, customary optimal transport techniques \cite{ambrosiogigli2013, villani2009} are applicable to study the second problem, which in turn yields corresponding results for the first  \cite[Rem.~2.2]{cavalletti2020}. Note that the preimages $l^{-1}([0,\infty))$ and $l^{-1}((0,\infty))$ encode causality and chronology of points in $\mms^2$, respectively.
\end{remark}

A coupling $\pi\in\Pi(\mu,\nu)$ of $\mu,\nu\in\scrP(\mms)$ is called \emph{$\ell_p$-optimal} if $\pi\in\Pi_\leq(\mu,\nu)$ and 
\begin{align*}
\ell_p(\mu,\nu) = \Vert\tsep\Vert_{\Ell^p(\mms^2,\pi)} = \Vert l \Vert_{\Ell^p(\mms^2,\pi)}.
\end{align*}
For existence of such couplings, for our work it will suffice to keep the following in mind. If $(\mms,\met,\ll,\leq,\tsep)$ is locally causally closed and globally hyperbolic, and if $\mu,\nu\in\scrP_\comp(\mms)$ with $\Pi_\leq(\mu,\nu)\neq\emptyset$, existence of an $\smash{\ell_p}$-optimal coupling $\pi$ of $\mu$ and $\nu$ holds, and $\ell_p(\mu,\nu) <\infty$ \cite[Prop.~2.3]{cavalletti2020}.

A key property of $\smash{\ell_p}$ is the \emph{reverse triangle inequality} \cite[Prop. 2.5]{cavalletti2020} somewhat  reminiscent of \eqref{Eq:Reverse tau} and of $l$ satisfying the reverse triangle inequality for \emph{every}, i.e.~not necessarily causally related, $x,y,z\in\mms$: for \emph{every} $\mu,\nu,\sigma\in\scrP(\mms)$,
\begin{align}\label{Eq:Reverse triangle lp}
\ell_p(\mu,\sigma) \geq \ell_p(\mu,\nu) + \ell_p(\nu,\sigma).
\end{align}

\subsubsection{Timelike $p$-dualizability}\label{Sub:Timelike dual} The concept of so-called \textit{\textnormal{(}strong\textnormal{)} timelike $p$-dualiza\-bi\-lity}, where $p\in(0,1]$, of pairs $(\mu,\nu)\in\scrP(\mms)^2$ originates in \cite{cavalletti2020} and generalizes the notion of \emph{$p$-separation} introduced in \cite[Def.~4.1]{mccann2020}. Pairs satisfying this condition admit good duality properties \cite[Prop.~2.19, Prop.~2.21, Thm.~2.26]{cavalletti2020}, which leads to a characterization of $\smash{\ell_p}$-geodesics (see \autoref{Sub:Geodesics} below) in the smooth case \cite[Thm.~4.3, Thm.~5.8]{mccann2020}, cf.~\autoref{Th:Equiv}.

In view of the following \autoref{Def:TL DUAL}, cf.~\cite[Def.~2.18, Def.~2.27]{cavalletti2020}, we refer to \cite[Def.~2.6]{cavalletti2020} for the definition of \emph{cyclical monotonicity} of a set in $\smash{\mms_\leq^2}$ with respect to $l^p$, see also \cite[Def.~5.1]{villani2009}. It will not be relevant in our work.

Given $a,b\colon\mms\to\R$ we define  $a\oplus b\colon\mms^2\to\R$ by
\begin{align*}
(a\oplus b)(x,y) := a(x) + b(y).
\end{align*}

\begin{definition}\label{Def:TL DUAL} Given any $p\in (0,1]$, a pair $(\mu,\nu)\in\scrP(\mms)$ is called 
\begin{enumerate}[label=\textnormal{\alph*.}]
\item  \emph{timelike $p$-dua\-lizable} by $\pi\in \Pi_\ll(\mu,\nu)$ if  $\pi$ is $\smash{\ell_p}$-optimal, and there exist Borel functions $a,b\colon\mms\to\R$ such that $a\oplus b\in\Ell^1(\mms^2,\mu\otimes\nu)$ as well as $l^p\leq a\oplus b$ on $\supp\mu\times\supp\nu$,
\item  \emph{strongly timelike $p$-dualizable} by $\pi\in\Pi_\ll(\mu,\nu)$ if the pair $(\mu,\nu)$ is timelike $p$-dualizable by $\pi$, and there is an $l^p$-cyclically monotone Borel set $\Gamma\subset \smash{\mms_\ll^2\cap(\supp\mu_0\times\supp\mu_1)}$ such that every given  $\sigma\in \Pi_\leq(\mu_0,\mu_1)$ is $\ell_p$-optimal if and only if $\sigma[\Gamma]=1$ \textnormal{(}in particular, $\sigma$ is chronological\textnormal{)}, and
\item \emph{timelike $p$-dualizable} if $(\mu,\nu)$ is timelike $p$-dualizable by some coupling $\pi\in\Pi_\ll(\mu,\nu)$; analogously for strong timelike $p$-dualizability.
\end{enumerate}

Moreover, any $\pi$ as in the above items is called \emph{timelike $p$-dualizing}. 
\end{definition}

\begin{remark}\label{Re:Strong timelike} If $\mu,\nu\in\scrP_\comp(\mms)$ then the pair $(\mu,\nu)$ is timelike $p$-dualizable if and only if there is an $\ell_p$-optimal coupling $\smash{\pi\in\Pi_\leq(\mu,\nu)}$ which is concentrated on $\smash{\mms_\ll^2}$.

If $\mu,\nu\in\scrP_\comp(\mms)$ on a causally closed, globally hyperbolic Lorentzian geodesic space $(\mms,\met,\ll,\leq,\tsep)$ with $\smash{\supp\mu\times\supp\nu\subset\mms_\ll^2}$, then  $(\mu,\nu)$ is automatically strongly timelike $p$-dualizable, $p\in (0,1]$ \cite[Cor.~2.29]{cavalletti2020}, see also \cite[Lem.~4.4, Thm.~7.1]{mccann2020}. 
\end{remark}

\subsubsection{Geodesics revisited}\label{Sub:Geodesics} The definitions we consider need the notion of an \emph{$\smash{\ell_p}$-geodesic}, $p\in (0,1]$, to be made precise. We propose \autoref{Def:Lorentzian geodesic} as Lorentzian version of Wasserstein geodesics in metric spaces. Precise technical results will be deferred to \autoref{App:B}.

In this subsection, we assume $(\mms,\met,\ll,\leq,\tsep)$ to  be causally closed, $\scrK$-globally hyperbolic, regular, and geodesic, i.e.~\autoref{Ass:ASS} below holds. 

Recall the continuous reparametrization map $\smash{\sfr\colon \TGeo(\mms) \to \TGeo^\tsep(\mms)}$ and the   definition $\TGeo^\tsep(\mms) := \sfr(\TGeo(\mms))$ from \autoref{Sub:GEO}. Under the above hypotheses, $\Geo(\mms)$ is closed, and $\TGeo(\mms)$ is a Borel subset of $\Cont([0,1];\mms)$, see  \autoref{Le:Borel}. Hence, $\TGeo^\tsep(\mms)$ is a Suslin set and therefore universally measurable. 

Given any $\mu_0,\mu_1\in\scrP(\mms)$, we define
\begin{align*}
\OptGeo_{\ell_p}(\mu_0,\mu_1) &:= \{\bdpi\in\scrP(\Geo(\mms)) : (\eval_0,\eval_1)_\push\bdpi\in\Pi_\leq(\mu_0,\mu_1)\\
&\qquad\qquad \textnormal{ is }\ell_p\textnormal{-optimal}\},\\
\OptTGeo_{\ell_p}(\mu_0,\mu_1) &:= \{\bdpi\in\OptGeo_{\ell_p}(\mu_0,\mu_1) : \bdpi[\TGeo(\mms)]=1\}\\ 
&\phantom{:}=\{\bdpi\in\OptGeo_{\ell_p}(\mu_0,\mu_1) : (\eval_0,\eval_1)_\push\bdpi[\mms_\ll^2]=1\},\\
\OptTGeo_{\ell_p}^\tsep(\mu_0,\mu_1) &:= \sfr_\push\OptTGeo_{\ell_p}(\mu_0,\mu_1).
\end{align*}
Note that the second last identity precisely follows from regularity. Lastly, we say that a measure $\bdpi\in\scrP(\Cont([0,1];\mms))$ \emph{represents} a curve $(\mu_t)_{t\in[0,1]}$ in $\scrP(\mms)$ provided $\mu_t=(\eval_t)_\push\bdpi$ for every $t\in[0,1]$.

\begin{definition}\label{Def:Lorentzian geodesic} We term a weakly continuous path $(\mu_t)_{t\in[0,1]}$ in $\scrP(\mms)$
\begin{enumerate}[label=\textnormal{\alph*.}]
\item \emph{causal $\ell_p$-geodesic} if it is represented by some $\smash{\bdpi\in\OptGeo_{\ell_p}(\mu_0,\mu_1)}$,
\item \emph{timelike $\ell_p$-geodesic} if it is represented by some $\smash{\bdpi\in\OptTGeo_{\ell_p}(\mu_0,\mu_1)}$, and
\item\label{Spans} \emph{timelike proper-time parametrized $\ell_p$-geodesic} if it is represented by some $\smash{\bdpi\in\OptTGeo_{\ell_p}^\tsep(\mu_0,\mu_1)}$.
\end{enumerate}

A measure $\bdpi$ as in the last item will be called \emph{timelike $\smash{\ell_p}$-optimal geodesic plan}.
\end{definition}

\begin{remark}\label{Re:Lip} Every causal or timelike $\smash{\ell_p}$-geodesic connecting compactly supported $\mu_0$ and $\mu_1$ is Lipschitz continuous with respect to the $q$-Wasserstein metric $W_q$ induced by $(\mms,\met)$ for every $q\in[1,\infty]$.
\end{remark}

\begin{remark} The above hypotheses, in particular regularity,  will imply the actual existence of (timelike) $\smash{\ell_p}$-geodesics  and their proper-time parametrized counterparts between timelike $p$-dualizable pairs of measures, cf.~\autoref{Le:Villani lemma for geodesic}. 
\end{remark}

Note that $\bdpi[\TGeo^\tsep(\mms)]=1$ for every $\smash{\bdpi\in\OptTGeo_{\ell_p}^\tsep(\mu_0,\mu_1)}$, and therefore $\bdpi$-a.e.~$\gamma\in\Cont([0,1];\mms)$ obeys $\tsep(\gamma_s,\gamma_t) = (t-s)\,\tsep(\gamma_0,\gamma_1)>0$ for every $s,t\in [0,1]$ with $s<t$ by \eqref{Eq:Proper time par}. Since $\smash{(\sfr\circ\gamma)_r = \gamma_r}$ for every $\gamma\in\TGeo(\mms)$ and every $r\in\{0,1\}$ by definition of $\sfr$,  $(\eval_0,\eval_1)_\push\bdpi\in\Pi_\ll(\mu_0,\mu_1)$ is $\smash{\ell_p}$-optimal for $\smash{\bdpi\in \OptTGeo_{\ell_p}^\tsep(\mu_0,\mu_1)}$. A standard argument using the reverse triangle inequality \eqref{Eq:Reverse tau} thus yields
\begin{align*}
\ell_p(\mu_s,\mu_t) = (t-s)\,\ell_p(\mu_0,\mu_1) > 0
\end{align*}
for every timelike proper-time parametrized $\smash{\ell_p}$-geodesic $(\mu_t)_{t\in[0,1]}$ and every $s,t\in[0,1]$ with $s<t$. Hence every such curve $(\mu_t)_{t\in[0,1]}$ is an $\smash{\ell_p}$-geodesic in the sense of \cite[Def.~2.13]{cavalletti2020} and \cite[Def.~1.1]{mccann2020}, but the converse is unclear to us. 


\subsubsection{Rényi entropy}\label{Subsub:Renyi} From now on, we draw our attention to \emph{measured} Lorentzian pre-length spaces, defined as follows.

\begin{definition} We call a sextuple $(\mms,\met,\meas,\ll,\leq,\tsep)$ a \emph{measured Lorentzian pre-length space} if $(\mms,\met,\ll,\leq,\tsep)$ is a Lorentzian pre-length space according to \autoref{Def:LLSSP}, and $\meas$ is a fully supported Radon measure on $\mms$ as in \autoref{Ch:Intro}. Accordingly, \emph{measured Lorentzian length spaces} and \emph{measured Lorentzian geodesic spaces} are defined.
\end{definition}

All over the sequel, we use the convenient abbreviation
\begin{align}\label{Eq:X}
\scrX := (\mms,\met,\meas,\ll,\leq,\tsep).
\end{align}
All Lorentzian notions from \autoref{Sub:Lorentzian nonsmooth} phrased for $\scrX$ will be evidently understood with respect to the Lorentzian structure $(\mms,\met,\ll,\leq,\tsep)$.

Given any $N\in [1,\infty)$,  define the \emph{$N$-Rényi entropy} $\scrS_N\colon\scrP(\mms)\to [-\infty,0]$ with respect to $\meas$ through
\begin{align*}
\scrS_N(\mu) := -\int_\mms \rho^{-1/N}\d\mu = - \int_\mms \rho^{1-1/N}\d\meas
\end{align*}
subject to the decomposition $\mu = \rho\,\meas + \mu_\perp$. This functional will play a central role in our work. Observe  that $\scrS_1(\mu) = -\meas[\supp(\rho\,\meas)]$. More generally, for every $N\in[1,\infty)$, by Jensen's inequality we have $\smash{\scrS_N(\mu) \geq -\meas[\supp\mu]^{1/N}}$ for every $\mu\in\scrP_\comp(\mms)$. Moreover, $\scrS_N$ is jointly weakly lower semicontinuous in the following sense \cite[Thm.~B.33]{lott2009} provided (for notational convenience, cf.~\autoref{Sub:Geodesics}) $\mms$ is compact. Assume two sequences $(\mu_n)_{n\in\N}$ and $(\meas_n)_{n\in\N}$ in $\scrP(\mms)$ to converge weakly to $\mu_\infty\in\scrP(\mms)$ and $\meas_\infty\in\scrP(\mms)$, respectively. Then for the Rényi entropies $\smash{\scrS_N^n}$ with respect to $\meas_n$, $n\in\N_\infty$, we have
\begin{align*}
\scrS_N^\infty(\mu_\infty) \leq \liminf_{n\to\infty}\scrS_N^n(\mu_n).
\end{align*}

\subsubsection{Boltzmann entropy} Another functional on $\scrP(\mms)$ we occasionally consider, cf.~\autoref{Le:Stu}, \autoref{Sub:Equiv TCDs}, and \autoref{Sec:Equiv TMCP's}, and whose relativistic consequences have been studied deeply in \cite{cavalletti2020, mccann2020, mondinosuhr2018}, is the \emph{Boltzmann entropy} $\Ent_\meas\colon \scrP(\mms)\to [-\infty,\infty]$ with respect to $\meas$. It is defined by
\begin{align*}
\Ent_\meas(\mu) := \begin{cases}\displaystyle\int_\mms \rho\log\rho\d\meas & \textnormal{if }\mu = \rho\,\meas\in\scrP^\ac(\mms,\meas),\, (\rho\log\rho)^+\in\Ell^1(\mms,\meas),\\
\infty & \textnormal{otherwise}.
\end{cases}
\end{align*}
Let $\Dom(\Ent_\meas)$ denote the usual finiteness domain of $\Ent_\meas$. 

For every $\mu\in\scrP_\comp(\mms)$, we have $\Ent_\meas(\mu)\geq -\log\meas[\supp\mu]$. Moreover, as for the Rényi entropy $\Ent_\meas$ is weakly lower semicontinuous on $\scrP(C)$ for every closed $C\subset\mms$ with $\meas[C]<\infty$ \cite[Lem.~4.1]{sturm2006a}; in fact, for every $\mu\in\scrP_\comp(\mms)$,
\begin{align}\label{Eq:Ent SN limit}
\Ent_\meas(\mu) = \lim_{N\to\infty}\big[N - N\,\scrS_N(\mu)\big] = \sup\{N-N\,\scrS_N(\mu) : N\in [1,\infty)\}.
\end{align}

\subsection{Timelike essential nonbranching} Finally, we introduce the following Lo\-rentz\-ian version of essential  nonbranching for metric measure spaces \cite{rajala2014}, termed \emph{timelike essential nonbranching}. The ``curvewise nonbranching'' part is adapted from \cite[Def.~1.10]{cavalletti2020}.

Given $s,t\in [0,1]$ with $s<t$, let $\smash{\Restr_s^t\colon \Cont([0,1];\mms)\to \Cont([0,1];\mms)}$ be the map restricting a curve $\gamma\colon[0,1]\to\mms$ to $[s,t]$ and then stretching it to $[0,1]$, i.e.
\begin{align}\label{Eq:Restr def}
\Restr_s^t(\gamma)_r := \gamma_{(1-r)s + rt}.
\end{align}

\begin{definition}\label{Def:Essentially nonbranching} A subset $G\subset\TGeo^\tsep(\mms)$ is \emph{timelike nonbranching} if for every $r\in (0,1)$, $\smash{\Restr_s^t\big\vert_G}$ is injective for every $(s,t)\in \{(0,r),(r,1)\}$.

We call $\scrX$ \emph{timelike $p$-essentially  nonbranching}, $p\in (0,1]$, if every plan $\bdpi$ representing some timelike proper-time parametrized $\smash{\ell_p}$-geodesic between measures in $\smash{\scrP_\comp^\ac(\mms,\meas)}$ is concentrated on a time\-like nonbranching subset of $\TGeo^\tsep(\mms)$.
\end{definition}

\begin{remark}\label{Re:reversed} The set $G:= \TGeo^\tsep(\mms)$ is timelike nonbranching if the Lorentzian space $(\mms,\met,\ll,\leq,\tsep)$ is timelike nonbranching according to \cite[Def.~1.10]{cavalletti2020}. In particular, timelike nonbranching Lorentzian pre-length spaces are timelike $p$-essentially nonbranching for every $p\in (0,1]$. 
\end{remark}

Important results derived later which assume $p$-timelike essential nonbranching are the local-to-global property, \autoref{Th:Local to global}, uniqueness of $\smash{\ell_p}$-optimal couplings and timelike proper-time parametrized $\smash{\ell_p}$-geodesics, \autoref{Th:Uniqueness couplings} and \autoref{Th:Uniqueness geodesics}, and the equivalence of some of our curvature-dimension conditions and measure-contraction properties to their entropic versions set up by Cavalletti and Mondino, \autoref{Th:Equivalence TCD* and TCDe} and \autoref{Th:Equivalence TMCP* and TMCPe}.

The following result will be useful at many occasions. See \cite[Lem.~2.6]{bacher2010}, \cite[Lem.~3.11]{erbar2015}, and \cite[Lem.~2.11]{sturm2006a} for related results for metric measure spaces. It follows from a Lorentzian adaptation of the mixing procedure in \cite[Ch.~4]{rajala2014}. The proof, which we carry out for completeness, is postponed to \autoref{Postponed}.  

\begin{lemma}\label{Le:Mutually singular} Assume timelike $p$-essential nonbranching for some $p\in (0,1)$. Let $\bdpi$ be a timelike $\smash{\ell_p}$-optimal geodesic plan between compactly supported measures which can be decomposed as $\smash{\bdpi = \lambda_1\,\bdpi^1 + \dots +\lambda_n\,\bdpi^n}$ for some $n\in\N$, $\lambda_1,\dots,\lambda_n \in (0,1)$, and timelike $\smash{\ell_p}$-optimal geodesic plans $\smash{\bdpi^1,\dots,\bdpi^n}$. Furthermore, let $t\in (0,1)$ and assume that $\smash{(\eval_r)_\push\bdpi^1,\dots,(\eval_r)_\push\bdpi^n\in \scrP^\ac(\mms,\meas)}$ for every $r\in\{0,t,1\}$. Lastly, assume $\smash{\bdpi^1,\dots,\bdpi^n}$ to be mutually singular. Then $\smash{(\eval_t)_\push\bdpi^1,\dots,(\eval_t)_\push\bdpi^n}$ are mutually singular.
\end{lemma}

\section{Timelike curvature-dimension condition}\label{Ch:TCD conditions}

\begin{assumption}\label{Ass:ASS}
From now on, in the notation of \eqref{Eq:X}, unless stated otherwise we assume $\scrX$ to be a causally closed, globally hyperbolic, regular measured Lorentzian geodesic space. 
\end{assumption}

Recall that $\scrK$-global hyperbolicity holds under this assumption.

\subsection{Definition and basic properties} \autoref{Def:TCD*} is the Lorentzian counterpart of the \emph{reduced} curvature-dimension condition \cite[Def.~2.3]{bacher2010}. On the other hand, \autoref{Def:TCD} is the natural Lorentzian adaptation of the $\CD(K,N)$ condition for metric measure spaces from \cite[Def.~1.3]{sturm2006b}. Both conditions deserve their names partly from their compatibility with smooth Lorentzian case, cf.~\autoref{Th:Equiv}. 

\begin{definition}\label{Def:TCD*} Let $p\in (0,1)$, $K\in\R$, and $N\in[1,\infty)$.
\begin{enumerate}[label=\textnormal{\alph*.}]
\item We say that $\scrX$ satisfies the \emph{reduced timelike curva\-ture-di\-men\-sion condition} $\smash{\TCD_p^*(K,N)}$ if for every timelike $p$-dualizable pair $(\mu_0,\mu_1) = (\rho_0\,\meas,\rho_1\,\meas) \in \smash{\scrP_\comp^\ac(\mms,\meas)^2}$, there is a timelike proper-time parametrized $\smash{\ell_p}$-geodesic $(\mu_t)_{t\in [0,1]}$ connecting $\mu_0$ to $\mu_1$ as well as a timelike $p$-dua\-lizing coupling $\pi\in \smash{\Pi_\ll(\mu_0,\mu_1)}$ such that for every $t\in [0,1]$ and every $N'\geq N$,
\begin{align*}
\scrS_{N'}(\mu_t) &\leq -\int_{\mms^2} \sigma_{K,N'}^{(1-t)}(\tsep(x^0,x^1))\,\rho_0(x^0)^{-1/N'} \d\pi(x^0,x^1) \\
&\qquad\qquad -\int_{\mms^2} \sigma_{K,N'}^{(t)}(\tsep(x^0,x^1))\,\rho_1(x^1)^{-1/N'} \d\pi(x^0,x^1).
\end{align*}
\item If the previous statement holds only for every strongly timelike $p$-dualizable  $\smash{(\mu_0,\mu_1)=(\rho_0\,\meas,\rho_1\,\meas)\in\scrP_\comp^\ac(\mms,\meas)^2}$, we say that $\scrX$ obeys the \emph{weak reduced timelike curvature-dimension condition} $\smash{\wTCD_p^*(K,N)}$.
\end{enumerate}
\end{definition}

\begin{definition}\label{Def:TCD} Let $p\in (0,1)$, $K\in\R$, and $N\in[1,\infty)$. 
\begin{enumerate}[label=\textnormal{\alph*.}]
\item We say that $\scrX$ obeys the \emph{timelike curva\-ture-di\-men\-sion condition} $\smash{\TCD_p(K,N)}$ if for every timelike $p$-dualizable pair $(\mu_0,\mu_1) = (\rho_0\,\meas,\rho_1\,\meas) \in \smash{\scrP_\comp^\ac(\mms,\meas)^2}$, there is a timelike proper-time parametrized $\smash{\ell_p}$-geodesic $(\mu_t)_{t\in [0,1]}$  connecting $\mu_0$ to $\mu_1$ as well as a timelike $p$-dua\-lizing coupling $\pi\in \smash{\Pi_\ll(\mu_0,\mu_1)}$ such that for every $t\in [0,1]$ and every $N'\geq N$,
\begin{align*}
\scrS_{N'}(\mu_t) &\leq -\int_{\mms^2} \tau_{K,N'}^{(1-t)}(\tsep(x^0,x^1))\,\rho_0(x^0)^{-1/N'} \d\pi(x^0,x^1) \\
&\qquad\qquad -\int_{\mms^2} \tau_{K,N'}^{(t)}(\tsep(x^0,x^1))\,\rho_1(x^1)^{-1/N'} \d\pi(x^0,x^1).
\end{align*}
\item If the previous statement holds only for every strongly timelike $p$-dualizable $(\mu_0,\mu_1)=(\rho_0\,\meas,\rho_1\,\meas)\in\scrP_\comp^\ac(\mms,\meas)^2$, we say that $\scrX$ obeys the \emph{weak timelike curvature-dimension condition} $\smash{\wTCD_p(K,N)}$.
\end{enumerate}
\end{definition}

\begin{remark} The conditions from  \autoref{Def:TCD*}  and \autoref{Def:TCD} are a priori formulated for compactly supported endpoints $\mu_0$ and $\mu_1$ only. In the timelike $p$-essentially nonbranching case, cf.~\autoref{Def:Essentially nonbranching}, this does in fact imply the respective defining statements for more general $\smash{\mu_0,\mu_1\in \scrP^\ac(\mms,\meas)}$ according to \autoref{Th:Equiv TCD with geo} and \autoref{Th:Equivalence TCD* and TCDe}.
\end{remark}

\begin{remark}\label{Re:TCD 0} Note that $\smash{\TCD_p^*(0,N)}$ and $\smash{\wTCD_p^*(0,N)}$ coincide with $\smash{\TCD_p(0,N)}$ and $\smash{\wTCD_p(0,N)}$, respectively, for every $p\in (0,1)$ and every $N\in[1,\infty)$, since
\begin{align*}
\sigma_{0,N'}^{(r)}(\vartheta) = \tau_{0,N'}^{(r)}(\vartheta) = r
\end{align*}
for every $r\in [0,1]$, every $N'\in [1,\infty)$, and every $\vartheta\geq 0$.
\end{remark}

In the remainder of this section, we examine elementary properties of the  conditions introduced above, and their respective relations. \autoref{Pr:TCD and TCD*} as well as  \autoref{Pr:Consistency TCD} hold on any measured Lorentzian pre-length space.

\begin{proposition}\label{Pr:TCD and TCD*} The following  hold for every $p\in (0,1)$, $K\in\R$, and $N\in[1,\infty)$; set $K^* := K(N-1)/N$.
\begin{enumerate}[label=\textnormal{(\roman*)}]
\item The condition $\smash{\TCD_p(K,N)}$ implies $\smash{\TCD_p^*(K,N)}$.
\item The condition $\smash{\wTCD_p(K,N)}$ implies $\smash{\wTCD_p^*(K,N)}$.
\item If $K>0$, the condition $\smash{\TCD_p^*(K,N)}$ implies $\smash{\TCD_p(K^*,N)}$.
\item If $K>0$, the condition $\smash{\wTCD_p^*(K,N)}$ implies $\smash{\wTCD_p(K^*,N)}$.
\end{enumerate}
\end{proposition}

\begin{proof} The first two statements are a direct consequence of the inequality
\begin{align*}
\tau_{K,N'}^{(r)}(\vartheta) \geq \sigma_{K,N'}^{(r)}(\vartheta)
\end{align*}
for every $r\in[0,1]$, every $N'\geq N$, and every $\vartheta\geq 0$, cf.~the proof of \cite[Prop.~2.5]{bacher2010}. The last two statements  follow similarly by noting that
\begin{align*}
\tau_{K^*,N'}^{(t)}(\vartheta)  \leq \sigma_{K,N'}^{(t)}(\vartheta).\tag*{\qedhere}
\end{align*}
\end{proof}

\begin{proposition}\label{Pr:Consistency TCD} Assume $\smash{\TCD_p^*(K,N)}$ for $p\in (0,1)$, $K\in\R$, and $N\in[1,\infty)$. Then the following statements hold.
\begin{enumerate}[label=\textnormal{\textcolor{black}{(}\roman*\textcolor{black}{)}}]
\item\label{La:DieEINS} \textnormal{\textbf{Consistency.}} The condition $\smash{\TCD_p^*(K',N')}$ holds for every $K'\leq K$ and every $N'\geq N$.
\item\label{La:DieZWEI} \textnormal{\textbf{Scaling.}} Given any $a,b,\theta>0$, the rescaled measured Lorentzian pre-length space $(\mms, a\met, b\meas,\ll,\leq, \theta\tsep)$ satisfies $\smash{\TCD_p^*(K/\theta^2,N)}$.
\end{enumerate}

Analogous statements are obeyed by the conditions $\smash{\wTCD_p^*(K,N)}$, $\smash{\TCD_p(K,N)}$, and $\smash{\wTCD_p(K,N)}$. 
\end{proposition}

\begin{proof} Concerning  \ref{La:DieEINS}, for all four conditions consistency in $K$ follows by nondecreasingness of $\smash{\sigma_{K,N}^{(r)}(\vartheta)}$ and $\smash{\tau_{K,N}^{{(r)}}(\vartheta)}$ in $K\in\R$ for fixed $r\in[0,1]$, $N\in [1,\infty)$, and $\vartheta\geq 0$. Consistency in $N$ is clear  since the defining inequality, in either case, is asked to hold for every $N''\geq N$, so is particularly satisfied for every  $N''\geq N'$.

Item \ref{La:DieZWEI} follows from the scaling properties 
\begin{align*}
\sigma_{K/\theta^2,N'}^{(r)}(\theta\,\vartheta) &= \sigma_{K,N'}^{(r)}(\vartheta),\\
\tau_{K/\theta^2,N'}^{(r)}(\theta\,\vartheta) &= \tau_{K,N'}^{(r)}(\vartheta).\qedhere
\end{align*}
\end{proof}

Lastly, recall that the $N$-Rényi entropy $\scrS_N(\mu)$, $N\in[1,\infty)$, at $\mu\in\scrP(\mms)$ only depends on the $\meas$-absolutely continuous part of $\mu$, which might be trivial. Along those timelike proper-time parametrized $\smash{\ell_p}$-geodesics in \autoref{Def:TCD*} or \autoref{Def:TCD}, this does not happen; in fact, under mild additional hypotheses, these always consist of $\meas$-absolutely continuous measures. This is addressed in \autoref{Le:Stu} and will be useful at various occasions. (We  also use variants of it without explicit notice occasionally, e.g.~in \autoref{Sub:Local global}.) Its proof uses \eqref{Eq:Ent SN limit} and is analogous to the one of  \cite[Prop.~1.6]{sturm2006b}, hence omitted. 

\begin{lemma}\label{Le:Stu} Fix $p\in (0,1)$, $K\in\R$, and $N\in[1,\infty)$. Let $(\mu_0,\mu_1) = (\rho_0\,\meas,\mu_1 \,\meas)\in(\scrP_\comp(\mms)\cap \Dom(\Ent_\meas))^2$,  let $\smash{\mu\in\scrP_\comp(\mms)}$, and let $\smash{\pi\in\Pi(\mu_0,\mu_1)}$ be a coupling of $\mu_0$ and $\mu_1$ such that for every $N'\geq N$, 
\begin{align*}
\scrS_{N'}(\mu)&\leq -\int_{\mms^2} \sigma_{K,N'}^{(1-t)}(\tsep(x^0,x^1))\,\rho_0(x^0)^{-1/N'}\d\pi(x^0,x^1)\\
&\qquad\qquad - \int_{\mms^2} \sigma_{K,N'}^{(t)}(\tsep(x^0,x^1))\,\rho_1(x^1)^{-1/N'}\d\pi(x^0,x^1).
\end{align*}
Then $\smash{\mu = \rho\,\meas\in\Dom(\Ent_\meas)}$ with
\begin{align*}
\Ent_\meas(\mu) \leq (1-t)\Ent_\meas(\mu_0) + t\Ent_\meas(\mu_1) - \frac{K}{2}\,t(1-t)\,\big\Vert \tsep\big\Vert_{\Ell^2(\mms^2,\pi)}^2.
\end{align*}
\end{lemma}

\begin{remark} Using \autoref{Le:Villani lemma for geodesic}, the assumptions in \autoref{Le:Stu} are satisfied for $\mu := \mu_t$, $t\in[0,1]$, and $\pi$ if these objects are coming from a timelike proper-time parametrized $\smash{\ell_p}$-geodesic as well as a timelike $p$-dualizing coupling, respectively, witnessing the inequality defining $\smash{\TCD_p^*(K,N)}$.
\end{remark}

\begin{corollary} For every $p\in (0,1)$, $K\in\R$, and $N\in[1,\infty)$, the $\smash{\wTCD_p^*(K,N)}$ condition implies  $\smash{\wTCD_p(K,\infty)}$  in the sense of \textnormal{\cite[Def.~4.1]{braun2022}}.
\end{corollary}


\subsection{Geometric inequalities} Now we derive fundamental geometric inequalities from the  conditions introduced in \autoref{Ch:TCD conditions}. We start with the more restrictive $\smash{\wTCD_p(K,N)}$-case which implies sharp versions of these facts, cf.~\autoref{Cor:HD};   \autoref{Sub:Versions reduced} provides nonsharp extensions for $\smash{\wTCD_p^*(K,N)}$-spaces. All proofs are  variations of standard arguments \cite{sturm2006b}, see also \cite{bacher2010,cavalletti2020}.

Some of the geometric inequalities to follow hold in fact under weaker timelike measure-contraction properties, cf.~\autoref{Pr:TMCP to TCD} and \autoref{Re:Geom inequ TMCP}.

\subsubsection{Sharp timelike Brunn--Minkowski inequality} 

To introduce the \emph{timelike Brunn--Minkowski inequality} in the next \autoref{Pr:Brunn-Minkowski}, given  $A_0,A_1\subset \mms$ and $t\in [0,1]$,  define the set of timelike $t$-intermediate points between $A_0$ and $A_1$ as
\begin{align*}
A_t := \{\gamma_t : \gamma\in\TGeo^\tsep(\mms),\, \gamma_0 \in A_0,\, \gamma_1\in A_1\}.
\end{align*}
Furthermore, define 
\begin{align}\label{Eq:THETA}
\Theta := \begin{cases} \sup\tsep(A_0\times A_1) & \textnormal{if }K<0,\\
\inf\tsep(A_0\times A_1) & \textnormal{otherwise}.
\end{cases}
\end{align}

\begin{proposition}\label{Pr:Brunn-Minkowski}  Assume $\smash{\wTCD_p(K,N)}$ for some $p\in (0,1)$, $K\in\R$, and $N\in [1,\infty)$. Let $A_0,A_1\subset\mms$ be relatively compact Borel sets with  $\meas[A_0]\,\meas[A_1] >0$, and assume strong timelike $p$-dualizability of 
\begin{align*}
(\mu_0,\mu_1) := (\meas[A_0]^{-1}\,\meas\mres A_0,\meas[A_1]^{-1}\, \meas\mres A_1).
\end{align*}
Then for every $t\in [0,1]$ and every $N'\geq N$, 
\begin{align}\label{Eq:BM}
\meas[A_t]^{1/N'} \geq \tau_{K,N'}^{(1-t)}(\Theta)\,\meas[A_0]^{1/N'} + \tau_{K,N'}^{(t)}(\Theta)\,\meas[A_1]^{1/N'}.
\end{align}

Assuming $\smash{\TCD_p(K,N)}$ in place of $\smash{\wTCD_p(K,N)}$, the same conclusion holds if $(\mu_0,\mu_1)$ is merely timelike $p$-dualizable.
\end{proposition}

\begin{proof} Without loss of generality, we assume $N' > 1$ and $\meas[A_t] < \infty$. Let a time\-like proper-time parametrized $\smash{\ell_p}$-geodesic $(\mu_t)_{t\in [0,1]}$ connecting $\mu_0$ to $\mu_1$ and a timelike $p$-dualizing  $\pi\in\smash{\Pi_\ll(\mu_0,\mu_1)}$ witness the semiconvexity inequality defining $\smash{\wTCD_p(K,N)}$. Note that $\supp\mu_t\subset A_t$ and, by $\smash{\wTCD_p(K,N)}$, that $\meas[A_t]>0$. Let $\rho_t$ denote the density of the absolutely continuous part of $\mu_t$ with respect to $\meas$. Employing  Jensen's inequality,  and $\smash{\wTCD_p(K,N)}$ again we  obtain
\begin{align*}
&\tau_{K,N}^{(1-t)}(\Theta)\,\meas[A_0]^{1/N'} + \tau_{K,N}^{(t)}(\Theta)\,\meas[A_1]^{1/N'}\\
&\qquad\qquad\leq \int_{A_t}\rho_t^{1-1/N'}\d\meas \leq \meas[A_t]\,\Big[\!\fint_{A_t}\rho_t\d\meas\Big]^{1-1/N'} \leq \meas[A_t]^{-1/N'},
\end{align*}
which is the desired claim.

The proof of the last statement is completely analogous.
\end{proof}

\begin{remark} Note that in the case $K\geq 0$, \eqref{Eq:BM} implies
\begin{align*}
\meas[A_t]^{1/N'} \geq (1-t)\,\meas[A_0]^{1/N'} + t\,\meas[A_1]^{1/N'}.
\end{align*}
\end{remark}

\begin{remark} Recall from \autoref{Re:Strong timelike} that the strong timelike $p$-dualizability hy\-pothesis for $(\mu_0,\mu_1)$ in  \autoref{Pr:Brunn-Minkowski} is satisfied if $\smash{A_0 \times A_1 \subset \mms_\ll^2}$ provided the space $\scrX$ is locally causally closed, globally hyperbolic, and geodesic.
\end{remark}

\subsubsection{Sharp timelike Bonnet--Myers inequality} Next, an immediate consequence of \autoref{Pr:Brunn-Minkowski} is the subsequent \emph{timelike Bonnet--Myers inequality}.

\begin{corollary}\label{Cor:Bonnet-Myers}  Assume the $\smash{\wTCD_p(K,N)}$ condition  for some $p\in(0,1)$, $K>0$, and $N\in [1,\infty)$. Then 
\begin{align*}
\sup\tsep(\mms^2) \leq \pi\sqrt{\frac{N-1}{K}}.
\end{align*}
\end{corollary}

\begin{proof} Suppose to the contrapositive that for some $\varepsilon > 0$, we have
\begin{align*}
\tsep(z_0,z_1) \geq \pi\sqrt{\frac{N-1}{K}}+4\varepsilon
\end{align*}
for two given points $z_0,z_1\in\mms$. By continuity of $\tsep$, we fix $\delta > 0$ and $x_0,x_1\in\mms$ such that $\smash{A_0 := \sfB^\met(x_0,\delta) \subset I^+(z_0)}$, $\smash{A_1 := \sfB^\met(x_1,\delta) \subset I^-(z_1)}$, and 
\begin{align*}
\inf\tsep(\sfB^\met(x_0,\delta) \times \sfB^\met(x_1,\delta)) \geq \pi\sqrt{\frac{N-1}{K}}+\varepsilon.
\end{align*}
This implies strong timelike $p$-dualizability of $\smash{(\meas[A_0]^{-1}\,\meas\mres{A_0},\meas[A_1]^{-1}\,\meas\mres{A_1})}$ by \autoref{Re:Strong timelike}. Hence \autoref{Pr:Brunn-Minkowski} is applicable, and the inherent  inequality $\smash{\Theta \geq \pi\sqrt{(N-1)/K}+\varepsilon}$ together with the definition of $\smash{\tau_{K,N}^{(r)}}$ gives $\smash{\meas[A_{1/2}]=\infty}$. On the other hand, the set $\smash{A_{1/2}}$ is relatively compact since $\smash{A_{1/2} \subset J(z_0,z_1)}$ by global hy\-per\-bolicity, which makes the situation $\meas[A_{1/2}]=\infty$ impossible.
\end{proof}

\begin{remark} \autoref{Cor:Bonnet-Myers} implies  that, under its assumptions, the $\tsep$-length of every causal curve in $\mms$ is no larger than $\smash{\pi\sqrt{(N-1)/K}}$. Moreover, if $K>0$ and $N=1$ then $\mms$ does not contain any chronologically related pair of points.
\end{remark}

\subsubsection{Sharp timelike Bishop--Gromov inequality}\label{Sub:Sharp BG inequ} To prove  volume growth estimates à la \emph{Bishop--Gromov}, cf.~\autoref{Th:BG} below, we  recall the notion of $\tsep$-star-shaped sets from \cite[Sec.~3.1]{cavalletti2020}; this is used to localize the volume of $\tsep$-balls 
\begin{align}\label{Eq:TAU BALLS}
\sfB^\tsep(x,r) := \{y\in I^+(x) : \tsep(x,y)< r\}\cup\{x\},
\end{align}
with $x\in\mms$ and $r>0$, which typically have infinite volume. 

A set $\smash{E\subset I^+(x)\cup\{x\}}$  is termed \emph{$\tsep$-star-shaped} with respect to $x\in\mms$ if for every $\gamma\in\TGeo^\tsep(\mms)$ with $\gamma_0 = x$ and $\gamma_1 \in E$, we have $\gamma_t\in E$ for every $t\in (0,1)$. For such $E$ and $x$, and given $r>0$, we define the quantities
\begin{align*}
\rmv_r &:= \meas\big[\bar{\sfB}^\tsep(x,r)\cap E\big],\\
\rms_r &:= \limsup_{\delta \to 0} \delta^{-1}\,\meas\big[(\bar{\sfB}^\tsep(x,r+\delta)\setminus \sfB^\tsep(x,r))\cap E\big].
\end{align*}
Whenever confusion is excluded, we avoid making notationally explicit the  dependency of $\rmv_r$ and $\rms_r$ on $E$ and $x$. Recalling \eqref{Eq:DIS COF}, we set $\mathfrak{s}_{K,N} := \mathfrak{s}_{K/N}$ and
\begin{align}\label{Eq:integral def}
\mathfrak{v}_{K,N}(r) := \Big[\!\int_0^r \mathfrak{s}_{K,N-1}(r)^{N-1}\d r\Big]^{1/N}.
\end{align}

In the next result, we employ the convention $1/0 := \infty$.

\begin{theorem}\label{Th:BG} Assume  $\smash{\wTCD_p(K,N)}$ for some $p\in (0,1)$, $K\in\R$, and $N\in (1,\infty)$. Let $\smash{E\subset I^+(x)\cup\{x\}}$ be a compact set which is $\tsep$-star-shaped with respect to $x\in\mms$. Then for every $r,R > 0$ with $\smash{r < R \leq \pi\sqrt{(N-1)/\max\{K,0\}}}$,
\begin{align*}
\frac{\rms_r}{\rms_R} &\geq \Big[\frac{\mathfrak{s}_{K,N-1}(r)}{\mathfrak{s}_{K,N-1}(R)}\Big]^{N-1},\\
\frac{\rmv_r}{\rmv_R} &\geq \Big[\frac{\mathfrak{v}_{K,N}(r)}{\mathfrak{v}_{K,N}(R)}\Big]^N.
\end{align*}
\end{theorem}

\begin{proof} We assume $K>0$, the rest is argued similarly. Define $t := r/R$,  let $\delta > 0$, and let $\varepsilon > 0$ be small enough such that $\smash{A_0 \times A_1 \subset \mms_\ll^2}$, where $\smash{A_0 := \sfB^\tsep(x,\varepsilon)\cap E}$ and $A_1 := (\bar{\sfB}^\tsep(x,R+\delta R) \setminus \sfB^\tsep(x,R))\cap E$. Thus $(\meas[A_0]^{-1}\,\meas\mres{A_0},\meas[A_1]^{-1}\,\meas\mres{A_1})$ is strongly time\-like $p$-dualizable according to  \autoref{Re:Strong timelike}. Now we  observe  that $\smash{A_t \subset (\sfB^\tsep(x,r+\delta r)\setminus \sfB^\tsep(x,r))\cap E}$. By the reverse triangle inequality for $\tsep$, we get $\Theta \leq R+\delta R$. Employing \autoref{Pr:Brunn-Minkowski}, we therefore obtain
\begin{align*}
&\meas\big[(\bar{\sfB}^\tsep(x,r+\delta r)\setminus \sfB^\tsep(x,r))\cap E\big]^{1/N}\\
&\qquad\qquad \geq \tau_{K,N}^{(1-r/R)}(R+\delta R)\,\meas\big[\sfB^\tsep(x,\varepsilon)\cap E\big]^{1/N}\\
&\qquad\qquad\qquad\qquad + \tau_{K,N}^{(r/R)}(R+\delta R)\, \meas\big[(\bar{\sfB}^\tsep(x,R+\delta R)\setminus \sfB^\tsep(x,R))\cap E\big]^{1/N}\\
&\qquad\qquad \geq \tau_{K,N}^{(r/R)}(R+\delta R)\, \meas\big[(\bar{\sfB}^\tsep(x,R+\delta R)\setminus \sfB^\tsep(x,R))\cap E\big]^{1/N}.
\end{align*}
Writing out the expression $\smash{\tau_{K,N}^{(r/R)}(R+\delta R)^N}$ and finally sending $\delta\to 0$ implies the claimed lower bound for $\rms_r/\rms_R$.

As in the proof of \cite[Thm.~2.3]{sturm2006b}, we argue that $\rmv$ is locally Lipschitz continuous on $(0,\infty)$. Hence, it is  differentiable $\smash{\Leb^1}$-a.e.~on $(0,\infty)$, and $\smash{\rmv_r = \int_0^r \dot{\rmv}_s\d s}$ for every $r > 0$. Therefore, the claimed lower bound on $\rmv_r/\rmv_R$  follows from the previous part in com\-bination with \cite[Lem.~18.9]{villani2009}.
\end{proof}

\begin{remark} If $K=0$, the estimates from \autoref{Th:BG}  become
\begin{align*}
\frac{\rms_r}{\rms_R} &\geq \Big[\frac{r}{R}\Big]^{N-1},\\
\frac{\rmv_r}{\rmv_R} &\geq \Big[\frac{r}{R}\Big]^N.
\end{align*}
In fact, these estimates hold for $N=1$ as well.
\end{remark}

Incidentally, from \autoref{Th:BG} we deduce the following natural and sharp upper bound on the $\tsep$-Hausdorff dimension as introduced and studied in \cite{mccann2021}. Our results improve \cite[Thm.~5.2, Cor.~5.3]{mccann2021}. We refer to \cite{mccann2021} for details about the notions presented in the subsequent \autoref{Cor:HD}; recall also \autoref{Ex:Low reg spt}.

\begin{corollary}\label{Cor:HD} Let $(\mms,\langle\cdot,\cdot\rangle)$ be a Lipschitz continuous, globally hyperbolic spacetime of dimension $n\in\N$, cf.~\autoref{Ex:Low reg spt}. Assume its induced Lo\-rentzian geodesic space, with $\meas$ being the Lorentzian volume measure induced by $\langle\cdot,\cdot\rangle$, to obey $\smash{\wTCD_p(K,N)}$ for some $p\in (0,1)$, $K\in\R$, and $N\in[1,\infty)$. Then the geometric dimension $\smash{\dim^\tsep\mms}$ in the sense of \textnormal{\cite[Def.~3.1]{mccann2021}} satisfies
\begin{align*}
n=\dim^\tsep\mms \leq N.
\end{align*}
\end{corollary}

\begin{remark}\label{Re:HD2} As in \cite[Thm.~5.2]{mccann2021}, which is proven under $\smash{\wTCD_p^e(K,N)}$, assuming $\smash{\wTCD_p^*(K,N)}$ instead of $\smash{\wTCD_p(K,N)}$ in  \autoref{Cor:HD} would only lead to the conclusion $n = \dim^\tsep\mms \leq N+1$ by \autoref{Th:Reduced BG}.
\end{remark}

\subsubsection{Versions in the reduced case}\label{Sub:Versions reduced} With evident modifications, the following nonsharp versions of \autoref{Pr:Brunn-Minkowski}, \autoref{Cor:Bonnet-Myers}, \autoref{Th:BG}, and \autoref{Cor:HD} are readily verified under the more general $\smash{\wTCD_p^*(K,N)}$ condition.  The deduced inequalities are analogous to their counterparts drawn from the $\smash{\wTCD_p^e(K,N)}$ condition in  \cite[Prop.~3.4, Prop.~3.5, Prop.~3.6]{cavalletti2020}.

In view of \autoref{Th:Reduced BG} below, recall the definition of $\mathfrak{v}_{K,N}$ from \eqref{Eq:integral def}.

\begin{proposition}\label{Propos} Assume $\smash{\wTCD_p^*(K,N)}$ for some $p\in (0,1)$, $K\in\R$, and $N\in[1,\infty)$. Let $A_0,A_1\subset\mms$ be relatively compact Borel sets which satisfy  $\meas[A_0]\,\meas[A_1]>0$, and assume strong timelike $p$-dualizability of 
\begin{align*}
(\mu_0,\mu_1) := (\meas[A_0]^{-1}\,\meas\mres A_0, \meas[A_1]^{-1}\,\meas\mres A_1).
\end{align*}
Let $\Theta$ be as in \eqref{Eq:THETA}. Then for every $t\in [0,1]$ and every $N'\geq N$,
\begin{align*}
\meas[A_t]^{1/N'} \geq \sigma_{K,N'}^{(1-t)}(\Theta)\,\meas[A_0]^{1/N'} + \sigma_{K,N'}^{(t)}(\Theta)\,\meas[A_1]^{1/N'}.
\end{align*}

Assuming $\smash{\TCD_p^*(K,N)}$ in place of $\smash{\wTCD_p^*(K,N)}$, the same conclusion holds if $(\mu_0,\mu_1)$ is merely timelike $p$-dualizable.
\end{proposition}

\begin{corollary}\label{Cor:Reduced BM} Assume the $\smash{\wTCD_p^*(K,N)}$ condition for some $p\in(0,1)$, $K>0$, and $N\in [1,\infty)$. Then 
\begin{align*}
\sup\tsep(\mms^2) \leq \pi\sqrt{\frac{N}{K}}.
\end{align*}
\end{corollary}

\begin{theorem}\label{Th:Reduced BG} Assume $\smash{\wTCD_p^*(K,N)}$ for some $p\in (0,1)$, $K\in\R$, and $N\in (1,\infty)$. Let $\smash{E\subset I^+(x)\cup\{x\}}$ be a compact set which is $\tsep$-star-shaped with respect to $x\in\mms$. Then for every $r,R > 0$ with $\smash{r < R \leq \pi\sqrt{N/\max\{K,0\}}}$,
\begin{align*}
\frac{\rms_r}{\rms_R} &\geq \Big[\frac{\mathfrak{s}_{K,N}(r)}{\mathfrak{s}_{K,N}(R)}\Big]^N,\\
\frac{\rmv_r}{\rmv_R} &\geq \Big[\frac{\mathfrak{v}_{K,N+1}(r)}{\mathfrak{v}_{K,N+1}(R)}\Big]^{N+1}.
\end{align*}
\end{theorem}

\subsection{Stability}\label{Sec:Stability TCD} In this section, we prove a key property of our timelike curvature-dimension conditions, namely the weak stability of the notions from \autoref{Def:TCD*} and \autoref{Def:TCD}, cf.~\autoref{Th:Stability TCD}. The relevant notion of convergence of measured Lorentzian pre-length spaces, see \autoref{Def:Convergence}, is due to \cite[Sec. 3.3]{cavalletti2020}. Given our \autoref{Ass:ASS}, we add the hypothesis of regularity to it.

\begin{definition}\label{Def:Lor isometry} For Lorentzian pre-length spaces $\smash{(\mms^i,\met^i,\ll^i,\leq^i,\tsep^i)}$, $i\in\{0,1\}$, we term a map $\smash{\iota\colon \mms^0\to\mms^1}$ a \emph{Lorentzian isometric embedding} if $\iota$ is a topolo\-gical embedding such that for every $x,y\in\mms^0$,
\begin{enumerate}[label=\textnormal{\alph*.}]
\item $x\leq^0 y$ if and only if $\iota(x) \leq^1\iota(y)$, and
\item $\tau^1(\iota(x),\iota(y)) = \tau^0(x,y)$.
\end{enumerate}
\end{definition}

Given any $k\in\N_\infty := \N\cup\{\infty\}$ let $(\mms_k,\met_k,\meas_k,\ll_k,\leq_k,\tsep_k)$ be a fixed measured Lorentzian pre-length space which, as in \eqref{Eq:X}, we will abbreviate by $\scrX_k$.

\begin{definition}\label{Def:Convergence} We define $(\scrX_k)_{k\in\N}$ to converge to $\scrX_\infty$ if the following holds.
\begin{enumerate}[label=\textnormal{\alph*\textcolor{black}{.}}]
\item\label{La:Blurr1reg} There is a causally closed, $\scrK$-globally hyperbolic, regular measured Lorentzian geodesic space $\smash{(\mms,\met,\ll,\leq,\tsep)}$ such that for every $k\in\N_\infty$ there exists a Lorentzian isometric embedding $\smash{\iota_k\colon \mms_k\to \mms}$.
\item In terms of the maps $\iota_k$ from \ref{La:Blurr1reg}, the sequence $((\iota_k)_\push\meas_k)_{k\in\N}$ converges to $(\iota_\infty)_\push\meas_\infty$ in duality with $\Cont_\comp(\mms)$, i.e.~for every $\varphi\in\Cont_\comp(\mms)$ we have
\begin{align*}
\lim_{k\to\infty}\int_\mms \varphi\d(\iota_k)_\push\meas_k = \int_\mms \varphi\d(\iota_\infty)_\push\meas_\infty.
\end{align*}
\end{enumerate}
\end{definition}

\begin{remark} If $(\scrX_k)_{k\in\N}$ converges to $\scrX_\infty$,  $\scrX_k$ is automatically causally closed, $\scrK$-globally hyperbolic, regular, and geodesic for every $k\in\N_\infty$.
\end{remark}

The following \autoref{Le:CM lemma} is an implicit consequence from the proof of \cite[Lem. 3.15]{cavalletti2020} and will be useful for the proof of \autoref{Th:Stability TCD}. 

\begin{lemma}\label{Le:CM lemma} Assume that $\scrX$ is causally closed, globally hyperbolic, and geodesic. Let the pair $\smash{(\mu_0,\mu_1)=(\rho_0\,\meas,\rho_1\,\meas)\in\scrP_\comp^\ac(\mms,\meas)^2}$ admit some $\smash{\ell_p}$-optimal coupling $\smash{\pi\in\Pi_\ll(\mu_0,\mu_1)}$, $p\in (0,1]$. Then there exist sequences $(\pi^n)_{n\in\N}$ in $\scrP(\mms^2)$ and $(a_n)_{n\in\N}$ in $[1,\infty)$ with the properties 
\begin{enumerate}[label=\textnormal{\textcolor{black}{(}\roman*\textcolor{black}{)}}]
\item $(a_n)_{n\in\N}$ converges to $1$,
\item $\pi^n[\mms_\ll^2]=1$ for every $n\in\N$,
\item $\smash{\pi^n = \rho^n\,\meas^{\otimes 2}\in \scrP^\ac(\mms^2,\meas^{\otimes 2})}$ with $\smash{\rho^n\in\Ell^\infty(\mms^2,\meas^{\otimes 2})}$ for every $n\in\N$,
\item $(\pi^n)_{n\in\N}$ converges weakly to $\pi$, 
\item\label{DENS} the density $\smash{\rho_0^n}$ and $\smash{\rho_0^n}$ of the first and second marginal of $\pi_n$ with respect to $\meas$ is no larger than $a_n\,\rho_0$ and $a_n\,\rho_1$, respectively, for every $n\in\N$, and
\item $\smash{\rho_0^n\to \rho_0}$ and $\smash{\rho_1^n\to\rho_1}$ in $\Ell^1(\mms,\meas)$ as $n\to\infty$.
\end{enumerate}
\end{lemma}

Furthermore, in the notation of \autoref{Le:CM lemma}, given any $K\in\R$, $N\in[1,\infty)$, and $t\in [0,1]$, as well as a measure $\pi\in\Pi(\mu_0,\mu_1)$ with marginals $\mu_0 = \rho_0\,\meas\in\smash{\scrP^\ac(\mms,\meas)}$ and $\mu_1 = \rho_1\,\meas\in\smash{\scrP^\ac(\mms,\meas)}$, for brevity we define
\begin{align}\label{Eq:TDdef}
\begin{split}
\scrT_{K,N}^{(t)}(\pi) &:= -\int_{\mms^2} \tau_{K,N}^{(1-t)}(\tsep(x^0,x^1))\,\rho_0(x^0)^{-1/N}\d\pi(x^0,x^1)\\
&\qquad\qquad -\int_{\mms^2} \tau_{K,N}^{(t)}(\tsep(x^0,x^1))\,\rho_1(x^1)^{-1/N}\d\pi(x^0,x^1).
\end{split}
\end{align}

We  note the following variant of \cite[Lem.~3.2]{sturm2006b}. The latter proof easily carries over to the setting of \autoref{Le:Const perturb}, where the marginal densities, instead of being assumed to be the same throughout the sequence as in \cite{sturm2006b}, are allowed to be perturbed by constants vanishing in the limit.  Compare with \autoref{Le:USC lemma}.

\begin{lemma}\label{Le:Const perturb} Suppose  $\meas\in\scrP(\mms)$. Let $(a_k)_{k\in\N}$ be a sequence in $[1,\infty)$ converging to $1$, and let $(b_k)_{k\in\N}$ be a sequence of nonnegative real numbers converging to $0$. Define $\smash{\nu_{k,0},\nu_{k,1}\in\scrP^\ac(\mms,\meas)}$, $k\in\N$, by
\begin{align*}
\nu_{k,0} &:= (a_k+b_k)^{-1}\,(a_k\,\mu_0 + b_k\,\meas) = (a_k+b_k)^{-1}\,(a_k\,\rho_0 + b_k)\,\meas,\\
\nu_{k,1} &:= (a_k+b_k)^{-1}\,(a_k\,\mu_1 + b_k\,\meas) = (a_k+b_k)^{-1}\,(a_k\,\rho_1 + b_k)\,\meas.
\end{align*}
Then for every sequence $(\pi_k)_{k\in\N}$ of  measures $\pi_k\in \Pi(\nu_{k,0},\nu_{k,1})$ which converges weakly to $\pi\in\Pi(\mu_0,\mu_1)$,  every $K\in\R$, every $N\in[1,\infty)$, and every $t\in[0,1]$, 
\begin{align*}
\limsup_{k\to\infty}\scrT_{K,N}^{(t)}(\pi_k) \leq \scrT_{K,N}^{(t)}(\pi).
\end{align*}
\end{lemma}

\begin{remark} Of course, the assertion from \autoref{Le:Const perturb} remains valid if $\smash{\tau_{K,N}^{(1-t)}}$ and $\smash{\tau_{K,N}^{(t)}}$ are replaced by $\smash{\sigma_{K,N}^{(1-t)}}$ and $\smash{\sigma_{K,N}^{(t)}}$ in \eqref{Eq:TDdef}, respectively.
\end{remark}

The subsequent proof of \autoref{Th:Stability TCD} below  follows  \cite[Thm.~3.1]{sturm2006b}. However, by the nature of our timelike curvature-dimension conditions, we additionally  have to ensure chronological relations of many measures under consideration. Hence, the proof becomes much longer and quite technical. It may be skipped at first reading; alternatively, the reader may consult the proof of the counterpart  \autoref{Th:Stability TMCP} of \autoref{Th:Stability TCD} below for a somewhat easier argument in which one does not have to trace chronological $\smash{\ell_p}$-optimal couplings.

Let $\scrS_N^n$ denote the $N$-Rényi entropy with respect to $\meas_n$, $n\in\N_\infty$.

\begin{theorem}\label{Th:Stability TCD} Assume the convergence of $(\scrX_k)_{k\in\N}$ to $\scrX_\infty$ as in  \autoref{Def:Convergence}. Moreover, let $(K_k,N_k)_{k\in\N}$ be a sequence in $\R\times [1,\infty)$ converging to $(K_\infty,N_\infty)\in\R\times[1,\infty)$. Suppose the existence of $p\in (0,1)$ such that $\scrX_k$ obeys $\TCD_p(K_k,N_k)$ for every $k\in\N$. 
Then $\smash{\scrX_\infty}$ satisfies $\smash{\wTCD_p(K_\infty,N_\infty)}$. 

The analogous statement in which $\smash{\TCD_p(K_k,N_k)}$ and $\smash{\wTCD_p(K_\infty,N_\infty)}$ are replaced by $\smash{\TCD_p^*(K_k,N_k)}$ and $\smash{\wTCD_p^*(K_\infty,N_\infty)}$, $k\in\N$, respectively, holds too.
\end{theorem}

\begin{proof} We only prove the first implication, the second is similar. To relax  notation, without restriction and further notice we identify $\mms_k$ with its image $\iota_k(\mms_k)$ in $\mms$, and the measure $\meas_k$ with its push-forward $(\iota_k)_\push\meas_k$ for every $k\in\N_\infty$.

\textbf{Step 1.} \textit{Reduction to compact $\mms$.} Fix any  strongly timelike $p$-dualizable pair  $(\mu_{\infty,0},\mu_{\infty,1}) = (\rho_{\infty,0}\,\meas_\infty,\rho_{\infty,1}\,\meas_\infty)\in\scrP_\comp^\ac(\mms,\meas_\infty)$. By compactness of $\supp\mu_{\infty,0}$ and $\supp\mu_{\infty,1}$, and by $\scrK$-global hyperbolicity of $\smash{\mms}$, we restrict all  arguments below to a compact subset $\smash{C\subset \mms}$ with $\smash{\meas_k[C]^{-1}\,\meas_k\mres C \to \meas_\infty[C]^{-1}\,\meas_\infty\mres C}$ weakly as $k\to\infty$. To relax notation, we will thus assume with no loss of generality  that $\smash{\mms}$ itself is compact, and that  $\smash{\meas_k\in\scrP(\mms)}$ for every $k\in\N_\infty$. In particular, without notationally expressing this property all the time, all measures will henceforth be compactly supported.

\textbf{Step 2.} \textit{Restriction of the assumptions on $\mu_{\infty,0}$ and $\mu_{\infty,1}$.} We will first assume that $\rho_{\infty,0},\rho_{\infty,1}\in\Ell^\infty(\mms,\meas_\infty)$. The general case is discussed in Step 8 below; we note for now that this conclusion will not conflict with our reductions from Step 1.

\textbf{Step 3.} \textit{Construction of a chronological recovery sequence.} The goal of this step is the construction of a sequence $\smash{(\mu_{k,0},\mu_{k,1})_{k\in\N}}$ of timelike $p$-dualizable pairs $\smash{(\mu_{k,0},\mu_{k,1})\in\scrP^\ac(\mms,\meas_k)^2}$, $k\in\N$, such that $\smash{\mu_{k,0}\to \mu_{\infty,0}}$ and $\smash{\mu_{k,1}\to \mu_{\infty,1}}$ weakly as $k\to\infty$, up to a subsequence. The highly nontrivial part we address here is that such a pair can in fact be constructed to be chronological.

\textbf{Step 3.1.} Let $W_2$ denote the $2$-Wasserstein distance on $\smash{(\mms,\met)}$. Since $(\meas_k)_{k\in\N}$ converges weakly to $\meas_\infty$ and since $\mms$ is compact, we have $W_2(\meas_k,\meas_\infty)\to 0$ as $k\to\infty$. Let $\mathfrak{q}_k\in\scrP(\mms^2)$ be a $W_2$-optimal coupling of $\meas_k$ and $\meas_\infty$, $k\in\N$. For later use, we disintegrate $\mathfrak{q}_k$ with respect to $\pr_1$ and $\pr_2$, respectively, writing
\begin{align*}
\rmd\mathfrak{q}_k(x, y) = \rmd\mathfrak{p}_{x}^k(y)\d\meas_k(x) = \rmd\mathfrak{p}_{y}^\infty(x)\d\meas_\infty(y)
\end{align*}
for certain Borel maps $\smash{\mathfrak{p}^k \colon \mms\to \scrP(\mms)}$ and $\smash{\mathfrak{p}^\infty \colon \mms\to \scrP(\mms)}$. (Although $\mathfrak{p}^\infty$  depends on $k$ as well, we do not make this explicit to not overload  our  notation.) These disin\-te\-grations induce nonrelabeled Borel maps $\smash{\mathfrak{p}^k\colon \scrP^\ac(\mms,\meas_\infty)\to \scrP^\ac(\mms,\meas_k)}$ and $\smash{\mathfrak{p}^\infty}\colon \scrP^\ac(\mms,\meas_k)\to \scrP^\ac(\mms,\meas_\infty)$ defined by
\begin{align*}
\rmd\mathfrak{p}^k(f\,\meas_\infty)(x) &:= \Big[\!\int_\mms  f(y)\d\mathfrak{p}_{x}^k(y) \Big]\d\meas_k(x),\\
\rmd\mathfrak{p}^\infty(g\,\meas_k)(y) &:= \Big[\!\int_\mms  g(x)\d\mathfrak{p}_{y}^\infty(x) \Big]\d\meas_\infty(y).
\end{align*} 

\textbf{Step 3.2.} Let $\pi_\infty\in\Pi_\ll(\mu_{\infty,0},\mu_{\infty,1})$ be timelike $p$-dualizing, and let $(\pi^n_\infty)_{n\in\N}$ be an asso\-ciated sequence in $\smash{\scrP^\ac(\mms^2,\meas_\infty^{\otimes 2})}$ as in \autoref{Le:CM lemma}, where we write
\begin{align*}
\pi^n_\infty &= \rho_\infty^n\,\meas_\infty^{\otimes 2},\\
\mu_{\infty,0}^n := (\pr_1)_\push\pi_\infty^n &= \rho_{\infty,0}^n\,\meas_\infty,\\
\mu_{\infty,1}^n := (\pr_2)_\push\pi_\infty^n &= \rho_{\infty,1}^n\,\meas_\infty.
\end{align*}
For $k\in\N$, define $\smash{\mu_{k,0}^n,\mu_{k,1}^n\in\scrP^\ac(\mms,\meas_k)}$ by
\begin{align*}
\mu_{k,0}^n &:= \mathfrak{p}^k(\mu_{\infty,0}^n) = \rho_{k,0}^n\,\meas_k,\\
\mu_{k,1}^n &:= \mathfrak{p}^k(\mu_{\infty,1}^n) = \rho_{k,1}^n\,\meas_k
\end{align*}
as well as  $\smash{\pi_k^n\in\scrP^\ac(\mms^2,\meas_k^{\otimes 2})}$ by
\begin{align*}
\rmd \pi_k^n(x^0,x^1) &:= \Big[\!\int_{\mms^2}\rho_\infty^n(y^0,y^1)\d(\mathfrak{p}_{x^0}^k\otimes\mathfrak{p}_{x^1}^k)(y^0,y^1)\Big]\d\meas_k^{\otimes 2}(x^0,x^1),
\end{align*}
or in other words,
\begin{align*}
\pi_k^n = (\pr_1,\pr_3)_\push\big[(\rho_\infty^n \circ (\pr_2,\pr_4))\,\mathfrak{q}_k\otimes\mathfrak{q}_k\big].
\end{align*}
A straightforward computation entails $\smash{\pi_k^n\in\Pi(\mu_{k,0}^n,\mu_{k,1}^n)}$. Also, for given $n\in\N$, weak convergence of $(\mathfrak{q}_k)_{k\in\N}$ to the unique $W_2$-optimal (diagonal) coupling of $\meas_\infty$ and $\meas_\infty$ by tightness \cite[Lem.~4.3, Lem.~4.4]{villani2009}, up to a nonrelabeled subse\-quence, imply  weak convergence of $\smash{(\pi_k^n)_{k\in\N}}$ to $\smash{\pi_\infty^n}$ up to a nonrelabeled subsequence. 

\textbf{Step 3.3.} By \autoref{Le:CM lemma}, a tightness and a diagonal argument, passage to a subsequence, and relabeling of indices, we find $n\colon\N\to\N$ such that, setting
\begin{align*}
\tilde{\pi}_k := \pi_k^{n_k},
\end{align*}
the sequence $\smash{(\tilde{\pi}_k)_{k\in\N}}$ converges weakly to $\pi_\infty$. By Portmanteau's theorem,
\begin{align}\label{Eq:Portmanteau in stability}
1 = \pi_\infty[\mms_\ll^2] \leq \liminf_{k\to\infty} \tilde{\pi}_k[\mms_\ll^2].
\end{align}
Hence, up passing to a  subsequence, we henceforth assume that $\tilde{\pi}_k[\mms_\ll^2] > 0$ for every $k\in\N$. We write the mar\-ginals $\smash{\tilde{\mu}_{k,0},\tilde{\mu}_{k,1}\in\scrP^\ac(\mms,\meas_k)}$ of $\tilde{\pi}_k$ as
\begin{align*}
\tilde{\mu}_{k,0} &= \tilde{\rho}_{k,0}\,\meas_k = \rho_{k,0}^{n_k}\,\meas_k,\\
\tilde{\mu}_{k,1} &= \tilde{\rho}_{k,1}\,\meas_k = \rho_{k,1}^{n_k}\,\meas_k.
\end{align*}


\textbf{Step 3.4.} Now we 
define $\smash{\hat{\pi}_k\in\scrP^\ac(\mms^2)}$ by
\begin{align*}
\hat{\pi}_k := \tilde{\pi}_k[\mms_\ll^2]^{-1}\,\tilde{\pi}_k\mres \mms^2_\ll.
\end{align*}
Note that $\smash{(\hat{\pi}_k)_{k\in\N}}$ converges weakly to $\smash{\pi_\infty}$, whence
\begin{align}\label{Eq:Der Limes}
\lim_{k\to\infty}\int_{\mms^2}\tsep^p\d\hat{\pi}_k = \int_{\mms^2}\tsep^p\d\pi_\infty.
\end{align}
Let $\smash{\hat{\mu}_{k,0},\hat{\mu}_{k,0}\in\scrP^\ac(\mms,\meas_k)}$ denote the marginals of $\smash{\hat{\pi}_k}$, with
\begin{align*}
\hat{\mu}_{k,0} &= \hat{\rho}_{k,0}\,\meas_k,\\
\hat{\mu}_{k,1} &= \hat{\rho}_{k,1}\,\meas_k.
\end{align*}

These measures are chronologically related, but we do not know whether they have a chronological \emph{$\smash{\ell_p}$-optimal} coupling, i.e.~whether they are timelike $p$-dualizable. To achieve this,  note that $\smash{\ell_p(\hat{\mu}_{k,0},\hat{\mu}_{k,1})} \in (0,\infty)$ by chronology of $\smash{\hat{\pi}_k}$ as well as compactness of $\mms^2$. Let $\smash{\check{\pi}_k\in\Pi_\leq(\hat{\mu}_{k,0},\hat{\mu}_{k,1})}$ be an $\smash{\ell_p}$-optimal coupling. By weak convergence of $\smash{(\hat{\pi}_k)_{k\in\N}}$, its marginal sequences $\smash{(\hat{\mu}_{k,0})_{k\in\N}}$ and $\smash{(\hat{\mu}_{k,1})_{k\in\N}}$ are tight, and so is $\smash{(\check{\pi}_k)_{k\in\N}}$ \cite[Lem.~4.4]{villani2009}. By Prokhorov's theorem, the latter converges weakly to some $\smash{\check{\pi}_\infty\in\Pi(\mu_{\infty,0},\mu_{\infty,1})}$; causal closedness of $\mms$ and Portmanteau's theorem imply $\smash{\check{\pi}_\infty[\mms_\leq^2]=1}$. Moreover,  \eqref{Eq:Der Limes} gives
\begin{align}\label{Eq:Lp calc}
\begin{split}
\ell_p(\mu_{\infty,0},\mu_{\infty,1})^p &= \int_{\mms^2} \tsep^p\d\pi_\infty\\ 
&= \lim_{k\to\infty}\int_{\mms^2} \tsep^p\d\hat{\pi}_k\\ &\leq \lim_{k\to\infty} \int_{\mms^2}\tsep^p\d\check{\pi}_k\\ 
&= \int_{\mms^2}\tsep^p\d\check{\pi}_\infty.
\end{split}
\end{align}
Here the inequality follows from $\smash{\ell_p}$-optimality of $\smash{\check{\pi}_k}$. Thus, $\smash{\check{\pi}_\infty}$ is an $\smash{\ell_p}$-optimal coupling of $\smash{\mu_{\infty,0}}$ and $\smash{\mu_{\infty,1}}$. Since the latter pair is  strongly timelike $p$-dualizable, every such coupling is chronological, whence $\smash{\check{\pi}_\infty[\mms_\ll^2]=1}$; thus,
\begin{align*}
1 = \check{\pi}_\infty[\mms_\ll^2] \leq \liminf_{k\to\infty} \check{\pi}_k[\mms_\ll^2].
\end{align*}
As usual, we may and will assume  $\smash{\check{\pi}_k[\mms_\ll^2] > 0}$ for every $k\in\N$.

\textbf{Step 3.5.} Lastly, we define $\smash{\bar{\pi}_k\in\scrP(\mms^2)}$ by
\begin{align*}
\bar{\pi}_k := \check{\pi}_k[\mms_\ll^2]^{-1}\,\check{\pi}_k\mres \mms_\ll^2.
\end{align*}
Since the restriction of an $\smash{\ell_p}$-optimal coupling is $\smash{\ell_p}$-optimal \cite[Lem.~2.10]{cavalletti2020}, $\smash{\bar{\pi}_k}$ is a chronological $\smash{\ell_p}$-optimal coupling of its marginals $\smash{\mu_{k,0},\mu_{k,1}\in\scrP^\ac(\mms,\meas_k)}$ with
\begin{align*}
\mu_{k,0} &= \rho_{k,0}\,\meas_k,\\
\mu_{k,1} &= \rho_{k,1}\,\meas_k.
\end{align*}
The pair $\smash{(\mu_{k,0},\mu_{k,1})}$ is timelike $p$-dualizable by $\bar{\pi}_k$ for every $k\in\N$, as desired.

\textbf{Step 4.} \textit{Invoking the $\TCD$ condition.} Fix  $K\in\R$ and $N\in (1,\infty)$ with $K < K_\infty$ and $N > N_\infty$. Then $K < K_k$ and $N > N_k$ for large enough $k\in\N$. Hence, again up to a subsequence, we may and will assume that the previous strict inequalities hold for every $k\in\N$, respectively.

By \autoref{Pr:Consistency TCD}, for every $k\in\N$ there are a timelike proper-time parametrized $\smash{\ell_p}$-geodesic $(\mu_{k,t})_{t\in[0,1]}$ con\-necting $\smash{\mu_{k,0}}$ to $\mu_{k,1}$ as well as a timelike $p$-dualizing coupling $\smash{\pi_k\in\Pi_\ll(\mu_{k,0},\mu_{k,1})}$ such that for every $t\in[0,1]$ and every $N'\geq N$,
\begin{align}\label{Eq:TCD COND INVOKING}
\scrS_{N'}^k(\mu_{k,t}) \leq \scrT_{K,N'}^{(t)}(\pi_k).
\end{align}

\textbf{Step 5.} \textit{Estimating $\smash{-\scrT_{K,N'}^{(t)}(\pi_k)}$ from below.} Tracing back the construction from Step 2 and Step 3, we infer the inequalities
\begin{align}\label{Eq:Inequ rhok0}
\begin{split}
\rho_{k,0} &\leq \check{\pi}_k[\mms_\ll^2]^{-1}\,\tilde{\pi}_k[\mms_\ll^2]^{-1}\,\tilde{\rho}_{k,0}\quad \meas\textnormal{-a.e.},\\
\rho_{k,1} &\leq \check{\pi}_k[\mms_\ll^2]^{-1}\,\tilde{\pi}_k[\mms_\ll^2]^{-1}\,\tilde{\rho}_{k,1} \quad\meas\textnormal{-a.e.}
\end{split}
\end{align}
However, when using these estimates to bound $\smash{-\scrT_{K,N'}^{(t)}(\pi_k)}$ from below, we change the involved densities in the integrands, but not the coupling with respect to which it is integrated; we thus have to make up for this in the coupling as well. In turn, this requires us to take into account that both $\smash{\tilde{\pi}_k}$ and $\smash{\check{\pi}_k}$ are not chro\-nological, but only close to being so for large $k\in\N$.

\textbf{Step 5.1.} We modify $\pi_k$ into a  measure with marginals 
\begin{align*}
\nu_{k,0} &:=  (1+\delta_k+\varepsilon_k+\zeta_k)^{-1}\,(\tilde{\rho}_{k,0}+ \delta_k + \varepsilon_k+\zeta_k)\,\meas_k = \varrho_{k,0}\,\meas_k,\\
\nu_{k,1} &:=  (1+\delta_k+\varepsilon_k+\zeta_k)^{-1}\,(\tilde{\rho}_{k,1}+ \delta_k + \varepsilon_k+\zeta_k)\,\meas_k = \varrho_{k,1}\,\meas_k,
\end{align*}
where $\delta_k,\varepsilon_k \in[0,1]$ and $\zeta_k \geq 0$ are defined by
\begin{align}\label{Eq:DEFS}
\begin{split}
\delta_k &:= \tilde{\pi}_k[\{\tsep=0\}],\\
\varepsilon_k &:= \check{\pi}_k[\{\tsep=0\}],\\
\zeta_k &:= W_2(\meas_k,\meas_\infty).
\end{split}
\end{align}
Note that $(\delta_k)_{k\in\N}$, $(\varepsilon_k)_{k\in\N}$, and $(\zeta_k)_{k\in\N}$ converge to $0$. Define  $\varpi_k\in\scrP(\mms^2)$ by
\begin{align*}
\varpi_k &:= (1+\delta_k+\varepsilon_k+\zeta_k)^{-1}\,\big[\tilde{\pi}_k[\mms_\ll^2]\,\check{\pi}_k[\mms_\ll^2]\,\pi_k + \tilde{\pi}_k\mres \{\tsep=0\}\\
&\qquad\qquad + (\delta_k + \varepsilon_k+\zeta_k)\,\meas_k^{\otimes 2}  + \tilde{\pi}_k[\mms_\ll^2]\,\check{\pi}_k\mres \{\tsep=0\}\big].
\end{align*}
By tracing back the definitions of $\tilde{\pi}_k$ and $\check{\pi}_k$ in Step 2, one verifies that $\varpi_k$ is a coupling, not necessarily $\smash{\ell_p}$-optimal, of $\nu_{k,0}$ and $\nu_{k,1}$. 

\textbf{Step 5.2.} By \eqref{Eq:Inequ rhok0}, we have 
\begin{alignat*}{3}
\rho_{k,0} &\leq \check{\pi}_k[\mms_\ll^2]^{-1}\,\tilde{\pi}_k[\mms_\ll^2]^{-1}\,(1+\delta_k + \varepsilon_k+\zeta_k)\,\varrho_{k,0} & \quad & \meas\textnormal{-a.e.},\\
\rho_{k,1} &\leq \check{\pi}_k[\mms_\ll^2]^{-1}\,\tilde{\pi}_k[\mms_\ll^2]^{-1}\,(1+\delta_k + \varepsilon_k+\zeta_k)\,\varrho_{k,1} & \quad & \meas\textnormal{-a.e.}
\end{alignat*}
In the sequel, to clear up notation, given $a,b\in\R$ we write  $a\geq_k b$ if there exists a sequence $(c_k)_{k\in\N}$ of positive real numbers converging to $1$ with $a\geq c_k\,b$ for every $k\in\N$. By definition of the inherent distortion coefficients, we obtain
\begin{align*}
-\scrT_{K,N'}^{(t)}(\pi_k) &= \int_{\mms^2}\tau_{K,N'}^{(1-t)}(\tsep(x^0,x^1))\,\rho_{k,0}(x^0)^{-1/N'}\d\pi_k(x^0,x^1)\\
&\qquad\qquad +\int_{\mms^2}\tau_{K,N'}^{(t)}(\tsep(x^0,x^1))\,\rho_{k,1}(x^1)^{-1/N'}\d\pi_k(x^0,x^1)\\
&\geq_k \int_{\mms^2} \tau_{K,N'}^{(1-t)}(\tsep(x^0,x^1))\,\varrho_{k,0}(x^0)^{-1/N'}\d\pi_k(x^0,x^1)\\
&\qquad\qquad + \int_{\mms^2} \tau_{K,N'}^{(t)}(\tsep(x^0,x^1))\,\varrho_{k,1}(x^1)^{-1/N'}\d\pi_k(x^0,x^1)\\
&\geq_k\int_{\mms^2} \tau_{K,N'}^{(1-t)}(\tsep(x^0,x^1))\,\varrho_{k,0}(x^0)^{-1/N'}\d\varpi_k(x^0,x^1)\\
&\qquad\qquad + \int_{\mms^2}\tau_{K,N'}^{(t)}(\tsep(x^0,x^1))\,\varrho_{k,0}(x^0)^{-1/N'}\d\varpi_k(x^0,x^1)\\
&\qquad\qquad - (1-t)\int_{\{\tsep=0\}}\varrho_{k,0}(x^0)^{-1/N'} \d\tilde{\pi}_k(x^0,x^1)\\
&\qquad\qquad -  t\int_{\{\tsep=0\}} \varrho_{k,1}(x^1)^{-1/N'}\d\tilde{\pi}_k(x^0,x^1)\\
&\qquad\qquad - (1-t)\int_{\{\tsep=0\}}\varrho_{k,0}(x^0)^{-1/N'}\d\check{\pi}_k(x^0,x^1)\\
&\qquad\qquad - t\int_{\{\tsep=0\}}\varrho_{k,1}(x^1)^{-1/N'}\d\check{\pi}_k(x^0,x^1)\\
&\qquad\qquad - c\,(\delta_k + \varepsilon_k+\zeta_k)\int_{\mms} \varrho_{k,0}(x^0)^{-1/N'}\d\meas_k(x^0)\\
&\qquad\qquad - c\,(\delta_k + \varepsilon_k+\zeta_k)\int_{\mms} \varrho_{k,1}(x^1)^{-1/N'}\d\meas_k(x^1),
\end{align*}
where, recalling our choice $K < K_k$ for large enough $k\in\N$ and \autoref{Cor:Bonnet-Myers},
\begin{align}\label{Eq:c value}
c := \max\{\sup\tau_{K,N'}^{(1-t)}\circ\tsep(\mms^2),\sup\tau_{K,N'}^{(t)}\circ\tsep(\mms^2)\}.
\end{align}
Since $\varrho_{k,0} \geq_k \delta_k + \varepsilon_k+\zeta_k$ $\meas_k$-a.e.~and $\varrho_{k,1} \geq_k \delta_k+\varepsilon_k+\zeta_k$ $\meas_k$-a.e., by definition of $\delta_k$ and $\varepsilon_k$ we obtain the estimates
\begin{align*}
(\delta_k + \varepsilon_k+\zeta_k)^{1-1/N'}&\geq_k (1-t)\int_{\{\tsep=0\}} \varrho_{k,1}(x^0)^{-1/N'}\d\tilde{\pi}_k(x^0,x^1) \\
&\qquad\qquad + t\int_{\{\tsep=0\}} \varrho_{k,1}(x^0)^{-1/N'}\d\tilde{\pi}_k(x^0,x^1),\\
(\delta_k + \varepsilon_k+\zeta_k)^{1-1/N'} &\geq_k (1-t)\int_{\{\tsep=0\}} \varrho_{k,1}(x^0)^{-1/N'}\d\check{\pi}_k(x^0,x^1)\\
&\qquad\qquad  + t\int_{\{\tsep=0\}} \varrho_{k,1}(x^0)^{-1/N'}\d\check{\pi}_k(x^0,x^1),\\
2\,(\delta_k + \varepsilon_k+\zeta_k)^{1-1/N'} &\geq_k (\delta_k+\varepsilon_k+\zeta_k)\int_{\mms}\varrho_{k,0}(x^0)^{-1/N'}\d\meas_k(x^0)\\
&\qquad\qquad + (\delta_k+\varepsilon_k+\zeta_k)\int_{\mms}\varrho_{k,1}(x^1)^{-1/N'}\d\meas_k(x^1),
\end{align*}
and therefore
\begin{align*}
-\scrT_{K,N'}^{(t)}(\pi_k) &\geq_k\int_{\mms^2} \tau_{K,N'}^{(1-t)}(\tsep(x^0,x^1))\,\varrho_{k,0}(x^0)^{-1/N'}\d\varpi_k(x^0,x^1)\\
&\qquad\qquad + \int_{\mms^2}\tau_{K,N'}^{(t)}(\tsep(x^0,x^1))\,\varrho_{k,0}(x^0)^{-1/N'}\d\varpi_k(x^0,x^1)\\
&\qquad\qquad -2\,(c+1)\,(\delta_k + \varepsilon_k+\zeta_k)^{1-1/N'}.
\end{align*}

\textbf{Step 5.3.} Let $\smash{\mathfrak{r}^0,\mathfrak{r}^1\colon \mms\to \scrP(\mms)}$   denote the ($k$-dependent)  disintegrations of $\varpi_k$ with respect to $\pr_1$ and $\pr_2$, respectively, i.e.
\begin{align*}
\rmd \varpi_k(x^0,x^1) = \rmd\mathfrak{r}^0_{x^0}(x^1)\d\nu_{k,0}(x^0) = \rmd\mathfrak{r}^1_{x^1}(x^0)\d\nu_{k,1}(x^1),
\end{align*}
and define
\begin{align}\label{Eq:Disint tau}
\begin{split}
v_0(x^0) &:= \int_{\mms} \tau_{K,N'}^{(1-t)}(\tsep(x^0,x^1))\d\mathfrak{r}_{x^0}^0(x^1),\\
v_1(x^1) &:= \int_{\mms} \tau_{K,N'}^{(t)}(\tsep(x^0,x^1))\d\mathfrak{r}_{x^1}^1(x^0).
\end{split}
\end{align}
Moreover, define $\smash{\nu_{\infty,0}^k,\nu_{\infty,1}^k\in\scrP^\ac(\mms,\meas_\infty)}$ by
\begin{align*}
\nu_{\infty,0}^k &:= (1+\delta_k + \varepsilon_k+\zeta_k)^{-1}\,(\rho_{\infty,0}^{n_k} + \delta_k +\varepsilon_k+\zeta_k)\,\meas_\infty = \varrho_{\infty,0}^k\,\meas_\infty,\\
\nu_{\infty,1}^k &:= (1+\delta_k + \varepsilon_k+\zeta_k)^{-1}\,(\rho_{\infty,1}^{n_k} + \delta_k +\varepsilon_k+\zeta_k)\,\meas_\infty = \varrho_{\infty,1}^k\,\meas_\infty,
\end{align*}
and observe that
\begin{align}\label{Eq:Observation!}
\begin{split}
\mathfrak{p}^k(\nu_{\infty,0}^k) &= \nu_{k,0} = \varrho_{k,0}\,\meas_k,\\
\mathfrak{p}^k(\nu_{\infty,1}^k) &= \nu_{k,1} = \varrho_{k,1}\,\meas_k.
\end{split}
\end{align}
Then by Jensen's inequality,
\begin{align*}
&\int_{\mms^2} \tau_{K,N'}^{(1-t)}(\tsep(x^0,x^1))\,\varrho_{k,0}(x^0)^{-1/N'}\d\varpi_k(x^0,x^1)\\
&\qquad\qquad\qquad\qquad + \int_{\mms^2}\tau_{K,N'}^{(t)}(\tsep(x^0,x^1))\,\varrho_{k,0}(x^0)^{-1/N'}\d\varpi_k(x^0,x^1)\\
&\qquad\qquad = \sum_{i=0}^1 \int_{\mms} \varrho_{k,i}(x^i)^{1-1/N'}\,v_i(x^i)\d\meas_k(x^i)\\
&\qquad\qquad  = \sum_{i=0}^1 \int_{\mms}\Big[\!\int_{\mms} \varrho_{\infty,i}^k(y^i)\d\mathfrak{p}_{x^i}^k(y^i)\Big]^{1-1/N'}v_i(x^i)\d\meas_k(x^i)\\
&\qquad\qquad \geq \sum_{i=0}^1 \int_{\mms^2} \varrho_{\infty,i}^k(y^i)^{1-1/N'}\,v_i(x^i)\d\mathfrak{p}_{x^i}^k(y^i)\d\meas_k(x^i)\\
&\qquad\qquad = \sum_{i=0}^1 \int_{\mms^2}  \varrho_{\infty,i}^k(y^i)^{1-1/N'}\,v_i(x^i)\d\mathfrak{p}_{y^i}^\infty(x^i)\d\meas_\infty(y^i).
\end{align*}

\textbf{Step 5.4.} Next, we estimate the latter sum from below. Using \eqref{Eq:Disint tau} and \eqref{Eq:Observation!},
\begin{align*}
&\int_{\mms^2} \varrho_{\infty,0}^k(y^0)^{1-1/N'}\,v_0(x^0)\d\mathfrak{p}_{y^0}^\infty(x^0)\d\meas_\infty(y^0)\\
&\qquad\qquad =\int_{\mms^4}\d\mathfrak{p}_{x^1}^k(y^1)\d\mathfrak{r}_{x^0}^0(x^1)\d\mathfrak{p}_{y^0}^\infty(x^0)\d\meas_\infty(y^0)\\
&\qquad\qquad\qquad\qquad \Big[\varrho_{\infty,0}^k(y^0)^{-1/N'}\,\tau_{K,N'}^{(1-t)}(\tsep(x^0,x^1))\,\frac{\varrho_{\infty,0}^k(y^0)\,\varrho_{\infty,1}^k(y^1)}{\varrho_{k,1}(x^1)}\Big].
\end{align*}

Let $\varepsilon > 0$. Since $\smash{\tau_{K,N'}^{(t)}\circ\tsep}$ is uniformly continuous by compactness of $\smash{\mms^2}$, our choice of $K$, and possibly invoking \autoref{Cor:Bonnet-Myers}, we fix $\delta> 0$ such that
\begin{align*}
\big\vert \tau_{K,N'}^{(1-t)}(\tsep(x^0,x^1)) - \tau_{K,N'}^{(1-t)}(\tsep(y^0,y^1))\big\vert &\leq \varepsilon,\\
\big\vert \tau_{K,N'}^{(t)}(\tsep(x^0,x^1)) - \tau_{K,N'}^{(t)}(\tsep(y^0,y^1))\big\vert &\leq \varepsilon
\end{align*}
for every  $(x^0,y^0,x^1,y^1)\in A_\delta$, where
\begin{align}\label{Eq:A_delta}
A_\delta := \{ z \in\mms^4 :  \met(z_1,z_2) + \met(z_3,z_4) \leq \delta\}.
\end{align}
Lastly, somewhat suggestively, for $k\in\N$ we define a coupling $\smash{\eta_k\in\Pi(\nu_{\infty,0}^k,\nu_{\infty,1}^k)}$, which is notably independent of $\varepsilon$, by
\begin{align*}
\rmd\eta_k(y^0,y^1) &:= \int_{\mms^2} \frac{\varrho_{\infty,0}^k(y^0)\,\varrho_{\infty,1}^k(y^1)}{\varrho_{k,0}(x^0)\,\varrho_{k,1}(x^1)}\d\mathfrak{p}_{x^0}^k(y^0)\d\mathfrak{p}_{x^1}^k(y^1)\d\varpi_k(x^0,x^1)\\
&\phantom{:}= \int_{\mms^2} \frac{\varrho_{\infty,0}^k(y^0)\,\varrho_{\infty,1}^k(y^1)}{\varrho_{k,1}(x^1)}\d\mathfrak{p}_{x^1}^k(y^1)\d\mathfrak{r}_{x^0}^0(x^1)\d\mathfrak{p}_{y^0}^\infty(x^0) \d\meas_\infty(y^0)\\
&\phantom{:}= \int_{\mms^2} \frac{\varrho_{\infty,0}^k(y^0)\,\varrho_{\infty,1}^k(y^1)}{\varrho_{k,0}(x^0)}\d\mathfrak{p}_{x^0}^k(y^0)\d\mathfrak{r}_{x^1}^1(x^0)\d\mathfrak{p}_{y^1}^\infty(x^1)\d\meas_\infty(y^1).
\end{align*}
The previous computations then entail
\begin{align*}
&\int_{\mms^2} \varrho_{\infty,0}^k(y^0)^{1-1/N'}\,v_0(x^0)\d\mathfrak{p}_{y^0}^\infty(x^0)\d\meas_\infty(y^0)\\
&\qquad\qquad \geq -\varepsilon+\int_{A_\delta} \d\mathfrak{p}_{x^1}^k(y^1)\d\mathfrak{r}_{x^0}^0(x^1)\d\mathfrak{p}_{y^0}^\infty(x^0)\d\meas_\infty(y^0)\\
&\qquad\qquad\qquad\qquad\qquad\qquad \Big[\varrho_{\infty,0}^k(y^0)^{-1/N'}\,\tau_{K,N'}^{(1-t)}(\tsep(y^0,y^1))\,\frac{\varrho_{\infty,0}^k(y^0)\,\varrho_{\infty,1}^k(y^1)}{\varrho_{k,1}(x^1)}\Big]\\
&\qquad\qquad\qquad\qquad + \int_{A_\delta^\sfc} \d\mathfrak{p}_{x^1}^k(y^1)\d\mathfrak{r}_{x^0}^0(x^1)\d\mathfrak{p}_{y^0}^\infty(x^0)\d\meas_\infty(y^0)\\
&\qquad\qquad\qquad\qquad\qquad\qquad\Big[\varrho_{\infty,0}^k(y^0)^{-1/N'}\,\tau_{K,N'}^{(1-t)}(\tsep(x^0,x^1))\,\frac{\varrho_{\infty,0}^k(y^0)\,\varrho_{\infty,1}^k(y^1)}{\varrho_{k,1}(x^1)}\Big]\\
&\qquad\qquad = - \varepsilon + \int_{\mms^2} \varrho_{\infty,0}^k(y^0)^{-1/N'}\,\tau_{K,N'}^{(1-t)}(\tsep(y^0,y^1))\d\eta_k(y^0,y^1)\\
&\qquad\qquad\qquad\qquad + \int_{A_\delta^\sfc} \d\mathfrak{p}_{x^1}^k(y^1)\d\mathfrak{r}_{x^0}^0(x^1)\d\mathfrak{p}_{y^0}^\infty(x^0)\d\meas_\infty(y^0)\\
&\qquad\qquad\qquad\qquad\qquad\qquad \Big[\varrho_{\infty,0}^k(y^0)^{-1/N'}\,\big[\tau_{K,N'}^{(1-t)}(\tsep(x^0,x^1))- \tau_{K,N'}^{(1-t)}(\tsep(y^0,y^1))\big]\\
&\qquad\qquad\qquad\qquad\qquad\qquad \times \frac{\varrho_{\infty,0}^k(y^0)\,\varrho_{\infty,1}^k(y^1)}{\varrho_{k,1}(x^1)}\Big].
\end{align*}

\textbf{Step 5.5.} Let us estimate the latter error term. We first recall the definition $\zeta_k := W_2(\meas_k,\meas_\infty)$ from \eqref{Eq:DEFS}. By definition of $A_\delta$ and the value $c$ from \eqref{Eq:c value}, the $\meas$-a.e.~valid inequality for the density $\smash{\rho_{\infty,1}^{n_k}}$ involved in the definition of $\smash{\varrho_{\infty,1}^k}$ from \autoref{Le:CM lemma}, as well as Hölder's inequality,
\begin{align*}
&\Big\vert\!\int_{A_\delta^\sfc} \d\mathfrak{p}_{x^1}^k(y^1)\d\mathfrak{r}_{x^0}^0(x^1)\d\mathfrak{p}_{y^0}^\infty(x^0)\d\meas_\infty(y^0)\\
&\qquad\qquad\qquad\qquad \Big[\varrho_{\infty,0}^k(y^0)^{-1/N'}\,\big[\tau_{K,N'}^{(1-t)}(\tsep(x^0,x^1))- \tau_{K,N'}^{(1-t)}(\tsep(y^0,y^1))\big]\\
&\qquad\qquad\qquad\qquad \times \frac{\varrho_{\infty,0}^k(y^0)\,\varrho_{\infty,1}^k(y^1)}{\varrho_{k,1}(x^1)}\Big]\Big\vert\\
&\qquad\qquad \leq 2\,c\,\delta^{-1}\!\int_{\mms^4} \d\mathfrak{p}_{x^1}^k(y^1)\d\mathfrak{r}_{x^0}^0(x^1)\d\mathfrak{p}_{y^0}^\infty(x^0)\d\meas_\infty(y^0)\\
&\qquad\qquad\qquad\qquad\Big[\varrho_{\infty,0}^k(y^0)^{-1/N'}\,\big[\met(x^0,y^0) + \met(x^1,y^1)\big]\,\frac{\varrho_{\infty,0}^k(y^0)\,\varrho_{\infty,1}^k(y^1)}{\varrho_{k,1}(x^1)}\Big]\\
&\qquad\qquad \leq_k 2\,c\,\delta^{-1}\!\int_{\mms^2}\varrho_{\infty,0}^k(y^0)^{1-1/N'}\,\met(x^0,y^0)\d\mathfrak{q}_k(x^0,y^0)\\
&\qquad\qquad\qquad\qquad + 2\,c\,\delta^{-1}\,\zeta_k^{-1/N'}\!\int_{\mms} \rho_{\infty,1}(y^1)\,\met(x^1,y^1)\d\mathfrak{q}_k(x^1,y^1)\\
&\qquad\qquad \leq_k 2\,c\,\delta^{-1}\,\big\Vert\rho_{\infty,0}\big\Vert_{\Ell^\infty(\mms,\meas_\infty)}^{1-1/N'}\,\zeta_k + 2\,c\,\delta^{-1}\,\Vert\rho_{\infty,1}\Vert_{\Ell^\infty(\mms,\meas_\infty)}\,\zeta_k^{1-1/N'}\!.
\end{align*}

\textbf{Step 5.6.} Taking the estimates from Step 5.4 and Step 5.5 together, we get
\begin{align*}
&\int_{\mms^2}\varrho_{\infty,0}^k(y^0)^{1-1/N'}\,v_0(x^0)\d\mathfrak{p}_{y^0}^\infty(x^0)\d\meas_\infty(y^0)\\
&\qquad\qquad \geq_k -\varepsilon + \int_{\mms^2}\varrho_{\infty,0}^k(y^0)^{-1/N'}\,\tau_{K,N'}^{(1-t)}(\tsep(y^0,y^1))\d\eta_k(y^0,y^1)\\
&\qquad\qquad\qquad\qquad - 2\,c\,\delta^{-1}\,\big\Vert\rho_{\infty,0}\big\Vert_{\Ell^\infty(\mms,\meas_\infty)}^{1-1/N'}\,\zeta_k\\
&\qquad\qquad\qquad\qquad - 2\,c\,\delta^{-1/N'}\,\Vert\rho_{\infty,1}\Vert_{\Ell^\infty(\mms,\meas_\infty)}\,\zeta_k^{1-1/N'}\!.
\end{align*}
Analogously, for the second summand in the last part of Step 5.3 we obtain
\begin{align*}
&\int_{\mms^2} \varrho_{\infty,1}^k(y^1)^{1-1/N'}\,v_1(x^1)\d\mathfrak{p}_{y^1}^\infty(x^1)\d\meas_\infty(y^1)\\
&\qquad\qquad \geq_k -\varepsilon + \int_{\mms^2}\varrho_{\infty,1}^k(y^1)^{-1/N'}\,\tau_{K,N'}^{(t)}(\tsep(y^0,y^1))\d\eta_k(y^0,y^1)\\
&\qquad\qquad\qquad\qquad - 2\,c\,\delta^{-1}\,\big\Vert\rho_{\infty,1}\big\Vert_{\Ell^\infty(\mms,\meas_\infty)}^{1-1/N'}\,\zeta_k\\
&\qquad\qquad\qquad\qquad - 2\,c\,\delta^{-1}\,\Vert\rho_{\infty,0}\Vert_{\Ell^\infty(\mms,\meas_\infty)}\,\zeta_k^{1-1/N'}\!.
\end{align*}

\textbf{Step 6.} \textit{Passage to the limit.} Let $(a_k)_{k\in\N}$ be a given sequence of normalization constants in $[1,\infty)$ as provided by \autoref{Le:CM lemma}, i.e. 
\begin{align*}
\rho_{\infty,0}^{n_k} &\leq a_k\,\rho_{\infty,0}\quad\meas\textnormal{-a.e.},\\
\rho_{\infty,1}^{n_k} &\leq a_k\,\rho_{\infty,1}\quad\meas\textnormal{-a.e.}
\end{align*}

\textbf{Step 6.1.} For $k\in\N$, we define $\smash{\tilde{\nu}_{\infty,0}^k,\tilde{\nu}_{\infty,1}^k}\in\smash{\scrP^\ac(\mms,\meas_\infty)}$ by
\begin{align*}
\tilde{\nu}_{\infty,0}^k &:=(a_k^2+\delta_k+\varepsilon_k+\zeta_k)^{-1}\,(a_k^2\,\rho_{\infty,0} + \delta_k+\varepsilon_k+\zeta_k)\,\meas_\infty = \tilde{\varrho}_{\infty,0}^k\,\meas_\infty,\\
\tilde{\nu}_{\infty,1}^k &:= (a_k^2+\delta_k+\varepsilon_k+\zeta_k)^{-1}\,(a_k^2\,\rho_{\infty,1} + \delta_k+\varepsilon_k+\zeta_k)\,\meas_\infty = \tilde{\varrho}_{\infty,1}^k\,\meas_\infty.
\end{align*}
We turn $\eta_k$ into a coupling $\smash{\alpha_k \in\Pi(\nu_{\infty,0}^k,\nu_{\infty,1}^k)}$ by setting
\begin{align*}
\alpha_k &:= (a_k^2+\delta_k+\varepsilon_k+\zeta_k)^{-1} \big[(1+\delta_k+\varepsilon_k+\zeta_k)\,\eta_k\\
&\qquad\qquad + a_k^2\,\mu_{\infty,0}\otimes\mu_{\infty,1} - \tilde{\mu}_{k,0}\otimes\tilde{\mu}_{k,1}\big].
\end{align*}
This construction yields
\begin{align*}
&\int_{\mms^2}\varrho_{k,0}^k(y^0)^{-1/N'}\,\tau_{K,N'}^{(1-t)}(\tsep(y^0,y^1))\d\eta_k(y^0,y^1)\\
&\qquad\qquad\qquad\qquad + \int_{\mms^2}\varrho_{k,1}^k(y^1)^{-1/N'}\,\tau_{K,N'}^{(t)}(\tsep(y^0,y^1))\d\eta_k(y^0,y^1)\\
&\qquad\qquad \geq_k \int_{\mms^2} \tilde{\varrho}_{\infty,0}^k(y^0)^{-1/N'}\,\tau_{K,N'}^{(1-t)}(\tsep(y^0,y^1))\d\alpha_k(y^0,y^1)\\
&\qquad\qquad\qquad\qquad - c\int_\mms \varrho_{\infty,0}^k(y^0)^{-1/N'}\,\big\vert \rho_{\infty,0}(y^0) - \tilde{\rho}_{k,0}(y^0)\big\vert\d\meas_\infty(y^0)\\
&\qquad\qquad\qquad\qquad + \int_{\mms^2}\tilde{\varrho}_{\infty,1}^k(y^1)^{-1/N'}\,\tau_{K,N'}^{(t)}(\tsep(y^0,y^1))\d\alpha_k(y^0,y^1)\\
&\qquad\qquad\qquad\qquad - c\int_\mms \varrho_{\infty,1}^k(y^1)^{-1/N'}\,\big\vert \rho_{\infty,1}(y^1) - \tilde{\rho}_{k,1}(y^1)\big\vert\d\meas_\infty(y^1).
\end{align*}

\textbf{Step 6.2.} Taking together the above estimates with those obtained in Step 5 and invoking the $\Ell^1$-convergence asserted in  \autoref{Le:CM lemma}, we obtain
\begin{align*}
-\limsup_{k\to\infty} \scrT_{K,N'}^{(t)}(\pi_k) \geq - \limsup_{k\to\infty}\scrT_{K,N'}^{(t)}(\alpha_k) - 2\varepsilon.
\end{align*}
Now, by Prokhorov's theorem, the sequence $(\alpha_k)_{k\in\N}$ converges weakly to some $\alpha\in\Pi(\mu_{\infty,0},\mu_{\infty,1})$ up to a nonrelabeled subsequence. Since $\varepsilon$ was arbitrary and did not influence the construction of $\alpha_k$, by \autoref{Le:Const perturb} we thus get
\begin{align}\label{Eq:Eqeqeqeq}
-\limsup_{k\to\infty}\scrT_{K,N'}^{(t)}(\pi_k)\geq -\scrT_{K,N'}^{(t)}(\alpha).
\end{align}

\textbf{Step 6.3.} Now we send $k\to\infty$ in \eqref{Eq:TCD COND INVOKING}. For every $k\in\N$, there exists some plan $\smash{\bdpi_k\in\OptTGeo_{\ell_p}^\tsep(\mu_{k,0},\mu_{k,1})}$ representing the curve $(\mu_{k,t})_{t\in[0,1]}$. \autoref{Le:Villani lemma for geodesic} allows us to extract a nonrelabeled subsequence of $\smash{(\bdpi_k)_{k\in\N}}$ converging weakly to some $\smash{\bdpi\in\OptTGeo_{\ell_p}^\tsep(\mu_{\infty,0},\mu_{\infty,1})}$. The assignment $\mu_{\infty,t} := (\eval_t)_\push\bdpi_\infty$ produces a timelike proper-time parametrized $\smash{\ell_p}$-geodesic $(\mu_{\infty,t})_{t\in[0,1]}$ connecting $\mu_{\infty,0}$ to $\mu_{\infty,1}$. Joint lower semicontinuity of Rényi's entropy on $\scrP(\mms)$ by compactness of $\mms$ and \eqref{Eq:Eqeqeqeq} thus yield, for every $t\in[0,1]$ and every $N'\geq N$,
\begin{align}\label{Eq:Renyi inequ}
\scrS_{N'}^\infty(\mu_{\infty,t})\leq \limsup_{k\to\infty} \scrS_{N'}^k(\mu_{k,t}) \leq \limsup_{k\to\infty}\scrT_{K,N'}^{(t)}(\pi_k) \leq \scrT_{K,N'}^{(t)}(\alpha).
\end{align}

\textbf{Step 7.} \textit{Proof of the $\smash{\ell_p}$-optimality of $\alpha$.} To conclude the desired property leading towards $\smash{\TCD_p(K,N)}$ for $\scrX_\infty$, at least under the restrictions of Step 2, we have to prove the causality and the  $\smash{\ell_p}$-optimality of $\alpha$. Both is argued similarly to Step 5.4; we concentrate on the proof of the estimate entailing $\smash{\ell_p}$-optimality later, and then outline how to modify this argument to prove $\smash{\alpha[\mms_\ll^2]=1}$.

Given any $\varepsilon>0$, fix $\delta > 0$ with the property
\begin{align*}
\big\vert\tsep(x^0,x^1) - \tsep(y^0,y^1)\big\vert \leq \varepsilon
\end{align*}
for every $(x^0,y^0,x^1,y^1)\in\mms^4$  belonging to the set $\smash{A_\delta}$ from \eqref{Eq:A_delta}. As in Step 5.5 and tracing back the definitions of all considered couplings, we obtain
\begin{align*}
&\int_{\mms^2}\tsep^p(y^0,y^1)\d\alpha(y^0,y^1) = \lim_{k\to\infty}\int_{\mms^2}\tsep^p(y^0,y^1)\d\alpha_k(y^0,y^1)\\
&\qquad\qquad \geq \liminf_{k\to\infty} \int_{\mms^2}\tsep^p(y^0,y^1)\d\eta_k(y^0,y^1)\\
&\qquad\qquad \geq \liminf_{k\to\infty} \int_{\mms^2} \tsep^p(x^0,x^1)\d\varpi_k(x^0,x^1) - \varepsilon\\
&\qquad\qquad\qquad\qquad -2\,c\,\delta^{-1}\,\Vert\rho_{\infty,0}\Vert_{\Ell^\infty(\mms,\meas_\infty)}\limsup_{k\to\infty} W_2(\meas_k,\meas_\infty)\\
&\qquad\qquad\qquad\qquad -2\,c\,\delta^{-1}\,\Vert\rho_{\infty,1}\Vert_{\Ell^\infty(\mms,\meas_\infty)}\limsup_{k\to\infty} W_2(\meas_k,\meas_\infty)\\
&\qquad\qquad \geq \liminf_{k\to\infty}\int_{\mms^2}\tsep^p(x^0,x^1)\d\pi_k(x^0,x^1)-\varepsilon\\
&\qquad\qquad = \liminf_{k\to\infty} \int_{\mms^2}\tsep^p(x^0,x^1)\d\bar{\pi}_k(x^0,x^1)-\varepsilon\\
&\qquad\qquad \geq \liminf_{k\to\infty}\int_{\mms^2}\tsep^p(x^0,x^1)\d\check{\pi}_k(x^0,x^1)-\varepsilon\\
&\qquad\qquad\geq \ell_p(\mu_{\infty,0},\mu_{\infty,1})^p-\varepsilon.
\end{align*}
In the equality in the third last step, we have used that both $\smash{\pi_k,\bar{\pi}_k\in\Pi_\ll(\mu_{k,0},\mu_{k,1})}$ are $\smash{\ell_p}$-optimal; the last inequality is due to \eqref{Eq:Lp calc}.

The relation $\smash{\alpha[\mms_\ll^2]=1}$ is argued similarly by replacing, given $\varepsilon > 0$, $\tsep$ by a nonnegative function $\phi\in\Cont(\mms^2)$ obeying $\smash{\phi(\mms_\ll^2)=\{1\}}$, $\sup\phi(\mms)\leq 1$, and
\begin{align*}
\alpha[\mms_\ll^2] \geq \int_{\mms^2}\phi\d\alpha - \varepsilon.
\end{align*}
Both results together and the arbitrariness of $\varepsilon$ yield the claim.

\textbf{Step 8.} \textit{Relaxation of the assumptions on $\mu_{\infty,0}$ and $\mu_{\infty,1}$.} Now we outline how to get rid of the assumption $\smash{\rho_{\infty,0},\rho_{\infty,1}\in\Ell^\infty(\mms,\meas_\infty)}$ from Step 2. 

If $\smash{\rho_{\infty,0}}$ and $\smash{\rho_{\infty,1}}$ are not $\meas$-essentially bounded, given any $i\in\N$ we set
\begin{align*}
E_i := \{\rho_{\infty,0} \leq i,\, \rho_{\infty,1}\leq i\}.
\end{align*}
Let $\smash{\pi\in\Pi_\ll(\mu_{\infty,0},\mu_{\infty,1})}$ be an $\smash{\ell_p}$-optimal coupling, and set
\begin{align*}
\pi_i := \pi[E_i]^{-1}\,\pi\mres E_i
\end{align*}
provided $\pi[E_i]>0$. By restriction \cite[Lem.~2.10]{cavalletti2020}, $\pi_i$ is an $\smash{\ell_p}$-optimal coupling of its marginals $\smash{\lambda_{\infty,0}^i,\lambda_{\infty,1}^i\in\scrP^\ac(\mms,\meas_\infty)}$; moreover, the pair $\smash{(\lambda_{\infty,0}^i,\lambda_{\infty,1}^i)}$ is strongly timelike $p$-dualizable.  Let $\smash{\dot{\pi}_i\in\Pi_\ll(\lambda_{\infty,0}^i,\lambda_{\infty,1}^i)}$ be an $\smash{\ell_p}$-optimal coupling as constructed in the previous steps. Define $\smash{\beta_i\in\Pi_\ll(\mu_{\infty,0},\mu_{\infty,1})}$ by
\begin{align*}
\beta_i := \pi[E_i]\,\dot{\pi}_i + \pi\mres E_i^\sfc,
\end{align*}
and note that $\beta_i$ is an $\smash{\ell_p}$-optimal coupling of its marginals. Moreover, we observe that, for every $t\in[0,1]$ and every $N'\geq N$,
\begin{align*}
\scrT_{K,N'}^{(t)}(\dot{\pi}_i) &\leq_i -\int_{\mms^2} \rho_{\infty,0}(y^0)^{-1/N'}\,\tau_{K,N'}^{(1-t)}(\tsep(y^0,y^1))\d\dot{\pi}_i(y^0,y^1)\\
&\qquad\qquad -\int_{\mms^2} \rho_{\infty,1}(y^1)^{-1/N'}\,\tau_{K,N'}^{(1-t)}(\tsep(y^0,y^1))\d\dot{\pi}_i(y^0,y^1)\\
&\leq_i -\int_{\mms^2} \rho_{\infty,0}(y^0)^{-1/N'}\,\tau_{K,N'}^{(1-t)}(\tsep(y^0,y^1))\d\beta_i(y^0,y^1)\\
&\qquad\qquad + c \int_{E_i^\sfc} \rho_{\infty,0}(y^0)^{-1/N'}\d\pi(y^0,y^1)\\
&\qquad\qquad -\int_{\mms^2}\rho_{\infty,1}(y^1)^{-1/N'}\,\tau_{K,N'}^{(t)}(\tsep(y^0,y^1))\d\beta_i(y^0,y^1)\\
&\qquad\qquad + c\int_{E_i^\sfc}\rho_{\infty,1}(y^1)^{-1/N'}\d\pi(y^0,y^1).
\end{align*}

By Prokhorov's theorem, $(\beta_i)_{i\in\N}$ converges weakly to some $\smash{\beta\in\Pi(\mu_{\infty,0},\mu_{\infty,1})}$ which, thanks to Portmanteau's theorem, is a causal coupling. By stability \cite[Thm. 2.14]{cavalletti2020}, we thus infer the $\smash{\ell_p}$-optimality of $\beta$. Together with \autoref{Le:Const perturb} and Hölder's inequality, this addresses the right-hand side of the desired inequality; the left-hand side is treated with the  compactness argument from Step 6.3.

\textbf{Step 9.} \textit{Passing from $K$ and $N$ to $K_\infty$ and $N_\infty$.} Lastly, the restrictions from Step 4 the previous arguments imply $\smash{\TCD_p(K,N)}$ for $\scrX_\infty$ for every $K<K_\infty$ and every $N>N_\infty$. Given $\smash{(\mu_{\infty,0},\mu_{\infty,1})\in\scrP^\ac(\mms,\meas_\infty)}$ strongly timelike $p$-dualizable, letting $K$ and $N$ approach $K_\infty$ and $N_\infty$, respectively, and using compactness arguments as in Step 6.2 and Step 6.3 gives the existence of a proper-time parametrized $\smash{\ell_p}$-geodesic $(\mu_{\infty,t})_{t\in[0,1]}$ from $\mu_{\infty,0}$ to $\mu_{\infty,1}$ and a timelike $p$-dualizing coupling $\smash{\pi\in\Pi_\ll(\mu_{\infty,0},\mu_{\infty,1})}$ witnessing the desired inequality for the Rényi entropy for every $K < K_\infty$ and every $N' > N_\infty$. The semicontinuity properties of both sides in the respective parameters yields $\smash{\TCD_p(K_\infty,N_\infty)}$. The proof is terminated.
\end{proof}

\begin{remark}\label{RE:POTEN} Very roughly speaking, the relevant convergences  in the above proof are justified by \emph{uniform continuity}, recall e.g.~Step 5.4. This is considerably weaker than \emph{Lipschitz continuity} of the quantity inside the distortion coefficients, which has been used in the proof of \cite[Thm.~3.1]{sturm2006b} in the metric case at a similar point, yet is useless here since $\tsep$ is no distance. Our proof thus has the potential to admit straightforward adaptations to settings with continuous potentials other than $\tsep$ which do not come from a distance either.
\end{remark}

\begin{remark}\label{Re:A remarkable} A byproduct of the above proof is that $\smash{\TCD_p^*(K,N)}$ and $\smash{\TCD_p(K,N)}$ are weakly stable under convergence of \emph{normalized} Lorentzian pre-length spaces with respect to the following Lorentzian modification of Sturm's transport distance $\boldsymbol{\mathsf{D}}$ \cite[Def. 3.2]{sturm2006a}: it is still defined in terms of $\met$, but the inherent metric measure isometric embeddings should also respect the given Lorentzian structures, and are required to map into spaces obeying the regularity assumptions from item \ref{La:Blurr1reg} of \autoref{Def:Convergence}.
\end{remark}

\subsection{Good geodesics}\label{Sub:Good TCD} In this section, following \cite{braun2022}, see also  \cite{cavalletti2017,rajala2012a,rajala2012b}, we show the existence of timelike proper-time parametrized $\smash{\ell_p}$-geodesics with $\meas$-densities uniformly in $\Ell^\infty$ in time under mild assumptions on their endpoints. This treatise does not require any nonbranching condition on the underlying space. The proofs of the results to follow are mostly analogous to those of \cite{braun2022} and hence only outlined. 

We regard the corresponding result from \autoref{Th:Good geos TCD} as a key to develop the notion of \emph{timelike weak gradients} on Lorentzian geodesic spaces by using so-called \emph{test plans} --- many of which should exist by \autoref{Th:Good geos TCD} --- following \cite{ambrosio2014a}, and to prove a Lorentzian analogue of the \emph{Sobolev-to-Lipschitz property} \cite[Subsec.~4.1.3]{gigli2013}.

The crucial result providing the critical threshold for \eqref{Eq:Thresh} is the following.

\begin{lemma}\label{Le:Lalelu} Assume $\smash{\wTCD_p^*(K,N)}$ for  $p\in(0,1)$, $K\in\R$, and $N\in[1,\infty)$. Let $\smash{(\mu_0,\mu_1)=(\rho_0\,\meas,\rho_1\,\meas)\in\scrP_\comp^\ac(\mms,\meas)}$ be strongly timelike $p$-dualizable with $\rho_0,\rho_1\in\Ell^\infty(\mms,\meas)$. Let  $D$ be any real number no smaller than $\sup\tsep(\supp\mu_0\times\supp\mu_1)$. Then there exists a timelike proper-time parametrized  $\smash{\ell_p}$-geodesic $(\mu_t)_{t\in[0,1]}$ from $\mu_0$ to $\mu_1$ such that $\mu_t=\rho_t\,\meas\in\Dom(\Ent_\meas)$ for every $t\in[0,1]$, and
\begin{align*}
\meas\big[\{\rho_{1/2}>0\}\big] \geq \rme^{-D\sqrt{K^-N}/2}\,\max\{ \Vert\rho_0\Vert_{\Ell^\infty(\mms,\meas)},\Vert\rho_1\Vert_{\Ell^\infty(\mms,\meas)}\}^{-1}.
\end{align*}
\end{lemma}

\begin{proof} Recall that $\smash{\wTCD_p^*(K,N)}$ implies $\smash{\wTCD_p^*(K^-,N)}$ by \autoref{Pr:Consistency TCD}. Let $\smash{(\mu_t)_{t\in[0,1]}}$ be a timelike proper-time parametrized $\smash{\ell_p}$-geodesic from $\mu_0$ to $\mu_1$ and $\smash{\pi\in\Pi_\ll(\mu_0,\mu_1)}$ be some $\smash{\ell_p}$-optimal coupling obeying the inequality defining the condition $\smash{\wTCD_p^*(K^-,N)}$ as in  \autoref{Def:TCD*}. Thanks to \autoref{Le:Stu}, we obtain $\mu_t\in\Dom(\Ent_\meas)$ for every $t\in[0,1]$; in particular, we have  $\smash{\mu_t = \rho_t\,\meas\in\scrP_\comp^\ac(\mms,\meas)}$. Moreover, by \autoref{Re:Lower bounds sigma},
\begin{align*}
\scrS_N(\mu_{1/2})&\leq -\int_{\mms^2} \sigma_{K^-,N}^{(1/2)}(\tsep(x^0,x^1))\,\rho_0(x^0)^{-1/N}\d\pi(x^0,x^1)\\
&\qquad\qquad -\int_{\mms^2} \sigma_{K^-,N}^{(1/2)}(\tsep(x^0,x^1))\,\rho_1(x^1)^{-1/N}\d\pi(x^0,x^1)\\
&\leq -\rme^{-D\sqrt{K^-/N}/2}\,\max\{\Vert\rho_0\Vert_{\Ell^\infty(\mms,\meas)}, \Vert\rho_1\Vert_{\Ell^\infty(\mms,\meas)}\}^{-1/N}.
\end{align*}
On the other hand, we have $\smash{\scrS_N(\mu_{1/2}) \geq -\meas[\{\rho_{1/2}>0\}]^{1/N}}$ by Jensen's inequality, and rearranging terms provides the claim.
\end{proof}

\begin{theorem}\label{Th:Good geos TCD} Under the hypotheses of \autoref{Le:Lalelu}, there exists a timelike proper-time parametrized $\smash{\ell_p}$-geodesic $(\mu_t)_{t\in[0,1]}$ from $\mu_0$ to $\mu_1$ such that for every $t\in[0,1]$, $\smash{\mu_t = \rho_t\,\meas\in\scrP_\comp^\ac(\mms,\meas)}$ and
\begin{align}\label{Eq:Thresh}
\Vert\rho_t\Vert_{\Ell^\infty(\mms,\meas)} \leq \rme^{D\sqrt{K^-N}/2}\,\max\{\Vert\rho_0\Vert_{\Ell^\infty(\mms,\meas)},\Vert\rho_1\Vert_{\Ell^\infty(\mms,\meas)}\},
\end{align}
where
\begin{align*}
D := \sup\tsep(\supp\mu_0\times\supp\mu_1).
\end{align*}
\end{theorem}

\begin{proof} We construct the required geodesic at dyadic times in $[0,1]$ by a bisection argument and then define the rest of the geodesic by completion. By induction, given $n\in\N_0$ we assume that measures $\smash{\bdalpha^0, \dots, \bdalpha^n\in\OptTGeo_{\ell_p}^\tsep(\mu_0,\mu_1)}$ have already been constructed with the following properties. For every $\smash{k\in\{0,\dots,2^n\}}$, the pair $\smash{((\eval_{(k-1)2^{-n}})_\push\bdalpha^n,(\eval_{k2^{-n}})_\push\bdalpha^n)\in\scrP_\comp^\ac(\mms,\meas)^2}$ is strongly timelike $p$-dualizable, and
\begin{align*}
\sup\tsep(\supp\,(\eval_{(k-1)2^{-n}})_\push\bdalpha^n\times\supp\,(\eval_{k2^{-n}})_\push\bdalpha^n) \leq 2^{-n}\,D.
\end{align*}

\textbf{Step 1.} \textit{Minimization of an appropriate functional.} Set
\begin{align*}
c_{n+1} := \rme^{2^{-n-2}D\sqrt{K^-N}}\,\max\{\Vert\rho_0\Vert_{\Ell^\infty(\mms,\meas)},\Vert\rho_1\Vert_{\Ell^\infty(\mms,\meas)}\}
\end{align*}
and define the functional $\smash{\scrF_{c_{n+1}}\colon \scrP(\mms)\to [0,1]}$ by
\begin{align}\label{Eq:Functional FC}
\scrF_{c_{n+1}}(\mu) := \big\Vert (\rho- c_{n+1})^+\big\Vert_{\Ell^1(\mms,\meas)} + \mu_\perp[\mms].
\end{align}
subject to the decomposition $\mu = \rho\,\meas + \mu_\perp$. It measures how far $\rho$ deviates from being bounded from above by $c$, and how much the input $\mu$ fails to be $\meas$-absolutely continuous. Let $\smash{k\in\{1,\dots,2^{n+1}-1\}}$ be odd. As for \cite[Lem.~3.13]{braun2022}, we prove that $\smash{\scrF_{c_{n+1}}\circ (\eval_{1/2})_\push}$ admits a minimizer  $\smash{\bdpi_k^{n+1}\in \OptTGeo_{\ell_p}^\tsep(\mu_{(k-1)2^{-n-1}},\mu_{(k+1)2^{-n-1}}}$). By gluing, we find a measure $\smash{\bdalpha^{n+1}\in\OptTGeo_{\ell_p}^\tsep(\mu_0,\mu_1)}$ such that
\begin{align*}
(\Restr_{(k-1)2^{-n-1}}^{(k+1)2^{-n-1}})_\push\bdalpha^{n+1} = \bdpi_k^{n+1}
\end{align*}
for every $k$ as above. Here, given $s,t\in [0,1]$ with $s<t$, $\smash{\Restr_s^t}$ is defined in \eqref{Eq:Restr def}.

For every odd $k\in\{1,\dots,2^{n+1}-1\}$, the pairs 
\begin{align*}
&((\eval_{(k-1)2^{-n-1}})_\push\bdalpha^{n+1},(\eval_{k2^{-n-1}})_\push\bdalpha^{n+1}),\\
&((\eval_{k2^{-n-1}})_\push\bdalpha^{n+1},(\eval_{(k+1)2^{-n-1}})_\push\bdalpha^{n+1})
\end{align*}
are both strongly timelike $p$-dualizable \cite[Lem.~3.1]{braun2022}. Furthermore, since $\smash{\bdpi_k^{n+1}}$ is concentrated on curves in $\TGeo^\tsep(\mms)$ with $\tsep$-length at most $2^{-n}\,D$, we obtain
\begin{align*}
\sup\tsep(\supp\,(\eval_{(k-1)2^{-n-1}})_\push\bdalpha^{n+1}\times\supp\,(\eval_{k2^{-n-1}})_\push\bdalpha^{n+1}) &\leq 2^{-n-1}\,D,\\
\sup\tsep(\supp\,(\eval_{k2^{-n-1}})_\push\bdalpha^{n+1}\times\supp\,(\eval_{(k+1)2^{-n-1}})_\push\bdalpha^{n+1})&\leq 2^{-n-1}\,D.
\end{align*}

After Step 2, this will allow us to iteratively produce a sequence $(\bdalpha^n)_{n\in\N}$ in $\smash{\OptTGeo_{\ell_p}^\tsep(\mu_0,\mu_1)}$ according to the preceding scheme.

\textbf{Step 2.} \textit{Uniform density bound.} Given any $n\in\N_0$, any $c'>c_{n+1}$, and any odd  $k\in\smash{\{1,\dots,2^{n+1}-1\}}$, arguing as in \cite[Prop.~3.14]{braun2022} and using \autoref{Le:Lalelu} we get
\begin{align}\label{Eq:Fc inequality}
\begin{split}
&\inf\{ \scrF_{c'}((\eval_{1/2})_\push\bdpi) : \bdpi\in \OptTGeo_{\ell_p}^\tsep((\eval_{(k-1)2^{-n-1}})_\push\bdalpha^{n+1},\\
&\qquad\qquad (\eval_{(k+1)2^{-n-1}})_\push\bdalpha^{n+1}) \}=0.
\end{split}
\end{align}
The support of all considered measures belongs to the compact set $J(\mu_0,\mu_1)$; thus, using \cite[Cor.~3.15]{braun2022} we get \eqref{Eq:Fc inequality} for $c'$ replaced by $c_{n+1}$. Thus, $\smash{(\eval_{k2^{-n-1}})_\push\bdalpha^{n+1}} = \rho_{k2^{-n-1}}\,\meas\in \smash{\scrP_\comp^\ac(\mms,\meas)}$ and, computing the geometric sum in the exponent,
\begin{align}\label{Eq:DENS BD}
\Vert\rho_{k2^{-n-1}}\Vert_{\Ell^\infty(\mms,\meas)} \leq \rme^{D\sqrt{K^-N}/2}\,\max\{\Vert\rho_0\Vert_{\Ell^\infty(\mms,\meas)}, \Vert\rho_1\Vert_{\Ell^\infty(\mms,\meas)}\}.
\end{align}

\textbf{Step 3.} \textit{Completion.} Let $\smash{(\bdalpha^n)_{n\in\N}}$ be the sequence in $\smash{\OptTGeo_{\ell_p}^\tsep(\mu_0,\mu_1)}$ iteratively constructed according to Step 1 and Step 2. Using \autoref{Le:Villani lemma for geodesic}, this sequence converges weakly, up to a  subsequence, to some timelike $\smash{\ell_p}$-optimal geodesic plan $\smash{\bdalpha\in\OptTGeo_{\ell_p}^\tsep(\mu_0,\mu_1)}$. In particular, the assignment $\mu_t := (\eval_t)_\push\bdalpha$ defines a timelike proper-time parametrized $\smash{\ell_p}$-geodesic $(\mu_t)_{t\in[0,1]}$ from $\mu_0$ to $\mu_1$. Since the functional $\scrF_c$ given in terms of the threshold
\begin{align*}
c := \rme^{D\sqrt{K^-N}/2}\,\max\{\Vert\rho_0\Vert_{\Ell^\infty(\mms,\meas)}, \Vert\rho_1\Vert_{\Ell^\infty(\mms,\meas)}\}
\end{align*}
analogously to \eqref{Eq:Functional FC} is weakly lower semicontinuous in $\scrP(J(\mu_0,\mu_1))$ \cite[Lem.~3.13]{braun2022}, \eqref{Eq:DENS BD} yields $\scrF_c(\mu_t) =0$ for every $t\in[0,1]$, which is the claim.
\end{proof}

\subsection{Equivalence with the entropic TCD condition}\label{Sub:Equiv TCDs} In this section, we prove the equivalence of $\smash{\TCD_p^*(K,N)}$ with its entropic counterpart $\smash{\TCD_p^e(K,N)}$ introduced in \cite{cavalletti2020}, provided the underlying space is timelike $p$-essentially nonbranching according to \autoref{Def:Essentially nonbranching}. We also obtain equivalence of the respective weak with their strong versions; see \autoref{Th:Equivalence TCD* and TCDe}. Yet another characterization will be obtained in \autoref{Pr:MDPTS} below.

In this section, in addition to our standing assumptions on $\scrX$ we suppose the latter to be timelike $p$-essentially nonbranching for a fixed $p\in (0,1)$. 

For the convenience of the reader, let us recall the following notions \cite[Def.~3.2, Prop.~3.3]{cavalletti2020}, for which we define $\scrU_N \colon \scrP(\mms)\to [0,\infty]$, $N\in (0,\infty)$, by
\begin{align}\label{Eq:Expo Boltzmann}
\scrU_N(\mu) := \rme^{-\Ent_\meas(\mu)/N}.
\end{align}

\begin{definition}\label{Def:CD entropic} Let $K\in\R$ and $N\in (0,\infty)$. 
\begin{enumerate}[label=\textnormal{\alph*.}]
\item We say that $\scrX$ satis\-fies the \emph{entropic timelike curvature-dimension condition} $\smash{\TCD_p^e(K,N)}$ if for every timelike $p$-dualizable pair $(\mu_0,\mu_1)\in\Dom(\Ent_\meas)^2$, there exist a timelike proper-time parametrized $\smash{\ell_p}$-geo\-desic $(\mu_t)_{t\in [0,1]}$ connecting $\mu_0$ and $\mu_1$  as well as a timelike $p$-dualizing coupling $\pi\in\smash{\Pi_\ll(\mu_0,\mu_1)}$ such that for every $t\in [0,1]$,
\begin{align*}
\scrU_N(\mu_t) \geq \sigma_{K,N}^{(1-t)}\big[\Vert\tsep\Vert_{\Ell^2(\mms^2,\pi)}\big]\,\scrU_N(\mu_0) + \sigma_{K,N}^{(t)}\big[\Vert\tsep\Vert_{\Ell^2(\mms^2,\pi)}\big]\,\scrU_N(\mu_1).
\end{align*}
\item If the previous claim holds for every strongly timelike $p$-dualizable  $(\mu_0,\mu_1)\in(\scrP_\comp(\mms)\cap\Dom(\Ent_\meas))^2$, we say that $\scrX$ satisfies the \emph{weak entropic timelike curvature-dimension condition} $\smash{\wTCD_p^e(K,N)}$.
\end{enumerate}
\end{definition}

\begin{theorem}\label{Th:Equivalence TCD* and TCDe} The  following statements are equivalent for every given $K\in\R$ and $N\in[1,\infty)$.
\begin{enumerate}[label=\textnormal{\textcolor{black}{(}\roman*\textcolor{black}{)}}]
\item\label{La:1} The condition $\smash{\TCD_p^*(K,N)}$ holds.
\item\label{La:1.5} The condition $\smash{\wTCD_p^*(K,N)}$ holds.
\item\label{La:2} For every timelike $p$-dualizable  $(\mu_0,\mu_1) = (\rho_0\,\meas,\rho_1\,\meas)\in\scrP^\ac(\mms,\meas)^2$ there exists some timelike $\smash{\ell_p}$-optimal geodesic plan $\bdpi\in\smash{\OptTGeo_{\ell_p}^\tsep(\mu_0,\mu_1)}$ such that for every $t\in [0,1]$,  we have $\smash{(\eval_t)_\push\bdpi=\rho_t\,\meas\in\scrP^\ac(\mms,\meas)}$, and for every $N'\geq N$, the inequality
\begin{align*}
\rho_t(\gamma_t)^{-1/N'}&\geq \sigma_{K,N'}^{(1-t)}(\tsep(\gamma_0,\gamma_1))\,\rho_0(\gamma_0)^{-1/N'}  + \sigma_{K,N'}^{(t)}(\tsep(\gamma_0,\gamma_1))\,\rho_1(\gamma_1)^{-1/N'}
\end{align*}
holds for $\bdpi$-a.e.~$\gamma\in\TGeo^\tsep(\mms)$.
\item\label{La:3} The condition $\smash{\TCD_p^e(K,N)}$ holds.
\item\label{La:3.5} The condition $\smash{\wTCD_p^e(K,N)}$ holds.
\end{enumerate}
\end{theorem}

\begin{remark} The exceptional set in \ref{La:2} may depend on $t$ and $N'$. (However, with some effort it seems possible to achieve independence of the exceptional set on $t$, e.g.~by following \cite[Cor.~9.5]{cavallettimilman2021}.) 

By \autoref{Pr:iii to iv} and the proofs of \autoref{Pr:ii to iii} and \autoref{Pr:v to iii}, any of the claims in \autoref{Th:Equivalence TCD* and TCDe} is equivalent to the following weaker version of \ref{La:2}: for every pair $\smash{(\mu_0,\mu_1)=(\rho_0\,\meas,\rho_1\,\meas)\in\scrP_\comp^\ac(\mms,\meas)^2}$ such that $\smash{\supp\mu_0\times\supp\mu_1\subset\mms_\ll^2}$ and $\rho_0,\rho_1\in\Ell^\infty(\mms,\meas)$, there exists $\smash{\bdpi\in\OptTGeo_{\ell_p}^\tsep(\mu_0,\mu_1)}$ such that for every $t\in[0,1]$, we have $\smash{(\eval_t)_\push\bdpi\in\scrP^\ac(\mms,\meas)}$ and
\begin{align*}
\rho_t(\gamma_t)^{-1/N}&\geq \sigma_{K,N}^{(1-t)}(\tsep(\gamma_0,\gamma_1))\,\rho_0(\gamma_0)^{-1/N}  + \sigma_{K,N}^{(t)}(\tsep(\gamma_0,\gamma_1))\,\rho_1(\gamma_1)^{-1/N}
\end{align*}
holds for $\bdpi$-a.e.~$\gamma\in\TGeo^\tsep(\mms)$, where $(\eval_t)_\push\bdpi = \rho_t\,\meas$.

An analogous note applies to \autoref{Th:Equiv TCD with geo},  \autoref{Th:Equivalence TMCP* and TMCPe}, and \autoref{Th:Equivalence TMCP}.
\end{remark}

It is clear that \ref{La:1} implies \ref{La:1.5}, and that \ref{La:3} yields \ref{La:3.5}. Moreover, \ref{La:2}  implies \ref{La:1} by integration.  The proofs of \ref{La:1.5} implying \ref{La:2},  \ref{La:2} implying \ref{La:3}, and \ref{La:3.5} implying \ref{La:2} are, for the sake of clarity,  outsourced to \autoref{Pr:ii to iii}, \autoref{Pr:iii to iv}, and \autoref{Pr:v to iii}, respectively. 

\begin{remark}\label{Re:Uniq} Below, we occasionally use two results which are only established later in \autoref{Ch:TMCP}, namely  the following.
\begin{itemize}
\item First, the condition $\smash{\wTCD_p^*(K,N)}$ implies the timelike measure-contraction property $\smash{\TMCP^*(K,N)}$ according to \autoref{Def:TMCP}, cf.~\autoref{Pr:TMCP to TCD}.
\item Second, on timelike $p$-essentially nonbranching spaces satisfying the latter condition --- or $\smash{\TMCP^e(K,N)}$ according to \autoref{Def:TMCPe} ---, chronological $\smash{\ell_p}$-optimal couplings and timelike $\smash{\ell_p}$-optimal geodesic plans with sufficiently well-behaved endpoints are unique, cf.~\autoref{Th:Uniqueness couplings}, \autoref{Th:Uniqueness geodesics}, and \autoref{Re:From TNB to TENB}. 
\end{itemize}
As no result of this section is used in \autoref{Ch:TMCP}, no circular reasoning occurs.

The key point of these results here is that the timelike proper-time parametrized $\smash{\ell_p}$-geodesics and the $\smash{\ell_p}$-optimal couplings from \autoref{Def:TCD*}, \autoref{Def:TCD}, as well as  \autoref{Def:TMCP}, which are a priori unrelated, are in fact induced by the same timelike $\smash{\ell_p}$-optimal geodesic plan.

This means that we must have these uniqueness results at our disposal \emph{a priori} to establish \autoref{Th:Equivalence TCD* and TCDe}. In particular, these do \emph{not} follow from \autoref{Th:Equivalence TCD* and TCDe}, \autoref{Re:From TNB to TENB}, and \cite[Thm.~3.19, Thm.~3.20]{cavalletti2020}.
\end{remark}

\begin{proposition}\label{Pr:ii to iii} Assume $\smash{\wTCD_p^*(K,N)}$ for some $K\in\R$ and $N\in [1,\infty)$. Then for every timelike $p$-dualizable pair $(\mu_0,\mu_1) = (\rho_0\,\meas,\rho_1\,\meas)\in\scrP^\ac(\mms,\meas)^2$ there exists some $\bdpi\in\smash{\OptTGeo_{\ell_p}^\tsep(\mu_0,\mu_1)}$ such that for every $t\in [0,1]$, we have $(\eval_t)_\push\bdpi= \rho_t\, \meas\in\smash{\scrP^\ac(\mms,\meas)}$, and for every $N'\geq N$,  the inequality
\begin{align*}
\rho_t(\gamma_t)^{-1/N'}\geq \sigma_{K,N'}^{(1-t)}(\tsep(\gamma_0,\gamma_1))\,\rho_0(\gamma_0)^{-1/N'} + \sigma_{K,N'}^{(t)}(\tsep(\gamma_0,\gamma_1))\,\rho_1(\gamma_1)^{-1/N'}
\end{align*}
holds for $\bdpi$-a.e.~$\gamma\in\TGeo^\tsep(\mms)$, where $(\eval_t)_\push\bdpi = \rho_t\,\meas$.
\end{proposition}

\begin{proof} \textbf{Step 1.} \emph{Reinforcement of the assumptions on $\mu_0$ and $\mu_1$.} We first assume that $\smash{\supp\mu_0\times\supp\mu_1\subset\mms_\ll^2}$, that $\mu_0,\mu_1\in\scrP_\comp(\mms)$, and that $\rho_0,\rho_1\in\Ell^\infty(\mms,\meas)$. As mentioned above, timelike $p$-dualizable couplings and timelike $\ell_p$-optimal geodesic plans will then be unique all over Step 1, Step 2, and Step 3 without further notice.

\textbf{Step 2.} \textit{Restriction.} Observe that the pair  $(\mu_0,\mu_1)$ is in fact strongly timelike $p$-dualizable thanks to  \autoref{Re:Strong timelike}. Fix a $\cap$-stable generator $\{M_n : n\in\N\}$  of the  Borel $\sigma$-field of $\mms$; to be precise, the $M_n$'s  should be chosen to generate the relative Borel $\sigma$-field of the compact set $\supp\mu_0\cup\supp\mu_1$, but we ignore this minor technicality by assuming compactness of $\mms$ instead until Step 4. Given $n\in\N$ we cover $\mms$ by the mutually disjoint Borel sets 
\begin{align*}
L_1 &:= M_1\cap \dots \cap M_n,\\
L_2 &:= M_1\cap\dots\cap M_{n-1}\cap M_n^\sfc,\,\dots\\
L_{2^n} &:= M_1^\sfc \cap \dots \cap M_n^\sfc.
\end{align*}
Given the unique timelike $p$-dualizing coupling $\smash{\pi\in\Pi_\ll(\mu_0,\mu_1)}$ of $\mu_0$ and $\mu_1$, let us define $\smash{\mu_0^{ij},\mu_1^{ij}\in\scrP_\comp(\mms)\cap\Dom(\Ent_\meas)}$ by 
\begin{align*}
\mu_0^{ij}[A] &:= \lambda_{ij}^{-1}\,\pi[(A\cap  L_i) \times L_j] = \varrho_0^{ij}\,\meas,\\
\mu_1^{ij}[A] &:= \lambda_{ij}^{-1}\,\pi[L_i\times (A\cap L_j)] = \varrho_1^{ij}\,\meas
\end{align*}
for $i,j\in\{1,\dots,2^n\}$ provided $\lambda_{ij} := \pi[L_i\times L_j]>  0$. Then 
\begin{align}\label{Eq:MSING}
\begin{split}
\mu_0^{ij} &\perp \mu_0^{i'j'},\\
\mu_1^{ij} &\perp \mu_1^{i'j'}
\end{split}
\end{align}
for every $i,i',j,j'\in\{1,\dots,2^n\}$ with $i\neq i'$ or $j\neq j'$. Also,  $\smash{(\mu_0^{ij},\mu_1^{ij})}$ is strongly timelike $p$-dualizable by \autoref{Re:Strong timelike}. By the consistency part in \autoref{Pr:Consistency TCD} and the above mentioned uniqueness results, and $\wTCD_p^*(K,N)$,  there exists some $\smash{\bdpi^{ij}\in\OptTGeo_{\ell_p}^\tsep(\mu_0^{ij},\mu_1^{ij})}$ with 
\begin{align}\label{Eq:THa inequ}
\begin{split}
&\int_{\TGeo^\tsep(\mms)} \sigma_{K,N'}^{(1-t)}(\tsep(\gamma_0,\gamma_1))\,\varrho_0^{ij}(\gamma_0)^{-1/N'}\d\bdpi^{ij}(\gamma)\\
&\qquad\qquad\qquad\qquad + \int_{\TGeo^\tsep(\mms)} \sigma_{K,N'}^{(t)}(\tsep(\gamma_0,\gamma_1))\,\varrho_1^{ij}(\gamma_1)^{-1/N'}\d\bdpi^{ij}(\gamma)\\
&\qquad\qquad \leq \int_{\TGeo^\tsep(\mms)} \varrho_t^{ij}(\gamma_t)^{-1/N'}\d\bdpi^{ij}(\gamma)
\end{split}
\end{align}
for every $t\in [0,1]$ and every $N'\geq N$; note that since $\smash{\mu_0^{ij},\mu_1^{ij}\in \scrP_\comp(\mms)\cap\Dom(\Ent_\meas)}$ as a consequence of the assumption $\rho_0,\rho_1\in\Ell^\infty(\mms,\meas)$, we have $(\eval_t)_\push\bdpi^{ij} = \smash{\rho_t^{ij}}\,\meas \in\smash{\scrP^\ac(\mms,\meas)}$ thanks to \autoref{Le:Stu}.

\textbf{Step 3.} \textit{Pasting plans together and conclusion.} 
By uniqueness, we get
\begin{align*}
\bdpi = \lambda_{11}\,\bdpi^{11} + \lambda_{12}\,\bdpi^{12} + \dots + \lambda_{2^n2^n}\,\bdpi^{2^n2^n}
\end{align*}
for every $n\in\N$, since the right-hand side belongs to $\smash{\OptTGeo_{\ell_p}^\tsep(\mu_0,\mu_1)}$. In particular, setting $\smash{\mu_t := (\eval_t)_\push\bdpi = \rho_t\,\meas}$, $t\in[0,1]$ --- the $\meas$-ab\-solute continuity stemming from \autoref{Le:Stu} and again uniqueness ---, we have
\begin{align*}
\rho_t = \lambda_{11}\,\varrho_t^{11} + \lambda_{12}\,\varrho_t^{12} + \dots + \lambda_{2^n2^n}\,\varrho_t^{2^n2^n}\quad\meas\text{-a.e.}
\end{align*}
Moreover, \eqref{Eq:MSING} and \autoref{Le:Mutually singular} imply
\begin{align*}
(\eval_t)_\push\bdpi^{ij} \perp (\eval_t)_\push\bdpi^{i'j'}
\end{align*}
for every $t\in (0,1)$ and every $i,i',j,j'\in\{1,\dots,2^n\}$ with $i\neq i'$ or $j\neq j'$, i.e.
\begin{align*}
\meas\big[\{\varrho_t^{ij} > 0\} \cap \{\varrho_t^{i'j'}>0\}\big] = 0.
\end{align*}
Setting $G_{ij} := (\eval_0,\eval_1)^{-1}(L_i\times L_j)$, by construction  we obtain from \eqref{Eq:THa inequ} that
\begin{align*}
&\int_{G_{ij}} \sigma_{K,N'}^{(1-t)}(\tsep(\gamma_0,\gamma_1))\,\rho_0(\gamma_0)^{-1/N'}\d\bdpi(\gamma)\\
&\qquad\qquad\qquad\qquad + \int_{G_{ij}} \sigma_{K,N'}^{(t)}(\tsep(\gamma_0,\gamma_1))\,\rho_1(\gamma_1)^{-1/N'}\d\bdpi(\gamma)\\
&\qquad\qquad \leq \int_{G_{ij}} \rho_t(\gamma_t)^{-1/N'}\d\bdpi(\gamma).
\end{align*}
Since $i,j\in\{1,\dots,2^n\}$ and $n\in\N$ were arbitrary, the claim follows.

\textbf{Step 4.} \textit{Removing the restrictions on $\mu_0$ and $\mu_1$.} We initially address the first two restrictions in Step 1, then the last one. 

\textbf{Step 4.1.} If the pair $\smash{(\mu_0,\mu_1) \in\scrP^\ac(\mms,\meas)^2}$ with $\rho_0,\rho_1\in\Ell^\infty(\mms,\meas)$ is merely timelike $p$-dualizable, let $\pi\in\Pi_\ll(\mu_0,\mu_1)$ be a corresponding timelike $p$-dualizing coupling. Since $\smash{\mms_\ll^2}$ is open and $\met$ is separable, we cover $\smash{\supp\pi\cap \mms_\ll^2}$ by countably many (not necessarily disjoint) relatively compact rectangles $A_i\times B_i$ with $\smash{\bar{A}_i\times\bar{B}_i \subset \mms_\ll^2}$, where $A_i,B_i\subset\mms$ are open, $i\in\N$. Set $\pi^0 := 0$, and define the measure $\pi^i$ on $\mms$ inductively by
\begin{align*}
\pi^i := (\pi - \pi^{i-1})\mres (A_i\times B_i).
\end{align*}
Then $\smash{\pi = \pi^1 + \pi^2 + \dots}$ by construction. Define $\smash{\hat{\pi}^i\in\scrP_\comp(\mms^2)}$ by
\begin{align*}
\hat{\pi}^i := \pi^i[\mms^2]^{-1}\,\pi^i
\end{align*}
provided $\smash{\pi^i[\mms^2]> 0}$, and consider its marginals
\begin{align*}
\mu_0^i := (\pr_1)_\push\hat{\pi}^i = \varrho_0^i\,\meas,\\
\mu_1^i := (\pr_2)_\push\hat{\pi}^i = \varrho_1^i\,\meas.
\end{align*} 
By \autoref{Re:Strong timelike}, the pair $\smash{(\mu_0^i,\mu_1^i)}$ is strongly timelike $p$-dualizable for every $i\in\N$. Thus, the previous steps imply the existence of $\smash{\bdpi^i\in\OptTGeo_{\ell_p}^\tsep(\mu_0^i,\mu_1^i)}$ such that for every $t\in (0,1)$, we have $(\eval_t)_\push\bdpi^i = \smash{\varrho_t^i}\,\meas\in\smash{\scrP^\ac(\mms,\meas)}$,  and for every $N'\geq N$,
\begin{align}\label{Eq:The Inequality rhoT}
\varrho_t^i(\gamma_t)^{-1/N'} \geq \sigma_{K,N'}^{(1-t)}(\tsep(\gamma_0,\gamma_1))\,\varrho_0^i(\gamma_0)^{-1/N'} + \sigma_{K,N'}^{(t)}(\tsep(\gamma_0,\gamma_1))\,\varrho_1^i(\gamma_1)^{-1/N'}
\end{align}
holds for $\smash{\bdpi^i}$-a.e.~$\gamma\in\TGeo^\tsep(\mms)$. By construction,
\begin{align*}
\mu_0^i &\perp \mu_0^j,\\
\mu_1^i &\perp \mu_1^j
\end{align*}
for every $i,j\in\N$ with $i\neq j$, and from  \autoref{Le:Mutually singular} we derive
\begin{align*}
(\eval_t)_\push\bdpi^i \perp (\eval_t)_\push\bdpi^j.
\end{align*}
Note that the measure
\begin{align*}
\bdpi := \pi^1[\mms^2]\,\bdpi^1 + \pi^2[\mms^2]\,\bdpi^2 + \dots
\end{align*}
is a timelike $\smash{\ell_p}$-optimal geodesic plan interpolating $\mu_0$ and $\mu_1$. It satisfies $(\eval_t)_\push\bdpi = \rho_t\,\meas\in \smash{\scrP^\ac(\mms,\meas)}$ for every $t\in[0,1]$, and 
\begin{align*}
\rho_t = \pi^1[\mms^2]\,\varrho_t^1 + \pi^2[\mms^2]\,\varrho_t^2 + \dots\quad\meas\textnormal{-a.e.}
\end{align*}
Together with \eqref{Eq:The Inequality rhoT} this discussion implies, for $\bdpi$-a.e.~$\gamma\in\TGeo^\tsep(\mms)$,
\begin{align*}
\rho_t(\gamma)^{-1/N'} \geq \sigma_{K,N'}^{(1-t)}(\tsep(\gamma_0,\gamma_1))\,\rho_0(\gamma_0)^{-1/N'} + \sigma_{K,N'}^{(t)}(\tsep(\gamma_0,\gamma_1))\,\rho_1(\gamma_1)^{-1/N'}.
\end{align*}

\textbf{Step 4.2.} Finally, given any timelike $p$-dualizable pair $(\mu_0,\mu_1)= (\rho_0\,\meas,\rho_1\,\meas)\in\smash{\scrP^\ac(\mms,\meas)^2}$, for every $i,j\in\N$ we set 
\begin{align*}
E_i &:= \{i-1 \leq \rho_0^i < i\},\\
F_j &:= \{j-1 \leq \rho_1^j < j\}.
\end{align*}
By restriction \cite[Lem. 2.10]{cavalletti2020}, $\smash{(\mu_0^i,\mu_1^j)}$ is timelike $p$-dualizable, where 
\begin{align*}
\mu_0^i &:= \mu_0[E_i]^{-1}\,\mu_0\mres E_i = \rho_0^i\,\meas,\\
\mu_1^j &:= \mu_1[F_j]^{-1}\,\mu_1\mres F_j = \rho_1^j\,\meas
\end{align*}
provided $\mu_0[E_i] > 0$ as well as $\mu_1[F_j]>0$. Clearly  $\smash{\rho_0^i,\rho_1^j\in\Ell^\infty(\mms,\meas)}$, and a similar argument as above  leads to the conclusion. We omit the details.
\end{proof}

\begin{proposition}\label{Pr:iii to iv} Let $K\in\R$ and $N\in [1,\infty)$. Assume that for every timelike $p$-dualizable pair $(\mu_0,\mu_1)=(\rho_0\,\meas,\rho_1\,\meas)\in\Dom(\Ent_\meas)^2$, there is $\smash{\bdpi\in\OptTGeo_{\ell_p}^\tsep(\mu_0,\mu_1)}$ such that for every $t\in [0,1]$, we have $(\eval_t)_\push\bdpi = \rho_t\,\meas\in\scrP^\ac(\mms,\meas)$ and
\begin{align}\label{Eq:rhot in}
\rho_t(\gamma_t)^{-1/N}\geq \sigma_{K,N}^{(1-t)}(\tsep(\gamma_0,\gamma_1))\,\rho_0(\gamma_0)^{-1/N} + \sigma_{K,N}^{(t)}(\tsep(\gamma_0,\gamma_1))\,\rho_1(\gamma_1)^{-1/N}
\end{align}
holds for $\bdpi$-a.e.~$\gamma\in\TGeo^\tsep(\mms)$. Then $\scrX$ obeys $\smash{\TCD_p^e(K,N)}$.
\end{proposition}

\begin{proof} The assignment $\mu_t := (\eval_t)_\push\bdpi$ creates a timelike proper-time parametrized $\smash{\ell_p}$-geodesic $(\mu_t)_{t\in[0,1]}$ from $\mu_0$ to $\mu_1$. Moreover, $\smash{\pi := (\eval_0,\eval_1)_\push\bdpi\in\Pi_\ll(\mu_0,\mu_1)}$ forms a timelike $p$-dualizing coupling of $\mu_0$ and $\mu_1$. Taking logarithms on both sides of \eqref{Eq:rhot in} and recalling \eqref{Eq:Distortion coeff property}, we get
\begin{align*}
-\frac{1}{N}\log\rho_t(\gamma_t) \geq \rmG_t\Big[\!-\!\frac{1}{N}\log\rho_0(\gamma_0), -\frac{1}{N}\log\rho_1(\gamma_1), \frac{K}{N}\,\tsep^2(\gamma_0,\gamma_1)\Big]
\end{align*}
for $\bdpi$-a.e.~$\gamma\in\TGeo^\tsep(\mms)$; recall that the function $\rmG_t$ is defined in \eqref{Eq:GtHt}. Integrating this inequality with respect to $\bdpi$ and using Jensen's inequality together with joint convexity of $\rmG_t$ yields
\begin{align*}
-\frac{1}{N}\Ent_\meas(\mu_t) \geq \rmG_t\Big[\!-\!\frac{1}{N}\Ent_\meas(\mu_0), -\frac{1}{N}\Ent_\meas(\mu_1), \frac{K}{N}\,\big\Vert\tsep\big\Vert_{\Ell^2(\mms^2,\pi)}^2\Big].
\end{align*}
Here, in the case $K>0$ and $\smash{\Vert\tsep\Vert_{\Ell^2(\mms^2,\pi)} = \infty}$ the right-hand side is consistently interpreted as $\infty$, which implies  $\Ent_\meas(\mu_t) = -\infty$; for $K<0$ and $\smash{\Vert\tsep\Vert_{\Ell^2(\mms^2,\pi)}}=\infty$ the right-hand side has an evident interpretation. Exponentiating both sides now yields the desired $\smash{\TCD_p^e(K,N)}$ condition.
\end{proof}

\begin{proposition}\label{Pr:v to iii} Assume $\smash{\wTCD_p^e(K,N)}$ for some $K\in\R$ and $N\in [1,\infty)$. Then for every timelike $p$-dualizable pair $(\mu_0,\mu_1) = (\rho_0\,\meas,\rho_1\,\meas)\in\scrP^\ac(\mms,\meas)^2$ there exists $\bdpi\in\smash{\OptTGeo_{\ell_p}^\tsep(\mu_0,\mu_1)}$ such that for every $t\in [0,1]$, we have $(\eval_t)_\push\bdpi = \rho_t\,\meas$, and for every $N'\geq N$, the inequality 
\begin{align*}
\rho_t(\gamma_t)^{-1/N'}\geq \sigma_{K,N'}^{(1-t)}(\tsep(\gamma_0,\gamma_1))\,\rho_0(\gamma_0)^{-1/N'} + \sigma_{K,N'}^{(t)}(\tsep(\gamma_0,\gamma_1))\,\rho_1(\gamma_1)^{-1/N'}
\end{align*}
holds for $\bdpi$-a.e.~$\gamma\in\TGeo^\tsep(\mms)$, where $(\eval_t)_\push\bdpi = \rho_t\,\meas$. 
\end{proposition}

\begin{proof} The proof is similar to the one of \autoref{Pr:ii to iii} --- whose notation we retain here ---, whence we only outline the main differences. Again, it suffices to consider the case  $\smash{\supp\mu_0\times\supp\mu_1\subset\mms_\ll^2}$,  $\mu_0,\mu_1\in\scrP_\comp(\mms)$, and  $\smash{\rho_0,\rho_1\in\Ell^\infty(\mms,\meas)}$. In particular, these restrictions imply $\mu_0,\mu_1\in\Dom(\Ent_\meas)$. The parallel uniqueness results outlined in \autoref{Re:Uniq} under $\smash{\TCD_p^e(K,N)}$ which are implicitly used below are due to \cite[Thm.~3.19, Thm.~3.20]{cavalletti2020}, see also \autoref{Re:From TNB to TENB}.

Let $\smash{\bdpi\in\OptTGeo_{\ell_p}^\tsep(\mu_0,\mu_1)}$ be a timelike $\smash{\ell_p}$-optimal geodesic plan interpolating $\mu_0$ and  $\mu_1$. Given any $i,j\in\{1,\dots,2^n\}$, we define
\begin{align*}
\mu_0^{ij} &:= \lambda_{ij}^{-1}\,(\pr_1)_\push\big[(\eval_0,\eval_1)_\push\bdpi\mres(L_i\times L_j)\big] = \varrho_0^{ij}\,\meas\\
\mu_1^{ij} &:= \lambda_{ij}^{-1}\,(\pr_2)_\push\big[(\eval_0,\eval_1)_\push\bdpi\mres(L_i\times L_j)\big] = \varrho_1^{ij}\,\meas
\end{align*}
provided $\smash{\lambda_{ij} := \bdpi[G_{ij}] > 0}$. Recall that $\smash{\wTCD_p^e(K,N)}$ implies $\smash{\wTCD_p^e(K,N')}$ by  \cite[Lem.~3.10]{cavalletti2020}.  Arguing as for \autoref{Pr:ii to iii}, invoking $\smash{\wTCD_p^e(K,N')}$ for the strongly timelike $p$-dualizable  pair $\smash{(\mu_0^{ij}, \mu_1^{ij})} \in \smash{(\scrP_\comp(\mms)\cap\Dom(\Ent_\meas))^2}$, employing the uniqueness both of $\bdpi$ as well as the involved timelike $p$-dualizing coupling $\pi^{ij}\in\smash{\Pi_\ll(\mu_0^{ij},\mu_1^{ij})}$ by \autoref{Re:From TNB to TENB}, and taking logarithms on both sides of the resulting ``pasted inequality'' we infer that
\begin{align*}
&-\frac{\lambda_{ij}^{-1}}{N'}\int_{G_{ij}}\log \rho_t(\gamma_t)\d\bdpi(\gamma)\\
&\qquad\qquad \geq \rmG_t\Big[\!-\!\frac{\lambda_{ij}^{-1}}{N'}\int_{G_{ij}} \log\rho_0(\gamma_0)\d\bdpi(\gamma),-\frac{\lambda_{ij}^{-1}}{N'}\int_{G_{ij}} \log\rho_1(\gamma_1)\d\bdpi(\gamma),\\
&\qquad\qquad\qquad\qquad \lambda_{ij}^{-1}\,\frac{K}{N'}\int_{G_{ij}}\tsep^2(\gamma_0,\gamma_1)\d\bdpi(\gamma)\Big]\\
&\qquad\qquad\geq \lambda_{ij}^{-1}\int_{G_{ij}} \rmG_t\Big[\!-\!\frac{1}{N'}\log\rho_0(\gamma_0), -\frac{1}{N'}\log\rho_1(\gamma_1), \frac{K}{N'}\,\tsep^2(\gamma_0,\gamma_1)\Big]\d\bdpi(\gamma).
\end{align*}
Here the deduction that $\smash{(\eval_t)_\push\bdpi = \rho_t\,\meas\in\scrP^\ac(\mms,\meas)}$ comes directly from the definition of the functional $\scrU_N$, and the last inequality follows from joint convexity of $\rmG_t$ and Jensen's inequality. Since $i,j\in\{1,\dots,2^n\}$ and $n\in\N$ were arbitrary, the claim follows by taking exponentials of both integrands, respectively.
\end{proof}

By evident adaptations of the above arguments, we also obtain the following  result for the conditions from \autoref{Def:TCD}.

\begin{theorem}\label{Th:Equiv TCD with geo} The  following statements are equivalent for every given $K\in\R$ and $N\in[1,\infty)$.
\begin{enumerate}[label=\textnormal{\textcolor{black}{(}\roman*\textcolor{black}{)}}]
\item The condition $\smash{\TCD_p(K,N)}$ holds.
\item The condition $\smash{\wTCD_p(K,N)}$ holds.
\item For every timelike $p$-dualizable  $(\mu_0,\mu_1) = (\rho_0\,\meas,\rho_1\,\meas)\in\scrP^\ac(\mms,\meas)^2$ there exists some timelike $\smash{\ell_p}$-optimal geodesic plan $\bdpi\in\smash{\OptTGeo_{\ell_p}^\tsep(\mu_0,\mu_1)}$ such that  for every $t\in [0,1]$, we have $\smash{(\eval_t)_\push\bdpi =\rho_t\, \meas\in\scrP^\ac(\mms,\meas)}$, and for every $N'\geq N$, the inequality
\begin{align*}
\rho_t(\gamma_t)^{-1/N'}&\geq \tau_{K,N'}^{(1-t)}(\tsep(\gamma_0,\gamma_1))\,\rho_0(\gamma_0)^{-1/N'}  + \tau_{K,N'}^{(t)}(\tsep(\gamma_0,\gamma_1))\,\rho_1(\gamma_1)^{-1/N'}
\end{align*}
holds for $\bdpi$-a.e.~$\gamma\in\TGeo^\tsep(\mms)$.
\end{enumerate}
\end{theorem}

\subsection{From local to global}\label{Sub:Local global} The goal of this section is to establish the equivalence of $\smash{\TCD_p^*(K,N)}$ to its corresponding local version in \autoref{Def:TCD loc} below, i.e.~to establish a Lorentzian analogue of the local-to-global property from \cite[Ch.~5]{bacher2010}. 

To this aim, in addition to our standing assumptions on the base space, suppose again $\scrX$ to be timelike $p$-essentially nonbranching for  $p\in (0,1)$.


\begin{definition}\label{Def:TCD loc} Let $K\in\R$ and $N\in[1,\infty)$. 
\begin{enumerate}[label=\textnormal{\alph*.}]
\item We say that $\scrX$ satisfies $\smash{\TCD_p^*(K,N)}$ \emph{locally}, briefly $\smash{\TCD_{p,\loc}^*(K,N)}$, if there exists an open cover $(\mms_i)_{i\in I}$ of $\mms$ with the following property. For every $i\in I$ and every pair $(\mu_0,\mu_1)=(\rho_0\,\meas,\rho_1\,\meas)\in\smash{\scrP_\comp^\ac(\mms_i,\meas)^2}$ such that $\smash{\supp\mu_0\times\supp\mu_1\subset (\mms_i)_\ll^2}$, there is a timelike proper-time parametrized $\smash{\ell_p}$-geodesic $(\mu_t)_{t\in[0,1]}$ in  $\scrP(\mms)$ from $\mu_0$ to $\mu_1$ and an $\smash{\ell_p}$-optimal coupling $\pi\in\Pi_\ll(\mu_0,\mu_1)$ such that for every $t\in[0,1]$ and every $N'\geq N$,
\begin{align}\label{Eq:RED}
\begin{split}
\scrS_{N'}(\mu_t) &\leq -\int_{\mms^2} \sigma_{K,N'}^{(1-t)}(\tsep(x^0,x^1))\,\rho_0(x^0)^{-1/N'}\d\pi(x^0,x^1)\\
&\qquad\qquad -\int_{\mms^2}\sigma_{K,N'}^{(t)}(\tsep(x^0,x^1))\,\rho_1(x^1)^{-1/N'}\d\pi(x^0,x^1).
\end{split}
\end{align}
\item If the previous statement holds for $\smash{\sigma_{K,N'}^{(1-t)}}$ and $\smash{\sigma_{K,N'}^{(t)}}$ replaced by $\smash{\tau_{K,N'}^{(1-t)}}$ and $\smash{\tau_{K,N'}^{(t)}}$, respectively, $\scrX$ is termed to satisfy $\smash{\TCD_p(K,N)}$ \emph{locally}, briefly $\smash{\TCD_{p,\loc}(K,N)}$.
\end{enumerate} 
\end{definition}

\begin{proposition}\label{Pr:MDPTS} Given any $K\in\R$ and $N\in[1,\infty)$, the property $\smash{\TCD_p^*(K,N)}$ holds if and only if the following does. For every $\smash{\mu_0,\mu_1\in\scrP_\comp^\ac(\mms,\meas)}$ with $\supp\mu_0\times\supp\mu_1 \subset\smash{\mms_\ll^2}$ there  exists $\smash{\bdpi\in\OptTGeo_{\ell_p}^\tsep(\mu_0,\mu_1)}$ such that for every $N'\geq N$,
\begin{align*}
\scrS_{N'}((\eval_{1/2})_\push\bdpi) \leq \sigma_{K,N'}^{(1/2)}(\theta)\,\scrS_{N'}(\mu_0) + \sigma_{K,N'}^{(1/2)}(\theta)\,\scrS_{N'}(\mu_1),
\end{align*}
where
\begin{align*}
\theta := \begin{cases} \sup \tsep(\supp\mu_0\times\supp\mu_1) & \textnormal{if }K<0,\\
\inf \tsep(\supp\mu_0\times\supp\mu_1) & \textnormal{otherwise}.
\end{cases}
\end{align*}
\end{proposition}

\begin{proof} We only outline the main differences of the proof to its counterpart from \cite[Prop.~2.10]{bacher2010}. The forward implication is clear. Conversely, similarly to the proof of \autoref{Th:Good geos TCD}, by successively gluing plans corresponding to ``midpoints'' for which the Rényi entropy obeys the given inequality as in the proof of \cite[Prop.~2.10]{bacher2010}, we inductively construct a sequence $\smash{(\bdalpha^n)_{n\in\N}}$ in $\smash{\OptTGeo_{\ell_p}^\tsep(\mu_0,\mu_1)}$ such that for every $n\in\N$ and every odd $k\in\{1,\dots,2^n-1\}$,
\begin{align*}
\scrS_{N'}((\eval_{k2^{-n}})_\push\bdalpha^n) &\leq \sigma_{K,N'}^{(1/2)}(2^{-n+1}\,\theta)\,\scrS_{N'}((\eval_{(k-1)2^{-n}})_\push\bdalpha^n)\\
&\qquad\qquad + \sigma_{K,N'}^{(1/2)}(2^{-n+1}\,\theta)\,\scrS_{N'}((\eval_{(k+1)2^{-n}})_\push\bdalpha^n)\\
&\leq \sigma_{K,N'}^{(1/2)}(2^{-n+1}\,\theta)\,\sigma_{K,N'}^{(1-(k-1)2^{-n})}(\theta)\,\scrS_{N'}(\mu_0)\\
&\qquad\qquad\qquad\qquad + \sigma_{K,N'}^{(1/2)}(2^{-n+1}\,\theta)\,\sigma_{K,N'}^{((k-1)2^{-n})}(\theta)\,\scrS_{N'}(\mu_1)\\
&\qquad\qquad + \sigma_{K,N'}^{(1/2)}(2^{-n+1}\,\theta)\,\sigma_{K,N'}^{(1-(k+1)2^{-n})}(\theta)\,\scrS_{N'}(\mu_0)\\
&\qquad\qquad\qquad\qquad + \sigma_{K,N'}^{(1/2)}(2^{-n+1}\,\theta)\,\sigma_{K,N'}^{((k+1)2^{-n})}(\theta)\,\scrS_{N'}(\mu_1)\\
&= \sigma_{K,N'}^{(1-k2^{-n})}(\theta)\,\scrS_{N'}(\mu_0) + \sigma_{K,N'}^{(k2^{-n})}(\theta)\,\scrS_{N'}(\mu_1).
\end{align*}
holds for every $N'\geq N$. The latter equality is a crucial property of the distortion coefficients $\smash{\sigma_{K,N'}^{(r)}}$ \cite[Lem.~3.1]{rajala2012b} which does not hold for $\smash{\tau_{K,N'}^{(r)}}$. Moreover, note that to proceed by induction in this strategy, it is important to know that the respective ``midpoints'' come from timelike $\smash{\ell_p}$-optimal geodesic plans, which ensures that all chronology re\-lations are preserved through this process.

As in the proof of \autoref{Th:Good geos TCD}, $\smash{(\bdalpha^n)_{n\in\N}}$ admits an accumulation point $\bdalpha \in \smash{\OptTGeo_{\ell_p}^\tsep(\mu_0,\mu_1)}$, giving rise to a timelike proper-time parametrized $\smash{\ell_p}$-geodesic $(\mu_t)_{t\in[0,1]}$ from $\mu_0$ to $\mu_1$ via $\mu_t := (\eval_t)_\push\bdalpha$. Thanks to the weak lower semicontinuity of the Rényi entropy for probability measures with uniformly compact support, the latter satisfies, for every $t\in[0,1]$ and every $N'\geq N$,
\begin{align*}
\scrS_{N'}(\mu_t)\leq \sigma_{K,N'}^{(1-t)}(\theta)\,\scrS_{N'}(\mu_0) + \sigma_{K,N'}^{(t)}(\theta)\,\scrS_{N'}(\mu_1).
\end{align*}

As for \cite[Prop.~2.8]{bacher2010}, where the replacement of \cite[Rem.~2.9]{bacher2010} is \autoref{Le:Stu},  and \autoref{Pr:ii to iii}, and using \autoref{Th:Equivalence TCD* and TCDe} we obtain $\smash{\TCD_p^*(K,N)}$.
\end{proof}

\begin{definition}\label{Def:Chron lp geod} We term $\smash{\scrP_\comp^\ac(\mms,\meas)}$ \emph{chronologically $\smash{\ell_p}$-geodesic} if  every $\mu_0,\mu_1\in\smash{\scrP_\comp^\ac(\mms,\meas)}$ with $\supp\mu_0\times\supp\mu_1\subset\smash{\mms_\ll^2}$ are joined by a timelike proper-time parame\-tri\-zed $\smash{\ell_p}$-geodesic $(\mu_r)_{r\in[0,1]}$ which consists of $\meas$-absolutely continuous measures.
\end{definition}

\begin{theorem}\label{Th:Local to global} Let $K\in\R$ and $N\in[1,\infty)$. Then $\smash{\TCD_p^*(K,N)}$ holds if and only if $\smash{\TCD_{p,\loc}^*(K,N)}$ is satisfied and $\smash{\scrP_\comp^\ac(\mms,\meas)}$ is chronologically $\smash{\ell_p}$-geodesic.
\end{theorem}

By following the argument for \cite[Cor.~5.4]{cavalletti2017}, in fact we expect the assumption of chronological $\smash{\ell_{p}}$-geodesy to be redundant. We retain it for simplicity and outsource this slight generalization to future work.

We will prove \autoref{Th:Local to global} only in the case $K\geq 0$, the other situation follows  by analogous computations. To this aim, we formulate the following property in  \autoref{Def:Cn} for which, given any $m\in\N_0$, we define
\begin{align*}
I_m := \{k\,2^{-m} : k\in\{0,\dots,2^m\}\}.
\end{align*}
Moreover, given a collection $\mu:= (\mu_t)_{t\in[0,1]}$ in $\scrP(\mms)$, let $\smash{G_\mu^m}$ denote the set of all $\gamma\in\TGeo^\tsep(\mms)$ with $\gamma_t\in\supp\mu_t$ for every $t\in I_m$.

\begin{definition}\label{Def:Cn} Let $C\subset\mms$ be a compact set.  Given $m\in\N_0$, we say that the property $\rmP_m(C)$ is satisfied if for every timelike proper-time parametrized $\smash{\ell_p}$-geodesic $\mu:=(\mu_r)_{r\in[0,1]}$ consisting of $\meas$-absolutely continuous measures with $\mu_0,\mu_1\in\smash{\scrP^\ac(\mms,\meas)}$ with $\smash{\supp\mu_0\times\supp\mu_1\subset\mms_\ll^2\cap C^2}$, and every $s,t\in I_n$ with $t-s =2^{-m}$ there is a timelike $\smash{\ell_p}$-optimal geodesic plan $\smash{\bdpi_{s,t}\in\OptTGeo_{\ell_p}^\tsep(\mu_s,\mu_t)}$ with
\begin{align*}
\scrS_{N'}((\eval_{1/2})_\push\bdpi_{s,t}) \leq \sigma_{K,N'}^{(1/2)}(\theta_{s,t}^m)\,\scrS_{N'}(\mu_s) + \sigma_{K,N'}^{(1/2)}(\theta_{s,t}^m)\,\scrS_{N'}(\mu_t),
\end{align*}
where
\begin{align}\label{Eq:Thetat}
\theta_{s,t}^m := \inf\{\tsep(\gamma_s,\gamma_t) : \gamma\in G_\mu^m\}.
\end{align}
\end{definition}

\begin{lemma}\label{Le:PnC} Retaining the notation of \autoref{Def:Cn}, $\rmP_m(C)$ implies $\rmP_{m-1}(C)$ for every $m\in\N$.
\end{lemma}

\begin{proof} Let a timelike proper-time parametrized $\smash{\ell_p}$-geodesic $\smash{\mu:=(\mu_r)_{r\in[0,1]}}$ according to  $\rmP_n(C)$ be given. Let $s,t \in I_{m-1}$ such that $\smash{t-s = 2^{1-m}}$, and let $\smash{\theta:=\theta_{s,t}^{m-1}}$ be defined with respect to $\mu$. Inductively, we build a sequence $\smash{(\mu^i)_{i\in\N_0}}$ of timelike proper-time parametrized $\smash{\ell_p}$-geodesics $\smash{\mu^i := (\mu_r^i)_{r\in[0,1]}}$ obeying $\smash{\mu^i_r = \mu_r}$ for every $i\in\N_0$ and every $r\in[0,s]\cup[t,1]$ as follows. Initially, set $\smash{\mu^0 := \mu}$. 

\textbf{Step 1.} \textit{Construction for odd $i\in\N_0$.} If the element $\smash{\mu^{2i}}$ has been constructed for a given $i\in\N_0$, as in the proof of \cite[Thm.~2.10]{ambrosiogigli2013} and using that the respective ``midpoints'' result from a timelike $\smash{\ell_p}$-optimal geodesic plan interpolating its endpoints --- see also \autoref{Le:Concat} --- we construct a timelike proper-time parametrized $\smash{\ell_p}$-geodesic $\smash{\mu^{2i+1}:= (\mu_r^{2i+1})_{r\in[0,1]}}$ which has the subsequent properties. For every $r\in[0,s]\cup[t,1]$, we have $\smash{\mu_r^{2i+1}=\mu_r}$, and moreover, for every $N'\geq N$, 
\begin{align*}
\scrS_{N'}(\mu_{s+2^{-m-1}}^{2i+1}) &\leq \sigma_{K,N'}^{(1/2)}\Big[\frac{\theta}{2}\Big]\,\scrS_{N'}(\mu_s) + \sigma_{K,N'}^{(1/2)}\Big[\frac{\theta}{2}\Big]\,\scrS_{N'}(\mu_{s+2^{-m}}^{2i}),\\
\scrS_{N'}(\mu_{s+3\times 2^{-m-1}}^{2i+1}) &\leq \sigma_{K,N'}^{(1/2)}\Big[\frac{\theta}{2}\Big]\,\scrS_{N'}(\mu_{s+2^{-m}}^{2i}) + \sigma_{K,N'}^{(1/2)}\Big[\frac{\theta}{2}\Big]\,\scrS_{N'}(\mu_t).
\end{align*}
Here we have used the property $\rmP_m(C)$, the inequalities
\begin{align*}
2\,\theta_{s,s+2^{-m}}^m &\geq \theta,\\
2\, \theta_{s+2^{-m},t}^m &\geq \theta, 
\end{align*}
where both quantities on the left-hand sides are defined with respect to $\smash{\mu^{2i}}$, and nonincreasingness of $\smash{\sigma_{K,N'}^{(1/2)}(\vartheta)}$ in $\vartheta\geq 0$ by our assumption $K\geq 0$.

\textbf{Step 2.} \textit{Construction for even $i\in\N_0$.} Given $\smash{\mu^{2i+1}}$ exhibited in Step 1 above, similarly to this step we construct a timelike proper-time parametrized $\smash{\ell_p}$-geodesic $\smash{\mu^{2i+2} := (\mu_r^{2i+2})_{r\in[0,1]}}$ with the property that $\smash{\mu_r^{2i+2} = \mu_r^{2i+1}}$ for every $r\in[0,s+2^{-m-1}]\cup\smash{[s+3\times 2^{-m-1},1]}$ such that for every $N'\geq N$,
\begin{align*}
\scrS_{N'}(\mu_{s+2^{-m}}^{2i+2}) \leq \sigma_{K,N'}^{(1/2)}\Big[\frac{\theta}{2}\Big]\,\scrS_{N'}(\mu_{s+2^{-m-1}}^{2i+1}) + \sigma_{K,N'}^{(1/2)}\Big[\frac{\theta}{2}\Big]\,\scrS_{N'}(\mu_{s+3\times 2^{-m-1}}^{2i+1}).
\end{align*}

\textbf{Step 3.} \textit{Conclusion.} Pasting together the inequalities from Step 1 and Step 2 yields, for every $i\in\N_0$ and every $N'\geq N$,
\begin{align*}
\scrS_{N'}(\mu^{2i+2}_{s+2^{-m}}) &\leq 2\,\sigma_{K,N'}^{(1/2)}\Big[\frac{\theta}{2}\Big]^2\,\scrS_{N'}(\mu_{s+2^{-m}}^{2i})\\
&\qquad\qquad + \sigma_{K,N'}^{(1/2)}\Big[\frac{\theta}{2}\Big]^2\,\scrS_{N'}(\mu_s) + \sigma_{K,N'}^{(1/2)}\Big[\frac{\theta}{2}\Big]^2\,\scrS_{N'}(\mu_t). 
\end{align*}
Iterating this inequality gives
\begin{align*}
\scrS_{N'}(\mu_{s+2^{-m}}^{2i})&\leq 2^i\,\sigma_{K,N'}^{(1/2)}\Big[\frac{\theta}{2}\Big]^2\,\scrS_{N'}(\mu_{s+2^{-m}})\\
&\qquad\qquad + \frac{1}{2}\sum_{k=1}^i 2^k\,\sigma_{K,N'}^{(1/2)}\Big[\frac{\theta}{2}\Big]^{2k}\,\big[\scrS_{N'}(\mu_s) + \scrS_{N'}(\mu_t)\big].
\end{align*}

As $\smash{\mu_0^{2i} = \mu_0}$ and $\smash{\mu_1^{2i}= \mu_1}$ for every $i\in\N_0$ by construction, using \autoref{Le:Villani lemma for geodesic} we get  existence of a sequence $\smash{(\bdpi^{2i})_{i\in\N_0}}$ in $\smash{\OptTGeo_{\ell_p}^\tsep(\mu_0,\mu_1)}$ with $\smash{\mu_r^{2i} = (\eval_r)_\push\bdpi^{2i}}$ for every $r\in[0,1]$ which  converges weakly, up to a nonrelabeled subsequence, to a timelike $\smash{\ell_p}$-optimal geodesic plan $\smash{\bdpi\in\OptTGeo_{\ell_p}^\tsep(\mu_0,\mu_1)}$. By $\scrK$-global hyperbolicity, lower semicontinuity of $\scrS_{N'}$ on $\scrP(J(\mu_0,\mu_1))$, and the same computations as for  \cite[Clm.~5.2]{bacher2010} for the distortion coefficients on the right-hand side of the above bound, setting $\smash{\bdpi_{s,t} := (\Restr_s^t)_\push\bdpi\in\OptTGeo_{\ell_p}^\tsep(\mu_s,\mu_t)}$ we get
\begin{align*}
\scrS_{N'}((\eval_{1/2})_\push\bdpi_{s,t}) \leq \sigma_{K,N'}^{(1/2)}(\theta)\,\scrS_{N'}(\mu_s) + \sigma_{K,N'}^{(1/2)}(\theta)\,\scrS_{N'}(\mu_t).
\end{align*}
Here $\smash{\Restr_s^t}$ is defined in \eqref{Eq:Restr def}. This accomplishes $\rmP_{n-1}(C)$.
\end{proof}

\begin{proof}[Proof of \autoref{Th:Local to global}] The forward implication is a consequence of \autoref{Le:Stu}. 

Concerning the backward implication, by the partition  procedure in the proof of \autoref{Pr:ii to iii} it suffices to prove the desired inequality defining the condition $\smash{\TCD_p^*(K,N)}$ for every $\smash{\mu_0,\mu_1\in\scrP_\comp^\ac(\mms,\meas)}$ with $\supp\mu_0\times\supp\mu_1\subset\smash{\mms_\ll^2}$. 

Thanks to $\scrK$-global hyperbolicity, every member of any timelike proper-time parametrized $\smash{\ell_p}$-geodesic from $\mu_0$ to $\mu_1$ has support in $J(\mu_0,\mu_1)$. By $\smash{\TCD_{p,\loc}^*(K,N)}$ and by compactness of $J(\mu_0,\mu_1)$, there exist $\delta > 0$, a disjoint cover $L_1,\dots,L_n\subset\mms$ of $J(\mu_0,\mu_1)$, where $n\in\N$, and closed sets $\smash{C_k \subset\mms}$ containing $\smash{\sfB^\met(L_k,\delta)}$, $k\in\{1,\dots,n\}$, such that every two measures in $\smash{\scrP_\comp^\ac(C_k,\meas)}$ are joined by a timelike proper-time parametrized $\smash{\ell_p}$-geodesic obeying \eqref{Eq:RED}. Let $m\in \N$ with
\begin{align}\label{Eq:m choice}
2^{-m} \leq \delta.
\end{align}

We claim  $\smash{\rmP_m(J(\mu_0,\mu_1))}$. Let $(\mu_r)_{r\in[0,1]}$ be a timelike proper-time parametrized $\smash{\ell_p}$-geodesic from $\mu_0$ to $\mu_1$ consisting of $\meas$-absolutely continuous measures. Let $\bdpi\in\smash{\OptTGeo_{\ell_p}^\tsep(\mu_0,\mu_1)}$ represent  $(\mu_r)_{r\in[0,1]}$, and define $\smash{\varpi\in\scrP(\mms^{2^m+1})}$ by
\begin{align}\label{Eq:PROJ}
\varpi := (\eval_0,\eval_{2^{-m}}, \eval_{2^{-m+1}},\dots,\eval_{(2^m-1)2^{-m}},\eval_1)_\push\bdpi.
\end{align}
Let $s,t\in I_m$ with $t-s = 2^{-m}$. 
For $k\in\{1,\dots,n\}$, define $\smash{\nu_s^k,\nu_t^k\in\scrP_\comp^\ac(\mms,\meas)}$ by
\begin{align*}
\nu_s^k &:= \alpha_k^{-1}\,(\pr_{2^ms})_\push\varpi\mres L_k,\\
\nu_t^k &:= \alpha_k^{-1}\,(\pr_{2^mt})_\push\big[(\pr_{2^ms},\pr_{2^mt})_\push\bdpi\mres (L_k\times \mms)\big],
\end{align*}
provided $\smash{\alpha_k := \mu_s[L_k]>0}$. Note that $\smash{\supp\nu_s^k\subset \bar{L}_k}$ and $\smash{\supp\nu_s^k\times\supp\nu_t^k\subset\mms_\ll^2}$. Since $r:=\inf\tsep(\supp\mu_0\times\supp\mu_1) > 0$, every curve in $\TGeo^\tsep(\mms)$ starting in $\supp\mu_0$ and ending in $\supp\mu_1$ belongs to the set $G_r$ defined in \autoref{Cor:Cptness and equicty}, which is uniformly equi\-continuous.  Hence, using \eqref{Eq:m choice} and \eqref{Eq:PROJ} we get  $\smash{\supp\nu_s^k, \supp\nu_t^k \subset \bar{\sfB}^\met(L_k,\delta)\subset C_k}$. By $\smash{\TCD_{p,\loc}^*(K,N)}$ in the form described above, there is some timelike $\smash{\ell_p}$-optimal geodesic plan $\smash{\bdpi_{s,t}^k\in\OptTGeo_{\ell_p}^\tsep(\nu_s^k,\nu_t^k)}$ such that
\begin{align}\label{Eq:Mrg}
\scrS_{N'}((\eval_{1/2})_\push\bdpi_{s,t}^k)\leq \sigma_{K,N'}^{(1/2)}(\theta_{s,t}^m)\,\scrS_{N'}(\nu_s^k) + \sigma_{K,N}^{(1/2)}(\theta_{s,t}^m)\,\scrS_{N'}(\nu_t^k)
\end{align}
for every $N'\geq N$, where $\theta_{s,t}^m$ is defined as in \eqref{Eq:Thetat} with respect to $\mu := (\mu_r)_{r\in[0,1]}$.

Define $\smash{\bdpi_{s,t}\in\OptTGeo_{\ell_p}^\tsep(\mu_s,\mu_t)}$ by
\begin{align*}
\bdpi_{s,t} := \alpha_1\,\bdpi_{s,t}^1 + \dots + \alpha_n\,\bdpi_{s,t}^n.
\end{align*}
By construction and \autoref{Le:Mutually singular}, respectively, $\smash{(\eval_r)_\push\bdpi_{s,t}^1,\dots,(\eval_r)_\push\bdpi_{s,t}^n}$ are mutually singular for every fixed $r\in\{0,1/2\}$, whence
\begin{align*}
\scrS_{N'}(\mu_s) &= \alpha_1^{1-1/N'}\,\scrS_{N'}(\nu_s^1) + \dots + \alpha_1^{1-1/N'}\,\scrS_{N'}(\nu_s^n),\\
\scrS_{N'}((\eval_{1/2})_\push\bdpi_{s,t}) &= \alpha_1^{1-1/N'}\,\scrS_{N'}((\eval_{1/2})_\push\bdpi_{s,t}^1) + \dots + \alpha_n^{1-1/N'}\,\scrS_{N'}((\eval_{1/2})_\push\bdpi_{s,t}^n).
\end{align*} 
On the other hand, since mutual singularity of $\smash{\nu_t^1,\dots,\nu_t^n}$ may fail, 
\begin{align*}
\scrS_{N'}(\mu_t) \geq \alpha_1^{1-1/N'}\,\scrS_{N'}(\nu_t^1) + \dots + \alpha_n^{1-1/N'}\,\scrS_{N'}(\nu_t^n).
\end{align*}
Merging these inequalities with \eqref{Eq:Mrg} yields
\begin{align*}
\scrS_{N'}((\eval_{1/2})_\push\bdpi_{s,t}) \leq \sigma_{K,N'}^{(1/2)}(\theta_{s,t}^m)\,\scrS_{N'}(\mu_s) + \sigma_{K,N'}^{(1/2)}(\theta_{s,t}^m)\,\scrS_{N'}(\mu_t),
\end{align*}
and the desired property $\rmP_m(J(\mu_0,\mu_1))$ follows.

To close the proof of \autoref{Th:Local to global}, let $(\mu_r)_{r\in[0,1]}$ be a timelike proper-time parametrized $\smash{\ell_p}$-geodesic from $\mu_0$ to $\mu_1$ which consists of $\meas$-absolutely continuous measures, as hypothesized. Iteratively, the property $\rmP_m(J(\mu_0,\mu_1))$ proven above implies $\rmP_0(J(\mu_0,\mu_1))$ by \autoref{Le:PnC}. By \autoref{Propos}, this already implies the condition $\smash{\TCD_p^*(K,N)}$.
\end{proof}

The following \autoref{Cor:K-} easily follows from \autoref{Th:Local to global} and the stability result from \autoref{Th:Stability TCD}. \autoref{Pr:Pr} is proven along the same lines as \cite[Prop.~5.5]{bacher2010} by using uniform continuity of $\tsep$ on compact subsets of $\smash{\mms^2}$; the point  is that for small $\vartheta\geq 0$, the quantities  $\smash{\sigma_{K,N}^{(r)}(\vartheta)}$ and $\smash{\tau_{K,N}^{(r)}(\vartheta)}$ ``almost coincide''. See also the computations in \cite{deng}.

\begin{corollary}\label{Cor:K-} Let $K\in\R$ and $N\in[1,\infty)$. Then chronological $\smash{\ell_p}$-geodesy of $\smash{\scrP_\comp^\ac(\mms,\meas)}$ and the condition $\smash{\TCD_{p,\loc}^*(K',N)}$ hold for every $K' < K$ if and only if $\smash{\TCD_p^*(K,N)}$ is satisfied.
\end{corollary}

\begin{proposition}\label{Pr:Pr} Assume $\smash{\scrP_\comp^\ac(\mms,\meas)}$ to be chronologically $\smash{\ell_p}$-geodesic. Given any $K\in\R$ and $N\in[1,\infty)$,  $\smash{\TCD_{p,\loc}^*(K',N)}$ holds for every $K'<K$ if and only if $\smash{\TCD_{p,\loc}(K',N)}$ is satisfied for every $K'<K$.
\end{proposition}

\begin{remark}\label{Eq:Blurt} We do not know whether \autoref{Th:Local to global} holds for the stronger conditions $\smash{\TCD_{p,\loc}(K,N)}$ and $\smash{\TCD_p(K,N)}$, respectively; compare with \cite[Rem.~5.6]{bacher2010}. It would follow from \autoref{Th:Local to global} if we succeeded in proving that $\smash{\TCD_p^*(K,N)}$ and $\smash{\TCD_p(K,N)}$ are \emph{equivalent} on timelike $p$-essentially nonbranching spaces, similarly to the metric measure result \cite[Thm.~1.1]{cavallettimilman2021}; recall \autoref{Pr:TCD and TCD*}. 
\end{remark}

\subsection{Universal cover} Now we discuss an application of \autoref{Th:Local to global} to the universal cover of a Lorentzian pre-length space, assuming that the latter exists. By construction of the Lorentzian structure on the universal cover, the latter inherits the curvature properties of its base space locally, and hence globally if $\scrX$ is timelike $p$-essentially nonbranching; see \autoref{Th:TCD universal cover} below for details. The universal cover plays a key role in the smooth \emph{Gannon--Lee singularity theorems} \cite{gannon1975, gannon1976, lee1976}. (Admittedly, these assume the null energy condition instead of a timelike energy condition. See also the recent work \cite{schinnerl2021} for spacetimes with $\smash{\Cont^1}$-metrics, and the references therein.) More elementary, our discussion also creates further examples of spaces satisfying the $\smash{\TCD_p^*(K,N)}$ condition.

\subsubsection{Covering spaces}  We start by recalling the following basic definition  due to \cite[p.~62, p.~80]{spanier1966} (given for general topological spaces therein).

\begin{definition} Let $(\mms,\met)$ be a metric space. 
\begin{enumerate}[label=\textnormal{\alph*.}]
\item Given a topological space $\smash{\hat{\mms}}$, a map $\smash{\cov\colon\hat{\mms}\to\mms}$ is called \emph{covering map} if every point in  $\mms$ has an open neighborhood $U\subset\mms$ such that $\smash{\cov^{-1}(U)}$ is the disjoint union of open subsets of $\smash{\hat{\mms}}$ on each of which $\cov$ restricts to a homeomorphism.
\item A topological space $\smash{\hat{\mms}}$ is called \emph{covering space} for $\mms$ if there exists a covering map $\smash{\cov\colon\hat{\mms}\to\mms}$.
\item A topological space $\smash{\hat{\mms}}$ is called \emph{universal cover} for $\mms$ if it is a covering space for $\mms$, and for every covering space $\smash{\check{\mms}}$ for $\mms$ \textnormal{(}with covering map $\smash{\textnormal{\texttt{q}}}$\textnormal{)} there exists a continuous map $\smash{F \colon \hat{\mms}\to \check{\mms}}$ such that
\begin{align*}
\cov = F\circ \textnormal{\texttt{q}}.
\end{align*}
\end{enumerate}
\end{definition}

\begin{remark} A universal cover might fail to exist in general \cite[Ex.~2.5.17]{spanier1966}, but if it exists, it is unique up to homeomorphism. If $\mms$ is connected, locally path-connected \cite[p.~65]{spanier1966}, and semi-locally simply connected \cite[p.~78]{spanier1966}, then it admits a simply connected universal cover \cite[Cor.~2.5.14]{spanier1966}. If one drops semi-local simple  connectedness, a universal cover might still exist, but fail to be simply connected \cite[Ex.~2.5.18]{spanier1966}.
\end{remark}

Every covering projection $\smash{\cov\colon\hat{\mms}\to\mms}$ is a local homeomorphism \cite[Lem.~2.1.8]{spanier1966}. Hence, defining $\smash{\hat{\met}\colon\hat{\mms}^2\to[0,\infty)}$ by 
\begin{align}\label{Eq:d hut}
\hat{\met}(\hat{x},\hat{y}) := \met(\cov(\hat{x}),\cov(\hat{y})),
\end{align}
the  covering space $\smash{\hat{\mms}}$ of a metric space is automatically a metric space with respect to $\smash{\hat{\met}}$, and $\smash{\hat{\met}}$ induces the topology of $\smash{\hat{\mms}}$. 

\begin{remark}\label{Re:Bruh} By definition of a covering map and \eqref{Eq:d hut}, for every compact $C\subset\mms$ there exists an open subset $U\subset\mms$ containing $C$ such that $\smash{\cov\big\vert_{\cov^{-1}(U)}}$ is an isometric isomorphism,  in particular bi-Lipschitz continuous. This implies that properness of $\met$ lifts to $\smash{\hat{\met}}$. Moreover, a curve $\smash{\hat{\gamma}}\colon[a,b]\to\smash{\hat{\mms}}$, $a,b\in\R$ with $a<b$, is $\smash{\hat{\met}}$-rectifiable if and only if $\smash{\gamma := \cov\circ\hat{\gamma}}$ is $\met$-rectifiable, and in either case $\smash{\Len_{\hat{\met}}(\hat{\gamma}) = \Len_\met(\gamma)}$.
\end{remark}

\subsubsection{Lift of Lorentzian structures} In this subsection, we fix a Lorentzian pre-length space $(\mms,\met,\ll,\leq,\tsep)$. In the remainder of this section, we assume the existence of a universal cover $\smash{\hat{\mms}}$ for $\mms$ with covering map $\cov$. 

We introduce a natural lifted Lorentzian pre-length structure on $\smash{\hat{\mms}}$;  various  Lorentzian properties will lift to $\smash{\hat{\mms}}$ as well. We confine ourselves to those introduced in \autoref{Sub:Lorentzian nonsmooth} and needed for our work, although we expect  more of the vast list of properties from  \cite{cavalletti2020,kunzinger2018} to lift to $\smash{\hat{\mms}}$.

The following \autoref{Pr:Lift} essentially follows from \autoref{Re:Bruh} and the way $\smash{\hll}$, $\smash{\hleq}$, and $\smash{\hat{\tsep}}$ below are defined. (In the language of \autoref{Def:Lor isometry}, $\cov$ is a Lorentzian isometric embedding.) In particular, a curve $\smash{\hat{\gamma}\colon[a,b]\to\hat{\mms}}$, $a,b\in\R$ with $a<b$, is causal or timelike if and only if $\gamma:= \cov\circ\hat{\gamma}$ is causal or timelike, respectively, and in either case $\smash{\Len_{\hat{\tsep}}(\hat{\gamma}) = \Len_\tsep(\gamma)}$. Moreover, the former is a geodesic if and only if the latter is. We omit the straightforward proof.

\begin{proposition}\label{Pr:Lift} Define two relations $\smash{\hat{\leq}}$ and $\smash{\hat{\ll}}$ on $\smash{\hat{\mms}}$ by $\hat{x}\hleq \hat{y}$ provided $\cov(\hat{x})\leq \cov(\hat{y})$, and $\hat{x} \hll\hat{y}$ provided $\cov(\hat{x})\ll \cov(\hat{y})$. Moreover, define  $\smash{\hat{\tsep}\colon\hat{\mms}^2\to [0,\infty]}$ by 
\begin{align*}
\hat{\tsep}(\hat{x},\hat{y}) &:= \tsep(\cov(\hat{x}),\cov(\hat{y})).
\end{align*}
Then $\smash{(\hat{\mms},\hat{\met},\hat{\ll},\hat{\leq},\hat{\tsep})}$ is a Lorentzian pre-length space. If $(\mms,\met,\ll,\leq,\tsep)$ is
\begin{enumerate}[label=\textnormal{\textcolor{black}{(}\roman*\textcolor{black}{)}}]
\item\label{La:Un} causally closed,
\item\label{La:Deux} locally causally closed,
\item\label{La:Trois} causally path-connected,
\item\label{La:NT} non-totally imprisoning,
\item\label{La:Quatre} localizable,
\item regular,
\item\label{La:Six} globally hyperbolic,
\item\label{La:Sept} $\scrK$-globally hyperbolic, 
\item\label{La:Len} length, or
\item\label{La:Geo} geodesic,
\end{enumerate}
then the respective property holds for $\smash{(\hat{\mms},\hat{\met},\hat{\ll},\hat{\leq},\hat{\tsep})}$ as well.
\end{proposition}

\subsubsection{Lift of curvature properties} Given a  measured Lorentzian pre-length space $\scrX$, using the covering map $\cov$ the Lorentzian structure on the universal cover from \autoref{Pr:Lift} can be naturally endowed with a reference measure $\smash{\hat{\meas}}$ as follows. Let $\smash{\hat{\Sigma}}$ be  the family of all sets $\smash{\hat{B}\subset \hat{\mms}}$ such that $\smash{\cov\big\vert_{\hat{B}}}$ is a local isometry with Borel image $\smash{\cov(\hat{B})\subset\mms}$. This family is closed under any number of intersections, and since $\smash{\hat{\met}}$ is proper by \autoref{Re:Bruh}, the $\sigma$-algebra $\smash{\sigma(\hat{\Sigma})}$ coincides with the Borel $\sigma$-algebra on $\smash{\hat{\mms}}$. The definition $\smash{\hat{\meas}[\hat{B}] := \meas[\cov(\hat{B})]}$, $\smash{B\in\hat{\Sigma}}$, thus evidently gives rise to the desired (lifted) Borel measure $\smash{\hat{\meas}}$ on $\smash{\hat{\mms}}$.

For the subsequent theorem, we add the hypothesis of timelike $p$-essential nonbranching, $p\in (0,1)$, to our standing assumptions on $\scrX$.

\begin{theorem}\label{Th:TCD universal cover} Let $p\in (0,1)$, $K\in\R$, and $N\in[1,\infty)$. Assume $\smash{\TCD_p^*(K,N)}$ for $\scrX$. Then the causally closed, globally hyperbolic measured Lorentzian geodesic space $\smash{\hat{\scrX} := (\hat{\mms},\hat{\met},\hat{\meas},\hll,\hleq,\hat{\tsep})}$ according to \autoref{Pr:Lift} is timelike $p$-essentially nonbranching, and it obeys $\smash{\TCD_p^*(K,N)}$.
\end{theorem}

\begin{proof} Composition with the covering map $\cov$ yields a canonical nonrelabeled and continuous map $\smash{\cov\colon \TGeo^\tsep(\hat{\mms})\to \TGeo^\tsep(\mms)}$. In an evident notation, given any timelike $\smash{\hat{\ell}_p}$-optimal geodesic plan $\smash{\hat{\bdpi}\in\scrP(\TGeo^\tsep(\hat{\mms}))}$, the plan $\smash{\bdpi := \cov_\push\hat{\bdpi}}$ lies in $\scrP(\TGeo^\tsep(\mms))$. Given a timelike nonbranching set $\smash{G\subset\TGeo^\tsep(\mms)}$ with $\bdpi[G]=1$, the set $\smash{\hat{G}} := \smash{\cov^{-1}(G)\subset\TGeo^\tsep(\hat{\mms})}$ is timelike nonbranching as well, and $\smash{\hat{\bdpi}[\hat{G}]=1}$. Hence, timelike $p$-essential nonbranching lifts to the Lorentzian structure from \autoref{Pr:Lift}.

Thanks to the construction given in \autoref{Pr:Lift} as well as \autoref{Re:Bruh}, the spaces $\smash{\hat{\scrX}}$ and $\smash{\scrX}$ are locally isometric as metric and Lorentzian pre-length spaces; recall \autoref{Sec:Stability TCD}. Therefore, the condition $\smash{\TCD_p^*(K,N)}$ lifts and makes $\smash{\hat{\scrX}}$ obey $\smash{\TCD_p^*(K,N)}$ locally. By lifting of timelike proper-time parametrized $\smash{\ell_p}$-geodesics witnessing the inequality defining $\smash{\TCD_p^*(K,N)}$ for $\scrX$, $\smash{\scrP_\comp^\ac(\hat{\mms},\hat{\meas})}$ is chronologically $\smash{\ell_p}$-geodesic according to \autoref{Def:Chron lp geod}. Using \autoref{Pr:Lift} again with \autoref{Th:Local to global},  $\smash{\TCD_p^*(K,N)}$ thus holds globally on the universal cover.
\end{proof}

\section{Timelike measure-contraction property}\label{Ch:TMCP}

In this chapter, we introduce and study Lorentzian analogues of the mea\-sure-contraction property for metric measure spaces  \cite{ohta2007,sturm2006b}. An entropic version has recently been introduced in \cite{cavalletti2020}, cf.~\autoref{Def:TMCPe} below. Again, for smooth space\-times the respective conditions are tightly linked to timelike lower Ricci curvature bounds, cf.~\autoref{Th:TMCP Smooth}.

Since the timelike measure contraction properties from \autoref{Def:TMCP} are independent of any transport exponent, cf.~\autoref{Re:Indep transp exp}, all statements in this chapter requiring such exponent will hold for \emph{every} $p\in (0,1)$.

\subsection{Definition and basic properties} 

\begin{definition}\label{Def:TMCP} Let $K\in\R$, and $N\in[1,\infty)$. 
\begin{enumerate}[label=\textnormal{\alph*.}]
\item We say that $\scrX$ obeys the \emph{reduced timelike measure-contraction property} $\smash{\TMCP^*(K,N)}$ if for every $\smash{\mu_0\in\scrP_\comp^\ac(\mms,\meas)}$ and every $x_1\in I^+(\mu_0)$ there is a timelike proper-time parametrized $\smash{\ell_{1/2}}$-geodesic $(\mu_t)_{t\in[0,1]}$ from $\mu_0$ to $\smash{\mu_1 := \delta_{x_1}}$ such that for every $t\in[0,1)$ and every $N'\geq N$,
\begin{align*}
\scrS_{N'}(\mu_t) \leq -\int_\mms \sigma_{K,N'}^{(1-t)}(\tsep(x^0, x_1))\,\rho_0(x^0)^{-1/N'}\d\mu_0(x^0).
\end{align*}
\item If the previous statement holds for $\smash{\tau_{K,N'}^{(1-t)}}$ instead of $\smash{\sigma_{K,N'}^{(1-t)}}$, $\scrX$ is said to satisfy the \emph{timelike measure-contraction property} $\smash{\TMCP(K,N)}$.
\end{enumerate}
\end{definition}

\begin{remark} As in \autoref{Re:TCD 0}, the conditions  $\smash{\TMCP^*(0,N)}$ and $\smash{\TMCP(0,N)}$ co\-incide for every $N\in[1,\infty)$.
\end{remark}

\begin{remark}[Compare also with {\cite[Rem.~2.4]{cavalletti2022}}]\label{Re:Indep transp exp} Every curve $(\mu_t)_{t\in [0,1]}$ between measures $\mu_0$ and $\mu_1$ as in \autoref{Def:TMCP} is a proper-time parametrized $\smash{\ell_{1/2}}$-geodesic if and only if it is a proper-time parametrized $\smash{\ell_p}$-geodesic for some $p\in (0,1)$ (and consequently, for \emph{every} $p\in (0,1)$). Indeed, the product measure $\mu_0\otimes\mu_1$ is only one coupling of $\mu_0$ and $\mu_1$, and it is chronological, hence $\smash{\ell_p}$-optimal for every $p$ as above. In particular, the timelike measure-contraction properties from \autoref{Def:TMCP} do not depend on the transport exponent.
\end{remark}

\begin{remark}\label{Re:mae x1} Under our standing assumptions on $\scrX$, $\smash{\TMCP^*(K,N)}$ holds if and only if the statement from \autoref{Def:TMCP} is  merely satisfied for $\meas$-a.e. $x_1\in \mms$ thanks to  \autoref{Le:Villani lemma for geodesic} and Fatou's lemma. The analogous statement holds for $\smash{\TMCP(K,N)}$ as well.
\end{remark}

As for \autoref{Pr:TCD and TCD*} and \autoref{Pr:Consistency TCD}, the statements from \autoref{Pr:TMCP to TMCP*} and  \autoref{Pr:Consistency TMCP} below are readily established (in all generality, the proofs work for general measured Lorentzian pre-length spaces). 

\begin{proposition}\label{Pr:TMCP to TMCP*} The following hold for every $K\in\R$, and $N\in[1,\infty)$; as in \autoref{Pr:TCD and TCD*}, set $K^* := K(N-1)/N$.
\begin{enumerate}[label=\textnormal{(\roman*)}]
\item  The condition $\smash{\TMCP(K,N)}$ implies $\smash{\TMCP^*(K,N)}$.
\item If $K>0$, the condition $\smash{\TMCP^*(K,N)}$ implies $\smash{\TMCP(K^*,N)}$.
\end{enumerate}
\end{proposition}

\begin{proposition}\label{Pr:Consistency TMCP} Assume $\smash{\TMCP^*(K,N)}$ for $K\in\R$, and $N\in[1,\infty)$. Then the following statements hold.
\begin{enumerate}[label=\textnormal{(\roman*)}]
\item \textnormal{\textbf{Consistency.}} The condition $\smash{\TMCP^*(K',N')}$ holds for every $K'\leq K$ and every $N'\geq N$.
\item \textnormal{\textbf{Scaling.}} Given any $a,b,\theta>0$, the rescaled measured Lorentzian pre-length space $(\mms, a\met, b\meas,\ll,\leq, \theta\tsep)$ satisfies $\smash{\TMCP^*(K/\theta^2,N)}$.
\end{enumerate}

Analogous statements hold for the $\smash{\TMCP(K,N)}$ condition.
\end{proposition}

\begin{remark}\label{Re:Unlikely} An analogue of \autoref{Le:Stu} is unlikely to hold under $\smash{\TMCP^*(K,N)}$ or $\smash{\TMCP(K,N)}$. In particular, we do not know a priori whether the measures $\mu_t$, $t\in[0,1)$, of a timelike proper-time parametrized $\smash{\ell_p}$-geodesic as in \autoref{Def:TMCP} are always --- or can be chosen  to be --- $\meas$-absolutely continuous. Instead, we establish the existence of timelike proper-time parametrized $\smash{\ell_p}$-geodesics consisting of $\meas$-absolutely continuous measures whose densities are uniformly $\Ell^\infty$ in time under $\smash{\TMCP^*(K,N)}$, at least if $\smash{\mu_0\in\scrP_\comp^\ac(\mms,\meas)}$ has $\meas$-essentially bounded density, by a variational method following \cite{braun2022}, see also \cite{cavalletti2017,rajala2012a,rajala2012b} in \autoref{Sub:Good}; cf.~\autoref{Th:Good TMCP}. By uniqueness of timelike proper-time parametrized $\smash{\ell_p}$-geodesics under timelike essential nonbranching conditions, see \autoref{Th:Uniqueness geodesics}, this  suffices to show an equivalence result analogous to \autoref{Th:Equivalence TCD* and TCDe} for our timelike measure-contraction property, see \autoref{Th:Equivalence TMCP* and TMCPe} and \autoref{Th:Equivalence TMCP}.
\end{remark}

The link of $\smash{\TMCP^*(K,N)}$ and $\smash{\TMCP(K,N)}$ to their respective $\TCD$ counterparts is more subtle. The proof of the corresponding \autoref{Pr:TMCP to TCD} follows  \cite[Prop.~3.11]{cavalletti2020}. We need the following result whose proof is a straigtforward adaptation of the argument for \cite[Lem.~3.3]{sturm2006b}, recall also \autoref{Le:Const perturb}. The difference to \cite{sturm2006b} is that we only require the marginal whose density is contained in the respective integrand to be constant.

\begin{lemma}\label{Le:USC lemma} Let $K\in\R$ and $N\in[1,\infty)$. Let  $\rho\colon\mms\to [0,\infty)$ be a Borel function with $\mu:= \rho\,\meas\in\scrP(\mms)$. Moreover, let $(\pi_n)_{n\in\N}$ be a sequence in $\scrP(\mms^2)$ converging weakly to  $\pi\in\scrP(\mms^2)$  such that for some $i\in\{1,2\}$, we have 
\begin{align*}
(\pr_i)_\push\pi_n = \mu
\end{align*}
for every $n\in\N$. Then for every $r\in[0,1]$,
\begin{align*}
&\int_{\mms^2} \tau_{K,N}^{(r)}(\tsep(x^0,x^1))\,(\rho\circ\pr_i)(x^0,x^1)^{-1/N} \d\pi(x^0,x^1)\\
&\qquad\qquad\leq \liminf_{n\to\infty}\int_{\mms^2} \tau_{K,N}^{(r)}(\tsep(x^0,x^1))\,(\rho\circ\pr_i)(x^0,x^1)^{-1/N}\d\pi_n(x^0,x^1).
\end{align*}

An analogous assertion holds with $\smash{\tau_{K,N}^{(r)}}$ replaced by $\smash{\sigma_{K,N}^{(r)}}$.
\end{lemma}

\begin{proposition}\label{Pr:TMCP to TCD}  The following  hold for every $p\in (0,1)$, $K\in\R$, and $N\in[1,\infty)$.
\begin{enumerate}[label=\textnormal{\textcolor{black}{(}\roman*\textcolor{black}{)}}]
\item\label{LART} The condition $\smash{\wTCD_p(K,N)}$ implies $\smash{\TMCP(K,N)}$.
\item\label{LARTT} The condition $\smash{\wTCD_p^*(K,N)}$ implies $\smash{\TMCP^*(K,N)}$.
\end{enumerate}
\end{proposition}

\begin{proof} We only prove \ref{LART}, the proof of \ref{LARTT} is similar. 

\textbf{Step 1.} \textit{Approximation of $\mu_0$ and $\mu_1$.}  
Given any $\varepsilon\in (0,1)$, let $C_\varepsilon \subset\mms$ be a compact set with $\mu_0[C_\varepsilon] \geq 1-\varepsilon$ and $\smash{C_\varepsilon \times \{x_1\}\subset\mms_\ll^2}$. Define $\smash{\mu_0^\varepsilon\in\scrP_\comp^\ac(\mms,\meas)}$ by
\begin{align*}
\mu_0^\varepsilon := \mu_0[C_\varepsilon]^{-1}\,\mu_0\mres C_\varepsilon = \rho_0^\varepsilon\,\meas
\end{align*}
Moreover, since $\smash{\mms_\ll^2}$ is open while $C_\varepsilon\times\{x_1\}$ is compact, there exists $\eta > 0$ such that  $\smash{C_\varepsilon\times \sfB^\met(x_1,\eta)\subset\mms_\ll^2}$. For $\delta\in (0,\eta)$, we set 
\begin{align*}
\mu_1^\delta := \meas\big[\sfB^\met(x_1,\delta)\big]^{-1}\,\meas\mres \sfB^\met(x_1,\delta) = \rho_1^\delta\,\meas.
\end{align*}

By \autoref{Re:Strong timelike}, the pair $\smash{(\mu_0^\varepsilon,\mu_1^\delta)}$ is strongly  timelike $p$-dualizable, and using $\smash{\wTCD_p(K,N)}$ we find a timelike proper-time parametrized $\smash{\ell_p}$-geodesic $\smash{(\mu_t^{\varepsilon,\delta})_{t\in[0,1]}}$ connecting $\smash{\mu_0^\varepsilon}$ to $\smash{\mu_1^\delta}$ as well as a timelike $p$-dualizing $\smash{\pi^{\varepsilon,\delta}\in\Pi_\ll(\mu_0^\varepsilon,\mu_1^\delta)}$, with support in $\smash{C_\varepsilon\times \sfB^\met(x_1,\eta)}$, such that for every $t\in[0,1]$ and every $N'\geq N$,
\begin{align*}
\scrS_{N'}(\mu_t^{\varepsilon,\delta}) &\leq -\int_{\mms^2} \tau_{K,N'}^{(1-t)}(\tsep(x^0,x^1))\,\rho_0^\varepsilon(x^0)^{-1/N'}\d\pi^{\varepsilon,\delta}(x^0,x^1)\\
&\qquad\qquad -\int_{\mms^2}\tau_{K,N'}^{(1-t)}(\tsep(x^0,x^1))\,\rho_1^\delta(x^1)^{-1/N'}\d\pi^{\varepsilon,\delta}(x^0,x^1)\\
&\leq -\int_{\mms^2}\tau_{K,N'}^{(1-t)}(\tsep(x^0,x^1))\,\rho_0^\varepsilon(x^0)^{-1/N'}\d\pi^{\varepsilon,\delta}(x^0,x^1).
\end{align*}

\textbf{Step 2.} \textit{Sending $\delta \to 0$.} Given any $\varepsilon > 0$ and a fixed sequence $(\delta_n)_{n\in\N}$ decreasing to $0$, from $\smash{(\mu_t^{\varepsilon,\delta_n})_{t\in[0,1]}}$, $n\in\N$, we construct a timelike proper-time parametrized $\smash{\ell_p}$-geodesic from $\mu_0^\varepsilon$ to $\mu_1$ as follows. Let $\smash{\bdpi^{\varepsilon,\delta_n}\in\OptTGeo_{\ell_p}^\tsep(\mu_0^\varepsilon,\mu_1^{\delta_n})}$ represent $\smash{(\mu_t^{\varepsilon,\delta_n})_{t\in[0,1]}}$. As $\smash{(\mu_1^{\delta_n})_{n\in\N}}$  converges weakly to $\mu_1$, this sequence is tight, and so is $\smash{(\bdpi^{\varepsilon,\delta_n})_{n\in\N}}$ by \autoref{Le:Villani lemma for geodesic}. Let $\smash{\bdpi^\varepsilon\in \OptTGeo_{\ell_p}^\tsep(\mu_0^\varepsilon,\mu_1)}$ be a weak limit of a nonrelabeled subsequence. Then the assignment $\smash{\mu_t^\varepsilon := (\eval_t)_\push\bdpi^\varepsilon}$, $t\in[0,1]$, gives rise to a timelike proper-time parametrized $\smash{\ell_p}$-geodesic connecting $\smash{\mu_0^\varepsilon}$ to $\mu_1$.


Every weak limit point of $\smash{(\pi^{\varepsilon,\delta_n})_{n\in\N}}$ equals $\smash{\mu_0^\varepsilon\otimes\mu_1 = \mu_0^\varepsilon\otimes\delta_{x_1}}$. 
Therefore, $\scrK$-global hyperbolicity,  weak lower semicontinuity of $\smash{\scrS_{N'}}$ on measures with uniformly bounded support, and \autoref{Le:USC lemma} yield, for every such $t$ and every $N'\geq N$, 
\begin{align*}
\scrS_{N'}(\mu_t^\varepsilon) &\leq \limsup_{n\to\infty} \scrS_{N'}(\mu_t^{\varepsilon,\delta_n})\\
&\leq -\liminf_{n\to\infty}\int_{\mms^2}\tau_{K,N'}^{(1-t)}(\tsep(x^0,x^1))\,\rho_0^\varepsilon(x^0)^{-1/N'}\d\pi^{\varepsilon,\delta_n}(x^0,x^1)\\
&\leq -\int_{\mms^2} \tau_{K,N'}^{(1-t)}(\tsep(x^0,x_1))\,\rho_0^\varepsilon(x^0)^{-1/N'}\d\mu_0^\varepsilon(x^0)
\end{align*}

\textbf{Step 3.} \textit{Sending $\varepsilon \to 0$.} Given a sequence $(\varepsilon_n)_{n\in\N}$ decreasing to $0$, note that $\smash{(\mu_0^{\varepsilon_n})_{n\in\N}}$ converges weakly to $\mu_0$.  Similarly to  Step 2, from $\smash{(\mu_t^{\varepsilon_n})_{n\in\N}}$ we construct a timelike proper-time parametrized $\smash{\ell_p}$-geodesic $(\mu_t)_{t\in[0,1]}$ connecting $\mu_0$ to $\mu_1$. As in Step 2 and using Fatou's lemma we get, for every $t\in[0,1]$ and every $N'\geq N$,
\begin{align*}
\scrS_{N'}(\mu_t) &\leq \limsup_{n\to\infty} \scrS_{N'}(\mu_t^{\varepsilon_n})\\
&\leq -\liminf_{n\to\infty}\int_{\mms^2}\tau_{K,N'}^{(1-t)}(\tsep(x^0,x_1))\,\rho_0^{\varepsilon_n}(x^0)^{-1/N'}\d\mu_0^{\varepsilon_n}(x^0)\\
&\leq -\int_{\mms^2}\tau_{K,N'}^{(1-t)}(\tsep(x^0,x_1))\,\rho_0(x^0)^{-1/N'}\d\mu_0(x^0).
\end{align*}
The claim follows from \autoref{Re:Indep transp exp}.
\end{proof}

\begin{remark}\label{Re:Geom inequ TMCP} With essentially identical proofs, the respective versions of the timelike Bonnet--Myers inequality, \autoref{Cor:Bonnet-Myers} and \autoref{Cor:Reduced BM}, and the timelike Bishop--Gromov inequality, \autoref{Th:BG} and \autoref{Th:Reduced BG}, hold as well under $\smash{\TMCP(K,N)}$ and $\smash{\TMCP^*(K,N)}$. In the timelike Brunn--Minkowski inequalities, of course the set $A_1$ is just a singleton, and $\meas[A_1] := 0$.
\end{remark}

\subsection{Stability}  Next, we discuss the stability of the notions from \autoref{Def:TMCP}. In contrast to the weak stability from \autoref{Th:Stability TCD}, both timelike measure-contraction properties are stable under the convergence introduced in \autoref{Def:Convergence}.

\begin{theorem}\label{Th:Stability TMCP} Assume the convergence of $(\scrX_k)_{k\in\N}$ to $\scrX_\infty$ as in  \autoref{Def:Convergence}. Moreover, let $(K_k,N_k)_{k\in\N}$ be a sequence in $\R\times [1,\infty)$ converging to $(K_\infty,N_\infty)\in\R\times[1,\infty)$. Suppose $\scrX_k$ obeys $\TMCP(K_k,N_k)$ for every $k\in\N$. 
Then $\smash{\scrX_\infty}$ satisfies $\smash{\TMCP(K_\infty,N_\infty)}$. 

The analogous statement in which  $\smash{\TMCP(K_k,N_k)}$ is respectively replaced by $\smash{\TMCP^*(K_k,N_k)}$, $k\in\N_\infty$, holds as well.
\end{theorem}

\begin{proof} It suffices to prove the first statement, the second  is argued analogously. 

In this proof, we often adopt notations from the proof of \autoref{Th:Stability TCD} without explicit notice. In particular, we again identify $\mms_k$ with its image $\iota_k(\mms_k)$ in $\mms$ and $\meas_k$ with its push-forward $(\iota_k)_\push\meas_k$ for every $k\in\N_\infty$. 

\textbf{Step 1.} \textit{Reduction to compact $\mms$.} Owing to \autoref{Re:mae x1}, given any $\mu_{\infty,0} = \rho_{\infty,0}\,\meas_\infty\in \scrP_\comp^\ac(\mms,\meas_\infty)$ and  $x_{\infty,1}\in I^+(\mu_{\infty,0}) \cap \supp\meas_\infty$, as for \autoref{Th:Stability TCD} we use the compactness of $\supp\mu_{\infty,0}$ and $\supp\mu_{\infty,1}$, where $\smash{\mu_{\infty,1} := \delta_{x_{\infty,1}}}$, to assume without restriction that $\mms$ is compact, and that $\meas_k\in\scrP(\mms)$ for every $k\in\N_\infty$. All  measures considered below will thus be compactly supported. Moreover, we may and will suppose that $W_2(\meas_k,\meas_\infty)\to 0$ as $k\to\infty$.

\textbf{Step 2.} \textit{Restriction of the assumptions on $\mu_{\infty,0}$ and $\mu_{\infty,1}$.} We will first assume that $\tau(\cdot,x_{\infty,1})$ is bounded away from zero on $\supp\mu_{\infty,0}$, that $\rho_{\infty,0}\in\Ell^\infty(\mms,\meas_\infty)$, and that $x_{\infty,1}$ can be approximated with respect to $\met$ by a sequence $(x_{k,1})_{k\in\N}$ of points $x_{k,1}\in \supp\meas_k$ such that $x_{\infty,1} \in I^-(x_{k,1})$ for every $k\in\N$. The general case is discussed in Step 7 below; we note for now that this conclusion will not conflict with our reductions from Step 1.

\textbf{Step 3.} \textit{Construction of a chronological recovery sequence.} 
In this step, given the sequence $(x_{k,1})_{k\in\N}$ from Step 2 we construct a sequence $(\mu_{k,0})_{k\in\N}$ of measures $\smash{\mu_{k,0} = \rho_{k,0}\,\meas_k\in\scrP^\ac(\mms,\meas_k)}$, $k\in\N$, such that $\mu_{k,0}\to\mu_{\infty,0}$ weakly as $k\to \infty$ possibly up to extracting a subsequence, and  $x_{k,1}\in I^+(\mu_{k,0})$ for every $k\in\N$. The constructed sequence will allow for the correct behavior of all  functionals under consideration, cf.~Step 5 below.

Given any $k\in\N$, let $\mathfrak{q}_k\in\scrP(\mms^2)$ be a $W_2$-optimal coupling of $\meas_k$ and $\meas_\infty$. We disintegrate $\mathfrak{q}_k$ with respect to $\pr_1$, writing
\begin{align*}
\rmd \mathfrak{q}_k(x,y) = \rmd\mathfrak{p}_x^k(y)\d\meas_k(x).
\end{align*}
Let $\smash{\mathfrak{p}^k\colon \scrP^\ac(\mms,\meas_\infty)\to\scrP^\ac(\mms,\meas_k)}$ denote the canonically induced map.

Given any $k\in\N$, define $\tilde{\mu}_{k,0}\in\scrP^\ac(\mms,\meas_k)$ by
\begin{align*}
\tilde{\mu}_{k,0} := \mathfrak{p}^k(\mu_{\infty,0}) = \tilde{\rho}_{k,0}\,\meas_k.
\end{align*}
Note that the measure $\mathfrak{r}_k := (\rho_{\infty,0}\circ\pr_2)\,\mathfrak{q}_k\in\scrP(\mms^2)$  constitutes a coupling of $\tilde{\mu}_{k,0}$ and $\mu_{\infty,0}$. With the $W_2$-optimality of $\mathfrak{q}_k$,  this  implies
\begin{align*}
W_2(\tilde{\mu}_{k,0},\mu_{\infty,0}) \leq  \big\Vert\rho_{\infty,0}\big\Vert_{\Ell^\infty(\mms,\meas_\infty)}^{1/2}\,W_2(\meas_k,\meas_\infty),
\end{align*}
and consequently $\smash{\tilde{\mu}_{k,0}\to \mu_{\infty,0}}$ weakly as $k\to\infty$. Hence, by our assumption on the se\-quence $(x_{k,1})_{k\in\N}$ from Step 2 and Portmanteau's theorem,
\begin{align*}
\liminf_{k\to\infty} \tilde{\mu}_{k,0}[I^-(x_{k,1})] \geq \liminf_{k\to\infty} \tilde{\mu}_{k,0}[I^-(x_{\infty,1})] \geq \mu_{\infty,0}[I^-(x_{\infty,1})] =1.
\end{align*}
Up to passing to a subsequence we may and will thus assume $\smash{\tilde{\mu}_{k,0}[I^-(x_{k,1})] >0}$ for every $k\in\N$. We then define $\smash{\mu_{k,0}\in\scrP^\ac(\mms,\meas_k)}$ by
\begin{align*}
\mu_{k,0} := \tilde{\mu}_{k,0}[I^-(x_{k,1})]^{-1} \tilde{\mu}_{k,0} \mres I^-(x_{k,1}).
\end{align*}
By construction, the sequence $(\mu_{k,0})_{k\in\N}$ converges to $\mu_{\infty,0}$ weakly as $k\to\infty$, and we have $x_{k,1}\in I^+(\mu_{k,0})$ for every $k\in\N$, as desired.

\textbf{Step 4.} \textit{Invoking the $\TMCP$ condition.} Fix $K\in\R$ and $N\in (1,\infty)$ such that $K< K_\infty$ and $N>N_\infty$. Up to passing to a subsequence if necessary, we may and will thus assume that $K<K_k$ and $N>N_k$ for every $k\in\N$.

By \autoref{Pr:Consistency TMCP}, for every $k\in\N$ there exists a timelike proper-time para\-metrized $\smash{\ell_{1/2}}$-geodesic $(\mu_{k,t})_{t\in[0,1]}$ connecting $\mu_{k,0}$ and $\smash{\mu_{k,1} = \delta_{x_{k,1}}}$ such that for every $t\in [0,1)$ and every $N'\geq N$,
\begin{align}\label{Eq:TMCP cond}
\scrS_{N'}^k(\mu_{k,t}) \leq -\int_{\mms} \tau_{K,N'}^{(1-t)}(\tsep(x^0,x_{k,1}))\,\rho_{k,0}(x^0)^{-1/N'} \d\mu_{k,0}(x^0).
\end{align}

\textbf{Step 5.} \textit{Estimating the previous right-hand side.} We estimate the negative of the right-hand side of \eqref{Eq:TMCP cond} from below, up to errors which become arbitrarily small.  By construction of $\mu_{k,0}$, we find a sequence $(a_k)_{k\in\N}$ of normalization constants converging to $1$ such that $\smash{\rho_{k,0} \leq a_k\,\tilde{\rho}_{k,0}}$ $\meas_k$-a.e.~for every $k\in\N$.

\textbf{Step 5.1.} First, we argue that $x_{k,1}$ can be replaced $x_{\infty,1}$ up to a small error. By our choice of $K$ and $N$ and the timelike Bonnet--Myers inequality from  \autoref{Re:Geom inequ TMCP},
\begin{align*}
c := \sup\tau_{K,N'}^{(1-t)}\circ\tsep(\mms^2)
\end{align*}
is finite. Furthermore, the involved function is jointly uniformly continuous. Thus, given any $\varepsilon> 0$, by uniform continuity we obtain, for sufficiently large $k\in\N$,
\begin{align*}
&\int_\mms \tau_{K,N'}^{(1-t)}(\tsep(x^0,x_{k,1}))\,\rho_{k,0}(x^0)^{-1/N'}\d\mu_{k,0}(x^0)\\
&\qquad\qquad \geq \int_\mms \tau_{K,N'}^{(1-t)}(\tsep(x^0,x_{\infty,1}))\,\rho_{k,0}(x^0)^{-1/N'}\d\mu_{k,0}(x^0)-\varepsilon.
\end{align*}

\textbf{Step 5.2.} We modify $\smash{\tilde{\mu}_{k,0}}$ into the measure
\begin{align*}
\nu_{k,0} := (1+\delta_k)^{-1}\,(\tilde{\rho}_{k,0}+\delta_k)\,\meas_k =\varrho_{k,0}\,\meas_k,
\end{align*} 
where $\delta_k\in[0,1]$ is defined by
\begin{align*}
\delta_k = \tilde{\mu}_{k,0}[I^-(x_{k,1})^\sfc].
\end{align*}
Note that $(\delta_k)_{k\in\N}$ converges to $0$. Employing the inequality $\rho_{k,0} \leq a_k\,\varrho_{k,0}$ $\meas_k$-a.e.,  the definition of $\mu_{k,0}$, and a similar notation as in the proof of \autoref{Th:Stability TCD}, 
\begin{align}\label{Eq:Integral rechnung}
&\int_{\mms} \tau_{K,N'}^{(1-t)}(\tsep(x^0,x_{\infty,1}))\,\rho_{k,0}(x^0)^{-1/N'} \d\mu_{k,0}(x^0)\nonumber\\
&\qquad\qquad \geq_k\int_{\mms}\tau_{K,N'}^{(1-t)}(\tsep(x^0,x_{\infty,1}))\,\varrho_{k,0}(x^0)^{-1/N'}\d\mu_{k,0}(x^0)\nonumber\\
&\qquad\qquad \geq_k\int_{\mms}\tau_{K,N'}^{(1-t)}(\tsep(x^0,x_{\infty,1}))\,\varrho_{k,0}(x^0)^{-1/N'}\d\nu_{k,0}(x^0)\\
&\qquad\qquad\qquad\qquad - \int_{I^-(x_{k,1})^\sfc} \tau_{K,N'}^{(1-t)}(\tsep(x^0,x_{\infty,1}))\,\varrho_{k,0}(x^0)^{-1/N'}\d\tilde{\mu}_{k,0}(x^0)\nonumber\\
&\qquad\qquad\qquad\qquad - \delta_k\int_{\mms}\tau_{K,N'}^{(1-t)}(\tsep(x^0,x_{\infty,1}))\,\varrho_{k,0}(x^0)^{-1/N'}\d\meas_k(x^0).\nonumber
\end{align}

\textbf{Step 5.3.} Possibly invoking the timelike Bishop--Gromov inequality outlined in \autoref{Re:Geom inequ TMCP}, we obtain the estimates
\begin{align*}
c\,\delta_k^{1-1/N'} &\geq \int_{I^-(x_{k,1})^\sfc}\tau_{K,N'}^{(1-t)}(\tsep(x^0,x_{\infty,1}))\,\varrho_{k,0}(x^0)^{-1/N'}\d\tilde{\mu}_{k,0}(x^0)\\
c\,\delta_k^{1-1/N'} &\geq \delta_k\int_{\mms} \tau_{K,N'}^{(1-t)}(\tsep(x^0,x_{\infty,1}))\,\varrho_{k,0}(x^0)^{-1/N'}\d\meas_k(x^0).
\end{align*}
Hence, our main task is to estimate the integral in \eqref{Eq:Integral rechnung} from below. To this aim, by  definition of $\nu_{k,0}$ and Jensen's inequality, we obtain
\begin{align*}
&\int_\mms \tau_{K,N'}^{(1-t)}(\tsep(x^0,x_{\infty,1}))\,\varrho_{k,0}(x^0)^{-1/N'}\d\nu_{k,0}(x^0)\\
&\qquad\qquad = \int_\mms \tau_{K,N'}^{(1-t)}(\tsep(x^0,x_{\infty,1}))\,\Big[\!\int_\mms (\rho_{\infty,0}(y^0)+\delta_k)\d\mathfrak{p}_{x^0}^k(y^0)\Big]^{1-1/N'}\!\d\meas_k(x^0)\\
&\qquad\qquad \geq \int_{\mms^2} \tau_{K,N'}^{(1-t)}(\tsep(x^0,x_{\infty,1}))\,\rho_{\infty,0}(y^0)^{1-1/N'}\d\mathfrak{q}_k(x^0,y^0).
\end{align*}
Let $(\phi_i)_{i\in\N}$ be a sequence in $\Cont_\bounded(\mms)$ such that
\begin{align*}
\big\Vert \phi_i- \rho_{\infty,0}^{1-1/N'}\big\Vert_{\Ell^1(\mms,\meas_\infty)} \leq 2^{-i}
\end{align*}
as well as $\sup\phi_i(\mms) \leq \Vert\rho_{\infty,0} \Vert_{\Ell^\infty(\mms,\meas_\infty)}$ for every $i\in\N$. Then we get
\begin{align*}
&\int_\mms \tau_{K,N'}^{(1-t)}(\tsep(x^0,x_{\infty,1}))\,\varrho_{k,0}(x^0)^{-1/N'}\d\nu_{k,0}(x^0)\\
&\qquad\qquad \geq \int_{\mms^2} \tau_{K,N'}^{(1-t)}(\tsep(x^0,x_{\infty,1}))\,\phi_i(y^0)\d\mathfrak{q}_k(x^0,y^0) - 2^{-i}\,c.
\end{align*}

\textbf{Step 5.4.} By tightness and stability of $W_2$-optimal couplings \cite[Lem.~4.3, Lem. 4.4]{villani2009}, $(\mathfrak{q}_k)_{k\in\N}$ converges weakly to the dia\-gonal coupling of $\meas_\infty$ and $\meas_\infty$ along a nonrelabeled subsequence. In particular, by Lebesgue's theorem,
\begin{align*}
&\liminf_{k\to\infty} \int_\mms \tau_{K,N'}^{(1-t)}(\tsep(x^0,x_{\infty,1}))\,\varrho_{k,0}(x^0)^{-1/N'}\d\nu_{k,0}(x^0)\\
&\qquad\qquad \geq \liminf_{i\to\infty} \int_\mms \tau_{K,N'}^{(1-t)}(\tsep(x^0,x_{\infty,1}))\,\phi_i(x^0)\d\meas_\infty(x^0) - c\,\limsup_{i\to\infty}2^{-i}\\
&\qquad\qquad =\int_\mms \tau_{K,N'}^{(1-t)}(\tsep(x^0,x_{\infty,1}))\,\rho_{\infty,0}(x^0)^{-1/N'}\d\mu_{\infty,0}(x^0).
\end{align*}

\textbf{Step 6.} \textit{Conclusion.} Let $(\bdpi_k)_{k\in\N}$ be a sequence of $\smash{\bdpi_k\in\OptTGeo_{\ell_{1/2}}^\tsep(\mu_{k,0},\mu_{k,1})}$ representing the $\smash{\ell_{1/2}}$-geodesic $(\mu_{k,t})_{t\in[0,1]}$ in Step 4, $k\in\N$. By compactness of $\mms$ and \autoref{Le:Villani lemma for geodesic}, this sequence converges to a timelike $\smash{\ell_{1/2}}$-optimal geodesic plan $\smash{\bdpi_\infty\in\OptTGeo_{\ell_{1/2}}^\tsep(\mu_{\infty,0},\mu_{\infty,1})}$ along a nonrelabeled subsequence. In particular, the assignment $\smash{\mu_{\infty,t} := (\eval_t)_\push\bdpi_\infty}$ gives rise to a timelike proper-time parametrized $\smash{\ell_{1/2}}$-geodesic $(\mu_{\infty,t})_{t\in[0,1]}$ from $\mu_{\infty,0}$ to $\mu_{\infty,1}$; note that every optimal coupling of $\mu_{\infty,0}$ and $\mu_{\infty,1}$ is concentrated on $\smash{\mms_\ll^2}$. By joint weak lower semicontinuity of the Rényi entropy and \eqref{Eq:TMCP cond}, given any $\varepsilon > 0$, $t\in[0,1)$, and $N'\geq N$,
\begin{align*}
\scrS_{N'}^\infty(\mu_{\infty,t}) &\leq \limsup_{k\to\infty} \scrS_{N'}^k(\mu_{k,t})\\
&\leq -\liminf_{k\to\infty} \int_\mms \tau_{K,N'}^{(1-t)}(\tsep(x^0,x_{k,1}))\,\rho_{k,0}(x^0)^{-1/N'}\d\mu_{k,0}(x^0)\\
&\leq \varepsilon - \int_\mms \tau_{K,N'}^{(1-t)}(\tsep(x^0,x_{\infty,1}))\,\rho_{\infty,0}(x^0)^{-1/N'}\d\mu_{\infty,0}(x^0).
\end{align*}



\textbf{Step 7.} \textit{Relaxation of the assumptions on $\mu_\infty,0$ and $\mu_{\infty,1}$.} First, we argue how to construct the hypothesized sequence $(x_{k,1})_{k\in\N}$ from Step 2. The idea is to approximate points which lie ``in between'' $\supp\mu_{\infty,0}$ and $x_{\infty,1}$ from the future. For these points, the above discussion applies, and we will be able to conclude the desired $\TMCP$ property by a tightness argument.

\textbf{Step 7.1.} As $\inf\tau(\supp\mu_{\infty,0}, x_{\infty,1})>0$, given any $i\in \N$ we may and will fix a sequence $\smash{(x_{\infty,1}^i)_{i\in\N}}$ of points $\smash{x_{\infty,1}^i}\in I(\mu_{\infty,0},\mu_{\infty,1})\cap\supp\meas_\infty$ converging to $\smash{x_{\infty,1}}$. Since $I(x_{\infty,1}^i,x_{\infty,1})$ is open and nonempty,  for every $i\in\N$ we  construct a sequence $\smash{(x_{\infty,1}^{i,j})_{j\in\N}}$ such that $\smash{x_{\infty,1}^{i,j} \in I(x_{\infty,1}^i, x_{\infty,1})\cap\supp\meas_\infty}$ and
\begin{align*}
x_{\infty,1}^{i,j} \ll x_{\infty,1}^{i,j-1}
\end{align*}
for every $j\in\N$. Since $\smash{x_{\infty,1}^{i,j}\in\supp\meas_\infty}$, the weak convergence of $(\meas_k)_{k\in\N}$ implies the existence of a sequence $\smash{(x_{k,1}^{i,j})_{k\in\N}}$ of points $\smash{x_{k,1}^{i,j}\in\supp\meas_k}$ which converges to $\smash{x_{\infty,1}^{i,j}}$. In particular, for a sufficiently large integer $k_j\in \N$, we have 
\begin{align*}
x_{\infty,1}^i\ll x_{\infty,1}^{i,j} \ll x_{k_j,1}^{i,j}
\end{align*}
for every $i,j\in\N$, as well as $\smash{\met(x_{k_j}^{i,j}, x_{\infty,1}^i) \to 0}$ as $j\to\infty$ for every $i\in\N$. 

The above arguments applied for a fixed $i\in\N$ (along a suitable subsequence in $k$ which depends on $i$) yield the existence of a timelike proper-time parametrized $\smash{\ell_{1/2}}$-geodesic $\smash{(\mu_{\infty,t}^i)_{t\in[0,1]}}$ from $\mu_{\infty,0}$ to $\smash{\mu_{\infty,1}^i:= \delta_{x_{\infty,1}^i}}$ such that for every $t\in[0,1)$ and every $N'\geq N_\infty$,
\begin{align*}
\scrS_{N'}^\infty(\mu_{\infty,t}^i) \leq -\int_\mms \tau_{K,N'}^{(1-t)}(\tsep(x^0,x_{\infty,1}^i))\,\rho_{\infty,0}(x^0)^{-1/N'}\d\mu_{\infty,0}(x^0).
\end{align*}
Similarly to Step 6 and using lower semicontinuity of Rényi's entropy and Fatou's lemma, we find a timelike proper-time parametrized $\smash{\ell_p}$-geodesic $\smash{(\mu_{\infty,t})_{t\in[0,1]}}$ connecting $\smash{\mu_{\infty,0}}$ to $\smash{\mu_{\infty,1}}$ such that the previous inequality holds for $\smash{\mu_{\infty,t}^i}$ and $\smash{x_{\infty,1}^i}$ replaced by $\smash{\mu_{\infty,t}}$ and $\smash{x_{\infty,1}}$,  respectively, for every $t\in[0,1)$ and every $N'\geq N_\infty$.

\textbf{Step 7.2.} Finally, we remove the two assumptions on $\mu_{\infty,0}$ from Step 2. Let $\mu_{\infty,0}\in\smash{\scrP^\ac(\mms,\meas_\infty)}$ and $\smash{x_{\infty,1}\in I^+(\mu_{\infty,0})\cap\supp\meas_\infty}$. Given  sufficiently large $n,m\in\N$, define $\smash{\mu_{\infty,0}^n,\mu_{\infty,0}^{n,m}\in\scrP^\ac(\mms,\meas_\infty)}$ by
\begin{align*}
\mu_{\infty,0}^n &:= \mu_{\infty,0}\big[\{\tsep(\cdot, x_{\infty,1}) \geq 2^{-n}\}\big]^{-1}\,\mu_{\infty,0} \mres \{\tsep(\cdot, x_{\infty,1}) \geq 2^{-n}\} = \rho_{\infty,0}^n\,\meas_\infty,\\
\mu_{\infty,0}^{n,m} &:= \big\Vert \!\min\{\rho_{\infty,0}^n,m\}\big\Vert_{\Ell^1(\mms,\meas_\infty)}^{-1}\,\min\{\rho_{\infty,0}^n,m\}\,\meas_\infty.
\end{align*}
By continuity of $\tsep$, $\smash{\mu_{\infty,0}^{n,m}}$ obeys  the hypotheses from Step 2 for every $n,m\in\N$. A diagonal procedure gives rise to a map $m\colon \N\to\N$ such that  $\smash{(\tilde{\mu}_{\infty,0}^n)_{n\in\N}}$, where 
\begin{align*}
\tilde{\mu}_{\infty,0}^n := \mu_{\infty,0}^{n,m_n},
\end{align*}
converges weakly to $\smash{\mu_{\infty,0}}$ as $n\to\infty$.

The above discussion thus yields the desired inequality along some timelike proper-time parametrized $\smash{\ell_{1/2}}$-geodesic connecting $\smash{\tilde{\mu}_{\infty,0}^n}$ to $\smash{\mu_{\infty,1} := \delta_{x_{\infty,1}}}$. The usual tightness and lower semicontinuity argument directly gives the claim; note that the convergence of the corresponding right-hand sides is granted by Levi's theorem.

\textbf{Step 8.} \textit{Passage from $K$ and $N$ to $K_\infty$ and $N_\infty$.} According to Step 4, all in all we thus infer $\smash{\TMCP(K,N)}$ for $\scrX_\infty$ for every $K<K_\infty$ and every $N>N_\infty$, and we deduce  $\smash{\TMCP(K_\infty,N_\infty)}$ as in Step 9 in the proof of \autoref{Th:Stability TCD}.
\end{proof}

\subsection{Good geodesics}\label{Sub:Good} Now we prove an analogue of \autoref{Th:Good geos TCD}, assuming the $\smash{\TMCP^*(K,N)}$ condition. As indicated in \autoref{Re:Unlikely}, this result will be used in the next \autoref{Sec:Equiv TMCP's} to prove the equivalence of the latter condition with the entropic timelike measure-contraction property from \cite{cavalletti2020}. To this aim, not only a uniform $\smash{\Ell^\infty}$-bound for the interpolating densities, but also a qualitative entropy bound has to be established; cf.~the proof of \autoref{Pr:More than TMCP} and \eqref{Eq:mm}.

Throughout this section, let $p\in (0,1)$ be arbitrary yet fixed.

As in \autoref{Sub:Good TCD}, we only outline the proof and refer to \cite[Sec.~4.3]{braun2022}, see also \cite{cavalletti2017}, for a similar discussion in the entropic case.

The following result is proven similarly to \autoref{Le:Lalelu}.

\begin{lemma}\label{Le:Lululu} Assume the $\smash{\TMCP^*(K,N)}$ condition for some  $K\in\R$, and $N\in[1,\infty)$. Let $\smash{\mu_0 = \rho_0\,\meas\in\scrP_\comp^\ac(\mms,\meas)}$, and let $x_1\in I^+(\mu_0)$. Finally, let $D$ be any real number no smaller than $\sup\tsep(\supp\mu_0\times\{x_1\})$. Then there is a timelike proper-time parametrized $\smash{\ell_p}$-geodesic $(\mu_t)_{t\in[0,1]}$ from $\mu_0$ to $\smash{\mu_1 := \delta_{x_1}}$ such that for every $t\in [0,1)$, denoting by $\rho_t$ the $\meas$-ab\-solutely continuous part of $\mu_t$,
\begin{align*}
\meas\big[\{\rho_t>0\}\big] \geq (1-t)^N\,\rme^{-tD\sqrt{K^-N}}\,\big\Vert\rho_0\big\Vert_{\Ell^\infty(\mms,\meas)}^{-1}.
\end{align*}
\end{lemma}

For the proof of \autoref{Th:Good TMCP}, we introduce some further terminology. Following \cite[Sec.~1.1]{cavalletti2020}, the \emph{causally reversed} Lorentzian structure $\smash{(\mms,\met,\ll^\leftarrow,\leq^\leftarrow,\tsep^\leftarrow)}$ corresponding to $(\mms,\met,\ll,\leq,\tsep)$ is given by
\begin{itemize}
\item $\smash{x\ll^\leftarrow y}$ if and only if $y\ll x$,
\item $\smash{x\leq^\leftarrow y}$ if and only if $y\leq x$, and
\item $\smash{\tsep^\leftarrow(x,y) := \tsep(y,x)}$.
\end{itemize}
Let $\smash{\ell_p^\leftarrow}$ be the total transport cost function \eqref{Eq:Totalcost} induced by $\smash{\tsep^\leftarrow}$, $p\in (0,1]$.

\begin{theorem}\label{Th:Good TMCP} Under the hypotheses of \autoref{Le:Lululu}, there exists a timelike proper-time parametrized $\smash{\ell_p}$-geodesic $(\mu_t)_{t\in[0,1]}$ from $\mu_0$ to $\smash{\mu_1:= \delta_{x_1}}$ such that for every $t\in[0,1]$, $\smash{\mu_t = \rho_t\,\meas\in\scrP^\ac(\mms,\meas)}$ and
\begin{align}\label{Eq:556}
\Vert\rho_t\Vert_{\Ell^\infty(\mms,\meas)} \leq \frac{1}{(1-t)^N}\,\rme^{Dt\sqrt{K^-N}}\,\Vert\rho_0\Vert_{\Ell^\infty(\mms,\meas)}
\end{align}
with the entropy inequality
\begin{align}\label{Eq:557}
\scrS_N(\mu_t) \leq (1-t)\,\rme^{-Dt\sqrt{K^-/N}}\,\scrS_N(\mu_0),
\end{align}
where
\begin{align*}
D:= \sup\tsep(\supp\mu_0\times\{x_1\}).
\end{align*}
\end{theorem}

\begin{proof} As for \cite[Thm.~4.18]{braun2022}, we first observe that the bisection argument from the proof of  \autoref{Th:Good geos TCD} does not apply here since $\smash{\mu_1\notin\scrP^\ac(\mms,\meas)}$, while we aim to establish that $\smash{\mu_t \in\scrP^\ac(\mms,\meas)}$ for every $t\in[0,1)$.

Given $n\in\N$ and $k\in\N_0$, we set $\smash{s_n^k := (1-2^{-n})^k}$. Fix $n\in\N$, and assume that a measure $\smash{\bdbeta_n^k\in \OptTGeo_{\ell_p^\leftarrow}^{\tsep^\leftarrow}(\mu_1,\mu_0)}$ has already been defined in such a way that for every $i\in\{0,\dots,k\}$, we have $\smash{(\eval_{s_n^i})_\push\bdbeta_n^k \in\scrP_\comp^\ac(\mms,\meas)}$ and
\begin{align*}
\sup\tsep^\leftarrow(\{x_1\}\times\supp\,(\eval_{s_n^i})_\push\bdbeta_n^k)\leq 2^{-n}\,s_n^{i-1}\,D
\end{align*}
for which \eqref{Eq:556} and \eqref{Eq:557} hold in the evident form.

\textbf{Step 1.} \textit{Minimization of an appropriate functional.} Arguing as in \cite[Prop.~3.14, Cor.~3.15]{braun2022} and using \autoref{Le:Lululu},
\begin{align}\label{Eq:BLUBBB}
\min\{ \scrF_{c_n^{k+1}}((\eval_{2^{-n}})_\push\bdpi ) : \bdpi \in \OptTGeo_{\ell_p}^\tsep((\eval_{s_n^k})_\push\bdbeta_n^k,\mu_1)\}= 0,
\end{align}
where we set 
\begin{align*}
c_n^{k+1} := \frac{1}{(1-2^{-n})^N}\,\rme^{2^{-n}s_n^kD\sqrt{K^-N}}\,\Vert\rho_0\Vert_{\Ell^\infty(\mms,\meas)}.
\end{align*}
Note that \autoref{Le:Lululu} does \emph{not}  assert absolute $\meas$-continuity of the $t$-slice of the timelike proper-time parametrized $\smash{\ell_p}$-geodesic under consideration. However, this does not affect the proof of \cite[Prop.~3.14]{braun2022} as a possible further  $\meas$-singular part gives an upper bound for the difference  $\scrE_{c'}(\bdpi) - \scrE_{c'}(\bdalpha)$ in Step 3 therein. 

Furthermore, as in \cite[Lem.~3.5]{braun2022} the functional $\scrS_N \circ \eval_{2^{-n}}$ admits a minimizer $\smash{\bdpi_n^{k+1}\in\OptTGeo_{\ell_p}^\tsep((\eval_{s_n^k})_\push\bdbeta_n^k,\mu_1)}$. Let $\bdsigma_{k+1}^n\in\smash{\OptTGeo_{\ell_p^\leftarrow}^\tsep(\mu_1,(\eval_{s_n^k})_\push\bdbeta_n^k)}$ denote the timelike $\smash{\ell_p^\leftarrow}$-optimal geodesic plan that is  obtained by ``time-reversal'' of $\smash{\bdpi_n^{k+1}}$. By gluing, we build $\smash{\bdbeta_n^{k+1}\in\OptTGeo_{\ell_p^\leftarrow}^{\tsep^\leftarrow}(\mu_1,\mu_0)}$ with
\begin{align*}
(\Restr_0^{s_n^k})_\push\bdbeta_n^{k+1} &= \bdpi_n^{k+1},\\
(\Restr_{s_n^k}^1)_\push\bdbeta_n^{k+1} &= (\Restr_{s_n^k}^1)_\push\bdbeta_n^k.
\end{align*}
Here, given any $s,t\in[0,1]$ with $s<t$, the map $\smash{\Restr_s^t}$ is defined in \eqref{Eq:Restr def}. Moreover, 
\begin{align*}
\sup \tsep^\leftarrow(\{x_1\}\times \supp\,(\eval_{s_n^{k+1}})_\push\bdbeta_n^{k+1})\leq 2^{-n}\,s_n^k\,D.
\end{align*}

\textbf{Step 2.} \textit{Uniform density bound.} From \eqref{Eq:BLUBBB} and arguing as for \cite[Prop.~3.16]{braun2022}, the functional $\smash{\scrF_{c_n^{k+1}}\circ(\eval_{2^{-n}})}$ vanishes at $\smash{\bdpi_n^{k+1}}$. As a consequence of our construction, $\smash{(\eval_{s_n^{k+1}})_\push\bdbeta_n^{k+1} = \rho_{s_n^{k+1}}\,\meas\in\scrP_\comp^\ac(\mms,\meas)}$ and, by geometric summation,
\begin{align}\label{Eq:Blubbbb}
\big\Vert\rho_{s_n^{k+1}}\big\Vert_{\Ell^\infty(\mms,\meas)} \leq \frac{1}{(s_n^{k+1})^N}\,\rme^{(1-s_n^{k+1})D\sqrt{K^-N}}\,\Vert\rho_0\Vert_{\Ell^\infty(\mms,\meas)}.
\end{align}
Moreover, arguing as in Step 6 in the proof of \cite[Thm.~3.1]{cavalletti2017},
\begin{align}\label{Eq:Ent bc}
\scrS_N((\eval_{s_n^{k+1}})_\push\bdbeta_n^{k+1}) \leq \scrS_N(\mu_0)\,s_n^{k+1}\,\rme^{-(1-s_n^{k+1})D\sqrt{K^-/N}}.
\end{align}

\textbf{Step 3.} \textit{Completion.} Iterating the procedure in Step 1 and Step 2, we construct a family $\smash{\{\bdbeta_n^k : n \in\N,\, k\in\N_0\}\subset\OptTGeo_{\ell_p^\leftarrow}^{\tsep^\leftarrow}(\mu_1,\mu_0)}$. Let $(\bdbeta^n)_{n\in\N}$ be an enumeration of the elements of the former set. By \autoref{Le:Villani lemma for geodesic} --- observe that $\smash{(\mms,\met,\ll^\leftarrow,\leq^\leftarrow,\tsep^\leftarrow)}$ inherits the regularity properties of its causally reversed structure --- the latter admits an accumulation point $\smash{\bdbeta\in\OptTGeo_{\ell_p^\leftarrow}^{\tsep^\leftarrow}(\mu_1,\mu_0)}$. Let $\smash{\bdalpha\in\OptTGeo_{\ell_p}^\tsep(\mu_0,\mu_1)}$ be the ``time-reversal'' of $\bdbeta$, which induces a timelike proper-time parametrized $\smash{\ell_p}$-geodesic $(\mu_t)_{t\in[0,1]}$ from $\mu_0$ to $\mu_1$ by definition. By weak lower semicontinuity of $\scrF_c$ in $\scrP(J(\mu_0,\mu_1))$ for appropriate values $c>0$, recall \eqref{Eq:BLUBBB} and \eqref{Eq:Blubbbb}, we get $\mu_t = \rho_t\,\meas\in\scrP_\comp^\ac(\mms,\meas)$, and $\Vert\rho_t\Vert_{\Ell^\infty(\mms,\meas)}$ obeys the desired upper bound for every $t\in[0,1)$. Also \eqref{Eq:Ent bc} passes to the limit by weak lower semicontinuity of $\scrS_N$ on uniformly compactly supported probability measures.
\end{proof}

\subsection{Uniqueness of $\ell_p$-optimal couplings and $\ell_p$-geodesics}\label{Sub:Uniqueneess} In this section, we prove uniqueness of \emph{chronological} $\smash{\ell_p}$-optimal couplings (if existent), \autoref{Th:Uniqueness couplings}, as well as of timelike $\smash{\ell_p}$-optimal geodesic plans,  \autoref{Th:Uniqueness geodesics}. We follow \cite[Sec.~3.4]{cavalletti2020}, see also \cite{cavalletti2017}. A byproduct of our discussion is an extension of  the results from \cite{cavalletti2020} from the timelike nonbranching to the timelike \emph{essential}  nonbranching case, cf.~\autoref{Re:From TNB to TENB}.

In this section, in addition to our standing assumptions, let $\scrX$ be timelike $p$-essentially nonbranching for some fixed $p\in (0,1)$. 

\begin{lemma}\label{Le:Uniqueness Diracs} Assume $\smash{\TMCP^*(K,N)}$ for some  $K\in\R$ and $N\in [1,\infty)$. Let $\mu_0 = \rho_0\,\meas\in\smash{\scrP_\comp^\ac(\mms,\meas)}$, and define $\mu_1 := \lambda_1\,\delta_{x_1} + \dots + \lambda_n\,\delta_{x_n}$ for  $\lambda_1,\dots,\lambda_n\in (0,1]$ with $\lambda_1+\dots+\lambda_n=1$ and pairwise distinct $x_1,\dots,x_n\in\mms$. Let $\pi\in\Pi_\ll(\mu_0,\mu_1)$ be an  $\smash{\ell_p}$-optimal coupling. Then there is a $\mu_0$-measurable map $T\colon \supp\mu_0\to\mms$ with 
\begin{align*}
\pi = (\Id,T)_\push\mu_0,
\end{align*}
and consequently,
\begin{align*}
\ell_p(\mu_0,\mu_1)^p = \int_\mms \tsep(x,T(x))^p\d\mu_0(x).
\end{align*}

In particular, $\pi$ is the unique chronological $\smash{\ell_p}$-optimal coupling of $\mu_0$ and $\mu_1$.
\end{lemma}

\begin{proof} Without restriction, we assume $N>1$.

By a standard argument, cf.~the proof of \cite[Lem.~3.17]{cavalletti2020}, it suffices to prove the existence of a $\mu_0$-measurable map $T\colon \supp\mu_0\to\mms$ such that $\pi=(\Id,T)_\push\mu_0$. To this aim, define the compact set $E\subset\supp\mu_0$ by
\begin{align*}
E := \{x\in \mms : \#[(\{x\}\times\mms)\cap \supp\pi] \geq 2\}.
\end{align*}
We claim that $\mu_0[E] = 0$, which directly gives the desired $T$.

Assume to the contrapositive that $\mu_0[E]>0$. We first reduce the discussion to the uniform distribution on some subset of $\mms$ as follows. Up to shrinking $E$, we assume without restriction that $\varepsilon\leq \rho_0\leq 1/\varepsilon$ $\meas$-a.e.~on $E$ for some $\varepsilon > 0$. A further possible shrinking of $E$ entails  the existence of  distinct points $z_1,z_2\in\{x_1,\dots,x_n\}$ and well-defined maps $T_1,T_2\colon E\to \mms$ with $T_1(x) = z_1$ and $T_2(x)=z_2$ for every $x\in E$. We may and will additionally assume that $\smash{E\times\{z_1\}, E\times\{z_2\}\subset \mms_\ll^2}$. By restric\-tion \cite[Lem.~2.10]{cavalletti2020}, the couplings
\begin{align*}
\pi_1 &:= \meas[E]^{-1}\,\One_{E\times\{z_1\}}\,(\rho_0\circ\pr_1)^{-1}\,\pi,\\
\pi_2 &:= \meas[E]^{-1}\,\One_{E\times\{z_2\}}\,(\rho_0\circ\pr_1)^{-1}\,\pi
\end{align*}
are $\smash{\ell_p}$-optimal with $\smash{\pi_1\in \Pi_\ll(\nu_0,\delta_{z_1})}$ and $\smash{\pi_2\in\Pi_\ll(\nu_0,\delta_{z_2})}$, respectively, where
\begin{align*}
\nu_0 := \meas[E]^{-1}\,\meas\mres E = \varrho_0\,\meas.
\end{align*}

From now on, we will work with $\nu_0$ instead of $\mu_0$. Let $\smash{\bdpi_1\in\OptTGeo_{\ell_p}^\tsep(\nu_0,\delta_{z_1})}$ and $\smash{\bdpi_2\in\OptTGeo_{\ell_p}^\tsep(\nu_0,\delta_{z_2})}$ be timelike $\smash{\ell_p}$-optimal geodesic plans which represent timelike proper-time parametrized $\smash{\ell_p}$-geodesics as given by \autoref{Th:Good TMCP}. Note that $(\eval_0)_\push\bdpi_1$ and $(\eval_0)_\push\bdpi_2$ are mutually singular since
\begin{align*}
\bdpi_i[\{\gamma\in\TGeo^\tsep(\mms) : \gamma_1 = z_i\}] = 1
\end{align*}
for every $i\in\{1,2\}$, and $z_1\neq z_2$. Note that $(\eval_t)_\push\bdpi_1$ and $(\eval_t)_\push\bdpi_2$ are $\meas$-absolutely continuous for every $t\in[0,1)$. Restriction of $\bdpi_1$ and $\bdpi_2$ to $[0,1-\delta]$ for arbitrary $\delta > 0$, cf.~\autoref{Le:Villani lemma for geodesic}, and  \autoref{Le:Mutually singular} imply
\begin{align}\label{Eq:Mutual singularity}
(\eval_t)_\push\bdpi_1 \perp (\eval_t)_\push\bdpi_2
\end{align}
for every $t\in(0,1]$. Given $i\in\{1,2\}$, let $\smash{\varrho_t^i}$ denote the density of the $\meas$-absolutely continuous part of $(\eval_t)_\push\bdpi_1$, which is nontrivial by $\smash{\TMCP^*(K,N)}$. Then
\begin{align*}
\int_\mms (\varrho_t^i)^{1-1/N}\d\meas &\geq (1-t)\,\rme^{-tD\sqrt{K^-/N}}\int_\mms \varrho_0^{1-1/N}\d\meas\\
&= (1-t)\,\rme^{-tD\sqrt{K^-/N}}\,\meas[E]^{1/N}.
\end{align*}
On the other hand, by Jensen's inequality,
\begin{align*}
\int_\mms (\varrho_t^i)^{1-1/N}\d\meas &\leq \meas\big[\{\varrho_t^i > 0\}\big]\,\meas\big[\{\varrho_t^i > 0\}\big]^{-1}\int_{\{\varrho_t^i > 0\}} (\varrho_t^i)^{1-1/N}\d\meas\\
&\leq \meas\big[\{\varrho_t^i > 0\}\big]\,\Big[\meas\big[\{\varrho_t^i > 0\}\big]^{-1}\int_{\{\varrho_t^i>0\}}\varrho_t^i\d\meas\Big]^{1-1/N}\\
&\leq \meas\big[\{\varrho_t^i > 0\}\big]^{1/N}.
\end{align*}
These  inequalities imply
\begin{align}\label{Eq:mm}
\liminf_{t\to 0} \meas\big[\{\varrho_t^i > 0\}\big] \geq \meas[E]= \meas\big[\{\varrho_0^i > 0\}\big].
\end{align}

By $\scrK$-global hyperbolicity, the sets
\begin{align*}
F &:= \{x\in \mms : x\in\supp\,(\eval_t)_\push\bdpi_i\textnormal{ for some }t\in[0,1],\, i\in\{1,2\}\},\\
G_\delta &:= \{x\in F : \tsep(y,x) \leq \delta\textnormal{ for some }y\in E\}
\end{align*}
are compact for every $\delta>0$. Since $\meas[G_\delta]\to \meas[E]$ as $\delta\to 0$ by Lebesgue's theorem, there exists some $\delta_0 \in (0,1)$ such that
\begin{align}\label{Eq:mmm}
\meas[G_{\delta_0}] \leq \frac{3}{2}\,\meas[E].
\end{align}
By construction, for every $i\in\{1,2\}$, every $t\in (0,1)$, and $(\eval_t)_\push\bdpi_i$-a.e.~$x\in \mms$ there exists $\gamma\in\TGeo^\tsep(\mms)$ such that $x=\gamma_t$, $\gamma_0\in E$, and $\smash{\gamma_1 = z_i}$. By definition of $\smash{G_{\delta_0}}$, we get $\supp(\eval_t)_\push\bdpi_i\subset G_{\delta_0}$ for every $t\in[0,\delta_0]$. Hence, \eqref{Eq:mm}, \eqref{Eq:Mutual singularity}, and \eqref{Eq:mmm} imply
\begin{align*}
\frac{3}{2}\,\meas[E] &< \meas[\{\rho_s^0>0\}] + \meas[\{\rho_s^1 >0\}]\\
&= \meas[\{\rho_s^0>0\} \cup \{\rho_s^1 > 0\}]\\ 
&\leq \meas[G_{\delta_0}]\\ 
&\leq \frac{3}{2}\,\meas[E]
\end{align*}
for some $s\in (0,\delta_0)$, a contradiction.
\end{proof}

\begin{lemma}\label{Pr:More than TMCP} Let $\smash{\TMCP^*(K,N)}$ hold for some $K\in\R$ and $N\in[1,\infty)$. Let $\smash{\mu_0 = \rho_0\,\meas \in\scrP_\comp^\ac(\mms,\meas)}$ and $\mu_1\in\scrP_\comp(\mms)$ with $\smash{\supp\mu_0\times\supp\mu_1\subset\mms_\ll^2}$. Then there exist a timelike proper-time parametrized $\smash{\ell_p}$-geodesic $(\mu_t)_{t\in[0,1]}$ connecting $\mu_0$ to $\mu_1$ as well as an $\smash{\ell_p}$-optimal coupling $\pi\in\Pi_\ll(\mu_0,\mu_1)$ such that for every $t\in [0,1)$ and every $N'\geq N$, we have $\mu_t\in\Dom(\Ent_\meas)$ and
\begin{align}\label{Eq:Claim conv TMCP}
\scrS_{N'}(\mu_t) \leq (1-t)\,\rme^{-tD\sqrt{K^-/N}}\,\scrS_N(\mu_0),
\end{align}
where
\begin{align*}
D := \sup\tsep(\supp\mu_0\times\supp\mu_1).
\end{align*}
\end{lemma}

\begin{proof} \textbf{Step 1.} \textit{Approximation of $\mu_1$.} Let $B_1,\dots,B_n\subset\mms$, $n\in\N$, be a given Borel partition of $\supp\mu_1$ with $\mu_1[B_i]>0$ for every $i\in\{1,\dots,n\}$. Given such an $i$, fix $\smash{x_1^i\in\supp\mu_1}$ and set $\lambda_i := \mu_1[B_i]$ as well as
\begin{align*}
\mu_1^n := \lambda_1\,\delta_{x_1^1} + \dots + \lambda_n\,\delta_{x_1^n}.
\end{align*}
Since every coupling of $\mu_0$ and $\smash{\mu_1^n}$ is chronological, an $\smash{\ell_p}$-optimal coupling of these exists uniquely by \autoref{Le:Uniqueness Diracs}. Let $\smash{\pi_n\in \Pi_\ll(\mu_0,\mu_1^n)}$ be this coupling, given by a $\mu_0$-measurable map $T_n\colon \supp\mu_0\to\mms$. For $i\in\{1,\dots,n\}$, we define 
\begin{align*}
A_i := T_n^{-1}(\{x_1^i\})\times\{x_1^i\},
\end{align*}
and observe that $A_1,\dots,A_n\subset \mms^2$ constitutes a Borel partition of $\supp\pi_n$. 

Define $\smash{\pi_n^i\in\scrP(\mms^2)}$, $\smash{\nu_0^i\in \scrP_\comp^\ac(\mms,\meas)}$, and $\smash{\nu_1^i\in\scrP_\comp(\mms)}$ by
\begin{align*}
\pi_n^i &:= \lambda_i^{-1}\,\pi_n\mres A_i,\\
\nu_0^i &:= (\pr_1)_\push\pi_n^i = \varrho_0^i\,\meas,\\
\nu_1^i &:= (\pr_2)_\push\pi_n^i = \delta_{x_1^i}.
\end{align*}
By construction, we have $\mu_0 = \lambda_1\,\nu_0^i + \dots + \lambda_n\,\nu_0^n$ and thus
\begin{align*}
\rho_0 = \lambda_1\,\varrho_0^i + \dots + \lambda_n\,\varrho_0^n\quad\meas\textnormal{-a.e.}
\end{align*}
Moreover $\smash{\supp\nu_0^i \cap \supp\nu_0^j = \emptyset}$ for every $i,j\in\{1,\dots,n\}$ with $i\neq j$, whence
\begin{align}\label{Eq:MUT SING}
\nu_0^i\perp\nu_0^j.
\end{align}

\textbf{Step 2.} \textit{Invoking the $\TMCP$ condition.} As $\smash{x_1^i\in I^+(\nu_0^i)}$ for every $i\in\{1,\dots,n\}$, using \autoref{Th:Good TMCP} there exists a timelike proper-time parametrized $\smash{\ell_p}$-geodesic $\smash{(\nu_t^i)_{t\in[0,1]}}$ from $\smash{\nu_0^i}$ to $\smash{\nu_1^i}$ such that for every $t\in[0,1)$,
\begin{align}\label{Eq:Interm TMCP}
\scrS_N(\nu_t^i) \leq (1-t)\,\rme^{-tD\sqrt{K^-/N}}\,\scrS_N(\mu_0).
\end{align}

With this information, we now build a timelike proper-time parametrized $\smash{\ell_p}$-geodesic $\smash{(\mu_t^n)_{t\in [0,1]}}$ from $\mu_0$ to $\smash{\mu_1^n}$ for which \eqref{Eq:Claim conv TMCP} holds with $\pi$ replaced by $\pi_n$. By definition, $\smash{(\nu_t^i)_{t\in[0,1]}}$ is represented by some $\smash{\bdalpha^i\in\OptTGeo_{\ell_p}^\tsep(\nu_0^i,\nu_1^i)}$. Thanks to  \eqref{Eq:MUT SING}, \autoref{Th:Good TMCP}, and \autoref{Le:Mutually singular}, we obtain
\begin{align}\label{Eq:MUT SING II}
(\eval_t)_\push\bdalpha^i \perp (\eval_t)_\push\bdalpha^j
\end{align}
for every $t\in(0,1)$ and every $i,j\in\{1,\dots,n\}$ with $i\neq j$. Then
\begin{align*}
\bdpi_n := \lambda_1\,\bdalpha^1 + \dots + \lambda_n\,\bdalpha^n
\end{align*}
is a timelike $\smash{\ell_p}$-optimal geodesic plan connecting $\mu_0$ to $\smash{\mu_1^n}$, and the assignment $\smash{\mu_t^n := (\eval_t)_\push\bdpi_n}$ gives rise to a timelike proper-time parametrized $\smash{\ell_p}$-geodesic connecting these probability measures. Let $\smash{\varrho_t^i}$ denote the nontrivial density of the $\meas$-absolutely continuous part of $\smash{(\eval_t)_\push\bdalpha^i}$. Then we obtain that
\begin{align}\label{oppppü}
\rho_t^n\,\meas := \big[\lambda_1\,\varrho_t^1 + \dots + \lambda_n\,\varrho_t^n\big]\,\meas
\end{align}
is the $\meas$-absolutely continuous part of $\mu_t^n$, $t\in[0,1]$. Now \eqref{Eq:MUT SING II} ensures 
\begin{align}\label{This ensures}
\meas\big[\{\varrho_t^i > 0\} \cap \{\varrho_t^j > 0\}\big] = 0.
\end{align}
for every $t\in [0,1]$ and every $i,j\in\{1,\dots,n\}$ with $i\neq j$.  Hence,  using  \eqref{Eq:Interm TMCP} and \eqref{Eq:MUT SING} we get, for every $t\in[0,1)$,
\begin{align*}
\scrS_{N}(\mu_t^n) &= \sum_{i=1}^n \lambda_i^{1-1/N}\,\scrS_{N}(\mu_t^i)\\ 
&\leq (1-t)\,\rme^{-tD\sqrt{K^-/N}}\sum_{i=1}^n \lambda_i^{1-1/N}\,\scrS_N(\mu_0^i)\\
&= (1-t)\,\rme^{-tD\sqrt{K^-/N}}\,\scrS_N(\mu_0).
\end{align*}

Lastly, \eqref{oppppü}, \autoref{Th:Good TMCP}, and \eqref{This ensures} imply $\sup\{\Vert\rho_t^n\Vert_{\Ell^\infty(\mms,\meas)}:n\in\N\} < \infty$ for every $t\in [0,1)$. By a standard convexity argument, this yields $\smash{\mu_t^n\in\Dom(\Ent_\meas)}$ and $\smash{\sup\{\vert\! \Ent_\meas(\mu_t^n)\vert : n\in\N\} < \infty}$ for every such $t$. 

\textbf{Step 3.} \textit{Conclusion.} If $\supp\mu_1$ consists of finitely many points, the claim simply follows by choosing $n\in\N$ such that $\smash{\mu_1 = \mu_1^n}$. 

Therefore, only the case $\#\supp\mu_1=\infty$ remains to be studied. By \cite[Thm.~2.16]{cavalletti2020}, there exists a sequence $\smash{(\bar{\mu}_1^n)_{n\in\N}}$ in $\scrP_\comp(\mms)$ which converges weakly to $\mu_1$ such that $\smash{\bar{\mu}_1^n\in\scrP^\ac(\mms,\mu_1^n)}$ for every $n\in\N$, and for every sequence $(\pi_n)_{n\in\N}$ of $\smash{\ell_p}$-optimal couplings $\smash{\pi_n \in\Pi_\leq(\mu_0,\bar{\mu}_1^n)}$, every weak limit of any subsequence of the latter is $\smash{\ell_p}$-optimal for $\mu_0$ and $\mu_1$. In particular, the previous discussion applies to $\smash{\bar{\mu}_1^n}$ in place of $\smash{\mu_1^n}$, $n\in\N$, and yields sequences $\smash{(\bdpi_n)_{n\in\N}}$ and $(\pi_n)_{n\in\N}$ of timelike $\smash{\ell_p}$-optimal geodesic plans $\smash{\bdpi_n\in\OptTGeo_{\ell_p}^\tsep(\mu_0,\bar{\mu}_1^n)}$ and $\smash{\ell_p}$-optimal couplings $\smash{\pi_n\in\Pi_\ll(\mu_0,\mu_1^n)}$ as in Step 2, respectively. By $\scrK$-global hyperbolicity, Prokhorov's theorem, and  \autoref{Le:Villani lemma for geodesic} these sequences converge weakly to some $\smash{\bdpi\in\OptTGeo_{\ell_p}^\tsep(\mu_0,\mu_1)}$ and an $\smash{\ell_p}$-optimal coupling $\pi\in\Pi_\ll(\mu_0,\mu_1)$, respectively, up to a nonrelabeled subsequence. Define a timelike proper-time parametrized $\smash{\ell_p}$-geodesic $(\mu_t)_{t\in[0,1]}$ from $\mu_0$ to $\mu_1$ by $\mu_t := (\eval_t)_\push\bdpi$. By weak lower semicontinuity of $\Ent_\meas$ on probability measures with uniformly bounded support, we have $\mu_t\in\Dom(\Ent_\meas)$ for every $t\in [0,1)$, in particular $\smash{\mu_t\in\scrP_\comp^\ac(\mms,\meas)}$. By analogous semicontinuity properties of $\scrS_N$ on measures with uniformly bounded support, with $\smash{\mu_t^n := (\eval_t)_\push\bdpi_n}$ we obtain
\begin{align*}
\scrS_N(\mu_t) &\leq \limsup_{n\to\infty} \scrS_N(\mu_t^n) \leq (1-t)\,\rme^{-tD\sqrt{K^-/N}}\,\scrS_N(\mu_0).\qedhere
\end{align*}
\end{proof}

The proof of the following \autoref{Th:Uniqueness couplings} follows the lines of \autoref{Le:Uniqueness Diracs} modulo some modifications we briefly discuss.

\begin{theorem}\label{Th:Uniqueness couplings} Assume $\smash{\TMCP^*(K,N)}$ for $K\in\R$ and $N\in [1,\infty)$. Suppose the pair $(\mu_0,\mu_1)\in\scrP_\comp^\ac(\mms,\meas)\times\scrP_\comp(\mms)$ to be timelike $p$-dualizable by $\smash{\pi\in\Pi_\ll(\mu_0,\mu_1)}$. Then there exists a $\mu_0$-measurable map $T\colon \mms \to \mms$ such that
\begin{align*}
\pi = (\Id,T)_\push\mu_0,
\end{align*}
and consequently,
\begin{align*}
\ell_p(\mu_0,\mu_1)^p = \int_\mms \tsep(x,T(x))^p\d\mu_0(x).
\end{align*}

In particular, $\pi$ is the unique chronological $\ell_p$-optimal coupling of $\mu_0$ and $\mu_1$.
\end{theorem}

\begin{proof} As in the proof of \autoref{Le:Uniqueness Diracs}, it suffices to prove the existence of $T$. Let $\Gamma\subset\mms_\ll^2$ be an $\smash{\ell_p}$-cyclically monotone set with $\pi[\Gamma]=1$ \cite[Prop.~2.8]{cavalletti2020}, and set
\begin{align*}
E := \{x\in\mms :  \pr_2[\Gamma \cap (\{x\}\times \mms)] \geq 2\},
\end{align*}
which is a Suslin set. We claim that $\mu_0[E]=0$, which directly gives the desired $T$.

Assume to the contrapositive that $\mu_0[E]>0$. By the von Neumann selection theorem \cite[Thm.~9.1.3]{bogachev2007b}, there exist $\mu_0$-measurable maps $T_1,T_2\colon E\to \mms$ whose graphs are both contained in $\Gamma$ and such that $T_1(x) \neq T_2(x)$ for every $x\in E$. By Lusin's theorem, there exists a compact set $C\subset E$ with $\mu_0[C]>0$ such that the restrictions $\smash{T_1\big\vert_C}$ and $\smash{T_2\big\vert_C}$ are continuous. This yields
\begin{align*}
\min\{\met(T_1(x),T_2(x)) : x\in C\} > 0.
\end{align*}
In particular, there exist $z_1,z_2\in\mms$ and $r>0$ with $\met(z_1,z_2)>r$, and a compact set $C'\subset C$ with $\mu_0[C]>0$ as well as $\smash{T_1(C')\subset\sfB^\met(z_1,r/2)}$ and $\smash{T_2(C')\subset\sfB^\met(z_2,r/2)}$. Up to possibly shrinking the radius $r$, we may and will assume without restriction that  $\smash{C'\times [\bar{\sfB}^\met(z_1,r/2)\cup \bar{\sfB}^\met(z_2,r/2)] \subset\mms_\ll^2}$.

As in the proof of \autoref{Le:Uniqueness Diracs}, we may  further shrink $C'$, hence assume that $\rho_0$ is bounded and bounded away from zero on $C'$. Consider $\smash{\nu_0\in\scrP_\comp^\ac(\mms,\meas)}$ with
\begin{align*}
\nu_0 := \meas[C']^{-1}\,\meas\mres C'.
\end{align*}
Furthermore, define $\smash{\mu_1^1,\mu_1^2\in\scrP_\comp(\mms)}$ by
\begin{align*}
\mu_1^1 &:= (T_1)_\push\mu_0,\\
\mu_1^2 &:= (T_2)_\push\mu_0.
\end{align*}
By construction $\smash{\supp\mu_1^1 \cap \supp\mu_1^2 = \emptyset}$, and the pairs $\smash{(\nu_0,\mu_1^1)}$ and $\smash{(\nu_0,\mu_1^2)}$ are both strongly $p$-timelike dualizable by \autoref{Re:Strong timelike}.  \autoref{Pr:More than TMCP} applies to both pairs; following the proof of \autoref{Le:Uniqueness Diracs} from now on gives the claim.
\end{proof}

\begin{theorem}\label{Th:Uniqueness geodesics} Assume $\smash{\TMCP^*(K,N)}$ for some $K\in\R$ and $N\in[1,\infty)$. Suppose the pair $(\mu_0,\mu_1)\in\scrP_\comp^\ac(\mms,\meas)\times\scrP_\comp(\mms)$ to be timelike $p$-dualizable. Then for every $\bdpi\in\smash{\OptTGeo_{\ell_p}^\tsep(\mu_0,\mu_1)}$ there is a $\mu_0$-measurable map $\smash{\mathfrak{T} \colon \supp\mu_0\to \TGeo^\tsep(\mms)}$ with
\begin{align*}
\bdpi = \mathfrak{T}_\push\mu_0.
\end{align*}

In particular, the set $\smash{\OptTGeo_{\ell_p}^\tsep(\mu_0,\mu_1)}$ is a singleton.
\end{theorem}

\begin{proof} By the usual argument outlined in the proof of \autoref{Le:Uniqueness Diracs}, it suffices to prove that every $\smash{\bdpi\in\OptTGeo_{\ell_p}^\tsep(\mu_0,\mu_1)}$ is induced by a map $\mathfrak{T}$ as above.

Assume to the contrapositive that one such $\bdpi$, henceforth fixed, is not given by a map. We disintegrate $\bdpi$ with respect to $\eval_0$, writing, with some abuse of notation,
\begin{align*}
\rmd\bdpi(\gamma) = \rmd\bdpi_x(\gamma)\d\mu_0(x)
\end{align*}
for a $\mu_0$-measurable map $\bdpi\colon \supp\mu_0\to \scrP(\TGeo^\tsep(\mms))$. By assumption on $\bdpi$, there is a compact set $C\subset\supp\mu_0$ such that $\mu_0[C]>0$ and $\#\supp \bdpi_x \geq 2$ for every $x\in C$. 
Hence, for every $x\in C$ there exists  $t_x\in (0,1)$ such that $\#\supp\,(\eval_t)_\push\bdpi_x \geq 2$; since $\smash{\TGeo^\tsep(\mms)\subset\Cont([0,1];\mms)}$ and $\smash{\supp\bdpi_x \subset \TGeo^\tsep(\mms)}$, there exists an open interval $I_x\subset (0,1)$ with $t_x\in I_x$ such that $\#\supp\,(\eval_t)_\push\bdpi_x  \geq 2$ for every $t\in I_x$.

Given any $t\in \Q\cap (0,1)$, we define
\begin{align*}
C_t &:= \{x \in C : \#\supp\,(\eval_t)_\push\bdpi_x\geq 2\}.
\end{align*}
Since $C_t$ unites to $C$ when ranging over $t\in \Q\cup (0,1)$, we find $s\in \Q\cap (0,1)$ with $\smash{\mu_0[C_s] >0}$. Finally, we define $\smash{\bdalpha\in\scrP(\TGeo^\tsep(\mms))}$ by
\begin{align*}
\rmd\bdalpha(\gamma) = \mu_0[C_s]^{-1}\,\One_{C_s}(x)\d\bdpi_x(\gamma)\d\mu_0(x).
\end{align*}
By restriction \cite[Lem.~2.10]{cavalletti2020}, $(\eval_0,\eval_1)_\push\bdalpha$ is an $\smash{\ell_p}$-optimal coupling of its marginals, the first one of which is $\smash{\mu_0[C_s]^{-1}\,\mu_0\mres C_s\in\scrP_\comp^\ac(\mms,\meas)}$. Since $\bdalpha$ is concentrated on $\TGeo^\tsep(\mms)$, a standard argument using the reverse triangle inequality \eqref{Eq:Reverse tau} for $\tsep$ implies that $(\eval_0,\eval_s)_\push\bdalpha$ is a chronological $\smash{\ell_p}$-optimal coupling of its marginals. By definition of $C_s$, it is not given by a map, which contradicts  \autoref{Th:Uniqueness couplings}.
\end{proof}

\begin{remark}\label{Re:From TNB to TENB} The above arguments --- which, as mentioned, are based on \cite{cavalletti2020} --- show that the uniqueness  results of \cite[Thm.~3.19, Thm.~3.20]{cavalletti2020} remain true if the timelike nonbranching assumption therein is replaced by the hypothesis of timelike $p$-essential nonbranching for the considered exponent $p\in (0,1)$.
\end{remark}

\subsection{Equivalence with the entropic TMCP condition}\label{Sec:Equiv TMCP's} In this section, additionally to our standing assumptions we suppose $\scrX$ to be timelike $p$-essentially nonbranching for some  fixed $p\in (0,1)$. Building upon the results in \autoref{Sub:Good} and \autoref{Sub:Uniqueneess}, we establish the equivalence of $\smash{\TMCP^*(K,N)}$ with the entropic timelike measure-contraction property from \cite[Def.~3.7]{cavalletti2020}, restated in \autoref{Def:TMCPe}. A byproduct of our argumentation is a pathwise characterization of $\smash{\TMCP^*(K,N)}$ and $\smash{\TMCP^e(K,N)}$, cf.~\autoref{Th:Equivalence TMCP* and TMCPe}, which is also available for $\smash{\TMCP(K,N)}$ according to \autoref{Th:Equivalence TMCP}.

Recall the definition \eqref{Eq:Expo Boltzmann} of the exponentiated Boltzmann entropy $\scrU_N$.

\begin{definition}\label{Def:TMCPe} Let $K\in\R$ and $N\in (0,\infty)$. We say that $\scrX$ satis\-fies the \emph{entropic timelike measure-contraction property} $\smash{\TMCP^e(K,N)}$ if for every $\mu_0\in\scrP_\comp^\ac(\mms,\meas)$ and every $x_1\in I^+(\mu_0)$, there exists an $\smash{\ell_{1/2}}$-geo\-desic $(\mu_t)_{t\in [0,1]}$ connecting $\mu_0$ and $\smash{\mu_1:= \delta_{x_1}}$ such that for every $t\in [0,1]$,
\begin{align*}
\scrU_N(\mu_t) \geq \sigma_{K,N}^{(1-t)}\big[\Vert\tsep\Vert_{\Ell^2(\mms^2,\mu_0\otimes\mu_1)}\big]\,\scrU_N(\mu_0).
\end{align*}
\end{definition}

This definition is independent of any transport exponent \cite[Rem.~2.4]{cavalletti2022}.

\begin{theorem}\label{Th:Equivalence TMCP* and TMCPe} The following statements are equivalent for every given $K\in\R$ and $N\in[1,\infty)$.
\begin{enumerate}[label=\textnormal{\textcolor{black}{(}\roman*\textcolor{black}{)}}]
\item\label{LAAA} The condition $\smash{\TMCP^*(K,N)}$ holds.
\item\label{LAAAA} For every $\mu_0 = \rho_0\,\meas \in\scrP^\ac(\mms,\meas)$ and every $x_1\in I^+(\mu_0)$ there is a timelike $\smash{\ell_p}$-optimal geodesic plan $\smash{\bdpi\in\OptTGeo_{\ell_p}^\tsep(\mu_0,\mu_1)}$, where $\smash{\mu_1 := \delta_{x_1}}$, such that for every $t\in[0,1)$, we have $\smash{(\eval_t)_\push\bdpi = \rho_t\,\meas \in\scrP^\ac(\mms,\meas)}$, and for every $N'\geq N$, the inequality
\begin{align*}
\rho_t(\gamma_t)^{-1/N'} \geq \sigma_{K,N'}^{(1-t)}(\tsep(\gamma_0,x_1))\,\rho_0(\gamma_0)^{-1/N'}
\end{align*}
holds for $\bdpi$-a.e.~$\gamma\in\TGeo^\tsep(\mms)$.
\item\label{LAAAAA} The condition $\smash{\TMCP^e(K,N)}$ holds.
\end{enumerate}
\end{theorem}

\begin{proof} We outline the necessary adaptations of the arguments in \autoref{Sub:Equiv TCDs} and leave the details to the reader. 

By integration, \ref{LAAAA} implies \ref{LAAA}. The implication of \ref{LAAA} to \ref{LAAAA} is argued analogously to \autoref{Pr:ii to iii}. Note that to obtain the asserted  absolute continuity of $(\eval_t)_\push\bdpi$ with respect to $\meas$ for every $t\in[0,1)$, we have combined \autoref{Th:Good TMCP} with the uniqueness \autoref{Th:Uniqueness geodesics}. The step from \ref{LAAAA} to \ref{LAAAAA} follows since
\begin{align*}
-\frac{1}{N}\log\rho_t(\gamma_t) \geq \rmH_t\Big[\!-\!\frac{1}{N}\log\rho_0(\gamma_0),\frac{K}{N}\,\tsep^2(\gamma_0,x_1)\Big]
\end{align*}
for $\bdpi$-a.e.~$\gamma\in\TGeo^\tsep(\mms)$, using that the function $\rmH_t$ defined in \eqref{Eq:GtHt} has similar convexity properties as the function $\rmG_t$ used to derive  \autoref{Pr:iii to iv}, and then arguing as in the proof of the latter result. Similarly, the implication from \ref{LAAAAA} to \ref{LAAAA} is shown analogously to the proof of \autoref{Pr:v to iii}.
\end{proof}

Evident adaptations of the previous arguments give the following.

\begin{theorem}\label{Th:Equivalence TMCP} The following statements are equivalent for every given $K\in\R$ and $N\in[1,\infty)$.
\begin{enumerate}[label=\textnormal{\textcolor{black}{(}\roman*\textcolor{black}{)}}]
\item The condition $\smash{\TMCP(K,N)}$ holds.
\item For every $\mu_0 =\rho_0\,\meas\in\scrP^\ac(\mms,\meas)$ and every $x_1\in I^+(\mu_0)$ there is a timelike $\smash{\ell_p}$-optimal geodesic plan $\smash{\bdpi\in\OptTGeo_{\ell_p}^\tsep(\mu_0,\mu_1)}$, where $\smash{\mu_1 := \delta_{x_1}}$, such that for every $t\in[0,1)$, we have $\smash{(\eval_t)_\push\bdpi = \rho_t\,\meas\in\scrP^\ac(\mms,\meas)}$, and for every $N'\geq N$, the inequality
\begin{align*}
\rho_t(\gamma_t)^{-1/N'} \geq \tau_{K,N'}^{(1-t)}(\tsep(\gamma_0,x_1))\,\rho_0(\gamma_0)^{-1/N'}
\end{align*}
holds for $\bdpi$-a.e.~$\gamma\in\TGeo^\tsep(\mms)$.
\end{enumerate}
\end{theorem}

\appendix

\section{Smooth Lorentzian spacetimes}\label{App:Smooth}

In this appendix, we study the properties from \autoref{Def:TCD*} and \autoref{Def:TMCP} on smooth Lorentzian manifolds. In \autoref{Th:Equiv}, 
we prove the equivalence of all notions from the former to upper dimension and timelike lower Ricci curvature bounds, thereby justifying the terminology ``curvature-dimension condition'' for the inherent notions. A similar equivalence for the timelike measure-contraction property is due to \autoref{Th:TMCP Smooth}.

We follow the conventions of \cite{mccann2020}. See \cite{oneill1983} for an exhaustive account on semi-Riemannian geometry, and \cite{eckstein2017, kellsuhr2020, mccann2020,  mondinosuhr2018, suhr2018} for details about optimal transport on smooth Lorentzian spacetimes.

In this chapter, let $(\mms,\langle\cdot,\cdot\rangle)$ be a \emph{smooth Lorentzian spacetime} of dimension $\smash{n\in\N_{\geq 2}}$, that is, $\mms$ is a smooth, connected, time-oriented manifold with smooth metric tensor $\langle\cdot,\cdot\rangle$ of signature $+,-,\dots,-$. We   always assume $\mms$ to be globally hyperbolic \cite[p.~411]{oneill1983}. By \autoref{Ex:Low reg spt}, this spacetime canonically induces a measured Lorentzian structure obeying \autoref{Ass:ASS}. Hence, all our previously obtained results apply to the setting in this chapter.

\subsection{Smooth timelike curvature-dimension condition} 

For $N > n$ and arbitrary $\smash{V\in\Cont^2(\mms)}$ we define the \emph{$N$-Bakry--Émery Ricci tensor}
\begin{align}\label{Eq:Weighted Ric}
\Ric^{N,V} := \Ric + \Hess V - \frac{1}{N-n}\,\rmD V\otimes\rmD  V.
\end{align}
In the borderline case $N=n$, we conventionally assume $V$ to be constant, in which case $\smash{\Ric^{N,V} := \Ric}$. Given any  $K\in\R$, we say that ``$\smash{\Ric^{N,V}\geq K}$ in every timelike direction'' if $\Ric^{N,V}(\xi,\xi) \geq K\,\vert \xi\vert^2$ for every $x\in \mms$ and every $\xi\in T_x\mms$ with $\vert \xi\vert > 0$. 

Set $\smash{\vol^V = \rme^V\,\vol}$, and note that $\smash{\scrP^\ac(\mms,\vol^V) = \scrP^\ac(\mms,\vol)}$.

In the following main theorem of this section, the conditions in \ref{La:V} and \ref{La:VI} are understood in terms of the $N$-Rényi entropy $\smash{\scrS_N^V}$ with respect to $\smash{\vol^V}$; the same applies to \autoref{Th:TMCP Smooth} below.

\begin{theorem}\label{Th:Equiv} For every given $p\in (0,1)$, $K\in\R$, and $N\in[n,\infty)$, the subsequent state\-ments are equivalent.
\begin{enumerate}[label=\textnormal{\textcolor{black}{(}\roman*\textcolor{black}{)}}]
\item\label{La:II} $\smash{\Ric^{N,V} \geq K}$ in every timelike direction.
\item\label{La:V}  The condition $\smash{\TCD_p^*(K,N)}$ holds.
\item\label{La:VI} The condition $\smash{\wTCD_p^*(K,N)}$.
\end{enumerate} 

In particular, if any of the conditions in \ref{La:V} or  \ref{La:VI} holds for some $p\in (0,1)$, then these statements are satisfied for every $p\in (0,1)$.
\end{theorem}


The only implication that still has to be established is \ref{La:II} implying  \ref{La:VI}. 
Indeed, the equivalence of 
\ref{La:V} and \ref{La:VI} holds by 
\autoref{Th:Equivalence TCD* and TCDe} and \autoref{Re:reversed}. This also gives the equivalence of \ref{La:V} and \ref{La:VI} to $\smash{\TCD_p^e(K,N)}$, which in turn implies   \ref{La:II} by the result \cite[Thm.~8.5]{mccann2020}. 

Let us now prove the remaining implication in  \autoref{Pr:From 1 to 4}. It follows from \autoref{Th:Equivalence TCD* and TCDe}, but we prefer to give an independent proof by using the smooth tools provided by  \cite{mccann2020}.  \autoref{Pr:From 1 to 4} is indeed sufficient to derive the $\smash{\wTCD_p^*(K,N)}$ condition by an argument as for \autoref{Pr:ii to iii}, see also \cite[Thm.~7.5]{mccann2020}.

\begin{proposition}\label{Pr:From 1 to 4} Let $K\in\R$ and $N\in[n,\infty)$. Assume the inequality $\smash{\Ric^{N,V}\geq K}$ in every timelike direction. Then for every $p\in (0,1)$, every $\mu_0,\mu_1\in\smash{\scrP_\comp^\ac(\mms,\vol^V)}$ such that $\smash{\supp\mu_0\times\supp\mu_1\subset\mms_\ll^2}$ and $\smash{\rho_0,\rho_1\in\Ell^\infty(\mms,\vol^V)}$, there is a timelike proper-time parametrized $\smash{\ell_p}$-geodesic $(\mu_t)_{t\in[0,1]}$ from $\mu_0$ to $\mu_1$ and a timelike $p$-dualizing coupling $\pi\in\Pi_\ll(\mu_0,\mu_1)$ such that for every $t\in[0,1]$ and every $N'\geq N$,
\begin{align*}
\scrS_{N'}^V(\mu_t) &\leq -\int_{\mms^2} \sigma_{K,N'}^{(1-t)}(\tsep(x^0,x^1))\,\rho_0(x^0)^{-1/N'}\d\pi(x^0,x^1)\\
&\qquad\qquad -\int_{\mms^2} \sigma_{K,N'}^{(t)}(\tsep(x^0,x^1))\,\rho_1(x^1)^{-1/N'}\d\pi(x^0,x^1).
\end{align*}
\end{proposition}

\begin{proof} \textbf{Step 1.} \textit{Setup.} We start with a short digression  about Lorentzian optimal transport from \cite{mccann2020}. Since $\mu_0$ and $\mu_1$ have compact support and $\smash{\supp\mu_0\times\supp\mu_1}$ is chronological, the pair $(\mu_0,\mu_1)$ is $p$-separated according to \cite[Def.~4.1]{mccann2020} by \cite[Lem.~4.4]{mccann2020}. Since also $\smash{\mu_0 \ll \vol}$,  $\mu_0$ and $\mu_1$ are both joined by a unique timelike proper-time parametrized $\smash{\ell_p}$-geodesic $(\mu_t)_{t\in[0,1]}$, cf.~\cite[Cor.~5.9]{mccann2020} and \autoref{Re:Compat}, and  coupled by a unique  $\smash{\ell_p}$-optimal coupling $\pi\in\Pi_\ll(\mu_0,\mu_1)$ \cite[Thm.~5.8]{mccann2020}. Additionally, these results yield the existence of some sufficiently regular vector field $X$ on $\mms$ such that, defining $T\colon [0,1]\times\mms\to \mms$ by 
\begin{align*}
T_t(x):= \exp_x t X(x),
\end{align*}
for every $t\in[0,1]$ we have
\begin{align}\label{Eq:Repr Lorentzian}
\begin{split}
\mu_t &= (T_t)_\push\mu_0,\\
\pi &= (\Id, T_1)_\push\mu_0.
\end{split}
\end{align}
Thanks to \cite[Prop.~6.1]{mccann2020}, the approximate derivative $A_t := \smash{\tilde{\rmD}T_t \colon T\mms \to (T_t)_*T\mms}$ as introduced e.g.~in \cite[Def.~3.8]{mccann2020}   exists, is invertible, and depends smoothly on $t\in[0,1]$ $\smash{\vol}$-a.e. Given any $x\in\mms$ and $t\in[0,1]$, define
\begin{align*}
\jmath_t(x) &:= \vert\! \det A_t(x)\vert\,\rme^{V(T_t(x))},\\
\varphi_t(x) &:= \log \jmath_t(x) = \log \vert\!\det A_t(x)\vert + V(T_t(x)),\\
B_t(x) &:= \dot{A}_t(x)\,A_t^{-1}(x).
\end{align*}
The function $\jmath_t$ should be thought of as a Jacobian, cf.~Step 4 below.

\textbf{Step 2.} \textit{Invoking the timelike Ricci curvature bound.} We henceforth suppose that $N'\geq N > n$; in the case $N'\geq N=n$ and  recalling our convention of $V$ being constant in this situation, the arguments to follow, e.g.~the inequalities in \eqref{Eq:In part}, are somewhat easier. To relax notation, the dependency of the point $x\in\mms$ at which the following computations may and will be  performed is not made explicit.

By \cite[Prop.~6.1]{mccann2020}, we know the formulas
\begin{align}\label{Eq:The formulas}
\begin{split}
\dot{\varphi}_t &= \tr B_t + \big\langle\nabla V,\dot{T}_t\big\rangle\circ T_t,\\
\ddot{\varphi}_t &= - \Ric(\dot{T}_t,\dot{T}_t) - \Hess V(\dot{T}_t,\dot{T}_t) - \tr B_t^2.
\end{split}
\end{align}
Since $(\tr B_t)^2 \leq n\,\tr B_t^2$ by the Cauchy--Schwarz inequality, 
\begin{align*}
\frac{1}{N}\,\dot{\varphi}_t^2 &\leq \frac{1+\varepsilon}{N}\,(\tr B_t)^2 + \frac{1+\varepsilon^{-1}}{N}\,\big\langle\nabla  V,\dot{T}_t\big\rangle^2\circ T_t\\
&\leq \frac{n(1+\varepsilon)}{N}\,\tr B_t^2 + \frac{1+\varepsilon^{-1}}{N}\,\big\langle\nabla V,\dot{T}_t\big\rangle^2\circ T_t\\
&= \tr B_t^2 + \frac{1}{N-n}\,\big\langle \nabla V,\dot{T}_t\big\rangle^2\circ T_t
\end{align*}
by choosing $\varepsilon := (N-n)/n > 0$. This entails
\begin{align}\label{Eq:In part}
\ddot{\varphi}_t + \frac{1}{N'}\,\dot{\varphi}_t^2 \leq \ddot{\varphi}_t + \frac{1}{N}\,\dot{\varphi}_t^2 &\leq -\Ric^{N,V}(\dot{T}_t,\dot{T}_t) \leq -K\,\vartheta^2,
\end{align}
where $\vartheta := \vert\dot{T}_t\vert = \tsep(\cdot,T_1)>0$ by geodesy \cite[Thm.~6.4]{mccann2020}. This easily implies
\begin{align*}
\frac{\partial^2}{\partial t^2} \jmath_t^{1/N'}\leq -\frac{K\vartheta^2}{N'}\,\jmath_t^{1/N'}
\end{align*}
which in turn entails
\begin{align*}
\jmath_t^{1/N'} \geq \sigma_{K,N'}^{(1-t)}(\vartheta)\,\jmath_0^{1/N'} + \sigma_{K,N'}^{(t)}(\vartheta)\,\jmath_1^{1/N'}.
\end{align*}

\textbf{Step 3.} \textit{Conclusion.} Writing $\smash{\mu_t = \rho_t\,\vol^V = \rho_t\,\rme^V\,\vol}$, $t\in[0,1]$, the change of variables formula with respect to $\vol$ gives $\smash{(\rho_t\circ T_t)\,\jmath_t = \rho_0\,\rme^V}$ $\vol$-a.e. Combining this with the previous inequality and the identity  $\vartheta = \tsep(\cdot,T_1)$ implies
\begin{align*}
\scrS_{N'}^V(\mu_t) &= -\int_\mms \rho_t^{1-1/N'}\d\vol^V\\
&= -\int_\mms (\rho_t\circ T_t)^{1-1/N'}\,\jmath_t\d\vol\\
&= -\int_\mms \rho_0^{1-1/N'}\,\rme^{-V/N'}\,\jmath_t^{1/N'}\d\vol^V\\
&\leq -\int_\mms \sigma_{K,N'}^{(1-t)}(\tsep\circ(\Id,T_1))\,\jmath_0^{1/N'}\,\rho_0^{1-1/N'}\,\rme^{-V/N'}\d\vol^V\\
&\qquad\qquad -\int_\mms  \sigma_{K,N'}^{(t)}(\tsep\circ(\Id,T_1))\,\jmath_1^{1/N'}\,\rho_0^{1-1/N'}\,\rme^{-V/N'}\d\vol^V\\
&= -\int_\mms \sigma_{K,N'}^{(1-t)}(\tsep\circ(\Id,T_1))\,\rho_0^{-1/N'}\d\mu_0\\
&\qquad\qquad -\int_\mms\sigma_{K,N'}^{(t)}(\tsep\circ(\Id,T_1))\,(\rho_1\circ T_1)^{-1/N'}\d\mu_0
\end{align*}
for every $t\in[0,1]$; employing \eqref{Eq:Repr Lorentzian} gives the desired result.
\end{proof}

\begin{remark} We made the preceding assumption $\smash{\rho_0,\rho_1\in\Ell^\infty(\mms,\vol^V)}$ only for cosmetic reasons in view of \autoref{Pr:ii to iii} and possible adaptations of \autoref{Pr:From 1 to 4} to other settings. It is not used in the above proof.
\end{remark}

\subsection{Smooth timelike measure-contraction property} With all results shown thus far, it is now straightforward to relate the timelike measure-contraction properties from \autoref{Def:TMCP} to timelike lower Ricci curvature bounds.

Recall our convention of $V$ being constant when $N=n$ in \eqref{Eq:Weighted Ric}.

\begin{theorem}\label{Th:TMCP Smooth} For every $K\in\R$, the subsequent statements are equivalent.
\begin{enumerate}[label=\textnormal{\textcolor{black}{(}\roman*\textcolor{black}{)}}]
\item\label{La:Die Eins} $\Ric\geq K$ in every timelike direction.
\item\label{La:Die Drei} The condition $\smash{\TMCP^*(K,n)}$ holds.
\end{enumerate}
\end{theorem}

\begin{proof} Assuming \ref{La:Die Eins}, we know from \autoref{Th:Equiv} that $\smash{\TCD_p^*(K,n)}$ holds, which implies \ref{La:Die Drei} thanks to \autoref{Pr:TMCP to TCD}. Furthermore, item \ref{La:Die Drei} implies $\smash{\TMCP^e(K,n)}$ according to \autoref{Th:Equivalence TMCP* and TMCPe}, whence \cite[Thm.~A.1]{cavalletti2020}  entails \ref{La:Die Eins}.
\end{proof}

\begin{remark} For general $K\in\R$ and $N\in[1,\infty)$, even for smooth spacetimes the condition $\smash{\TMCP^*(K,N)}$ is strictly weaker than $\smash{\TCD^*(K,N)}$, cf.~\cite[Rem.~A.3]{cavalletti2020}, \autoref{Th:Equivalence TCD* and TCDe}, and \autoref{Th:Equivalence TMCP* and TMCPe}.
\end{remark}

\section{Geodesics of probability measures}\label{App:B}

Here, we develop the theory of $\smash{\ell_p}$-geodesics in $\scrP(\mms)$, $p\in(0,1]$. As indicated in \autoref{Sub:GEO} and \autoref{Sub:Geodesics}, it is slightly different from the approaches taken in \cite{cavalletti2020, mccann2020}.

\subsection{Measurability}\label{Sub:Measura} For evident technical reasons, we first discuss measurability of the sets $\Geo(\mms)$ and $\TGeo(\mms)$ as introduced in \autoref{Sub:GEO}.

As usual, we will endow the space $\Cont([0,1];\mms)$ with the topology induced by the uniform distance $\met_\infty$ coming from the background metric $\met$.

\begin{lemma}\label{Le:Borel} $\smash{\Geo(\mms)}$ is closed,  $\smash{\TGeo(\mms)}$ is a Borel subset of $\Cont([0,1];\mms)$.
\end{lemma} 

\begin{proof} The sets of initial and final points, respectively, of any $\met_\infty$-conver\-gent sequence $\smash{(\gamma^k)_{k\in\N}}$ in $\Geo(\mms)$ are contained in compact sets, say $C_0,C_1\subset\mms$. Consequently, $\gamma^k([0,1])\subset J(C_0,C_1)$ for every $k\in\N$. The claim follows from $\scrK$-global hyperbolicity and the limit curve theorem \cite[Thm.~3.7]{kunzinger2018}.

For the second claim, note that $(\eval_0,\eval_1)^{-1}(\tsep((0,\infty)))$ is open. Regularity implies $\smash{\TGeo(\mms) = \Geo(\mms)\cap (\eval_0,\eval_1)^{-1}(\tsep((0,\infty)))}$, which terminates the proof.
\end{proof}

\begin{remark} One can easily see that there are more flexible ways, which do not rely on regularity for instance, to show that $\Geo(\mms)$ and $\TGeo(\mms)$ are Borel subsets of $\Cont([0,1];\mms)$, using the identities
\begin{align*}
\Geo(\mms) &= \{\gamma\in\Geo(\mms) : \Lip\,\gamma\leq n\textnormal{ for some }n\in\N\},\\
\TGeo(\mms) &= \{\gamma\in\Geo(\mms) : \textnormal{for every }s,t\in[0,1]\cap \Q\textnormal{ with } s<t\\
&\qquad\qquad \textnormal{there exists }\delta\in(0,1]\cap\Q \textnormal{ with }\tsep(\gamma_s,\gamma_t)\geq \delta\}.
\end{align*}
\end{remark}

\subsection{Causal and timelike geodesics minimizing a Lagrangian action} From now on, in addition to \autoref{Ass:ASS}, as in \autoref{Sub:Geodesics} it will not be restrictive to assume the compactness of $\mms$ for notational convenience. More precisely, $\mms$ is typically understood as a causal diamond between two compact subsets of an ambient space satisfying \autoref{Ass:ASS}.

We first recall parts of \autoref{Def:Lorentzian geodesic}, defining
\begin{align*}
\OptGeo_{\ell_p}(\mu_0,\mu_1) &:= \{\bdpi\in\scrP(\Geo(\mms)) : (\eval_0,\eval_1)_\push\bdpi\in\Pi_\leq(\mu_0,\mu_1)\\
&\qquad\qquad \textnormal{is } \ell_p\textnormal{-optimal} \}.
\end{align*}

\begin{definition}\label{Def:Lorentzian geodesic II} We term a weakly continuous path $(\mu_t)_{t\in[0,1]}$ in $\scrP(\mms)$ a \emph{causal $\ell_p$-geodesic} if it is represented by some $\smash{\bdpi\in\OptGeo_{\ell_p}(\mu_0,\mu_1)}$.
\end{definition}

In \autoref{Pr:Max}, we examine causal $\smash{\ell_p}$-geodesics to precisely minimize some \emph{Lagrangian action} according to the general theory developed in \cite[Ch.~7]{villani2009}. (The \emph{minimization} theory from \cite{villani2009} can easily be linked to our present Lorentzian setting which \emph{maximizes} costs by flipping signs.) Analogously to the metric Wasserstein situation, cf.~\cite[Thm.~2.10, Rem. 2.14]{ambrosiogigli2013} and \cite[Thm.~6.18]{villani2009}, this is an instance of informally saying that $\scrP(\mms)$ is a Lorentzian geodesic space with respect to $\smash{\ell_p}$ if and only if $\mms$ is so with respect to $\tsep$.

We consider a family $\boldsymbol{\scrA}$ of functionals $\scrA^{s,t}\colon\Cont([s,t];\mms)\to(-\infty,\infty]$, $s,t\in [0,1]$ with $s<t$, defined by
\begin{align*}
\scrA^{s,t}(\gamma) := \begin{cases} - \Len_\tsep(\gamma)^p & \textnormal{if }\gamma\textnormal{ is causal in the sense of \autoref{Sub:Length curves}},\\
\infty & \textnormal{otherwise}.
\end{cases}
\end{align*}
In particular, by non-total imprisonment there exists $c>0$ such that $\scrA^{s,t}(\gamma) < \infty$ only on curves $\gamma\colon [s,t]\to\mms$ with $\Lip\,\gamma \leq c$. Analogously, we define a family $\boldsymbol{\rmc}$ of cost functions $\smash{\rmc^{s,t}\colon \mms^2\to (-\infty,\infty]}$, $s,t\in[0,1]$ with $s<t$, by
\begin{align*}
\rmc^{s,t}(x,y) &:= \begin{cases} - \tsep(x,y)^p & \textnormal{if }x\leq y,\\
\infty & \textnormal{otherwise}.
\end{cases}
\end{align*}

\begin{lemma}\label{Le:Coerc Lagr} The family $\boldsymbol{\scrA}$ together with $\boldsymbol{\rmc}$ is a coercive Lagrangian action according to \textnormal{\cite[Def.~7.11, Def.~7.13]{villani2009}}.
\end{lemma}

\begin{proof} We only outline the argument. For every $s,t\in[0,1]$ with $s< t$, $\scrA^{s,t}$ and $\rmc^{s,t}$ are lower semicontinuous (with respect to the uniform topology and the natural one on $\mms^2$, respectively) by causal closedness, \cite[Prop.~3.17, Thm.~3.26]{kunzinger2018}, and continuity of $\tsep$. Property (i) in \cite[Def.~7.11]{villani2009} follows from additivity of $\Len_\tsep$ \cite[Lem.~2.25]{kunzinger2018} and from $\gamma$ failing to be causal if and only if this property fails on some subinterval of its domain of definition. The properties (ii) and (iii) in \cite[Def.~7.11]{villani2009} follow from $\scrX$ being a Lorentzian geodesic space. Boundedness of $\tsep$ on $\mms$ ensures uniform lower boundedness of $\boldsymbol{\scrA}$ according to (i) in \cite[Def.~7.13]{villani2009}. The compactness property (ii) from \cite[Def.~7.13]{villani2009} follows from non-total imprisonment, compactness of $\mms$, and Arzelà--Ascoli's theorem, and by existence of causal geodesics implicit in the definition of a Lorentzian geodesic space, cf.~\autoref{Def:Geod base space}.
\end{proof}

The following follows directly from \autoref{Le:Coerc Lagr}, cf.~\cite[Thm.~7.21]{villani2009}; for further consequences we do not need, we refer to \cite[Prop.~7.16, Thm.~7.30]{villani2009}.

Given $s,t\in [0,1]$ with $s< t$ as well as $\mu,\nu\in\scrP(\mms)$, set
\begin{align*}
\rmC^{s,t}(\mu,\nu) := \inf\!\Big\lbrace\!\int_{\mms^2}\rmc^{s,t}(x,y) \d\pi(x,y) : \pi\in\Pi(\mu,\nu)\Big\rbrace.
\end{align*}

\begin{proposition}\label{Pr:Max} For a weakly continuous path $(\mu_t)_{t\in[0,1]}$ in $\scrP(\mms)$, the following properties are equivalent.
\begin{enumerate}[label=\textnormal{(\roman*)}]
\item The curve $(\mu_t)_{t\in[0,1]}$ is a causal $\smash{\ell_p}$-geodesic according to \autoref{Def:Lorentzian geodesic II}.
\item For every $r,s,t\in[0,1]$ with $r< s < t$,
\begin{align*}
\rmC^{r,s}(\mu_r,\mu_s) + \rmC^{s,t}(\mu_s,\mu_t) = \rmC^{r,t}(\mu_r,\mu_t).
\end{align*}
\item For every $s,t\in[0,1]$ with $s<t$, the curve $(\mu_r)_{t\in[s,t]}$ is a minimizer of the following coercive action $\smash{\boldsymbol{\mathrm{A}}^{s,t}}$ defined by
\begin{align*}
\boldsymbol{\mathrm{A}}^{s,t}(\nu) := \sup\{ \rmC^{t_0,t_1}(\nu_{t_0},\nu_{t_1}) + \rmC^{t_1,t_2}(\nu_{t_1},\nu_{t_2}) + \dots + \rmC^{t_{n-1},t_n}(\nu_{t_{n-1}},\nu_{t_n})\}
\end{align*}
for curves $\nu := (\nu_r)_{r\in[s,t]}$ in $\scrP(\mms)$, where the supremum is taken over all $n\in\N$ and all $t_0,\dots,t_n\in[s,t]$ with $t_i < t_{i+1}$ for every $i\in\{0,\dots,n-1\}$.
\end{enumerate}

Given any $\mu_0,\mu_1\in\scrP(\mms)$ with $\Pi_\leq(\mu_0,\mu_1)\neq \emptyset$, there exists at least one curve $(\mu_t)_{t\in[0,1]}$ satisfying either of the above properties. 

Finally, if $C_0,C_1\subset\scrP(\mms)$ are weakly compact subsets with $\smash{\rmC^{0,1}(C_0\times C_1)\subset\R}$, then the set $\smash{\OptGeo_{\ell_p}(C_0\times C_1)}$ is weakly compact. 
\end{proposition}

\subsection{Proper-time parametrizations} Now we use   \autoref{Pr:Max} and \cite[Thm. 7.30]{villani2009} to establish properties of timelike (proper-time parametrized) $\smash{\ell_p}$-geodesics. To set up these notions, let $\met_\infty$ denote the uniform metric on $\Cont([0,1];\mms)$ induced by $\met$. As in \autoref{Sub:GEO}, for $\eta\in\TGeo(\mms)$ define $\psi_\eta\in\Cont([0,1]; [0,1])$ by
\begin{align*}
\psi_\eta(t) := \frac{\Len_\tsep(\eta\big\vert_{[0,t]})}{\tsep(\eta_0,\eta_1)} = \frac{\tsep(\eta_0,\eta_t)}{\tsep(\eta_0,\eta_1)}.
\end{align*}
This function is strictly increasing \cite[Prop.~2.34]{kunzinger2018} and hence admits an inverse $\psi_\eta^{-1}\in\Cont([0,1];[0,1])$. Define the map $\sfr\colon \TGeo(\mms)\to\Cont([0,1];\mms)$ by
\begin{align*}
(\sfr\circ\eta)_t := \eta_{\psi_\eta^{-1}(t)}.
\end{align*}
Recall from \cite[Cor.~3.35]{kunzinger2018} that the map $\sfr$ parametrizes $\eta\in\TGeo(\mms)$ proportional to $\tsep$-arclength, i.e.~the identity \eqref{Eq:Proper time par} holds for $\gamma := \sfr\circ\eta$.

We point out that \autoref{Le:Cty reparametrization} and \autoref{Cor:WEAK CONT r} do not require the  compactness of $\mms$ which has been assumed above.

\begin{lemma}\label{Le:Cty reparametrization} The map $\sfr$ is continuous with respect to $\met_\infty$.
\end{lemma}

\begin{proof} Let $\varepsilon > 0$ and $\eta\in\TGeo(\mms)$, 
and note that for every $\alpha\in\TGeo(\mms)$,
\begin{align*}
\met_\infty(\sfr\circ\eta,\sfr\circ\alpha) &=\sup_{t\in[0,1]} \met(\eta_{\psi_\eta^{-1}(t)}, \alpha_{\psi_\alpha^{-1}(t)})\\
&= \sup_{s\in[0,1]} \met(\eta_{\psi_\eta^{-1}(\psi_\alpha(s))},\alpha_s)\\
&\leq \sup_{s\in[0,1]} \met(\eta_{\psi_\eta^{-1}(\psi_\alpha(s))},\eta_s) + \met_\infty(\eta,\alpha)\\
&\leq \Lip\,\eta\sup_{s\in[0,1]} \big\vert\psi_\eta^{-1}(\psi_\alpha(s)) - s\big\vert + \met_\infty(\eta,\alpha).
\end{align*}
Uniform continuity of $\smash{\psi_\eta^{-1}}$ implies the existence of $\zeta > 0$ such that
\begin{align*}
\big\vert\psi_\eta^{-1}(r) - \psi_\eta^{-1}(r')\big\vert \leq (\Lip\,\eta)^{-1}\,\frac{\varepsilon}{2}
\end{align*}
provided $r,r'\in[0,1]$ satisfy $\vert r-r'\vert \leq \zeta$. In turn, there exists $\lambda > 0$ with
\begin{align*}
\sup_{s\in[0,1]} \big\vert\psi_\alpha(s) - \psi_\eta(s)\big\vert \leq \zeta
\end{align*}
provided $\alpha\in\TGeo(\mms)$ satisfies $\met_\infty(\eta,\alpha)\leq \lambda$. Hence, setting $\delta := \min\{1,\lambda,\varepsilon/2\}>0$, by the above considerations we obtain that if $\alpha\in\TGeo(\mms)$ obeys $\met_\infty(\eta,\alpha) \leq \delta$, then $
\met_\infty(\sfr\circ\eta,\sfr\circ\alpha)\leq \varepsilon$, which is the desired claim.
\end{proof}

Recall the definition $\TGeo^\tsep(\mms) := \sfr(\TGeo(\mms))$ from \autoref{Sub:Geodesics}.

\begin{corollary}\label{Cor:Cptness and equicty} For every $r>0$, the set
\begin{align*}
G_r := \{\gamma\in\TGeo^\tsep(\mms) : \tsep(\gamma_0,\gamma_1) \geq r\}
\end{align*}
is compact and uniformly equicontinuous.
\end{corollary}

\begin{proof} The compactness of the preimage of $G_r$ under $\sfr$ is a consequence of the compactness of $\mms$, non-total imprisonment, Arzelà--Ascoli's theorem, causal closedness, continuity of $\tsep$, and \cite[Prop.~3.17]{kunzinger2018}. In particular, $G_r$ is compact by continuity of $\sfr$.

To show uniform equicontinuity, let $\varepsilon > 0$, and let $c>0$ with $\Lip\,\eta \leq c$ for every $\eta\in\Geo(\mms)$. Write a given $\gamma\in G_r$ as $\gamma := \sfr\circ\eta$, where $\eta\in\TGeo(\mms)$. Since $\smash{\psi_\eta^{-1}}$ is uniformly continuous on $[0,1]$, there exists $\delta > 0$ such that 
\begin{align*}
\big\vert \psi_\eta^{-1}(t) - \psi_\eta^{-1}(s)\big\vert \leq c^{-1}\,\varepsilon
\end{align*}
provided $s,t\in [0,1]$ satisfy $\vert t-s\vert \leq \delta$. Therefore,
\begin{align*}
\met(\gamma_t,\gamma_s) = \met(\eta_{\psi_\eta^{-1}(t)},\eta_{\psi_\eta^{-1}(s)}) \leq c\,\big\vert\psi_\eta^{-1}(t) - \psi_{\eta}^{-1}(s)\big\vert \leq \varepsilon.
\end{align*}
By compactness of $G_r$, $\delta$ can be chosen independently of $\gamma$.
\end{proof}

\begin{corollary}\label{Cor:WEAK CONT r} The map $\sfr_\push\colon\scrP(\TGeo(\mms))\to\scrP(\TGeo^\tsep(\mms))$ is weakly continuous.
\end{corollary}

For the convenience of the reader, we record parts of \autoref{Def:Lorentzian geodesic}, defining
\begin{align*}
\OptTGeo_{\ell_p}(\mu_0,\mu_1) &:= \{\bdpi\in\OptGeo_{\ell_p}(\mu_0,\mu_1) : \bdpi[\TGeo(\mms)]=1\}\\
&\phantom{:}= \{\bdpi\in\OptGeo_{\ell_p}(\mu_0,\mu_1) : (\eval_0,\eval_1)_\push\bdpi[\mms_\ll^2]=1\},\\
\OptTGeo_{\ell_p}^\tsep(\mu_0,\mu_1) &:= \sfr_\push\OptTGeo_{\ell_p}(\mu_0,\mu_1).
\end{align*}

\begin{definition}\label{Def:Lorentzian geodesic III} We term a weakly continuous path $(\mu_t)_{t\in[0,1]}$ in $\scrP(\mms)$
\begin{enumerate}[label=\textnormal{\alph*.}]
\item \emph{timelike $\ell_p$-geodesic} if it is represented by some $\smash{\bdpi\in\OptTGeo_{\ell_p}(\mu_0,\mu_1)}$, and
\item \emph{timelike proper-time parametrized $\smash{\ell_p}$-geodesic} if it is represented by some $\smash{\bdpi\in\OptTGeo_{\ell_p}^\tsep(\mu_0,\mu_1)}$.
\end{enumerate}
\end{definition}

\begin{remark}\label{Re:Compat} In the smooth case, our notion of timelike proper-time parametrized $\smash{\ell_p}$-geodesics is consistent with the notion of an $\smash{\ell_p}$-geodesic from \cite[Def.~1.1]{mccann2020} as follows. Every timelike proper-time parametrized $\smash{\ell_p}$-geodesic is an $\smash{\ell_p}$-geodesic by \autoref{Sub:Geodesics}, and if $\smash{(\mu_0,\mu_1)\in\scrP_\comp(\mms)^2}$ with $\mu_0\in\scrP^\ac(\mms,\vol)$ is $p$-separated \cite[Def.~4.1]{mccann2020}, these notions agree by uniqueness of $\smash{\ell_p}$-geodesics with chronologically coupled endpoints \cite[Cor.~5.9]{mccann2020}.
\end{remark}

Recall the definition of the restriction map from \eqref{Eq:Restr def}.

\begin{proposition}\label{Le:Villani lemma for geodesic} Let $p\in (0,1]$ and suppose that $(\mms,\met,\ll,\leq,\tau)$ is a compact, causally closed, $\scrK$-globally hyperbolic, regular  Lorentzian geodesic space. Suppose that the pair $(\mu_0,\mu_1)\in\scrP(\mms)^2$ is timelike $p$-dualizable. Then the following  hold.
\begin{enumerate}[label=\textnormal{\textcolor{black}{(}\roman*\textcolor{black}{)}}]
\item\label{111111111111111111111} For every $\smash{\ell_p}$-optimal $\pi\in\Pi_\ll(\mu_0,\mu_1)$, there exists $\bdpi\in\OptTGeo_{\ell_p}^\tsep(\mu_0,\mu_1)$ such that $\smash{\pi = (\eval_0,\eval_1)_\push\bdpi}$.
\item There exists at least one timelike proper-time parametrized $\smash{\ell_p}$-geodesic from $\mu_0$ to $\mu_1$.
\item For every $\bdpi\in\smash{\OptTGeo_{\ell_p}^\tsep(\mu_0,\mu_1)}$, we have 
\begin{align*}
(\Restr_s^t)_\push\bdpi\in\OptTGeo_{\ell_p}^\tsep((\eval_s)_\push\bdpi,(\eval_t)_\push\bdpi)
\end{align*}
for every $s,t\in[0,1]$ with $s<t$.
\item If $\smash{\bdpi\in\OptTGeo_{\ell_p}^\tsep(\mu_0,\mu_1)}$ and if $\smash{\bdsigma}$ is any nontrivial measure on  $\smash{\Cont([0,1];\mms)}$ with $\bdsigma\leq\bdpi$, then $\bdsigma[\Cont([0,1];\mms)]^{-1}\,\bdsigma$ is an element of $\smash{\OptTGeo_{\ell_p}^\tsep(\sigma_0,\sigma_1)}$, where $\sigma_i := \bdsigma[\Cont([0,1];\mms)]^{-1}\,(\eval_i)_\push\bdsigma\in\scrP(\mms)$, $i\in\{0,1\}$.
\end{enumerate}

Furthermore, let $\smash{(\mu_0^n)_{n\in\N}}$ and $\smash{(\mu_1^n)_{n\in\N}}$ be sequences in $\scrP(\mms)$ which converge weakly to $\mu_0\in\scrP(\mms)$ and $\mu_1\in\scrP(\mms)$, respectively. If all $\smash{\ell_p}$-optimal couplings of $\mu_0$ and $\mu_1$ are chronological \textnormal{(}in particular, if $(\mu_0,\mu_1)$ is strongly timelike $p$-dualizable\textnormal{)},  every  sequence $\smash{(\bdpi^n)_{n\in\N}}$ with $\smash{\bdpi^n\in\OptTGeo_{\ell_p}^\tsep(\mu_0^n,\mu_1^n)}$ has an accu\-mulation point. 
Any such point belongs to $\smash{\OptTGeo_{\ell_p}^\tsep(\mu_0,\mu_1)}$ provided
\begin{align*}
\ell_p(\mu_0,\mu_1) \leq \liminf_{n\to\infty} \ell_p(\mu_0^n,\mu_1^n).
\end{align*}
\end{proposition}

\begin{proof} All statements follow by pushing forward the corresponding statements from \autoref{Pr:Max} as well as  \cite[Prop.~7.16, Thm.~7.30]{villani2009} from $\smash{\OptTGeo_{\ell_p}(\mu_0,\mu_1)} \subset \smash{\OptGeo_{\ell_p}(\mu_0,\mu_1)}$ to $\smash{\OptTGeo_{\ell_p}^\tsep(\mu_0,\mu_1)}$ by $\smash{\sfr_\sharp}$. Note that here, regularity is needed in the inherent  measurable selection argument for \ref{111111111111111111111} --- both for the existence of timelike geodesics joining timelike related points, and to see that the set of elements of $\smash{\TGeo^\tsep(\mms)}$ joining fixed points $x,y\in\mms$ with $x\ll y$ is uniformly closed (cf.~the proof of \autoref{Le:Borel}).

The only result requiring an additional comment is the last paragraph. Consider a sequence $\smash{(\bdalpha^n)_{n\in\N}}$ with $\smash{\bdalpha^n\in\OptTGeo_{\ell_p}(\mu_0^n,\mu_1^n)}$ and $\smash{\bdpi^n = \sfr_\push\bdalpha^n}$. This sequence admits an accumulation point $\smash{\bdalpha\in \OptGeo_{\ell_p}(\mu_0,\mu_1)}$ by \autoref{Pr:Max}. We claim $\smash{\ell_p}$-optimality of $\smash{(\eval_0,\eval_1)_\push\bdalpha}$. Since $\mms$ is compact, $\tsep$ is bounded from above, and thus \cite[Lem.~4.3]{villani2009} implies the inequality
\begin{align*}
\int_{\mms^2}\tsep^p\d(\eval_0,\eval_1)_\push\bdalpha \geq \limsup_{n\to\infty}\int_{\mms^2} \tsep^p \d(\eval_0,\eval_1)_\push\bdalpha^n.
\end{align*}
On the other hand, our assumption yields
\begin{align*}
\limsup_{n\to\infty}\int_{\mms^2}\tsep^p\d(\eval_0,\eval_1)_\push\bdalpha^n &\geq \liminf_{n\to\infty} \ell_p^p(\mu_0^n,\mu_1^n) \geq \ell_p^p(\mu_0,\mu_1),
\end{align*}
which gives the desired optimality. Consequently, we have $\bdalpha[\TGeo(\mms)]=1$ by our assumption on $\mu_0$ and $\mu_1$. Thus $\smash{\bdpi := \sfr_\push\bdalpha}$ is an accumulation point of $\smash{(\bdpi^n)_{n\in\N}}$ by \autoref{Cor:WEAK CONT r}. An analogous argument shows that all accumulation points of $\smash{(\bdpi^n)_{n\in\N}}$ belong to $\smash{\OptTGeo_{\ell_p}^\tsep(\mu_0,\mu_1)}$.
\end{proof}

\subsection{Proof of \autoref{Le:Mutually singular}}\label{Postponed} To prove \autoref{Le:Mutually singular}, we first need the following elementary observation about concatenation of elements of $\smash{\TGeo^\tsep(\mms)}$.

\begin{lemma}\label{Le:Concat} Let $t\in (0,1)$ and $\smash{\gamma^1,\gamma^2\in\TGeo^\tsep(\mms)}$ satisfy $\smash{\gamma_1^1 = \gamma_0^2}$, 
\begin{align}\label{Eq:t equ}
\begin{split}
\tsep(\gamma_0^1,\gamma_1^1) &= t\,\tsep(\gamma_0^1,\gamma_1^2),\\
\tsep(\gamma_0^2,\gamma_1^2) &= (1-t)\,\tsep(\gamma_0^1,\gamma_1^2).
\end{split}
\end{align}
Then the concatenation $\gamma \colon [0,1]\to \mms$ of $\gamma^1$ and $\gamma^2$ given by
\begin{align}\label{Eq:Concat}
\gamma_r := \begin{cases} \gamma_{r/t}^1 & \textnormal{if } r\leq t,\\
\gamma_{(r-t)/(1-t)}^2 & \textnormal{if }r>t
\end{cases}
\end{align}
belongs to $\smash{\TGeo^\tsep(\mms)}$ as well.
\end{lemma}

\begin{proof} Write $\smash{\gamma^1 = \sfr\circ \eta^1}$ and $\smash{\gamma^2 = \sfr\circ\eta^2}$ for certain $\smash{\eta^1,\eta^2\in\TGeo(\mms)}$. Define the concatenation $\smash{\eta\colon[0,1]\to\mms}$ of $\smash{\eta^1}$ and $\smash{\eta^2}$ analogously to \eqref{Eq:Concat}. By construction, $\smash{\eta}$ is $\met$-Lipschitz con\-tinuous and timelike. Moreover, \eqref{Eq:t equ} yields
\begin{align*}
\Len_\tsep(\eta) &= \Len_\tsep(\eta\big\vert_{[0,t]}) + \Len_\tsep(\eta\big\vert_{(t,1]})\\
&= \tsep(\eta_0^1,\eta_1^1) + \tsep(\eta_0^2,\eta_1^2)\\ 
&= \tsep(\eta_0^1,\eta_1^2)\\ 
&= \tsep(\eta_0,\eta_1),
\end{align*}
which yields $\eta\in\TGeo(\mms)$. Finally, we claim that $\gamma = \sfr\circ \eta$. To this aim, we first compute $\smash{\psi_\eta}$ and its inverse. Using \eqref{Eq:t equ}, for $r\leq t$ we get
\begin{align}\label{Eq:intm}
\psi_\eta(r) = \frac{\tsep(\eta_0,\eta_r)}{\tsep(\eta_0,\eta_1)} = \frac{\tsep(\eta_0^1,\eta_{r/t}^1)}{\tsep(\eta_0^1,\eta_1^2)} = t\,\psi_{\eta^1}(t^{-1}\,r),
\end{align}
while for $r > t$, in an analogous way we obtain
\begin{align*}
\psi_\eta(r) = \frac{\Len_\tsep(\eta\big\vert_{[0,t]})}{\tsep(\eta_0,\eta_1)} + \frac{\Len_\tsep(\eta\big\vert_{(t,r]})}{\tsep(\eta_0,\eta_1)} = t + (1-t)\,\psi_{\eta^2}((1-t)^{-1}\,(r-t)).
\end{align*}
Noting that $\psi_\eta([0,t]) = [0,t]$, substituting $r = \psi_\eta^{-1}(s)$, $s\leq t$, in \eqref{Eq:intm} yields
\begin{align*}
\psi_\eta^{-1}(s) = t\,\psi_{\eta^1}^{-1}(t^{-1}\,s),
\end{align*}
and consequently
\begin{align*}
\gamma_s = \gamma_{s/t}^1 = \eta^1_{\psi_{\eta^1}^{-1}(s/t)} = \eta^1_{\psi_\eta^{-1}(s)/t} = \eta_{\psi_\eta^{-1}(s)} = (\sfr\circ\eta)_s.
\end{align*}
An analogous argument gives $\gamma_s = (\sfr\circ\eta)_s$ for $s>t$; the claim follows.
\end{proof}

\begin{proof}[Proof of \autoref{Le:Mutually singular}] 
\textbf{Step 1.} \textit{Mixing construction.} By contradiction and up to relabeling the plans $\smash{\bdpi^1,\dots,\bdpi^n}$, we may and will assume that $\smash{(\eval_t)_\push\bdpi^1}\not\perp\smash{(\eval_t)_\push\bdpi^2}$. Define the measures $\smash{\bdpi^\Lef,\bdpi^\Rig \in\scrP(\TGeo^\tsep(\mms))}$ by
\begin{align*}
2\,\bdpi^\Lef &:= (\Restr_0^t)_\push\bdpi^1 + (\Restr_0^t)_\push\bdpi^2,\\
2\,\bdpi^\Rig &:= (\Restr_t^1)_\push\bdpi^1 + (\Restr_t^1)_\push\bdpi^2;
\end{align*}
by \autoref{Le:Villani lemma for geodesic}, these represent timelike proper-time parametrized $\smash{\ell_p}$-geodesics between their initial and final point marginals. Define $\smash{\mu\in\scrP_\comp^\ac(\mms,\meas)}$ by
\begin{align*}
2\,\mu := 2\,(\eval_1)_\push\bdpi^\Lef = 2\,(\eval_0)_\push\bdpi^\Rig = (\eval_t)_\push(\bdpi^1+\bdpi^2).
\end{align*}
Under a slight abuse of notation, let $\smash{\bdpi^\Lef_\cdot\colon\mms\to \scrP(\TGeo^\tsep(\mms))}$ and $\smash{\bdpi^\Rig_\cdot}\colon\mms \to \smash{\scrP(\TGeo^\tsep(\mms))}$ denote the disintegrations of $\smash{\bdpi^\Lef}$ and $\smash{\bdpi^\Rig}$ with respect to $\eval_1$ and $\eval_0$, respectively, given by the formulas
\begin{align*}
\rmd \bdpi^\Lef(\gamma) &= \rmd \bdpi^\Lef_x(\gamma)\d\mu(x),\\
\rmd \bdpi^\Rig(\gamma) &= \rmd\bdpi_x^\Rig(\gamma)\d\mu(x).
\end{align*}
Lastly, the map $\smash{\Sp\colon \Cont([0,1];\mms) \to \{(\gamma^1,\gamma^2)\in\Cont([0,1];\mms)^2 : \gamma_1^1 = \gamma_0^2\}}$ given by
\begin{align*}
\Sp(\gamma) := (\Restr_0^t(\gamma), \Restr_t^1(\gamma))
\end{align*}
is bi-Lipschitz continuous. Thus we define $\bdpi^\mix_\cdot\colon\mms\to \scrP(\Cont([0,1];\mms))$ as well as the ``mixed measure'' $\bdpi^\mix\in\scrP(\Cont([0,1];\mms))$ via
\begin{align*}
\bdpi^\mix_x &:= (\Sp^{-1})_\push(\bdpi_x^\Lef \otimes \bdpi_x^\Rig),\\
\d\bdpi^\mix(\gamma) &:= \bdpi_x^\mix(\gamma) \d\mu(x).
\end{align*}
The construction entails
\begin{align}\label{Eq:RESS}
\begin{split}
(\Restr_0^t)_\push\bdpi^\mix &= \bdpi^\Lef,\\
(\Restr_t^1)_\push\bdpi^\mix &= \bdpi^\Rig.
\end{split}
\end{align}

\textbf{Step 2.} \textit{Geodesy of $\smash{\bdpi^\mix}$.} We claim that $\smash{\bdpi^\mix\in \OptTGeo_{\ell_p}^\tsep(\mu_0,\mu_1)}$, where for $i\in\{0,1\}$ the corresponding marginals are $2\,\mu_i := (\eval_i)_\push(\bdpi^1 + \bdpi^2)\in\smash{\scrP_\comp^\ac(\mms,\meas)}$. By construction, $\smash{\bdpi^\mix}$ is concentrated on ``broken'' curves $\gamma\in\Cont([0,1];\mms)$ such that $\smash{\Restr_0^t(\gamma),\Restr_t^1(\gamma)}\in\TGeo^\tsep(\mms)$; in particular, every such $\gamma$ is timelike (but not necessarily Lipschitz, as usual). Part of the claim is thus to show that $\smash{\bdpi^\mix}$-a.e.~such $\gamma$ is in fact unbroken, i.e.~a genuine element of $\smash{\TGeo^\tsep(\mms)}$.

Since $\smash{\bdpi[\TGeo^\tsep(\mms)]=1}$ for the measure $\bdpi$ in \autoref{Le:Mutually singular} above, and since $\pi := (\eval_0,\eval_1)_\push\bdpi$ is a chronological optimal coupling of the compactly supported measures $(\eval_0)_\push\bdpi$ and $(\eval_1)_\push\bdpi$ by definition, by \cite[Prop. 2.8]{cavalletti2020} there exists $\Sigma\subset\TGeo^\tsep(\mms)$ with $\bdpi[\Sigma]=1$ such that $\smash{\Gamma := (\eval_0,\eval_1)(\Sigma)\subset\mms_\ll^2}$ is $l^p$-cyclically monotone. 

Define $\Sigma' := (\eval_0,\eval_1)^{-1}(\Gamma) \cap \TGeo^\tsep(\mms)\subset\smash{\TGeo^\tsep(\mms)}$. By assumption, $\smash{\bdpi^1}$ and $\smash{\bdpi^2}$ are absolutely continuous with respect to $\bdpi$, whence $\smash{\bdpi^1[\Sigma'] = \bdpi^2[\Sigma']=1}$. For every $\smash{\gamma^1,\gamma^2\in\Sigma'}$ with $\gamma_t^1 = \gamma_t^2$, we have $\smash{\gamma_0^1 \ll \gamma_1^2}$ and $\smash{\gamma_0^2 \ll \gamma_1^2}$ (by concatenation of appropriate preimages of $\smash{\gamma^1}$ and $\smash{\gamma^2}$ under $\sfr$). Hence, using $\smash{l^p}$-cyclical monotonicity and  \eqref{Eq:Reverse tau} we obtain
\begin{align*}
&\tsep(\gamma_0^1,\gamma_1^1)^p + \tsep(\gamma_0^2,\gamma_1^2)^p \geq l(\gamma_0^1,\gamma_1^2)^p + l(\gamma_0^2,\gamma_1^1)^p\\
&\qquad\qquad =\tsep(\gamma_0^1,\gamma_1^2)^p  + \tsep(\gamma_0^2,\gamma_1^1)^p\\
&\qquad\qquad\geq \big[\tsep(\gamma_0^1,\gamma_t^1) + \tsep(\gamma_t^2,\gamma_1^2)\big]^p + \big[\tsep(\gamma_0^2,\gamma_t^2) + \tsep(\gamma_t^1,\gamma_1^1)\big]^p\\
&\qquad\qquad = \big[t\,\tsep(\gamma_0^1,\gamma_1^1) + (1-t)\,\tsep(\gamma_0^2,\gamma_1^2)\big]^p\\
&\qquad\qquad\qquad\qquad + \big[t\,\tsep(\gamma_0^2,\gamma_1^2) + (1-t)\,\tsep(\gamma_0^1,\gamma_1^1)\big]^p\\
&\qquad\qquad \geq \tsep(\gamma_0^1,\gamma_1^1)^p + \tsep(\gamma_0^2,\gamma_1^2)^p.
\end{align*}
It follows that all inequalities are in fact equalities. In particular, by strict concavity of the function $\smash{r\mapsto r^p}$ on $[0,\infty)$ implicitly used in the last inequality we obtain $\smash{\tsep(\gamma_0^1,\gamma_1^1) = \tsep(\gamma_0^2,\gamma_1^2)}$. It is also easy to verify that $\mu[\eval_t(\Sigma')]=1$. These observations imply that for $\mu$-a.e.~$x\in\mms$ there is $L_x > 0$ such that $\smash{\bdpi_x^\Lef}$ and $\smash{\bdpi_x^\Rig}$ are concen\-trated on elements of $\smash{\TGeo^\tsep(\mms)}$ with $\tsep$-length $t\,L_x$ and $(1-t)\,L_x$, respectively. For such $x$ and $\smash{\bdpi_x^\mix}$-a.e.~$\gamma\in \Cont([0,1];\mms)$, we thus get
\begin{align*}
L_x &= \tsep(\gamma_0,\gamma_1)\\
&\geq \tsep(\gamma_0,\gamma_t) + \tsep(\gamma_t,\gamma_1) \\
&\geq \Len_\tsep(\gamma\big\vert_{[0,t]}) + \Len_\tsep(\gamma\big\vert_{(t,1]})\\
&= t\,L_x + (1-t)\, L_x \\
&= L_x.
\end{align*}
Again, all inequalities must be identities; in particular,
\begin{align*}
\tsep(\gamma_0,\gamma_t) &= t\,L_x = t\,\tsep(\gamma_0,\gamma_1),\\
\tsep(\gamma_t,\gamma_1) &= (1-t)\,L_x = (1-t)\,\tsep(\gamma_0,\gamma_1).
\end{align*}
By \autoref{Le:Concat}, $\smash{\bdpi^\mix}$ is thus concentrated on $\smash{\TGeo^\tsep(\mms)}$. By \eqref{Eq:RESS}, it thus follows that $\smash{\bdpi^\mix\in\OptTGeo_{\ell_p}^\tsep((\eval_0)_\push\bdpi^\mix,(\eval_1)_\push\bdpi^\mix)}$.

The $\smash{\ell_p}$-optimality of $\smash{(\eval_0,\eval_1)_\push\bdpi^\mix}$ follows from
\begin{align*}
&\ell_p((\eval_0)_\push\bdpi^\mix, (\eval_1)_\push\bdpi^\mix)^p  = \ell_p((\eval_0)_\push(\bdpi^1  + \bdpi^2)/2, (\eval_1)_\push(\bdpi^1 + \bdpi^2)/2)^p\\
&\qquad\qquad = \frac{1}{2}\int_{\TGeo^\tsep(\mms)} \tsep(\gamma_0,\gamma_1)^p \d(\bdpi^1+\bdpi^2)(\gamma)\\
&\qquad\qquad = \frac{1}{2}\int_\mms L_x^p \d(\eval_t)_\push[\bdpi^1+\bdpi^2](x) \\
&\qquad\qquad = \int_\mms L_x^p\d(\eval_t)_\push\bdpi^\mix(x)\\
&\qquad\qquad = \int_{\TGeo^\tsep(\mms)}\tsep^p(\gamma_0,\gamma_1)\d\bdpi^\mix(\gamma).
\end{align*}

\textbf{Step 3.} \textit{Conclusion.} Write $\smash{(\eval_t)_\push\bdpi^i = \rho_t^i\,\meas}$, $i\in\{1,2\}$. Since $\smash{\bdpi^1 \perp\bdpi^2}$, there exists a Borel set $E\subset\mms$ with $\meas[E]>0$ such that $\smash{\rho_t^1,\rho_t^2 > 0}$ on $E$, and for $\mu$-a.e.~$x\in E$, $\smash{\bdpi_x^\Lef}$ and $\smash{\bdpi_x^\Rig}$ is not a Dirac mass. It follows that $\smash{\bdpi^\mix}$ or its ``time reversal'' are not concentrated on a timelike nonbranching subset of $\smash{\TGeo^\tsep(\mms)}$, a contradiction to timelike $p$-essential nonbranching.
\end{proof}


\begin{thebibliography}{000}
		\bibitem{alexander2019}{\textsc{S. B. Alexander, M. Graf, M. Kunzinger, C. Sämann.} \textit{Generalized cones as Lorentzian length spaces: causality,
curvature, and singularity theorems.} Preprint, \texttt{arXiv:1909.09575}, 2019.}
		\bibitem{ambrosiogigli2013}{\textsc{L. Ambrosio, N. Gigli.} \textit{A user's guide to optimal transport.} Modelling and optimisation of flows on networks, 1--155,
Lecture Notes in Math., 2062, Fond. CIME/CIME Found. Subser., Springer, Heidelberg, 2013.}
		\bibitem{ambrosio2008}{\textsc{L. Ambrosio, N. Gigli, G. Savaré.} \textit{Gradient flows in metric spaces and in the space of probability measures.} Lectures in Mathematics ETH Zürich. Birkhäuser Verlag, Basel, 2005. viii+333 pp.}
		\bibitem{ambrosio2014a}{\bysame. \textit{Calculus and heat flow in metric measure spaces and applications to spaces with Ricci bounds from below.} Invent. Math. \textbf{195} (2014), no. 2, 289--391.}
		\bibitem{ambrosio2014b}{\bysame. \textit{Metric measure spaces with Riemannian Ricci curvature bounded from below.} Duke Math. J. \textbf{163} (2014), no. 7, 1405--1490.}
		\bibitem{ambrosio2015}{\bysame. \textit{Bakry-Émery curvature-dimension condition and Riemannian Ricci curvature bounds.} Ann. Probab. \textbf{43} (2015), no. 1, 339--404.}
		\bibitem{ambrosio2017}{\textsc{L. Ambrosio, S. Honda.} 
\textit{New stability results for sequences of metric measure spaces with uniform Ricci bounds from below.} Measure theory in non-smooth spaces, 1--51, Partial Differ. Equ. Meas. Theory, De Gruyter Open, Warsaw, 2017.}
		\bibitem{bacher2010}{\textsc{K. Bacher, K.-T. Sturm.} \textit{Localization and tensorization properties of the curvature-dimension condition for metric measure spaces.} 
J. Funct. Anal. \textbf{259} (2010), no. 1, 28--56.}
\bibitem{bardeen}{\textsc{J. M. Bardeen, B. Carter, S. W. Hawking.} \textit{The four laws of black hole mechanics.} Comm. Math. Phys. \textbf{31} (1973), 161--170.}
\bibitem{beem1985}{\textsc{J. K. Beem, P. E. Ehrlich, S. Markvorsen, G. J. Galloway.} \textit{Decomposition theorems
for Lorentzian manifolds with nonpositive curvature.} J. Differential Geom. \textbf{22} (1985), no. 1, 29--42.}
\bibitem{bekenstein}{\textsc{J. D. Bekenstein.} 
\textit{Black holes and entropy.} Phys. Rev. D (3) \textbf{7} (1973), 2333--2346.}
\bibitem{billingsley}{\textsc{P. Billingsley.} \textit{Convergence of probability measures.} Second edition. Wiley Series in Probability and Statistics: Probability and Statistics. A Wiley-Interscience Publication. John Wiley \& Sons, Inc., New York, 1999. x+277 pp.}
\bibitem{bogachev2007b}{\textsc{V. I. Bogachev.} \textit{Measure theory. II.} Springer-Verlag, Berlin, 2007,  xiv+575 pp.}
\bibitem{bombelli1987}{\textsc{L. Bombelli, J. Lee, D. Meyer, R. D. Sorkin.} \textit{Space-time as a causal set.} Phys. Rev. Lett. \textbf{59} (1987), no. 5, 521--524.}
\bibitem{braun2022}{\textsc{M. Braun.} \textit{Good geodesics satisfying the timelike curvature-dimension condition.} Nonlinear Anal. \textbf{229} (2023), Paper No. 113205, 30 pp.}
\bibitem{braunohta}{\textsc{M. Braun, S. Ohta.} \textit{Optimal transport and timelike lower Ricci curvature bounds on Finsler spacetimes.} Preprint, \texttt{arXiv:2305.04389}, 2023.}
\bibitem{brue2020}{\textsc{E. Brué, D. Semola.} \textit{Constancy of the dimension for $\RCD(K,N)$ spaces via regularity of Lagrangian flows.} Comm. Pure Appl. Math. \textbf{73} (2020), no. 6, 1141--1204.}
\bibitem{burago2001}{\textsc{D. Burago, Yu. Burago, S. Ivanov.} \textit{A course in metric geometry.} Graduate Studies in Mathematics, 33. American Mathematical Society, Providence, RI, 2001. xiv+415 pp.}
\bibitem{burtscher2014}{\textsc{A. Burtscher, P. G. LeFloch.} 
\textit{The formation of trapped surfaces in spherically-symmetric Einstein-Euler spacetimes with bounded variation.} J. Math. Pures Appl. (9) \textbf{102} (2014), no. 6, 1164--1217.}
\bibitem{burtscher2021}{\textsc{A. Burtscher, L. García-Heveling.} \textit{Time functions on Lorentzian length spaces.} Preprint, \texttt{arXiv:2108.02693}, 2021.}
\bibitem{cavallettimilman2021}{\textsc{F. Cavalletti, E. Milman.} \textit{The globalization theorem for the curvature-dimension condition.} Invent. Math. \textbf{226} (2021), no. 1, 1--137.}
\bibitem{cavalletti2017}{\textsc{F. Cavalletti, A. Mondino.} 
\textit{Optimal maps in essentially non-branching spaces.} 
Commun. Contemp. Math. \textbf{19} (2017), no. 6, 1750007, 27 pp.}
		\bibitem{cavalletti2020}{\bysame. \textit{Optimal transport in Lorentzian synthetic spaces, synthetic timelike Ricci curvature lower bounds and applications.} Preprint, \texttt{arXiv:2004.08934}, 2020.}
		\bibitem{cavalletti2022}{\bysame. \textit{A review of Lorentzian synthetic theory of timelike Ricci curvature bounds.} Gen. Relativity Gravitation \textbf{54} (2022), no. 11, Paper No. 137, 39 pp.}
		\bibitem{christodoulou2009}{\textsc{D. Christodoulou.} \textit{The formation of black holes in general relativity.}
EMS Monographs in Mathematics. European Mathematical Society (EMS), Zürich, 2009. x+589 pp.}
		\bibitem{grant2012}{\textsc{P. T. Chruściel, J. D. E. Grant.} \textit{On Lorentzian causality xwith continuous metrics.} Classical Quantum Gravity \textbf{29} (2012), no. 14, 145001, 32 pp.}
		\bibitem{cordero2001}{\textsc{D. Cordero-Erausquin, R. J. McCann, M. Schmuckenschläger.} \textit{A Riemannian interpolation inequality à la Borell, Brascamp and Lieb.} Invent. Math. \textbf{146} (2001), no. 2, 219--257.}
		\bibitem{dafermos2014}{\textsc{M. Dafermos.} \textit{The mathematical analysis of black holes in general relativity.} Proceedings of the International Congress of Mathematicians—Seoul 2014. Vol. III, 747--772, Kyung Moon Sa, Seoul, 2014.}
		\bibitem{deng}{\textsc{Q. Deng, K.-T. Sturm.} \textit{Localization and tensorization properties of the curvature-dimension condition for metric measure spaces, II.} 
		J. Funct. Anal. \textbf{260} (2011), no. 12, 3718--3725.}
		\bibitem{eckstein2017}{\textsc{M. Eckstein, T. Miller.} \textit{Causality for nonlocal phenomena.} Ann. Henri Poincaré \textbf{18} (2017), no. 9, 3049--3096.}
		\bibitem{erbar2015}{\textsc{M. Erbar, K. Kuwada, K.-T. Sturm.} \textit{On the equivalence of the entropic curvature-dimension condition and Bochner's inequality on metric measure spaces.} Invent. Math. \textbf{201} (2015), no. 3, 993--1071.}
		\bibitem{erbar2020}{\textsc{M. Erbar, K.-T. Sturm.} 
\textit{Rigidity of cones with bounded Ricci curvature.} 
J. Eur. Math. Soc. (JEMS) \textbf{23} (2021), no. 1, 219--235.}
\bibitem{eschenburg1988}{\textsc{J.-H. Eschenburg.} \textit{The splitting theorem for space-times with strong energy condition.}  J. Differential Geom. \textbf{27} (1988), no. 3, 477--491.}
\bibitem{finster2016}{\textsc{F. Finster.} \textit{The continuum limit of causal fermion systems.} From Planck scale structures to macroscopic physics. Fundamental Theories of Physics, 186. Springer, Cham, 2016. xi+548 pp.}
\bibitem{finster2018}{\bysame. \textit{Causal Fermion systems: a primer for Lorentzian geometers.} J. Phys. Conf. Ser. \textbf{968} (2018), 012004, 14 pp.}
\bibitem{galloway1989}{\textsc{G. J. Galloway.} \textit{The Lorentzian splitting theorem without the completeness assumption.} J. Differential Geom. \textbf{29} (1989), no. 2, 373--387.}
		\bibitem{gannon1975}{\textsc{D. Gannon.} \textit{Singularities in nonsimply connected space-times.} J. Mathematical Phys. \textbf{16} (1975), no. 12, 2364--2367.}
		\bibitem{gannon1976}{\bysame. \textit{On the topology of spacelike hypersurfaces, singularities, and black holes.} Gen. Relativity Gravitation \textbf{7} (1976), no. 2, 219--232.}
		\bibitem{garcia}{\textsc{L. García-Heveling, E. Soultanis.} \textit{Causal bubbles in globally hyperbolic spacetimes.} Gen. Relativity Gravitation \textbf{54} (2022), no. 12, Paper No. 155, 7 pp.}
		\bibitem{geroch1970}{\textsc{R. Geroch.} \textit{Domain of dependence.} J. Mathematical Phys. \textbf{11} (1970), 437--449.}
		\bibitem{geroch1987}{\textsc{R. Geroch, J. Traschen.} 
\textit{Strings and other distributional sources in general relativity.} Phys. Rev. D (3) \textbf{36} (1987), no. 4, 1017--1031.}
		\bibitem{gigli2013}{\textsc{N. Gigli.} \textit{The splitting theorem in non-smooth context.} Preprint, \texttt{arXiv:1302.5555}, 2013.}
		\bibitem{gigli2015}{\bysame. \textit{On the differential structure of metric measure spaces and applications.} Mem. Amer. Math. Soc. \textbf{236} (2015), no. 1113, vi+91 pp.}
		\bibitem{gigli2018}{\bysame. \textit{Nonsmooth differential geometry --- an approach tailored for spaces with Ricci curvature bounded from below.} Mem. Amer. Math. Soc. \textbf{251} (2018), no. 1196, v+161 pp.}
		\bibitem{giglimondino2015}{\textsc{N. Gigli, A. Mondino, G. Savaré.} \textit{Convergence of pointed non-compact metric measure spaces and stability of Ricci curvature bounds and heat flows.} Proc. Lond. Math. Soc. (3) \textbf{111} (2015), no. 5, 1071--1129.}
		\bibitem{graf2020}{\textsc{M. Graf.} \textit{Singularity theorems for $\smash{\rmC^1}$-Lorentzian metrics.} Comm. Math. Phys. \textbf{378} (2020), no. 2, 1417--1450.}
		\bibitem{grafling}{\textsc{M. Graf, E. Ling.}  \textit{Maximizers in Lipschitz spacetimes are either timelike or null.} Classical Quantum Gravity \textbf{35} (2018), no. 8, 087001, 6 pp.}
		\bibitem{griffiths2015}{\textsc{J. B. Griffiths, J.  Podolský.} \textit{Exact space-times in Einstein's general relativity.} Cambridge Monographs on Mathematical Physics. Cambridge University Press, Cambridge, 2009. xviii+525 pp.}
		\bibitem{hawking1966}{\textsc{S. W. Hawking.} \textit{The occurrence of singularities in cosmology. I.} Proc. Roy. Soc. London Ser. A \textbf{294} (1966), 511--521.}
		\bibitem{hawking1973}{\textsc{S. W. Hawking, G. F. R.  Ellis.} \textit{The large scale structure of space-time.} 
Cambridge Monographs on Mathematical Physics, No. 1. Cambridge University Press, London-New York, 1973. xi+391 pp.}
		\bibitem{hawking1970}{\textsc{S. W. Hawking, R. Penrose.} \textit{The singularities of gravitational collapse and cosmology.} Proc. Roy. Soc. London Ser. A \textbf{314} (1970), 529--548.}
		\bibitem{jacobson}{\textsc{T. Jacobson.} \textit{Thermodynamics of spacetime: the Einstein equation of state.} Phys. Rev. Lett. \textbf{75} (1995), no. 7, 1260--1263.}
		\bibitem{kellsuhr2020}{\textsc{M. Kell, S. Suhr.} \textit{On the existence of dual solutions for Lorentzian cost functions.} Ann. Inst. H. Poincaré C Anal. Non Linéaire \textbf{37} (2020), no. 2, 343--372.}
		\bibitem{ketterer2015a}{\textsc{C. Ketterer.} \textit{Cones over metric measure spaces and the maximal diameter theorem.} J. Math. Pures Appl. (9) \textbf{103} (2015), no. 5, 1228--1275.}
		\bibitem{ketterer2015b}{\bysame. \textit{Obata's rigidity theorem for metric measure spaces.} Anal. Geom. Metr. Spaces \textbf{3} (2015), no. 1, 278--295.}
		\bibitem{klainerman2015}{\textsc{S. Klainerman, I. Rodnianski, J. Szeftel.}\textit{ The bounded $L^2$ curvature conjecture.} Invent. Math. \textbf{202} (2015), no. 1, 91--216.}
				\bibitem{kunzinger2018}{\textsc{M. Kunzinger, C. Sämann.} \textit{Lorentzian length spaces.} Ann. Global Anal. Geom. \textbf{54} (2018), no. 3, 399--447.}
		\bibitem{kunzinger2015}{\textsc{M. Kunzinger, R. Steinbauer, M. Stojković, J. A. Vickers.} \textit{Hawking's singularity theorem for $\smash{\rmC^{1,1}}$-metrics.} Classical Quantum Gravity \textbf{32} (2015), no. 7, 075012, 19 pp.}
		\bibitem{kste}{\textsc{M. Kunzinger, R. Steinbauer.} \textit{Null distance and convergence of Lorentzian length spaces.} Ann. Henri Poincaré (2022). Online first, \texttt{https://doi.org/10.1007/s00023-022-01198-6}.}
		\bibitem{lee1976}{\textsc{C. W. Lee.} \textit{A restriction on the topology of Cauchy surfaces in general relativity.} Comm. Math. Phys. \textbf{51} (1976), no. 2, 157--162.}
		\bibitem{lichnerowicz1955}{\textsc{A. Lichnerowicz.} \textit{Théories relativistes de la gravitation et de l'électromagnétisme. Relativité générale et théories unitaires.} Masson et Cie, Paris, 1955. xii+298 pp.}
		\bibitem{lott2009}{\textsc{J. Lott, C. Villani.} \textit{Ricci curvature for metric-measure spaces via optimal transport.} Ann. of Math. (2) \textbf{169} (2009), no. 3, 903--991.}
		\bibitem{lu2021}{\textsc{Y. Lu, E. Minguzzi, S.-I. Ohta.} \textit{Geometry of weighted Lorentz-Finsler manifolds I: singularity theorems.} J. Lond. Math. Soc. (2) \textbf{104} (2021), no. 1, 362--393.}
		\bibitem{lu2021b}{\bysame. \textit{Geometry of weighted Lorentz-Finsler manifolds II: a splitting theorem.} Internat. J. Math. \textbf{34} (2023), no. 1, Paper No. 2350002, 29 pp.}
		\bibitem{mars1993}{\textsc{M. Mars, J. M. M. Senovilla.} \textit{Geometry of general hypersurfaces in spacetime: junction conditions.} Classical Quantum Gravity \textbf{10} (1993), no. 9, 1865--1897.}
		\bibitem{martin2006}{\textsc{K. Martin, P. Panangaden.} \textit{A domain of spacetime intervals in general relativity.} 
		Comm. Math. Phys. \textbf{267} (2006), no. 3, 563--586.}
		\bibitem{mccann2020}{\textsc{R. J. McCann.} \textit{Displacement convexity of Boltzmann's entropy characterizes the strong energy condition from general relativity.} Camb. J. Math. \textbf{8} (2020), no. 3, 609--681.}
		\bibitem{mccann2021}{\textsc{R. J. McCann, C. Sämann.} 
\textit{A Lorentzian analog for Hausdorff dimension and measure.} Pure Appl. Anal. \textbf{4} (2022), no. 2, 367--400.}
		\bibitem{min}{\textsc{E. Minguzzi.} \textit{Further observations on the definition of global hyperbolicity under low regularity.} Preprint, \texttt{arXiv:2302.09284}, 2023.}
		\bibitem{mondino2019}{\textsc{A. Mondino, A. Naber.} \textit{Structure theory of metric measure spaces with lower Ricci curvature bounds.} J. Eur. Math. Soc. (JEMS) \textbf{21} (2019), no. 6, 1809--1854.}
		\bibitem{mondinosuhr2018}{\textsc{A. Mondino, S. Suhr.}  \textit{An optimal transport formulation of the Einstein equations of general relativity.} J. Eur. Math. Soc. \textbf{25} (2023), no. 3, 933--994.}
		\bibitem{ohta2007}{\textsc{S.-I. Ohta.} \textit{On the measure contraction property of metric measure spaces.} 
Comment. Math. Helv. \textbf{82} (2007), no. 4, 805--828.}
\bibitem{ohta}{\bysame. \textit{Finsler interpolation inequalities.} Calc. Var. Partial Differential Equations \textbf{36} (2009), no. 2, 211--249.}
		\bibitem{oneill1983}{\textsc{B. O'Neill.} \textit{Semi-Riemannian geometry. With applications to relativity.} Pure and Applied Mathematics, 103. Academic Press, Inc. [Harcourt Brace Jovanovich, Publishers], New York, 1983. xiii+468 pp.}
		\bibitem{otto2000}{\textsc{F. Otto, C. Villani.} \textit{Generalization of an inequality by Talagrand and links with the logarithmic Sobolev inequality.} J. Funct. Anal. \textbf{173} (2000), no. 2, 361--400.}
		\bibitem{penrose1965}{\textsc{R. Penrose.} \textit{Gravitational collapse and space-time singularities.} Phys. Rev. Lett. \textbf{14} (1965), 57--59.}
		\bibitem{penrose1972}{\bysame. \textit{The geometry of impulsive gravitational waves.} General relativity (papers in honour of J. L. Synge), pp. 101--115. Clarendon Press, Oxford, 1972.}
		\bibitem{rajala2012a}{\textsc{T. Rajala.} \textit{Interpolated measures with bounded density in metric spaces satisfying the curvature-dimension conditions of Sturm.}  J. Funct. Anal. \textbf{263} (2012), no. 4, 896--924.}
		\bibitem{rajala2012b}{\bysame. \textit{Improved geodesics for the reduced curvature-dimension condition in branching metric spaces.} Discrete Contin. Dyn. Syst. \textbf{33} (2013), no. 7, 3043--3056.}
		\bibitem{rajala2014}{\textsc{T. Rajala, K.-T. Sturm.} 
\textit{Non-branching geodesics and optimal maps in strong $\mathrm{CD}(K,\infty)$-spaces.} Calc. Var. Partial Differential Equations \textbf{50} (2014), no. 3-4, 831--846.}
\bibitem{rendall2002}{\textsc{A. D. Rendall.} \textit{Theorems on existence and global dynamics for the Einstein equations.} Living Rev. Relativ. \textbf{5} (2002), 2002--6, 62 pp.}
\bibitem{smnn}{\textsc{C. Sämann.} \textit{Global hyperbolicity for spacetimes with continuous metrics.} Ann. Henri Poincaré \textbf{17} (2016), no. 6, 1429--1455.}
\bibitem{savare2014}{\textsc{G. Savaré.} \textit{Self-improvement of the Bakry-Émery condition and Wasserstein contraction of the heat flow in $\RCD(K,\infty)$ metric measure spaces.} Discrete Contin. Dyn. Syst. \textbf{34} (2014), no. 4, 1641--1661.}
\bibitem{schinnerl2021}{\textsc{B. Schinnerl, R. Steinbauer.} 
\textit{A note on the Gannon-Lee theorem.} Lett. Math. Phys. \textbf{111} (2021), no. 6, Paper No. 142, 17 pp.}
\bibitem{spanier1966}{\textsc{E. H. Spanier.} \textit{Algebraic topology.} McGraw-Hill Book Co., New York-Toronto, Ont.-London 1966 xiv+528 pp.}
\bibitem{sormani2016}{\textsc{C. Sormani, C. Vega.} \textit{Null distance on a spacetime.} Classical Quantum Gravity \textbf{33} (2016), no. 8, 085001, 29 pp.}
\bibitem{steinbauer}{\textsc{R. Steinbauer.} \textit{The singularity theorems of general relativity and their low regularity extensions.} Jahresber. Dtsch. Math. Ver. (2022), 47 pp.}
\bibitem{steinbauer2009}{\textsc{R. Steinbauer, J. A. Vickers.} \textit{On the Geroch-Traschen class of metrics.} 
Classical Quantum Gravity \textbf{26} (2009), no. 6, 065001, 19 pp.}
		\bibitem{sturm2006a}{\textsc{K.-T. Sturm.} \textit{On the geometry of metric measure spaces. I.}
Acta Math. \textbf{196} (2006), no. 1, 65--131.}
		\bibitem{sturm2006b}{\bysame. \textit{On the geometry of metric measure spaces. II.} Acta Math. \textbf{196} (2006), no. 1, 133--177.}
		\bibitem{suhr2018}{\textsc{S. Suhr.} \textit{Theory of optimal transport for Lorentzian cost functions.} Münster J. Math. \textbf{11} (2018), no. 1, 13--47.}
		\bibitem{suzuki2019}{\textsc{K. Suzuki.} \textit{Convergence of Brownian motions on metric measure spaces under Riemannian curvature-dimension conditions.} Electron. J. Probab. \textbf{24} (2019), Paper No. 102, 36 pp.}
		\bibitem{verlinde}{\textsc{E. Verlinde.} \textit{On the origin of gravity and the laws of Newton.} J. High Energy Phys. \textbf{2011} (2011), no. 4, 029, 27 pp.}
		\bibitem{vickers1990}{\textsc{J. A. Vickers.} \textit{Quasi-regular singularities and cosmic strings.} Classical Quantum Gravity \textbf{7} (1990), no. 5, 731--741.}
		\bibitem{vickers2000}{\textsc{J. A. Vickers, J. P.  Wilson.} \textit{Generalized hyperbolicity in conical spacetimes.} Classical Quantum Gravity \textbf{17} (2000), no. 6, 1333--1360.}
		\bibitem{villani2009}{\textsc{C. Villani.} \textit{Optimal transport. Old and new.} Grundlehren der mathematischen Wissenschaften, 338. Springer-Verlag, Berlin, 2009. xxii+973 pp.}
		\bibitem{vonrenesse2005}{\textsc{M.-K. von Renesse, K.-T. Sturm.} \textit{Transport inequalities, gradient estimates, entropy, and Ricci curvature.} Comm. Pure Appl. Math. \textbf{58} (2005), no. 7, 923--940.}
		\bibitem{woolgar2016}{\textsc{E. Woolgar, W. Wylie.} 
\textit{Cosmological singularity theorems and splitting theorems for $N$-Bakry-Émery spacetimes.} J. Math. Phys. \textbf{57} (2016), no. 2, 022504, 12 pp.}
	\end{thebibliography}
\end{document}